\documentclass[11pt]{article}
\usepackage[a4paper, width = 150mm, top = 30mm, bottom = 30mm]{geometry}
\usepackage{jheppubmod}
\usepackage[usenames,dvipsnames]{xcolor}
\usepackage{simpler-wick}
\usepackage{epsfig,amsfonts,amssymb,amscd,amsmath,mathrsfs,xspace,mathrsfs,amsthm,dsfont,bbm,bm}
\usepackage[small,bf,hang]{caption}
\usepackage{mathabx,amsmath}
\usepackage{slashed}
\usepackage{xfrac}    
\usepackage{comment}
\usepackage{xcolor}
\usepackage{empheq}  % Required for empheqbox

\usepackage{bbold}
\usepackage{lipsum,framed}
\usepackage{cancel}
\usepackage{enumitem}
\usepackage{tcolorbox}
\usepackage{mathrsfs} % mathscr
\usepackage{braket} % bras and kets for quantum mechanics, easier to use than \langle and \rangle
\usepackage{slashed} % slash notation
\usepackage[normalem]{ulem}
\usepackage{comment}
\usepackage{color}
\usepackage{psfrag}
\usepackage{array}
\usepackage{amssymb}
\usepackage{amsmath}
\usepackage{amsthm}
\usepackage{mathtools}
\usepackage[T1]{fontenc} % if needed
\usepackage{cancel}
\usepackage{cleveref}
\crefformat{section}{\S#2#1#3} % see manual of cleveref, section 8.2.1
\crefformat{subsection}{\S#2#1#3}
\crefformat{subsubsection}{\S#2#1#3}
\usepackage[labelsep=quad]{subcaption}
	
\usepackage{cite}
\usepackage{epsfig}
\usepackage{tikz}
\usetikzlibrary{arrows,plotmarks,positioning,calc}
\usepackage{comment}
\usepackage{IEEEtrantools}
\usepackage{graphicx}
\usepackage{physics}
\usepackage{float}
\usepackage{tcolorbox}
\tcbuselibrary{theorems}
\theoremstyle{definition}
\newtheorem{definition}{Definition}[section]

\newtheorem{theorem}{Theorem}[section]
\newtheorem{lemma}[theorem]{Lemma}

\usepackage{setspace}
\usepackage{setspace}
\usepackage{amsmath}
\usepackage{amssymb}
\usepackage{slashed}
\usepackage{xcolor}
\usepackage{cancel}
\usepackage{mathtools}
\usepackage{tikz}
\usetikzlibrary{positioning,decorations.pathmorphing,decorations.markings,arrows}

\usepackage{empheq}
\usepackage{hyperref}
\hypersetup{
pdfencoding=unicode,
colorlinks=true,
urlcolor=teal,
linkcolor=RoyalBlue,
citecolor=Maroon,
pdfdisplaydoctitle=true,
pdfstartview=FitH,
linktocpage=true
}
\usepackage{pdfpages}
\usepackage{vmargin}
\usepackage{comment}
\setmarginsrb     { 2.5cm}  % left margin
{ 2.5cm}  % top margin
{ 2.5cm}  % right margin
{ 2cm}  % bottom margin
{ 0cm}  % head height
{ 0cm}  % head sep
{ 0cm}  % foot height
{ 1.25cm}  % foot sep

\numberwithin{equation}{section}

\newcommand{\D}{\mathrm{d}}
\newcommand*{\medsubset}{\mathbin{\scalebox{1.1}{\ensuremath{\subset}}}}%
\newcommand*{\medcap}{\mathbin{\scalebox{1.3}{\ensuremath{\cap}}}}%
\newcommand*{\medcup}{\mathbin{\scalebox{1.3}{\ensuremath{\cup}}}}%

\title{Towards the Parametric Renormalization of the S-matrix - I} 

\author[S]{Pinaki Banerjee}\emailAdd{pinaki.banerjee@iitgn.ac.in}
\author[\Theta]{\!\!, Harsh}\emailAdd{harsh22@cmi.ac.in}
\author[\Theta]{\!\!, Alok Laddha}\emailAdd{aladdha@cmi.ac.in}

\affiliation[S]{Indian Institute of Technology Gandhinagar, Department of Physics, Gujarat 382355, India}

\affiliation[\Theta]{Chennai Mathematical Institute, Siruseri, SIPCOT, Chennai 603103}

\abstract{Zimmermann's forest formula is the corner stone of perturbative renormalization in QFT.  By renormalizing individual Feynman graphs, it generates the UV finite S-matrix. This approach to renormalization makes the graph and all its forests  center pieces in the theory of renormalization.  On the other hand  the positive geometry program delegate the role of Feynman graphs as secondary to the amplitude itself, which are generated by canonical forms associated to positive geometries.  State of the art in this program is the convergence of S-matrix theory in local QFTs and string theory as the scattering amplitudes in QFT arise as integrals over certain moduli spaces. These integrals are known as curve integrals. For theories such as $\textrm{Tr}(\Phi^{3})$ theory with massive colored scalars, these integrals are divergent in the UV and have to be regularized. It is then natural to ask if there is a ``forest-like formula'' for these integrals which produce a renormalized amplitude without needing to  explicitly invoke the forests associated to divergent subgraphs. In this paper, we initiate such a program by deriving forest-like formula for planar massive $\textrm{Tr}(\Phi^{3})$ amplitudes in $D = 4$ dimensions.  Our analysis relies on the insightful manifestation of the forest formula derived by Brown and Kreimer in  \cite{Brown:2011pj}, that lead us to a definition of ``tropical counter-term'' for the bare amplitude.}

\begin{document}

\setcounter{tocdepth}{2}
\maketitle\setcounter{page}{2}
\pagebreak

\makeatletter
\g@addto@macro\bfseries{\boldmath}
\makeatother

%%%%%%%%%%%%%%%%%%%%%%%%%%%%%%%%
%%%%%%%%%%%%%%%%%%%%%%%%%%%%%%%%
\section{Introduction} 
S-matrix in QFT comes in many avatars. These include, (a) perturbative series in terms of Feynman diagrams, (b)  on-shell functions that can be constructed recursively \cite{Britto:2005fq}, (c) integral of a top form over moduli space of punctured Riemann surface, as in the CHY (Cachazo, He, Yuan) formalism \cite{Cachazo:2013hca},  (d)  volume of certain convex polytope in momentum twistor space \cite{Hodges:2009hk, Arkani-Hamed:2010wgm, Arkani-Hamed:2013jha}, (e)  solution to the S-matrix bootstrap constraints and, (f)  in the case of AdS$_{d+1}$,  representation in terms of conformal correlator of $\textrm{CFT}_{d}$.  However  behind these appearances, all of which are derived from completely disparate starting points, lies a class of functions all of which have several very specific analyticity properties that reflect unitarity and locality of the S-matrix.  In fact, as shown in the remarkable paper \cite{Douglas:2022ynw}, renormalized scattering amplitudes of perturbative quantum field theory can be classified as tame functions defined with respect to so called O-minimal structures. 

Positive geometry program of the S-matrix is built on a premise that the specificity of an analytic S-matrix is due to the fact that behind the S-matrix resides a certain class of geometries which exist in the kinematic space of external data (momenta and spin of the particles). These geometries are a specific class of semi-algebraic domains in the kinematic space that admit a unique canonical form.  Canonical form of every positive geometry discovered so far in the kinematic space is either a tree-level amplitude or loop level integrand of \emph{some} local and unitary QFT! The geometries are known as positive geometries.

 In the case of ${\cal N} = 4$ super Yang Mills theory, it has now been established the geometry in question is the amplituhedron whose combinatorial and geometric properties define S-matrix of the theory. 
 
In the world of non-supersymmetric QFTs, positive geometries have already given us a radically new perspective on the analytic S-matrix since the seminal work of Arkani-Hamed, Bai, He and Yan (ABHY)  in 2017, \cite{Arkani-Hamed:2017mur}.  It is now becoming clearer that the  tree-level amplitudes and one loop integrands of a large class of scalar QFTs are canonical forms corresponding to an increasing catalogue of positive geometries that include amplituhedron \cite{Arkani-Hamed:2013jha}, the associahedron \cite{Stasheff1, Stasheff2, Arkani-Hamed:2017mur},  the accordiohedra \cite{manneville2017,Jagadale:2019byr,Raman:2019utu}, $\widehat{D}_{n}$ polytope \cite{afpst,Arkani-Hamed:2019vag,Jagadale:2020qfa}, and many others such as the recently discovered positive geometries known as cosmohdra, \cite{Arkani-Hamed:2024jbp}. For a review of the program, we refer the reader to \cite{Herrmann:2022nkh, Ferro:2020ygk} and for a mathematicians perspective which has been indispensable to many of the developments in this program,  \cite{Arkani-Hamed:2017tmz,  Bazier-Matte:2018rat, PPPP,Brown:2025jjg}. 

The class of positive geometries which generate scattering amplitudes of flat space QFTs are  defined via combinatorial structure associated to dissections of a disc with marked points, \cite{PPPP}. Some of the ``hard postulates'' of the analytic S-matrix program, namely \emph{locality} and \emph{unitarity} then emerge as corollaries due to facet structures of these geometries and their canonical forms. 

In the past five years in fact, we have come to realize a more striking universal aspect that underlying an entire class of tree-level S-matrices with scalar particles. That a single positive geometry, viz. the associahedron is a master polytope for a large class of tree-level amplitudes that range from $\textrm{Tr}(\Phi^{3})$ theory \cite{Arkani-Hamed:2017mur}, $\textrm{Tr}(\Phi^{p})$ theories \cite{Banerjee:2018tun, Raman:2019utu}, multi-scalar field theories \cite{Jagadale:2019byr, Jagadale:2021iab, Jagadale:2022rbl, Jagadale:2023hjr}, effective field theories \cite{Jagadale:2021iab}) and even non-linear sigma model (NLSM) \cite{Arkani-Hamed:2024nhp}. 
Although these developments clearly show that the the classification of functions in the kinematic space which correspond to local and unitary amplitudes is highly conceivable in this program,  there are several bottle-necks to overcome so that some of the most intricate questions posed in the analytic S-matrix program can be analyzed with a new perspective.  
For e.g. until recently, the program was most developed for tree-level amplitudes and one-loop integrands and, even in these cases,  it was not clear if positive geometries naturally defined tree-level amplitudes of  gauge theories such as Yang-Mills theory and effective field theories such as Non-linear Sigma model (NLSM). 

Over the past two years we have seen striking advances in answers to some of these questions for amplitudes dressed by color factors. In a series of seminal papers \cite{Arkani-Hamed:2023lbd}, \cite{Arkani-Hamed:2023mvg} and \cite{Arkani-Hamed:2024vna},  the authors showed that canonical forms defined by positive geometries can in fact be re-expressed as  integrals over a space of  \emph{``global Schwinger parameters''}. Unlike Schwinger parameters which are valued in ${\bf R}^{+}$, a global Schwinger parameter is valued in ${\bf R}$. In the case of tree-level amplitudes and one-loop integrals, these ``Schwinger-like''  integrals are entirely fixed by the very same data that realizes positive geometries as convex polytopes in the kinematic space. However the most striking aspect of this new perspective is that it has led us to a new definition of amplitudes for colored scalar theories at \emph{all order} in loop and topological `t hooft expansion, as opposed to the canonical form on the positive geometries which have hitherto been \emph{explicitly} analyzed only in tree-level and one-loop cases. (It should however be noted that we expect positive geometries to exist at all orders in the loop expansion thanks to the upcoming work by Arkani-Hamed, Frost, Plamondon, Salvatori and Thomas \cite{Surfaceohedra}).  This formulation of S-matrix is known as curve integral formalism. It has now been shown that the curve integral of $\textrm{Tr}(\Phi^{3})$ theory can be ``deformed'' to generate amplitudes of NLSM \cite{Arkani-Hamed:2024yvu} as well as Yang-Mills amplitudes as shown in \cite{Arkani-Hamed:2023swr}, \cite{Arkani-Hamed:2023jry}, \cite{Arkani-Hamed:2024tzl}, \cite{Backus:2025njt} and \cite{Laddha:2024qtn}.

The curve integrals are not derived from any principles that canonically express them in terms Feynman graphs, but they can always be written as sum over Feynman graphs where each graph labels the various cones that positive geometry generates naturally. In principle, the renormalization of these curve integral is already guaranteed by the Zimmermann's Forest formula.  That is, inside each cone, one can apply the forest formula \cite{Zimmermann:1968mu, Zimmermann:1969jj} as written by Appelquist \cite{Appelquist:1969iv}, Brown and Kreimer \cite{Brown:2011pj}, giving us a roadmap from combinatorics of dissection quivers to renormalized amplitudes in a large class of QFTs.  

However there is a catch. As such, the positive geometries that generate loop level integrands  delegate Feynman diagrams as secondary  derived objects, not indispensable for the existence of an S-matrix. In other words, if we are to hope that the existence of a (well-defined) S-matrix does not rely on existence of Feynman diagrams, then we  should be able to renormalize the curve integral formula without ever relying on its decomposition into Feynman graphs. 

Search for such a \emph{``renormalized positive geometry''} for the S-matrix in fact leads us to a more primitive question : Is there a way to  renormalize the UV divergences in the curve integrals in the global Schwinger space without taking direct input from the Zimmermann forest formula which depends explicitly on the Feynman graph, $\Gamma$ and all of it's forests? It is this question that we try to answer in this paper for the simplest case of color ordered amplitudes in $\textrm{Tr}(\Phi^{3})$ theory in $D = 4$ dimensions.

In order to appreciate this question better, let us pose it at one-loop, as already for $L = 1$,  the tension between forest formula for graphs and positive geometries becomes explicit. It is well known that the color ordered one-loop integrand of $\textrm{Tr}(\Phi^{3})$ theory is volume of the dual of a convex polytope. The polytope in this case turns out to be the convex realisation of the so called $\widehat{D}_{n}^{\star}$ polytope. Both, these polytope and it's realisation were discovered by Arkani-Hamed, Plamondon, Salvatori, He and Thomas  \cite{afpst,Arkani-Hamed:2019vag}, and a review of their construction can be found in \cite{Jagadale:2022rbl}.   The volume of the dual polytope $\widehat{D}_{n}^{\star}$ is a measure which is defined locally at every point in the kinematic space for one-loop integrands and hence is a function of the loop momentum $\ell^{\mu}$. The question we would like to ask is the following. Given, $\textrm{Vol}(\widehat{D}_{n}^{\star})$, is there a counter-term measure $\textrm{Vol}^{\textrm{c.t.}}(\widehat{D}_{n}^{\star})$ such that
\begin{align}
\int\, \frac{\mathrm{d}^{D} \ell}{(2\pi)^{D}}\, [\, \textrm{Vol}(\widehat{D}_{n}^{\star})\, -\, \textrm{Vol}(\widehat{D}_{n}^{\star})^{\textrm{c.t.}}\, ] = {\cal A}_{1,n}^{R}
\end{align}
where  ${\cal A}_{1,n}^{R}$  denotes the color ordered renormalized  $n$-point amplitude at one-loop.  We find $\int \frac{\mathrm{d}^{D} \ell}{(2\pi)^{D}} \textrm{Vol}(\widehat{D}_{n}^{\star})^{\textrm{c.t,}}$ in the ``curve integral form'' and then show  how to find the corresponding counter-terms for $ L > 1$. This paper is organized as follows. 

In section \ref{sec:forest-formula-review} we review the parametric representation of the Feynman integral $\Phi_{\Gamma}$ for a fixed one particle-irreducible (1PI) graph $\Gamma$ followed by a review of the parametric forest formula. We illustrate this formula for graphs in $D = 4$ massive $\Phi^{3}$ theory and review the simplest possible illustration of the remarkable \emph{ decomposed Feynman rules} which were derived by Brown and Kreimer in their seminal work in \cite{Brown:2011pj}.   In section \ref{sec:curve-int-formula}, we briefly review the curve integral formula discovered in \cite{Arkani-Hamed:2023lbd,Arkani-Hamed:2023mvg} with emphasis on the results relevant for subsequent sections. In particular, we review the derivation of the \emph{surface Symanzik polynomials} with respect to a specific reference dissection of the underlying surface. This dissection is dual to the so-called $n$-point $L$-loop, tadpole fat graph (for example see Fig.~\ref{fig:1l2p_fatgraph}).  The remaining sections are divided in terms of the results that we summarize below. 
\section*{Summary of the results}
\begin{enumerate}
\item  In section \ref{sffpa}, we derive, what we call, a  ``surface-forest formula'' which is the parametric version of Zimmermann's forest formula, but applied to the curve integrals. This formula offers us a proof of principle that curve integrals can be renormalized even though they sum over all the Feynman graphs which are not necessarily 1PI graphs. The surface forest formula essentially use the fact that given a planar fat graph $\Gamma$ with a set of forests $F(\Gamma)$,  there exists a set of graphs $\Gamma^{\prime}$ whose corresponding set of all forests is also $F(\Gamma)$. Combining this observation with the parametric forest formula for graphs, we obtain the surface forest formula, eqn.(\ref{sf_formula}). However as this formula implicitly depends on the 1PI divergent subgraphs and hence takes us away from the basic premise of the positive geometry program, we do not pursue it in greater detail. Reader interested in the primary results of this paper can skip this section in the first reading.

\item In sections \ref{dtci} and, \ref{sec:para_reno_1l_planar} we study divergence of planar one-loop amplitude which corresponds to sum over Feynman diagrams that do not include tadpole graphs. We do this by isolating ``chambers'' inside the global Schwinger space which is orthogonal to union over all the cones associated to tadpole graphs.  We call this region tadpole-free region. We then show how the curve integral over the tadpole free region has a clear separation between UV divergent domains and it's compliment. This leads us directly to a proposal for (what we call) a \emph{tropical counter term}, which as we show renormalizes one-loop planar amplitude.
\item In section \ref{sec:reno_beyond_1loop} we outline a detailed strategy as to how to renormalize curve integral beyond one-loop and implement it at two loops in section \ref{sec:2-loop-2-point-renormalization}. Although our analysis is for all $n$ at two loops, we give explicit formula for renormalized curve integral at $n=2$ in this case.
\item And finally, in appendix \ref{sec:pseudo triangulation}, we propose a pseudo-triangulation model for computing the $g$-vector fan associated to dissections of a disc with $n$ marked points and two punctures. We show how, this model leads us to the same fan as the one derived in \cite{Arkani-Hamed:2023lbd,Arkani-Hamed:2023mvg}. However our method is inspired by the pseudo-triangulation model for cluster polytopes \cite{Ceballos_2015}, and we believe it can lead us directly to convex realisation of the corresponding surfaceohedra along the lines of $\widehat{D}_{n}$ polytopal realisation at one-loop.
\end{enumerate}
In the conclusion section, we discuss our perspective on the renormalization of positive geometries and remaining appendices contain details of several other computations used in the main body of the paper. 

We end this section with a disclaimer. Our approach to parametric renormalization does not employ dimensional regularization. Renormalization of parametric Feynman integrals via dimensional regularization produces period integrals in the Lee-Pomeransky representation, \cite{Klausen:2019hrg}. It is plausible that, esp. for scalar amplitudes, a more efficient approach to parametric renormalization of the curve integral may indeed be based on dimensional regularization and associated tropical subtraction scheme. \cite{Borinsky:2020rqs}, \cite{Borinsky:2023jdv}, \cite{Hillman:2023ezp}. More recently, subtraction scheme for tropical Euler integrals (which generate log divergent amplitudes) have been devised in \cite{Salvatori:2024nva} which can provide an alternative route to parametric renormalization of curve integrals.  We also emphasize that the striking result proved in \cite{Douglas:2022ynw},  that the perturbative amplitudes are tame functions relies on the existence of renormalized Feynman graphs as period integrals. However, in this paper we follow the route adopted by Brown and Kreimer \cite{Brown:2011pj} and  use a momentum subtraction scheme as opposed to dimensional regularization. 

%%%%%%%%%%%%%%%%%%%%%%%%%%%%%%%%
%%%%%%%%%%%%%%%%%%%%%%%%%%%%%%%%
\section{Review of the Zimmermann's forest formula} \label{sec:forest-formula-review}
Given a  1PI graph $\Gamma_{L}$ with $L$ cycles, parametric representation of the (regularized) Feynman integral $\Phi_{\Gamma_{L}}(\epsilon)$ in $D$ dimensions is given in terms of a graph polynomials. In the theories without massless particles, the integral is convergent asymptotically and requires only the UV regularization at the \emph{origin} of the Schwinger parameter space ${\bf R}_{+}^{\vert E_{\textrm{int}}\vert}$, where $E_{\textrm{int}}$ is the set of the internal edges in $E(\Gamma_L)$. %\textcolor{red}{Can one of you complete this review?} 

%%%%%%%%%%%%%%%%%%%%%%%%%%%%%%%%
\subsection{Quick review of parametric representation}

In this subsection we very briefly review the parametric representation of Feynman graphs. The purpose is to setup the notations, for details we refer the readers to \cite{Nakanishi:1971, Bogner:2010kv}.

Consider a scalar Feynman graph $\Gamma$, in spacetime dimension $D$. The bare Feynman integral $\Phi_{\Gamma}$ is given by integration of the loop momenta ($\ell_a$) and the integrand contains products of Feynman propagators\footnote{Throughout the paper we will be using these planar variables $X_{ij}$, which are defined as follows, \begin{align*}
X_{ij} \displaystyle= \bigg(\sum_{a = i}^{j-1} p_a\bigg)^2 + m^2.
	\end{align*} } ($1/X_{i j}$). We can use the Schwinger parametrization to each propagator,
\begin{align}
\frac{1}{X_{ij}} = \int_{0}^{\infty} \D \alpha_{ij} \, e^{-\,  \alpha_{ij} X_{i j}},
\end{align}
where $\alpha_{ij}$ the schwinger parameters corresponding  the propagator $\frac{1}{X_{ij}}$. Since $X_{i j}$ are quadratic in momenta, we can readily perform the loop integrations. This leaves us with the left over integrals over all  the schwinger parameters $({\bm \alpha} := \{\alpha_{ij}\})$,
\begin{align}
\Phi_{\Gamma} &= \int_{0}^{\infty} \D {\bm \alpha} \, I_{\Gamma}\\
\text{with,\ } I_{\Gamma} &=  \frac{e^{ \frac{F_{\Gamma}}{U_{\Gamma}} \, - \,  \bm{m}\cdot \bm{\alpha}}}{U_{\Gamma}^{D/2}}.
\end{align}

In the above equation $U_\Gamma$   and $F_\Gamma$ are the first and second Symanzik polynomials respectively. They are defined as follows. 

\begin{itemize}
\item {\emph{First Symanzik polynomial} $U_{\Gamma}$: A connected and simply connected subgraph $T \medsubset \Gamma$ is called spanning tree of $\Gamma_L$, if $T$ and $\Gamma$ share the identical \emph{vertices} i.e. $T^{[0]} =  \Gamma^{[0]}$. Then the first Symanzik polynomial is defined as 
\begin{align}
U_{\Gamma} = \sum_{T} \prod_{e \, \notin \, T}  \alpha_e,
\end{align}
where the sum is over all spanning trees of $\Gamma$.
}
\item {\emph{Second Symanzik polynomial} $F_{\Gamma}$ : A spanning two-forest is a pair of connected and simply connected subgraphs $T_1, T_2 \medsubset \Gamma$ such that, 
\begin{align}
T_1 \medcap T_2 = \emptyset, \quad \text{and} \quad T_1^{[0]} \medcup  T_2^{[0]}  =  \Gamma^{[0]}.
\end{align}
Then the second Symanzik polynomial is defined as,
\begin{align}
F_{\Gamma} =  \sum_{ T_1 \medcup T_2 } P(T_1)  \cdot P(T_2) \,  \prod_{e \, \notin \, T_1^{[1]} \medcup  T_2^{[1]}}  \alpha_e,
\end{align}
where the sum is over all spanning two-forest  of $\Gamma$,  $P(T_i)$ is sum of all momenta flowing into the tree  $T_i$ and $T_i^{[1]}$ denotes the \emph{edges} of the subgraph $T_i$.
}
\end{itemize}
Note that below we use analogue of these polynomials that are associated with fat graphs and hence for corresponding surfaces. They are called \emph{surface Symanzik polynomials} and are denoted by $\mathcal{U}_\Gamma$ and $\mathcal{F}_\Gamma$ respectively.

%%%%%%%%%%%%%%%%%%%%%%%%%%%%%%%%
\subsection{Review of parametric forest formula}\label{sec:review_parametric_forest}
The Zimmermannn forest formula for Feynman integrals is one of the two fundamental pillars of renormalization theory.\footnote{The other being Wilsonian approach to renormalization.}  Since the seminal work by Appelquist \cite{Appelquist:1969iv}, there has been considerable effort in renormalizing parametric Feynman integral using the forest formula \cite{Brown:2011pj,Arkani-Hamed:2022cqe}. %\cite{brown-kreimer, mizera-nima, hillmann, toric}.
For a comprehensive modern review of the forest formula, we refer the reader to \cite{Brown:2015qmm}%\cite{brown2015}
.  For a comprehensive account of the forest formula to renormalize the UV behaviour as well as define infra-red safe quantities, we refer the reader to \cite{Borinsky:2020rqs, Hillman:2023ezp,Borinsky:2023jdv}. 

In this section, we review the forest formula in the context of renormalized Feynman graph for the $\phi^{3}$ theory in $D = 4$ dimensions. This is a rather trivial case study of the forest formula as the  one of the beautiful aspects of the  formula in fact resides in its ability to deal with overlapping divergences \cite{Kreimer:1998iv}, 
that  are absent in this theory.  We focus on this simplest  example of the forest formula as it will serve as  an inspiration for us in the subsequent sections when we renormalize the curve integrals.  

We review the Zimmermannn forest formula in the parametric representation  as derived by Brown and Kreimer in their seminal work \cite{Brown:2011pj}.   Among the many beautiful results derived in \cite{Brown:2011pj} (including the detailed analysis of quadratic divergences), perhaps the most striking result is the derivation of so called \emph{decomposed Feynman rules} which express the renormalized Feynman integral in terms of manifestly UV finite functions which either depend on a specific scale $S$ (which can be taken as \emph{e.g.} the total incoming energy in center of mass frame), or on so called \emph{angles}, $\frac{X_{ij}}{S}\, :=\, \Theta_{ij}$.  The derivation of the decomposed Feynman rules for generic graphs emerged thanks to the the equivalence between renormalized Feynman graphs and  characters of the Connes-Kreimer Hopf algebra \cite{Kreimer:1997dp,Connes:1998qv,Connes:1999yr}, a review of which is beyond the scope of this section. In the case of a super-renormalizable theory,  the derivation of these rules follows from some algebraic manipulations of the forest formula as we now review. 

Given any one-particle-irreducible (1PI) graph $\Gamma$,  a forest $f$ is a collection of 1PI divergent subgraphs such that no two elements in the collection have non-trivial intersection. In other words, if $f$ is a forest comprising of a set of 1PI divergent subgraphs $\{\gamma_{1}, \dots, \gamma_{\vert f\vert}\}$ then 
\begin{align}
\gamma_{i} \medsubset \gamma_{j}\ \textrm{or}\ \gamma_{j} \medsubset \gamma_{i}\ \textrm{or}\ \gamma_{i} \medcap \gamma_{j} = \emptyset.
\end{align} 
Zimmermann's forest formula expresses the renormalized Feynman integral $\Phi_{\Gamma}^R$ can be written as follows. Let $\Phi_\Gamma$ be the bare Feynman integral which is a function of $\{X_{ij}\, \vert 1 \leq\, i\, < j+1\, \leq\, n-1\}$.\footnote{We suppress the dependence on the mass parameters and will assume that the there is a unique single particle state with mass $m$.} 
Following Brown and Kreimer, we reparametrize kinematic variables as
\begin{align}
\{X_{ij}\}\, \rightarrow\, \{\, X_{13} = S,\, \frac{X_{ij}}{X_{13}}\, :=\, \Theta_{ij}\, , \,  \Theta_{m} = \frac{m^{2}}{S} \} .
\end{align}
And hence, 
\begin{align}
\Phi_{\Gamma}(X_{13}, X_{24}, \dots)\, =\, \Phi_{\Gamma}(S, \Theta_{24}, \dots\,  ,\Theta_{m}) \, .
\end{align}
The renormalized Feynman integral (which is the result of the forest formula)  is a function $\Phi_{\Gamma}^{R}(S, S_{0},  \{\Theta_{ij}\}\, \Theta_{m}, \Theta_{m^{0}})$ which is defined as follows. 

Given a graph $\Gamma$ with cubic scalar vertices, the renormalized Feynman graph in $D = 4$ dimensions is defined as, 
\begin{align}
\Phi_{\Gamma}^{R}(S, \{ \Theta\},\, \{\Theta_{0}\} )  = \sum_{f} (-1)^{\vert f\vert}\, \phi_{f}^{(0)}(S_{0}, \Theta_{0})\, \phi_{\sfrac{\Gamma}{f}}(S, \{ \Theta\}) 
\end{align}
where 
\begin{align}
\{\Theta\}\, :=\, \big\{\, \{\Theta_{ij}\}_{(i,j)\, \neq\, (1,3)},\, \Theta_{m}\, \big\}, 
\end{align}
and the sum is over all the forests.

We will now derive the decomposed Feynman rules for $\Phi_{R}(\Gamma_{L})$, where $\Gamma_{L}$ is two-point 1PI $L$-loop grap necklace graph obtained by decorating a bubble with $L-1$ bubbles on as shown in the figure \ref{fig:necklace-graph}. 

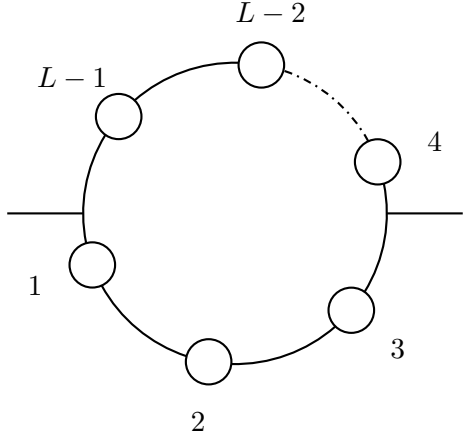
\begin{figure}
	\centering
	\begin{tikzpicture}
		\pgfmathsetmacro\R{2}
		\pgfmathsetmacro\s{0.3}
		\pgfmathsetmacro\delta{asin(\s / (1*\R))}
		
		% Draw the big circle as six arcs
		\foreach \angle in {80,140,200,260,320}{
			\draw[thick, black] ({\angle + \delta}:\R) arc ({\angle + \delta}:{\angle + 60 - \delta}:\R);
		}
		
		\foreach \angle in {20}{
			\draw[thick, dash dot, black] ({\angle + \delta}:\R) arc ({\angle + \delta}:{\angle + 60 - \delta}:\R);
		}
		
		% Draw the six small circles
		\foreach \angle in {20, 80, 140, 200, 260, 320} {
			\draw[thick, black] (\angle:\R) circle (\s cm);
		}
		
		% Add two small horizontal line segments outside the big circle
		\draw[thick, black] (-3.0,0) -- (-2.0,0); % Horizontal line on the left
		\draw[thick, black] (2.0,0) -- (3.0,0);  % Horizontal line on the right
		
		% Add nodes outside the small circles with labels
		
		\node at (140:\R + 0.8) {$L-1$};
		\node at (200:\R + 0.8) {$1$};
		\node at (260:\R + 0.8) {$2$};
		\node at (320:\R + 0.8) {$3$};
		\node at (20:\R + 0.8) {$4$};
		%\node at (25:\R + 0.8) {$.$};
		%\node at (30:\R + 0.8) {$.$};
		\node at (80:\R + 0.7) {$L-2$};
	\end{tikzpicture}
	\caption{$L$-loop necklace graph}
\label{fig:necklace-graph}
\end{figure}

\begin{lemma}
Let 
\begin{align}
\delta\Phi_{\gamma}(S, \Theta, S_{0}, \Theta_{0})\, :=\, \Phi_{\gamma}(S, \Theta)\, -\, \Phi_{\gamma}(S_{0}, \Theta_{0}) \, .
\end{align}
Then for any $L$-loop ``necklace'' graph $\Gamma_{L}$, the Zimmermann forest formula can be reexpressed in terms of manifestly UV-finite functions as follows. 
\begin{align}\label{decfeynrul}
\Phi^{R}_{\Gamma_{L}}(S, \Theta, S_{0}, \Theta_{0})\, =\, \sum_{k=0}^{L-1}\, ^{L}C_{k}\, \big(\delta\Phi_{\gamma}(S, S_{0}, \Theta, \Theta_{0})\big)^{k}\, I_{\sfrac{\Gamma_{L}}{\gamma^{k}}}(S, \Theta) \, .
\end{align}
where, $\gamma^{k}$ is the disjoint union of $k$ bubble graphs and $I_{\sfrac{\Gamma_{L}}{\gamma^{k}}}(S, \Theta)$ is the UV finite Feynman integral associated to the quotient graph $\sfrac{\Gamma_{L}}{\gamma^{k}}$, 
\begin{align}
I_{\sfrac{\Gamma_{L}}{\gamma^{k}}}(S, \Theta)\, =\, \Phi_{\Gamma_{L}}\, -\, \sum_{f = \gamma^{k}}\, (-1)^{\vert f\vert}\, \Phi_{f}\, \Phi_{\sfrac{\Gamma_{L}}{f}}.
\end{align}
\end{lemma}
\begin{proof}
We start with 
\begin{align}
\Phi_{\Gamma_{L}}&=\,  I_{\Gamma_{L}}\,  + (\, \Phi_{\Gamma_{L}}\, -\, I_{\Gamma_{L}}\, )\nonumber\\
&=\, I_{\Gamma_{L}}\, -\, \sum_{k=1}^{L-1}\, (-1)^{k}\, \Phi_{\gamma^{k}}(S, \Theta_{m})\, \Phi_{\sfrac{\Gamma}{\gamma^{k}}}(S, \Theta_{m}) \, .
\end{align}
Hence, the forest formula for the necklace graph in $D=4$ dimensions can be written as,
\begin{align}
\Phi_{\Gamma_{L}}^{R}\, =\, I_{\Gamma_{L}}\, -\, \sum_{k=1}^{L-1}\, (-1)^{k}\, (\Phi_{\gamma^{k}}(S, \Theta)\, -\, \Phi_{\gamma^{k}}(S_{0}, \Theta_{0})\, )\, \Phi_{\sfrac{\Gamma}{\gamma^{k}}}(S, \Theta) \, .
\end{align}
We can now sequentially substitute 
\begin{align}
\Phi_{\sfrac{\Gamma}{\gamma^{k}}}\, =\, I_{\sfrac{\Gamma}{\gamma^{k}}}\, -\, \sum_{k^{\prime}}\, (-1)^{k^{\prime}}\, ^{L-k}C_{k^{\prime}}\, (\Phi_{\gamma^{k^{\prime}}} - \Phi_{\gamma^{k^{\prime}}}^{0})\, \Phi_{\sfrac{\Gamma}{\gamma^{k^{\prime}}}},
\end{align}
and can immediately verify eqn. (\ref{decfeynrul}) as follows. For $k=1$, it is clear that only the substitution 
\begin{align}
\Phi_{\sfrac{\Gamma}{\gamma}} = I_{\sfrac{\Gamma_{L}}{\gamma}} + \dots , 
\end{align}
can contribute to the right hand side of eqn.(\ref{decfeynrul}). For $k=2$ we have two contributions.
\begin{align}
\Phi_{\Gamma_{L}}^{R}&=\, I_{\Gamma_{L}}\, +\, (L-1)\, \delta\Phi_{\gamma} I_{\sfrac{\Gamma_{L}}{\gamma}}\, +\, 
\big[\, ^{L-1}C_{1}\, ^{L-2}C_{1} \delta\Phi_{\gamma}\, \Phi_{\gamma}\, -\, ^{L-1}C_{2} \delta\Phi_{\gamma}\, (\, \phi_{\gamma} + \phi_{\gamma}^{0})\, \big]\, I_{\sfrac{\Gamma_{L}}{\gamma^{2}}}\, +\, \dots\nonumber\\
&=\,  I_{\Gamma_{L}}\, +\, ^{L-1}C_{1}\, \delta\Phi_{\gamma} I_{\sfrac{\Gamma_{L}}{\gamma}}\, +\, 
^{L-1}C_{2}\, \delta\Phi_{\gamma}^{2}\, I_{\sfrac{\Gamma_{L}}{\gamma^{2}}}\, +\, \dots \, . \nonumber
\end{align}
$k\, \geq\, 3$ terms in eqn.(\ref{decfeynrul}) can be evaluated similarly. We leave this verification to the interested reader. This completes the proof. 
\end{proof}
Some comments are in order. 
\begin{enumerate}
\item As shown in \cite{Brown:2011pj},  $\delta\Phi_{\gamma}$ can be written in terms of ``characters'' which are either scale of angle-dependent.
\begin{align}
\delta\Phi_{\gamma}\, = \Phi_{\textrm{fin}}(\Theta)\, +\, \Phi_{1-s}\bigg(\frac{S}{S_{0}}\bigg)\, - \Phi_{\textrm{fin}}(\Theta_{0}) \, ,
\end{align}
where $\gamma$ is the one-loop two-point graph with kinematic variables $S = p^{2}, \Theta = \frac{m^{2}}{S},\,  \Theta^{0} = \frac{m^{2}}{S_{0}}$, and 
\begin{align}
\Phi_{1-s} &=\, \int_{P^{1}_{+}}\, \D\mu(\vec{\tau})\, \frac{1}{{\cal U}_{\gamma}(\vec{\tau})^{\frac{D}{2}}}\, \log \bigg( \frac{S}{S_{0}} \bigg)\, ,\nonumber\\
\Phi_{\textrm{fin}}(\Theta)&=\, \int_{P^{1}_{+}}\, \D\mu(\vec{\tau})\, \frac{1}{{\cal U}_{\gamma}(\tau)^{\frac{D}{2}}}\, \log {\cal F}_{\gamma}(\vec{\tau}, \Theta) \, ,
\end{align}
where  $\tau_{1}, \tau_{2}\, \in\, [0,1]$ define the projective co-ordinates with $\sum_{i} \tau_{i}^{2} = 1$. ${\cal F}_{\gamma}(\vec{\tau}, \Theta)$ is obtained from the second Symanzik polynomial of $\gamma$ by dividing with an overall scale $S$ and mapping the Schwinger parameters $t_{1}, t_{2}$ to the projective space $P_{1}^{+}$ which is  co-ordinatized by $\tau_{1}, \tau_{2}$. The measure $d\mu(\vec{\tau})$ is defined as,
\begin{align}
d\mu(\vec{\tau})\, =\, \tau_{1} d\tau_{2} - \tau_{2} d \tau_{1} \, .
\end{align}
%\textcolor{magenta}{\it [PB: Little confused with the notation here.]}.

\item Each $I_{\sfrac{\Gamma}{\gamma^{k}}}$ is polynomial in $(S, \Theta)$ as opposed to a transcendental function and is independent of the renormalization point $S_{0}$. As a result the forest formula can be written in terms of forests $f = \gamma^{k}$ as, 
\begin{align}\label{decfeynrul1}
\Phi_{\Gamma_{L}}^{R}\, =\, \sum_{k}\, \widetilde{I}_{\sfrac{\Gamma_{L}}{\gamma^{k}}}\, \bigg[\, \Phi_{\gamma\, \textrm{fin}}(\Theta)\, + \Phi_{\gamma\, 1-s}\bigg(\frac{S}{S_{0}}\bigg)\, - \Phi_{\gamma, \textrm{fin}}(\Theta_{0})\, \bigg]^{k} \, ,
\end{align}
where $\widetilde{I} =\, ^{L}C_{k}\, I$.
For an overall convergent graph with logarithmic  sub-dievergences, eqn.(\ref{decfeynrul1}) is the simplest example of decomposed Feynman rules derived by Brown and Kreimer in \cite{Brown:2011pj}.\footnote{We are indebted to D. Kreimer for key clarifications regarding results in \cite{Brown:2011pj}.} 
\item This formula can be derived using the algebraic property of characters of the Connes-Kreimer Hopf algebra wherein the decomposed Feynman rules are derived in complete generality. However in the interest of the pedagogy, we do not touch upon the Hopf algebraic aspects of renormalization and instead refer the reader to a number of beautiful reviews such as \cite{Kreimer:2000ja, Kreimer-RG-notes, Ebrahimi-Fard:2005pga, Bergbauer:2005fb}. %\cite{kriemer1,kreimer2,fossey3,xxxxx}. 
\item For an $L$-loop $n$-point  amplitude generated by cubic graphs in $D=4$ dimensions, the renormalized on-shell amplitude is of the form\footnote{We assume that all the quadratically divergent tadpole graphs have been eliminated using off-shell counter-terms.} 
\begin{align}
{\cal A}^{R}_{L, n}(p_{1}, \dots, p_{n})\, =\, \sum_{\Gamma_{L}}\, \Phi^{R}_{\Gamma_{L}}(p_{1}, \dots, p_{n}, S, \{\Theta_{ij}, \Theta_{m}\}, S_{0}, \{\Theta_{ij}^{0},\, \Theta_{m}^{0}\}) . 
\end{align}
Our goal is to derive ${\cal A}^{R}$ by renormalizing the curve integral without its decomposition in terms of functions labeled by 1PI graphs. 
\end{enumerate}

%%%%%%%%%%%%%%%%%%%%%%%%%%%%%%%%
%%%%%%%%%%%%%%%%%%%%%%%%%%%%%%%%
\section{A review of curve integral formula}\label{sec:curve-int-formula}

In this section, we give a very brief account of the remarkable curve integral representation of the S-matrix discovered 
 in \cite{Arkani-Hamed:2023lbd, Arkani-Hamed:2023mvg,Arkani-Hamed:2024nhp}. We mostly follow the notations used in  \cite{Arkani-Hamed:2023lbd, Arkani-Hamed:2023mvg}. 
 
The curve integral formalism is the next stage in the evolution of the positive geometry program. It is built on the realization that there is a precise relationship between the  polytopal realization of the $g$-vector fan and  real-valued Plucker co-ordinates on the open string moduli space. The relationship is defined by realising that  for a $M$ dimensional moduli space, Plucker co-ordinates are Laurent polynomials in $M$ real variables which satisfty a system of non-linear equations known as Plucker relations. The low energy limit of the string amplitude leads to tropicalization of these polynomials and,  the resulting (piece-wise) linear functions are  precisely canonically dual to the set of all $g$-vectors. In other words, starting from the $g$-vector fan and finding its dual co-vectors lead us to a class of functions which are low energy limit of open string amplitudes! 

Below we provide lightning review of the formalism. Our review is not intended to be self-contained and our purpose is to only summarize the primary results of the original papers \cite{Arkani-Hamed:2023lbd, Arkani-Hamed:2023mvg} that are relevant for us directly.  We first summarize the derivation of the curve integral formula and then illustrate it with some examples. 

The ingredients that generate curve integrals corresponding to $n$ point, $L$ loop amplitudes at order $h$ in the 't hooft expansion are,  
\begin{itemize}
\item The  kinematic space of scattering data, a Riemann surface $\Sigma_{L,n,h}$ with $n$ marked points, $h$ boundaries and $L$ punctures, and a reference fat graph $\Gamma_{\textrm{ref}}(L,n,h)$.
\item $\Gamma_{\textrm{ref}}(L,n,h)$ assigns a unique momentum to each curve that can be drawn on $\Sigma_{L,n,h}$ in terms of the external momenta and a choice of loop momenta that is fixed once and for all.
\item  We will consider equivalence class $[C]$ of curves denoted as $C$ on $\Sigma_{L,n,h}$ which begin and end in the union of the set of marked points and the set of punctures, where the equivalence is defined by the homotopy. Without loss of generality, we will continue to always refer to the equivalence class as $C$ itself. It can be immediately seen that the set of all curves is then bijective with the set of all walks on $\Gamma_{\textrm{ref}}(L,n,h)$ whose end points are either on one of the external points of the graph or one or both of whose end points spiral around one of the punctures. 
\item A set of so-called $g$ vectors $\{\, g_{C}\, \}$ associated to set of all the curves that can be drawn on $\Sigma_{L,n,h}$ and the set of piece-wise linear one forms  $\alpha_{C}$ known as the headlight functions, that are dual to the set of all $g$ vectors in the sense that $\alpha_{C} (g_{C^{\prime}})\, =\, \delta_{C, C^{\prime}}$.
\item The previous two items can be used to write a parametric representation of ${\cal A}_{L,n,h}(p_{1}\, \dots,\, p_{n})$ which is in terms of an integral over ${\bf R}^{\vert E(\Gamma_{\textrm{ref}}(L, n, h))\vert}$ whose co-ordinates are known as global Schwinger parameters, \cite{Arkani-Hamed:2023lbd}.  Rather remarkably, the structural form of the integrand is same as the Symanzik representation of a Feynman integral and the two polynomials that define the curve integrand are known as first and second surface Symanzik polynomial\footnote{In \cite{Laddha:2024qtn}, it was shown that there is a third polynomial known as surface Corolla polynomial which can be used to ``spin up'' a color ordered scalar amplitude to color ordered gluon amplitude. Once again, this result parallels the parametric representation of color ordered gluon amplitude which involves the first and second Symanzik polynomials and the Corolla polynomial \cite{Kreimer:2012jw}.}.
\item We will only review the $h=0$ case as any fat graph that can be drawn on $\Sigma_{L,n,0} \, \equiv \Sigma_{L,n}$  generate planar amplitudes, which are the only amplitudes we renormalize in this paper. 
\end{itemize}
For the reference fat graph, we will choose the so-called tadpole fat graph such that it induces a unique (pseudo) triangulation $\Gamma^{\star}_{\textrm{ref}}(n,L)$ of  $\Sigma_{L,n}$ (see Fig.  \ref{fig:curve_C11_on_S_2,1} for an example).  It is however important to emphasize that the choice of reference graph is completely arbitrary and the final curve integral representation of the amplitude is independent this choice.\footnote{As we will see in the main sections of this paper, this specific choice of reference graph is perhaps the best choice to isolate UV divergences of the amplitude and hence best suited for a renormalization method which is inspired by the Forest formula.} We will denote the set of all chords that dissect $\Sigma_{(L,n)}$ as  ${\cal D}_{\Gamma_{\textrm{ref}}^{\star}(L,n)}$. This set is in bijection with the set of all the internal edges of $\Gamma_{\textrm{ref}}(L, n)$.\footnote{This simple fact is at the heart of positive geometry and emergence of curve integral formulae. We refer the interested reader to original references for details.} 

To the set of all chords in ${\cal D}_{\Gamma_{\textrm{ref}}^{\star}(L,n)}$,  in $\Gamma_{\textrm{ref}}^{\star}(L,n)$ we associate the cartesian basis of ${\bf R}^{\vert {\cal D}_{\Gamma^{\star}_{\textrm{ref}}}(L, n) \vert}$.  A specific rule for assigning values $(0, \pm 1)$ to  intersection of a curve $C$ in $\Gamma_{\textrm{ref}}(L,n)$ with chords in ${\cal D}_{\Gamma_{\textrm{ref}}^{\star}(L,n)}$ fixes the $g$-vector $g_{C}$. Span of $g_{C}$ over all the curves generate the $g$-vector fan \cite{PPPP}. 

$g$-vectors can be computed  using two algorithms. One of them is the so-called pseudo-triangulation approach which was pioneered by Ceballos and Pilaud in \cite{Ceballos_2015} and which was used to compute $g$-vectors for $\Sigma_{1,n}$ in  \cite{Jagadale:2022rbl}.  In  appendix \ref{sec:pseudo triangulation}, we propose a new pseudo-triangulation model to compute the set of all $g$-vectors for $\Sigma_{2,n}$ which generalizes the pseudo-triangulation model proposed in \cite{Ceballos_2015}.  We believe that this method can lead to a polytopal realisation whose normal fan is the fan spanned by the $g$-vectors. For $\widehat{D}_{n}$ polytope, this set of equations were reviewed in \cite{Jagadale:2022rbl} which reproduced the results of  \cite{afpst,Arkani-Hamed:2019vag}. It will be interesting to investigate that a polytopal realization obtained using $g$-vector fans in $\Sigma_{2,n}$ will produce the same convex realisation as the surfaceohedra \cite{Surfaceohedra}. We hope to return to this question elsewhere. 

The second method was given in \cite{Arkani-Hamed:2023lbd,Arkani-Hamed:2023mvg}, which we now review. This method is remarkably powerful in the sense that it produces the set of all $g$-vectors as well as set of all dual ``headlight functions''\footnote{Headlight functions are linear forms which are dual to $g$ vectors. The reason, determination of these headlight functions is a non-trivial problem is because of the fact that the linearity space of $g$-vectors is always less than their cardinality. One of the many beautiful results given in \cite{Arkani-Hamed:2023lbd} was to solve this highly non-trivial algebraic problem in full generality.} by first representing each curve $C$ as a word $W_{C}$ and then associate  (a) a set of signs to each word that are associated to components of $g_{C}$ along the ``axes'' labelled ${\cal D}_{\Gamma^{\star}_{\textrm{ref}}(L, n)}$, and (b) a $2\, \cross\, 2$ matrix representation for each letter in the word that produces the headlight function $\alpha_{C}$. We now summarize this method in a series of steps and explain each step via an example for $L=1$.  
\begin{itemize}
\item To each edge $e$ of $\Gamma_{\textrm{ref}}(L, n)$ we associate a variables $E_i$ (see figure~\ref{fig:1l2p_fatgraph}). We will refer to $t_i \equiv \log E_i , \  \forall\, i$ as \emph{global Schwinger parameters}. Note that $t_{i}\, \in\, {\bf R}$ as opposed to the usual Schwinger parameters which are in ${\bf R}^{+}$ when the external kinematics is in the Euclidean domain.\
\begin{figure}[h!]
\centering
\includegraphics[width=0.3\linewidth]{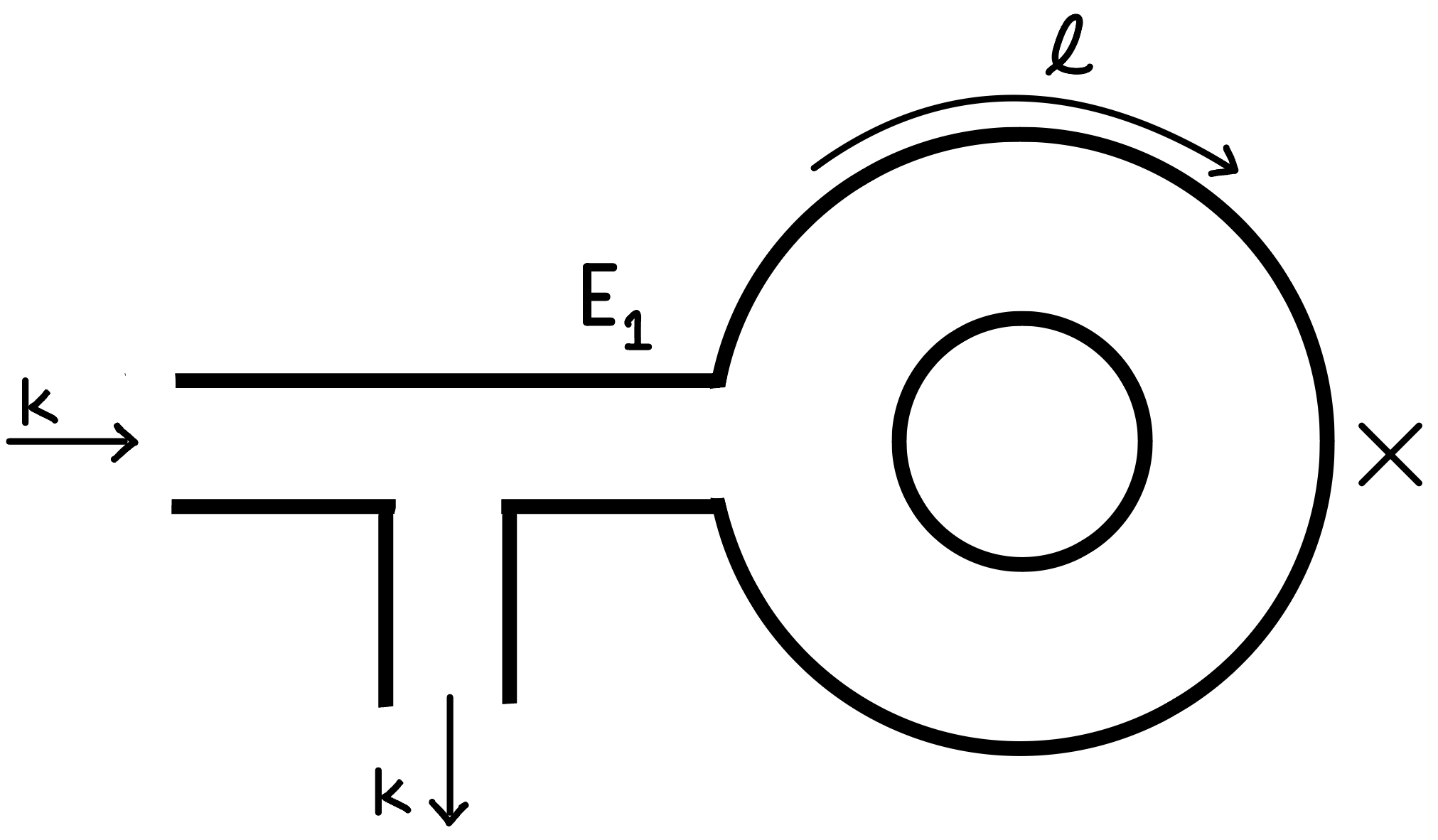}
\caption{One-loop two-point fatgraph}
\label{fig:1l2p_fatgraph}
\end{figure}
\item  Given a curve $C$ in $\Sigma_{L,n}$, there is a unique planar kinematic variable (that form a basis of the kinematic space spanned by external as well as loop momenta)  attached to it.  Furthermore, any curve $C$ on $\Sigma_{L,n}$ is represented by a unique word $W_{C}$ that is fixed by the pair $(\Sigma_{L,n},\, \Gamma_{\textrm{ref}})$. $W_{C}$ is determined by the sequence of the left and right turns $C$ takes to go from source of $C$ to the target of $C$ within $\Gamma_{\textrm{ref}}$, (see Fig. \ref{fig:curves_mountainscape}).

We can draw following curves inside $\Gamma_{\textrm{ref}}$. 
\begin{align}\label{setofcurvesn2l1}
\textrm{Set of curves in \,} \Gamma_{\textrm{ref}}(1,2)\, :=\, \{C^c_{10},C^c_{20},C^{cc}_{10},C^{cc}_{20},C_{11},C_{22}\} 
\end{align}
where $C_{i0}$ are spiral curves and  the super-script $c/cc$ denote the the orientation of the spiral, which can be clockwise or counter-clockwise respectively.  Fig. \ref{fig:curves_2p1l} shows an example of a set of curves, $\{C^c_{10},C^c_{20},C^{cc}_{10},C^{cc}_{20},C_{11},C_{22}\}$ which are drawn on the tadpole fat graph whose dual triangulates $\Sigma_{1,2}$. 
\begin{figure}[h!]
\centering
\includegraphics[width=0.6\linewidth]{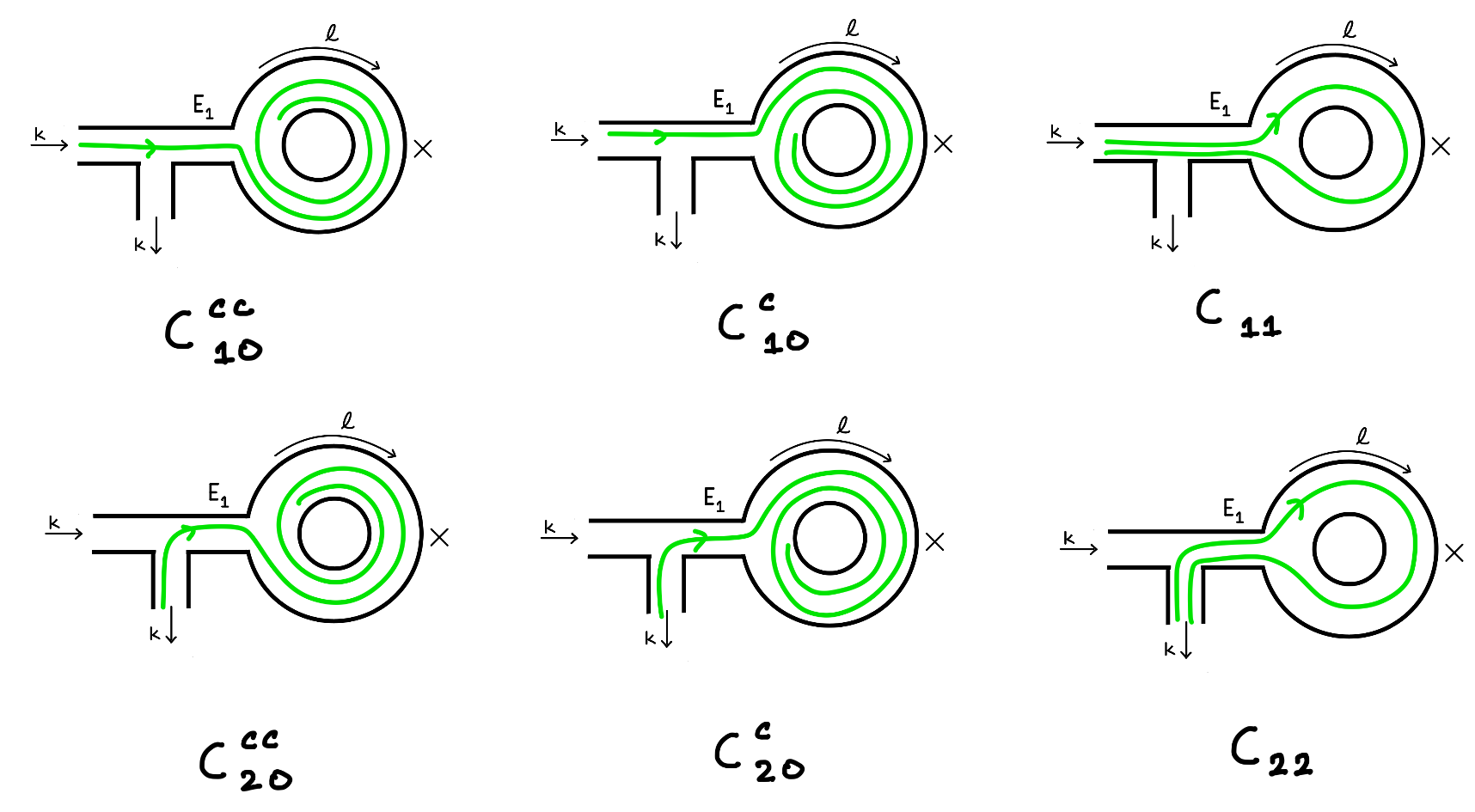}
\caption{One-loop two-point curves}
\label{fig:curves_2p1l}
\end{figure}
The corresponding words are then defined as, 
\begin{align}
&W_{C^{cc}_{10}} = 1\,L\,E_{1}\,R\,(XL)^{\infty}, \qquad W_{C^{cc}_{20}} = 2\,R\,E_{1}\,R\,(XL)^{\infty}, \qquad W_{C^{c}_{10}}  = 1\,L\,E_{1}\,L\,(XR)^{\infty}, \nonumber\\
&W_{C^{c}_{20}}  = 2\,R\,E_{1}\,L\,(XR)^{\infty},\qquad W_{C_{11}}      = 1\,L\,E_{1}\,R\,X\,R\,E_{1}\,R, \qquad W_{C_{22}}      = 2\,R\,E_{1}\,R\,X\,R\,E_{1}\,L\,2
\label{words_2p1l}
\end{align}
These words can also be represented as ``mountainscapes'' shown in Fig. \ref{fig:mountainscape_C11}.
\begin{figure}[h!]
	\centering
	\begin{subfigure}[t]{0.3\textwidth}
		\centering
		\includegraphics[width=0.7\linewidth]{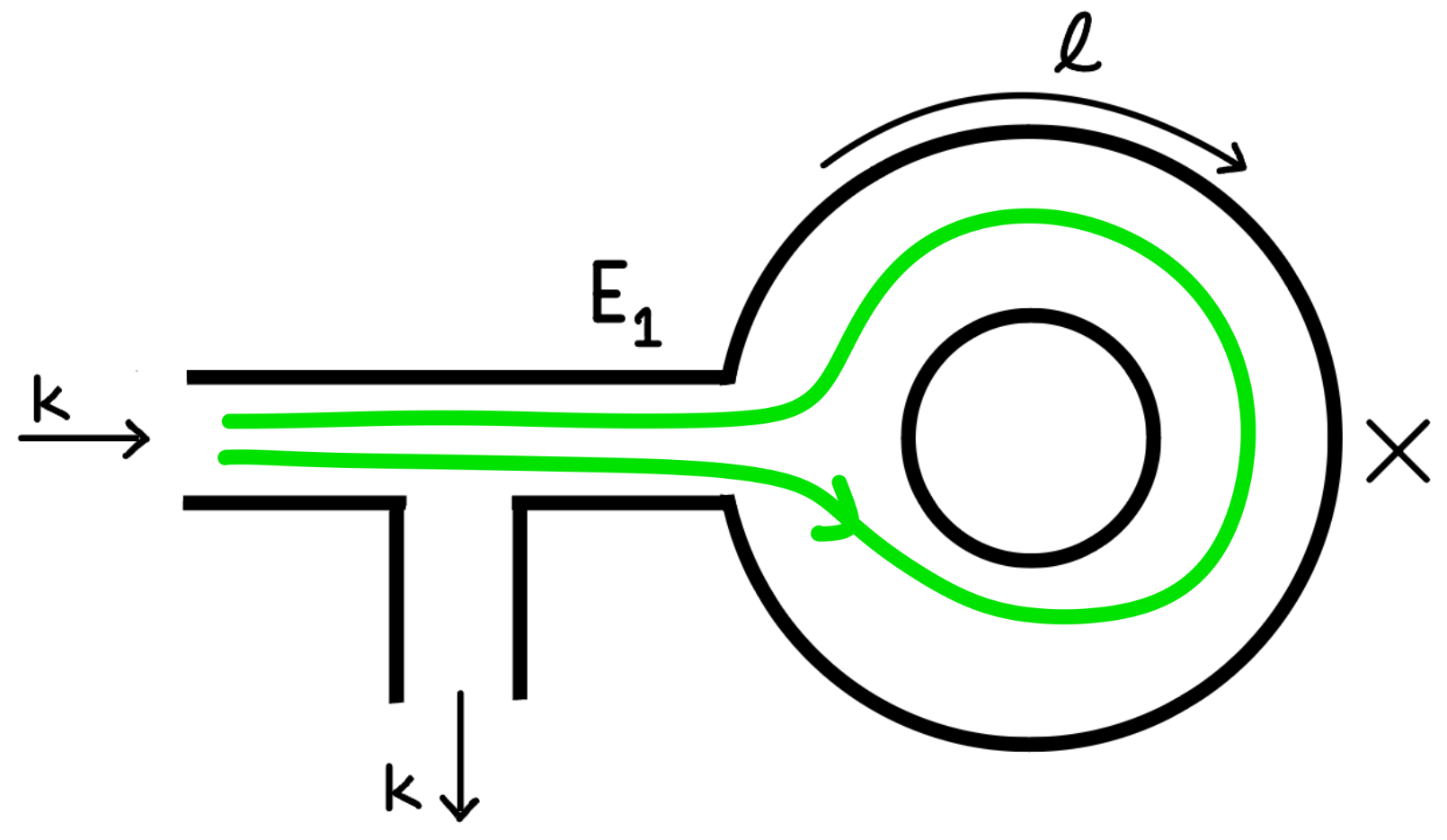}
		\caption{Curve $C_{11}$ in fatgraph}
		\label{fig:tadpole_curve_11}
	\end{subfigure}
	\hfill
	\begin{subfigure}[t]{0.3\textwidth}
		\centering
		\includegraphics[width=0.4\linewidth]{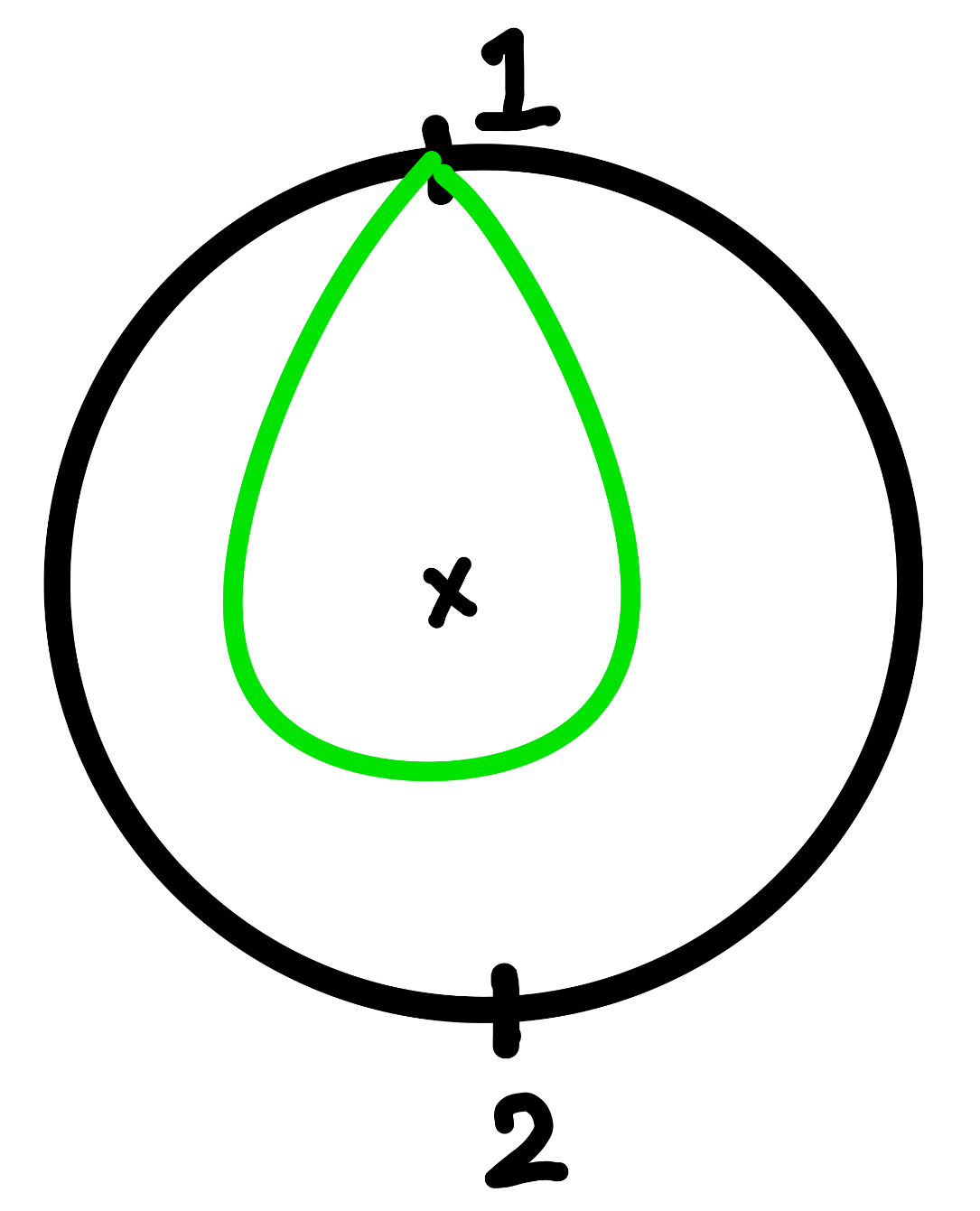}
		\caption{Curve $C_{11}$ on $\Sigma_{2,1}$}
		\label{fig:curve_C11_on_S_2,1}
	\end{subfigure}
	\hfill
	\begin{subfigure}[t]{0.3\textwidth}
		\centering
		\includegraphics[width=0.7\linewidth]{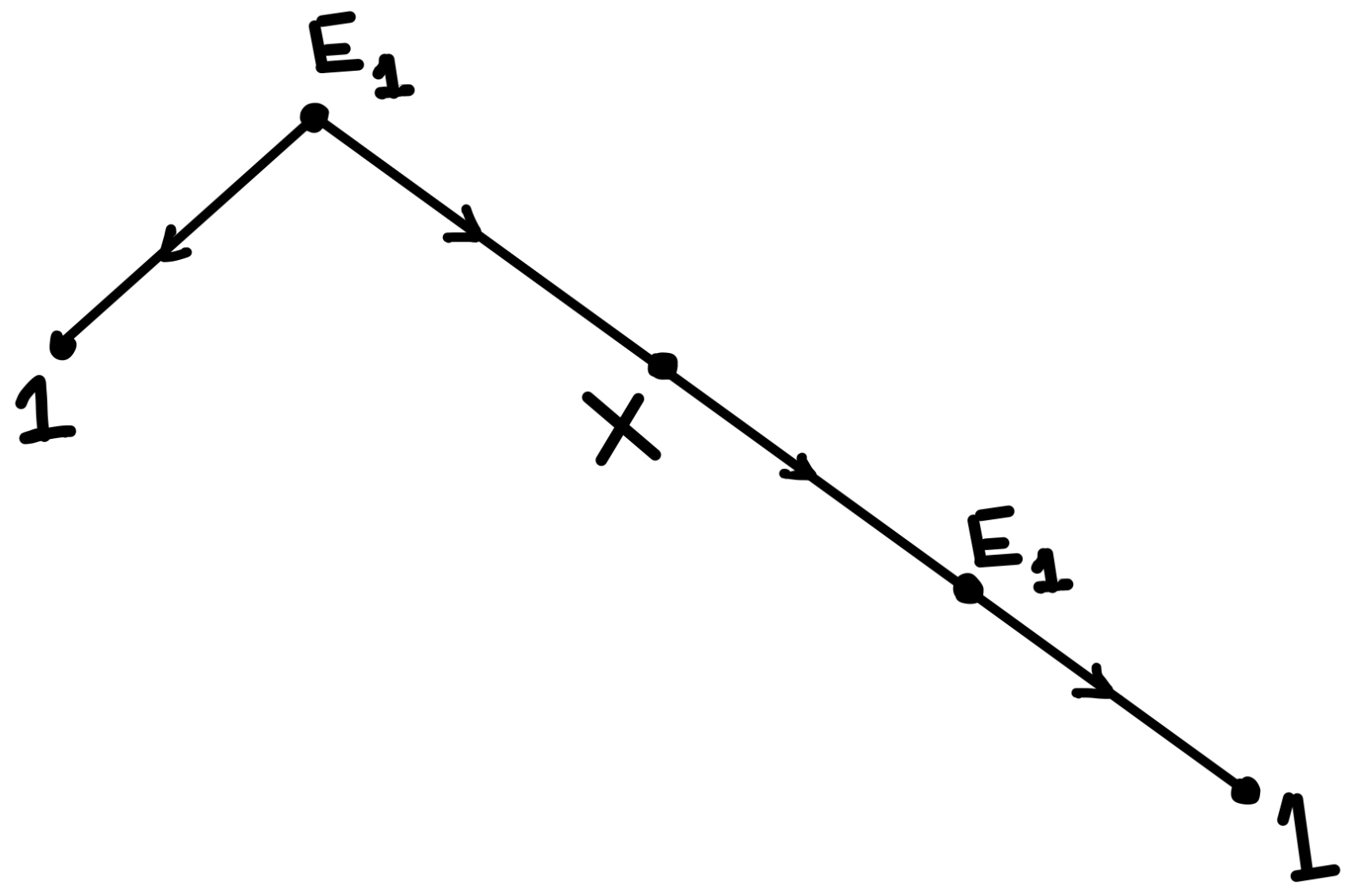}
		\caption{Mountainscape $C_{11}$ on $\Sigma_{2,1}$}
		\label{fig:mountainscape_C11}
	\end{subfigure}
	
	\caption{Curves $C_{11}$ on $\Sigma_{2,1}$ and its mountainscape.}
	\label{fig:curves_mountainscape}
\end{figure}

\item We can associate momentum $P_\mathcal{C}^\mu$
to a curve $\mathcal{C}$ via its word $W_\mathcal{C}$ by the following rule. 
\begin{align}
P_\mathcal{C}^\mu
= P_{\text{start}}^\mu
+ \sum_{\text{right turns}}
P_{\text {enterning from left}}^\mu,
\end{align}
In the above example, we have the following momentum assignments for the set of curves in eqn.(\ref{setofcurvesn2l1}). 
\begin{align}
&P^\mu_{C^{cc}_{10}} = k^\mu - l^\mu, \qquad P^\mu_{C^{cc}_{20}} = -\,l^\mu, \qquad P^\mu_{C^{c}_{10}}  = k^\mu -l^\mu, \\
&P^\mu_{C^{c}_{20}}  = -l^\mu,\qquad P^\mu_{C_{11}}      = 0,\qquad P^\mu_{C_{22}}      = 0 \, .
\end{align}

\item We can also extract the $g$-vectors associated to the curve $C$ from its word $W_{C}$ by following rule (see figure~\ref{fig:mountainscape_C11}). 
\begin{align}
\mathbf{g}_\mathcal{C}
= \sum_{\text{peaks p}} 	\mathbf{e}_p - \sum_{\text{valleys v}} 	\mathbf{e}_v \, .
\end{align}
The set of all the $g_{C}$ vectors spanning over all $C$s produce a complete fan as shown in e.g.  Fig. \ref{fig:g-vector_fan_2p1l}. Each top-dimensional cone inside the fan is spanned by $\vert {\cal D}_{\Gamma^{\star}_{\textrm{ref}}(L,n)} \vert$-tuple of curves that are precisely the propagators of corresponding Feynman diagram.

 As an example, g-vectors associated with each curve in the two-point one-loop fatgraph in basis $e_{E_1},\,e_x$ can be computed using the above formula and we find, 
\begin{align*}
&g_{C^{cc}_{10}} = \{1,-1\},\qquad g_{C^{cc}_{20}} = \{0,-1\},\qquad g_{C^{c}_{10}}  = \{0,1\}, \\
&g_{C^{c}_{20}}  = \{-1,1\},  \qquad g_{C_{11}}      = \{1,0\}, \qquad g_{C_{22}}      = \{-1,0\}
\end{align*}
In the figure below, we illustrate the corresponding $g$-vector fan whose linearity space is two dimensional and where we have labelled the axes by global Schwinger parameters.
\begin{figure}[h!]
\centering
\includegraphics[width=0.5\linewidth]{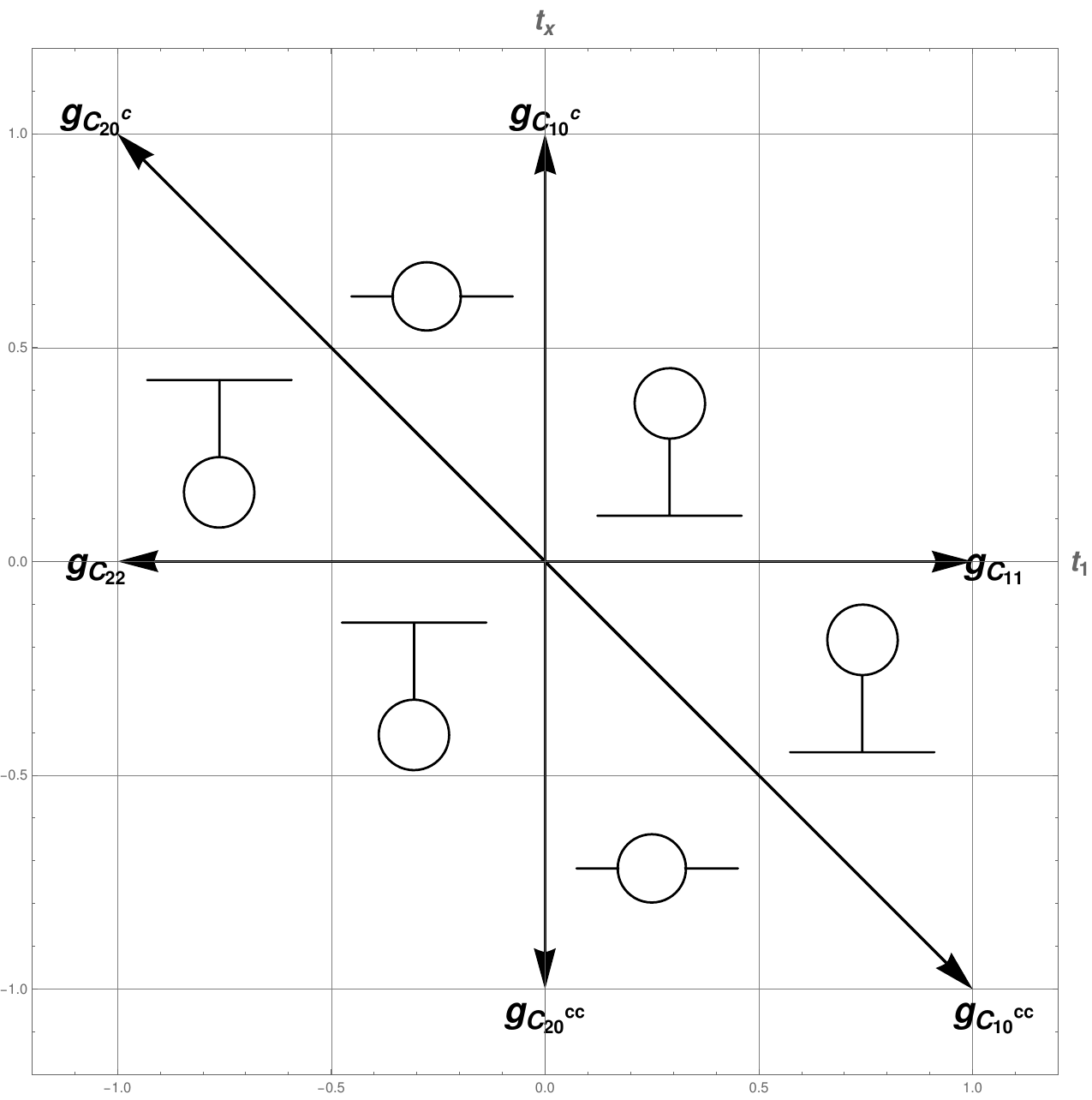}
\caption{Feynman Fan for $2$-point one-loop}
\label{fig:g-vector_fan_2p1l}
\end{figure}
In this case, there are precisely two curves for each momentum assignment and hence to avoid double counting we can restrict the domain to e.g. $t_x < 0$.  \footnote{For $\Sigma_{L > 1, n}$, redundancy in the choice of curves with same momentum assignment is removed by quotienting by mapping class group of $\Sigma_{L,n}$. As was shown in \cite{Arkani-Hamed:2023lbd}, this quotienting can be achieved by using the so-called tropical Mirzakhani kernel, which naturally selects a fundamental domain in the 
global Schwinger space.}
\item The headlight functions $\alpha_{C}$ naturally arise as tropicalization of so called \emph{surface variables} $u_{C}$. The surface variables are a solution to the system of non-linear equations.
\begin{align}\label{svpr}
u_{C} + \prod_{C^{\prime}\, \in\, \Sigma_{L,n}}\, (1)^{\sigma(C, C^{\prime})}\, u_{C^{\prime}} =  1, 
\end{align}
where $\sigma(C, C^{\prime})$ is the intersection number between $C,\, C^{\prime}$. 

In the case of $\Sigma_{0,n}$ $u_{C}$ are the well known real \emph{Plucker co-ordinates}. Clearly the solution to these system of equations $\forall\, (L,n)$ (or more generally $\forall\, (L,n, h)$) is a highly non-trivial problem. The solution to this problem came via advances in quiver representation theory and the theory of surface cluster algebras, \cite{Arkani-Hamed:2023lbd, Arkani-Hamed:2023mvg}. In a nut-shell the surface variables that satisfy eqn.(\ref{svpr})  can be obtained by following a simple algorithm, which expresses them as monomials in $E_{i}$ variables that are associated to internal edges of $\Gamma_{\textrm{ref}}(L, n)$. 

Given a curve $C$, we associate a matrix to the corresponding word $W_{C}$ by using the following assignments. As $C$ traveses a path inside the reference fat graph, we assign the following $2 \cross 2$ matrices at  each left and right turn.
\begin{align}
M_L(y_i)=\begin{bmatrix}
y_i & y_i \\
0 & 1 
\end{bmatrix}, \quad 
M_R(y_i)=\begin{bmatrix}
y_i & 0\\
1 & 1 
\end{bmatrix}
\end{align}
Note for the boundary edge we evaluate the respective
matrix at $1$ i.e.  $M_{L/R}(1)$ and for spiral curve we get an infinite product of these matrices. For a generic curve $C$, we can construct the matrix $M_C$ by multiplying the above matrices in its word $W_C$ as follows.
\begin{align}
M_{C} = M_{L/R} (0) . M_{L/R} (y_1). M_{L/R} (y_2) \ldots := \begin{bmatrix}
M_{C_{1,1}}& M_{C_{1,2}}\\
M_{C_{2,1}} & M_{C_{2,2}}
\end{bmatrix} 
\end{align}
\begin{align}
u_{C} = \frac{M_{C_{1,2}}\, M_{C_{2,1}}}{M_{C_{1,1}}\, M_{C_{2,2}}}.
\end{align}
For illustration, we compute the matrices $M_{C}$ for a set of curves in $\Sigma_{2,1}$. 
For spiral curves, the corresponding matrices involve an infinite product. 
\begin{align}
M_{C_{10}^{cc}} &= 
\begin{pmatrix}
1 + \dfrac{x}{1-x} & \epsilon \\[6pt]
1 + \dfrac{x(1+y_{1})}{1-x} & \epsilon(1+y_{1})
\end{pmatrix}, \quad
u_{C_{10}^{cc}} = \frac{1+xy_{1}}{1+y_{1}}, \\[10pt]
M_{C_{20}^{cc}} &=
\begin{pmatrix}
1 + \dfrac{x(1+y_{1})}{1-x} & \epsilon(1+y_{1}) \\[6pt]
\dfrac{x y_{1}}{1-x} & \epsilon y_{1}
\end{pmatrix}, \quad
u_{C_{20}^{cc}} = \frac{x(1+y_{1})}{1+xy_{1}}, \\[10pt]
M_{C_{11}} &=
\begin{pmatrix}
1 & 1+(1+x)y_{1} \\[6pt]
1 & (1+y_{1})(1+xy_{1})
\end{pmatrix}, \quad
u_{C_{11}} = \frac{1+(1+x)y_{1}}{(1+y_{1})(1+xy_{1})}, \\[10pt]
M_{C_{22}} &= 
\begin{pmatrix}
1+y_{1}(1+x(1+y_{1})) & y_{1}(1+x(1+y_{1})) \\[6pt]
x y_{1}^{2} & x y_{1}^{2}
\end{pmatrix}, \quad
u_{C_{22}} = \frac{y_{1}(1+x+xy_{1})}{(1+y_{1})(1+xy_{1})}.
\end{align}
\item The formula for the headlight function $\alpha_{C}$ is in terms of $u$ variables, 
\begin{align}
\alpha_C  &=  - \text{Trop}\,  u_C \nonumber\\
&= \text{Trop}\, M_{C_{1,1}}  + \text{Trop}\, M_{C_{2,2}}   - \text{Trop}\, M_{C_{1,2}}   - \text{Trop}\, M_{C_{2,1}} 
\end{align}
where for any monomial function $f(E_{1}, \dots, E_{K})$ of $K$ real variables, tropicalization of this monomial is defined implicitly via the limit, 
\begin{align}
e^{\textrm{Trop}(f)(t_{1}, \dots, t_{K})}\, :=\, \prod_{i=1}^{K}\, \lim_{t_{i}\, \rightarrow\, \infty} f(e^{t_{1}}, \dots, e^{t_{K}})\, .
\end{align}
After tropicalization of $u$-variables we obtain headlight function for the set of curves, 
$\{\, C_{10}, C_{20}, C_{11}, C_{22}\, \}$ as 
\begin{align}
&\alpha_{10}=\max(0,t_1)-\max(0,t_1+t_x)\\
&\alpha_{20}=-t_x-\max(0,t_1)+\max(0,t_1+t_x)\\
&\alpha_{11}=\max(0,t_1)+\max(0,t_1+t_x)-\max(0,t_1,t_x+t_1)\\
&\alpha_{22}=-t_1+\max(0,t_1)+\max(0,t_x+t_1)-\max(0,t_1,t_1+t_x)
\end{align}
$\epsilon$ is a parameter necessary to resum the  infinite series that arises due to words such as $(X L)^{\infty}$.\footnote{After obtaining the matrix representation of the word, 
any infinite sequence of letters, $x^{\infty}$ is first regularized as $\epsilon$ which is set fo zero after we obtain $u_{C}$.}
\item And finally, for $\Sigma_{L > 1, n, h=0}$ the measure in the global Schwinger space has to account for the action of the mapping class group (MCG). Hence we have to be able to rigorously define the ``gauge-fixed'' measure, 
\begin{align}
\mathrm{d}^E t  \to \frac{\mathrm{d}^E t }{\textrm{MCG}}\, .
\end{align}

\item The final result for $n$ point $L$-loop planar amplitude has the following striking form, 
\begin{align}
\mathcal{A}_{L, n} =  \int \frac{\mathrm{d}^E t }{\text{MCG}}\, \bigg[\frac{\pi^L}{\mathcal{U}(\alpha)}\bigg]^{D/2} \exp \bigg[\frac{\widebar{\mathcal{F}}(\alpha)}{\mathcal{U}(\alpha)}\bigg].
\end{align}
where, $\mathcal{U}(\alpha)$ and $\mathcal{F}(\alpha)$ called the first and second surface Symanzik polynomials. They are homogeneous polynomials in $\alpha$ and evidently are \emph{not} associated to any particular Feynman diagram. Where 
$\widebar{\mathcal{F}}(\alpha)=\mathcal{F}(\alpha)-\mathcal{U}(\alpha)\mathcal{Z}(\alpha)$, throughout the paper $\mathcal{F}(\alpha)=-\mathcal{F}_0(\alpha)$, where surface Symanzik polynomial ${\cal F}_0$ was first given in \cite{Arkani-Hamed:2023lbd}.
\item There is still one non-trivial step left viz. modding out by the mapping class group. Restricting to a fundamental domain in \textbf{$t$}-space is the obvious possibility, but it would amount to singling out Feynman diagrams and thus explicitly relies on the very structures which are secondary in this entire program. 

This can be done thanks to the \emph{tropical Mirzakhani kernel}, ${\cal K}$ \cite{Arkani-Hamed:2023lbd} which is morally akin to the Fadeev-Popov determinant  used in defining a ``measure'' on the space of gauge invariant orbits,
\begin{align}\label{eq:amp-parametric-form}
\mathcal{A}_{L, n} =  \int \mathrm{d}^E t\,  \mathcal{K}(\alpha) \, \bigg[\frac{\pi^L}{\mathcal{U}(\alpha)}\bigg]^{D/2} \exp \bigg[\frac{\mathcal{F}(\alpha)}{\mathcal{U}(\alpha)}\bigg].
\end{align}
$\mathcal{K}(\alpha)$  is a rational function of $\alpha$ and has support on a finite region of the fan. Therefore a few $\alpha_C$ will contribute to the integral. A rather beautiful fact about this kernel is that, as shown in \cite{Arkani-Hamed:2023mvg},  choice of one point tadpole fat graph as $\Gamma_{\textrm{ref}}(L,1)$ completely fixes it. As a result, ${\cal K}$ only depends on the global Schwinger parameters of a one point tadpole reference graph $\forall\, L$. Details can be found in \cite{Arkani-Hamed:2023mvg}.

\item  The formulae for the first and second \emph{surface Symanzik polynomials} were derived in \cite{Arkani-Hamed:2023lbd}. They have the same structural form as graph Symanzik polynomials except that graphs are replaced by surfaces that they triangulate! 
\begin{align}
\mathcal{F} 
&= - \sum_{S'\,\text{cuts } \Gamma \,\text{into two trees}} 
\left( \prod_{C \in S'} \alpha_{C} \right) 
\left( \sum_{C \in S'} K^{\mu}_{C} \right)^{2}, 
\label{def_F_by_cuts}\\[6pt]
\mathcal{U} 
&= \sum_{S\,\text{cuts } \Gamma \,\text{into a tree}} \,\,
\prod_{C \in S} \alpha_{C}.
\label{def_U_by_cuts}
\end{align}
\item The first surface Symanzik polynomial $\mathcal{U}$ takes a particularly simple form if we choose $\Gamma_{\textrm{ref}}(L, n)$ as the tadpole fat graph with $L$ loops. In the $L=1$ example, if the Schwinger parameter associated to the loop is $t_{x}$ then it can be shown that 
\begin{align}
{\cal U}(t_{1}, \dots, t_{n-1}, t_{x})\, =\, - t_{x}.
\end{align}
where $t_{1}, \dots, t_{n-1}$ are the global Schwinger parameters that correspond the $n-1$ internal edges of the reference graph. This remarkable result is due to invariance of ${\cal U}$ under attachment of any ``three point fat graph'' to the left of $\Gamma_{\textrm{ref}}(L,n)$ which in turn follows from its definition in eqn.(\ref{def_U_by_cuts}).\footnote{This result can be derived thanks to the so called telescopic property satisfies by headlight functions \cite{Arkani-Hamed:2023mvg}.}\\
We conclude this review with the formula for planar $n$-point one-loop amplitude in $\textrm{Tr}(\Phi^{3})$ theory. 
\begin{align}
{\cal A}_{n,L=1}=\int_{-\infty}^0 \D t_x \int_{\mathbb{R}^n}dt^n \left(\frac{1}{ -t_x}\right)^{\frac{D}{2}} \, e^{\frac{{\cal F}^{(n,1)}}{-t_x}-{\cal Z}^{(n,1)}}
\end{align} 
where 
\begin{align}
{\cal F}^{(n,1)} \,&=\, 
\sum_{i<j} \alpha_{i0}\, \alpha_{j0}\, z_{i} \cdot z_{j} \label{eq:F}\\[6pt]
{\cal Z}^{(n,1)} \,&=\, 
\sum_{\forall C_{ij}} \alpha_{C_{ij}}\, X_{C_{ij}}
+ \sum_{i}\alpha_{i0}\, m^{2}, \label{eq:Zij_1l}
\end{align}
where
\begin{align}
z_{i}^{\mu} \,&=\, \sum_{k=1}^{i-1} p_{k}^{\mu} \, . \label{eq:zm}
\end{align}
\end{itemize}
%%%%%%%%%%%%%%%%%%%%%%%%%%%%%%%%
%%%%%%%%%%%%%%%%%%%%%%%%%%%%%%%%
\section{A ``surface-forest formula'' : renormalization of planar amplitudes}\label{sffpa}

In this section, we argue that the defining properties of surface Symanzik polynomials lead us to an obvious generalization of forest formula which renormalizes the curve integrals. We analyze planar $n$-point amplitudes at $L$-loops for $\Tr(\Phi)^{3}$ theory in four dimensions. 

The underlying idea behind this generalization (of the forst formula) is rather simple : Any forest, namely a forest which is a disjoint union of 1PI divergent graphs such that each graph has $n \leq\, 2$ external points, can be realised as a (possibly disconnected) subgraph of several Feynman graphs which are related to each other by a sequence of mutations that keep the forest fixed. It is then rather natural to extend the forest formula to the curve integrands which simultaneously renormalize all the fat graphs which share the forest under consideration. In other words, the most naive extension of the forest formula to positive geometries  is to first \emph{``sum over all the Feynman graphs with a fixed common forest''} followed by sum over forests. We will refer to this formula as \emph{surface-forest formula}. 

However, as it will be clear from the preceding discussion that the surface-forest formula is rather naive in that it splits the curve integral into regions, all of which have common 1PI divergent subgraphs and then sums over all such contributions and hence this approach to renormalization does not utilize the full power of ``surfaceology'' \cite{Arkani-Hamed:2023lbd,Arkani-Hamed:2023mvg} due to its explicit reliance on 1PI divergent (fat) subgraphs inside a punctured surface.\footnote{We are grateful to Nima Arkani-Hamed for insisting why, from the perspective of positive geometry that we should not be satisfied with such a formula.}   Reader may wonder, why then do we discuss this formula. Our primary motivation in doing so is to provide a proof of principle. That is, the purpose of this section is to show that the curve integrals can be renormalized without breaking them explicitly into sum over Feynman graphs, by utilizing the fact that many Feynman graphs have common 1PI divergent subgraphs. We do however insist that this section can be skipped in the first reading as the subsequent sections which form the core of the paper do not explicitly rely on the surface-forest formula. 

We now introduce the concept of  a surface-forest. Let $\Sigma_{L, n}$ be a marked surface with $n$ external points ${\cal V}\, =\, \{v_{1}, \dots,\, v_{n}\, \}$ and $L$ punctures that form a set, ${\cal P} :=\, \{ v_{n+1},\, \dots,\, v_{n+L}\}$. Let ${\cal V}\, \medcup\, {\cal P}\, :=\, {\cal VP}$. 

\begin{definition}\label{def:surface_forest}
	A surface forest (s-forest) is a  partial dissection of $\Sigma_{L,n}$ which is generated by finite union of one punctured 2-gon, $\sigma_{1,2}$ and one-punctured 1-gon $\sigma_{1,1}$ that satisfy the following conditions. 
	\begin{enumerate}
		\item  Every $\sigma_{1,2}$ has two boundaries,  that connect $( v_{i},\, v_{j})\, \in\,  {\cal VP}\, \times\, {\cal VP}$ on either side of a puncture $v_{I} \in {\cal P}$.  We denote such a 2-gon as $\sigma(i,I, j)$. 
		\item  Every $\sigma_{1,1}$ has a single boundary that connects $v_{i}\, \in\, {\cal VP}$ to itself and which bounds a puncture $v_{J}\, \in\, {\cal P}$.  We denote such a 1-gon as $\sigma(i,J,i)$. 
		\item $\sigma(i,I,j)$  belongs to a surface forest if and only if there does not exist any $\sigma_{1,1}\, \subset\, \sigma(i,I,j)$. 
		\item Two such $(k\leq\, 2)$-gons  $\sigma(i_{1}, I_{1}, j_{1}),\, \sigma(i_{2}, I_{2}, j_{2})$ belong to a surface forest if and only if, 
		\begin{align}
			\sigma(i_{1}, I_{1}, j_{1})\, \medcap\, \sigma(i_{2}, I_{2}, j_{2})\, = 0\, \textrm{or}\,  \textrm{a singleton set in ${\cal PV}$.}
		\end{align}
	\end{enumerate} 
\end{definition}
See Fig.  (\ref{fig:p_l_sF1}--\ref{fig:tadpole_sf}) for some examples.  \\2

\begin{figure}[h!]
	\centering
	% Top row
	\subcaptionbox{A 2-gon with both edges $i,\, i+1 \in \mathcal{V}$ and $n+I \in \mathcal{P}$. \label{fig:p_l_sF1}}
	[0.3\textwidth]{\includegraphics[width=\linewidth]{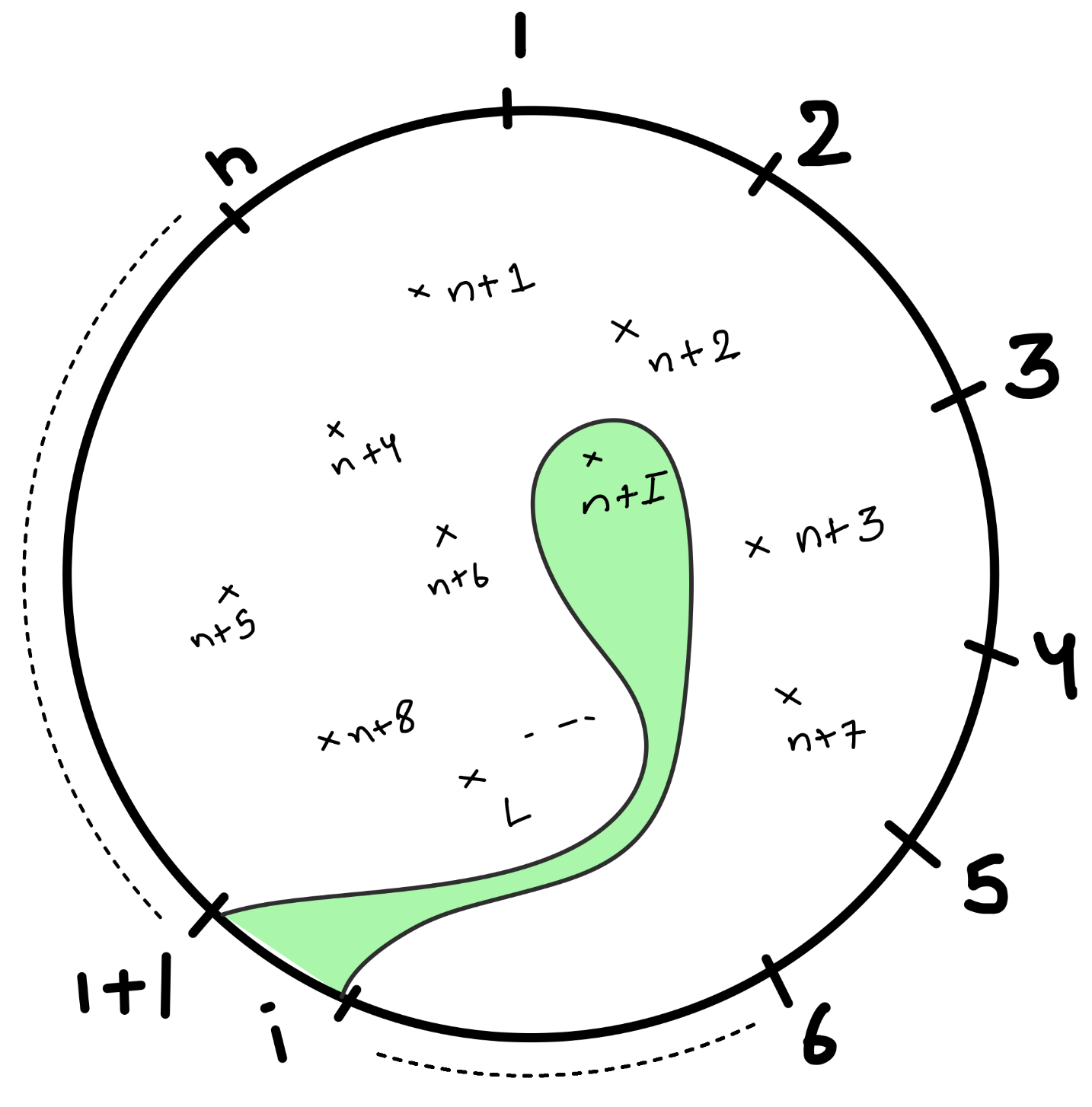}}
	\hfill
	\subcaptionbox{A 2-gon with one edge $i+1 \in \mathcal{V}$, the other $n+1 \in \mathcal{P}$, and $n+I \in \mathcal{P}$. \label{fig:p_l_sF2}}
	[0.3\textwidth]{\includegraphics[width=\linewidth]{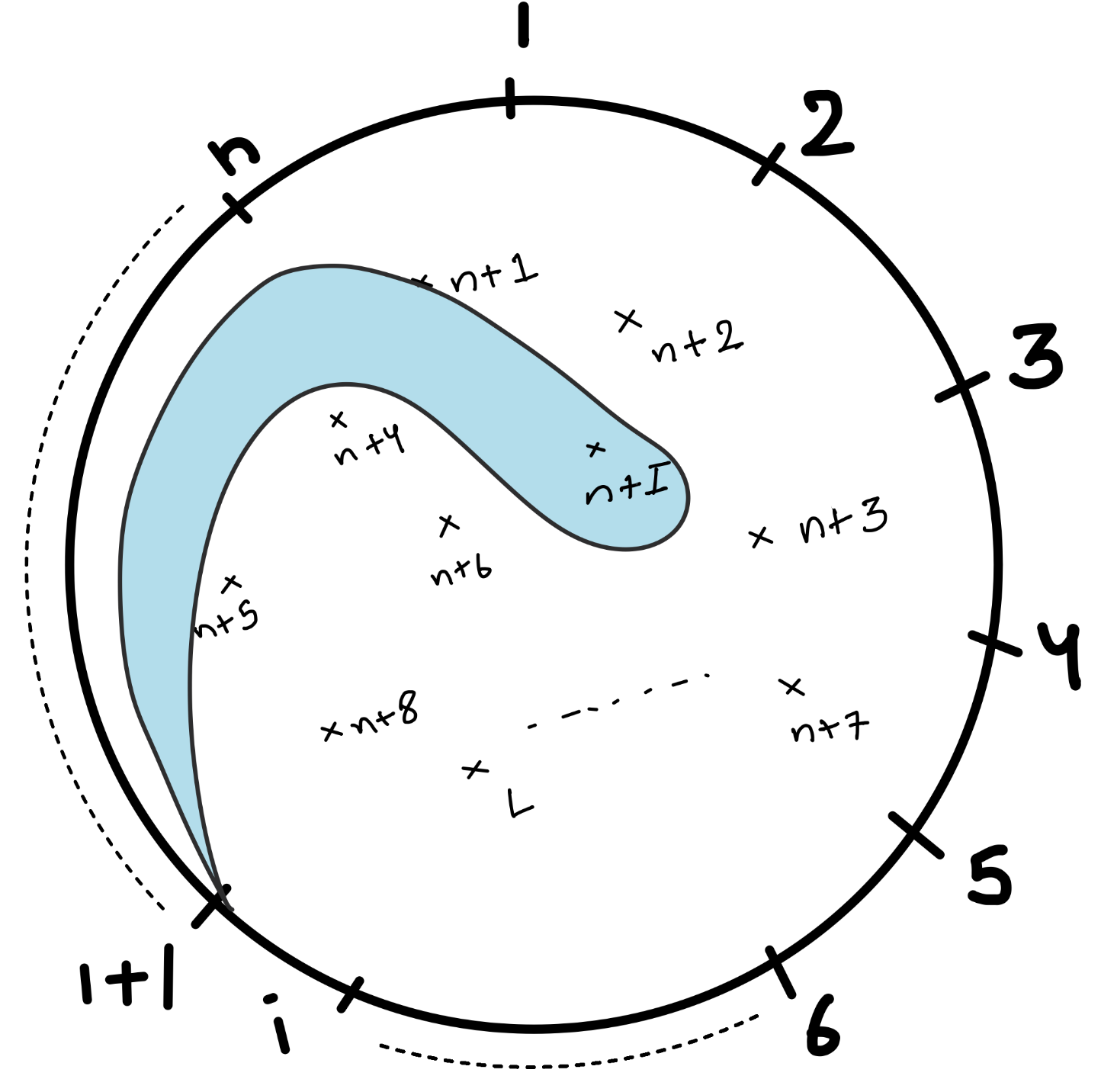}}
	\hfill
	\subcaptionbox{A 2-gon with both edges $1,3 \in \mathcal{V}$ and $n+2 \in \mathcal{P}$. \label{fig:p_l_sF3}}
	[0.3\textwidth]{\includegraphics[width=\linewidth]{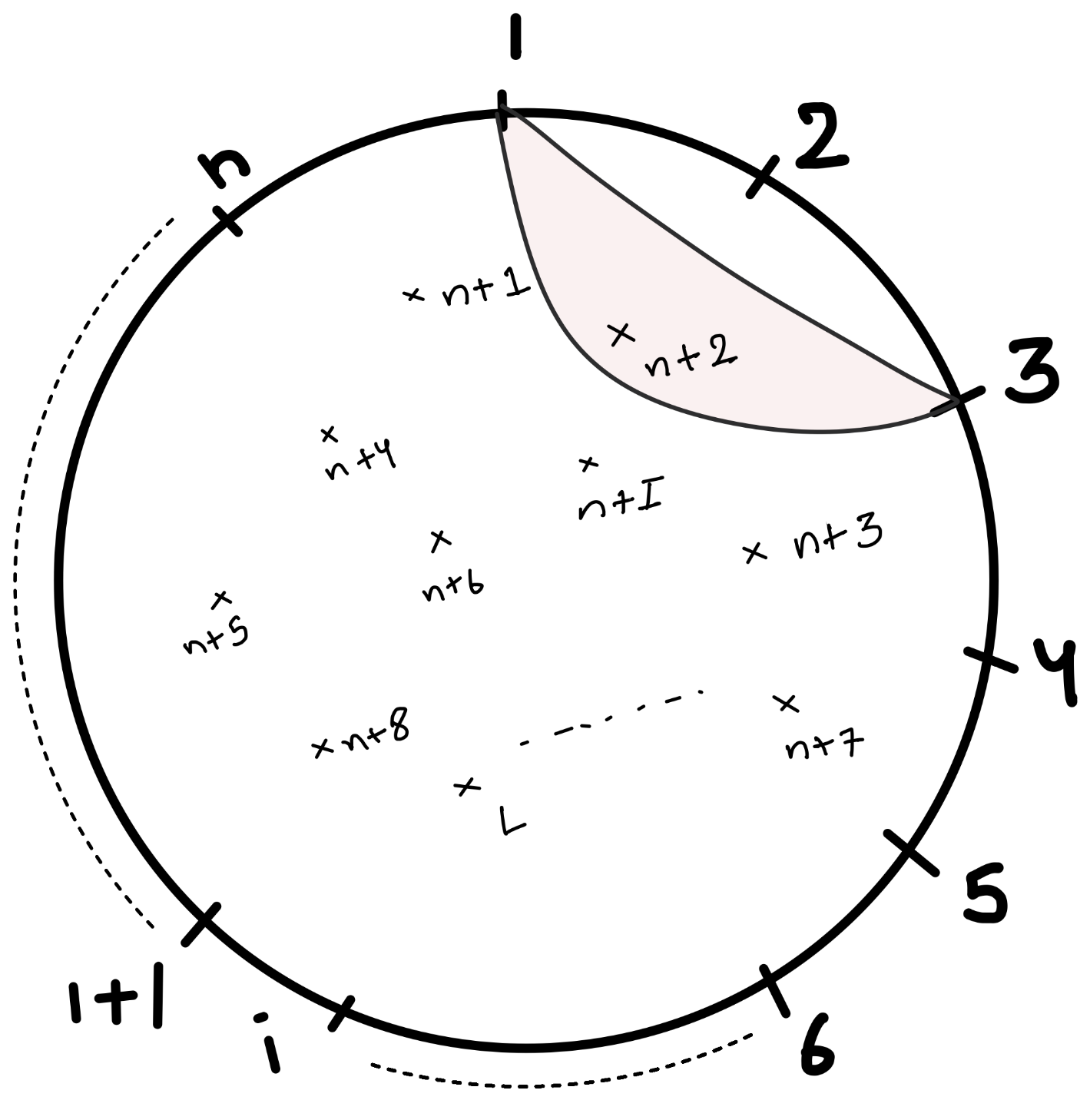}}
	
	\vspace{0.5cm}
	% Bottom row
	\subcaptionbox{A 2-gon with both edges $n+5,n+6 \in \mathcal{P}$ and $n+4 \in \mathcal{P}$. \label{fig:p_l_sF4}}
	[0.3\textwidth]{\includegraphics[width=\linewidth]{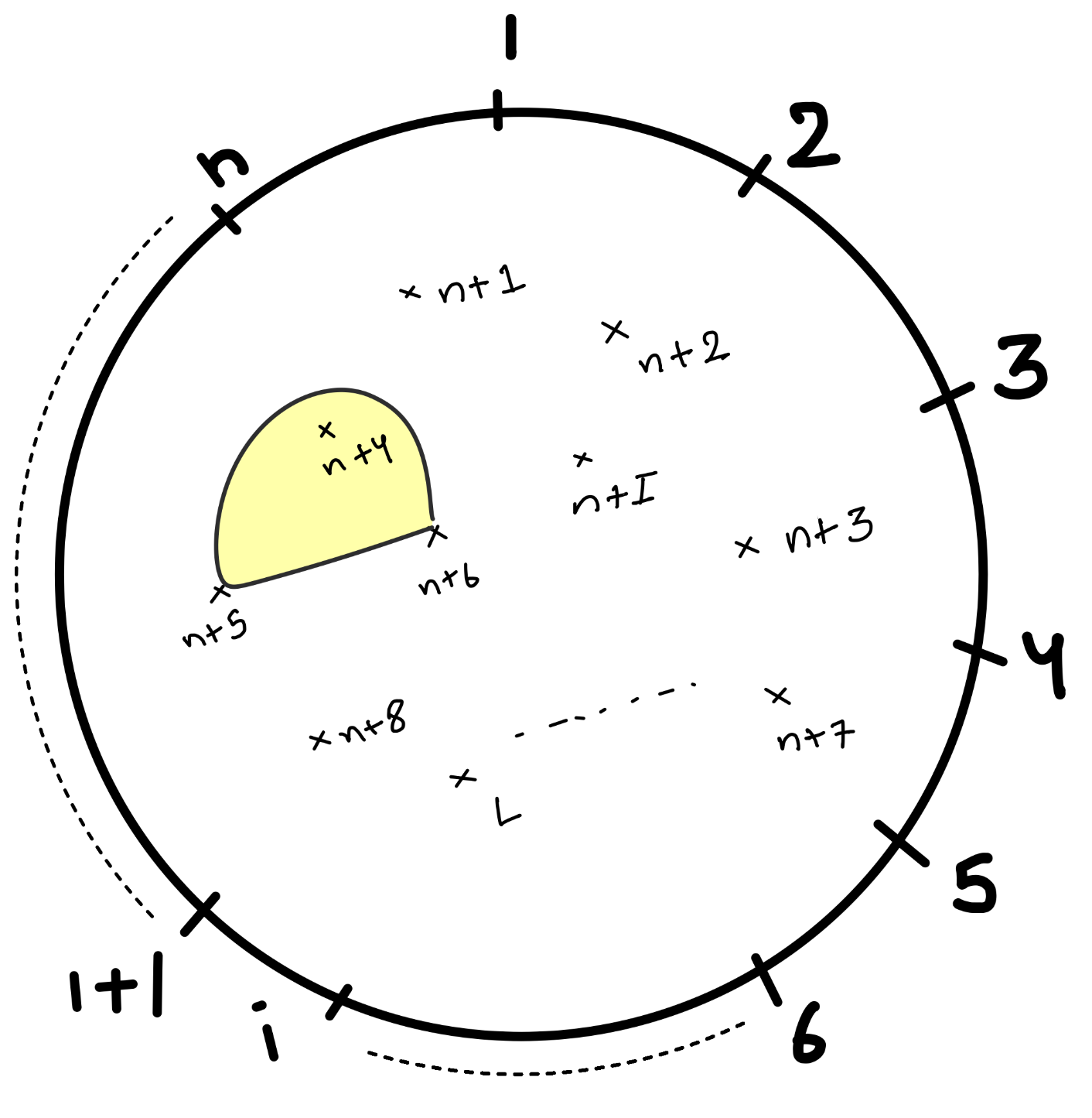}}
	\hfill
	\subcaptionbox{A 1-gon with edge $n+2 \in \mathcal{P}$ and puncture $n+4 \in \mathcal{P}$. \label{fig:tadpole_sf}}
	[0.3\textwidth]{\includegraphics[width=\linewidth]{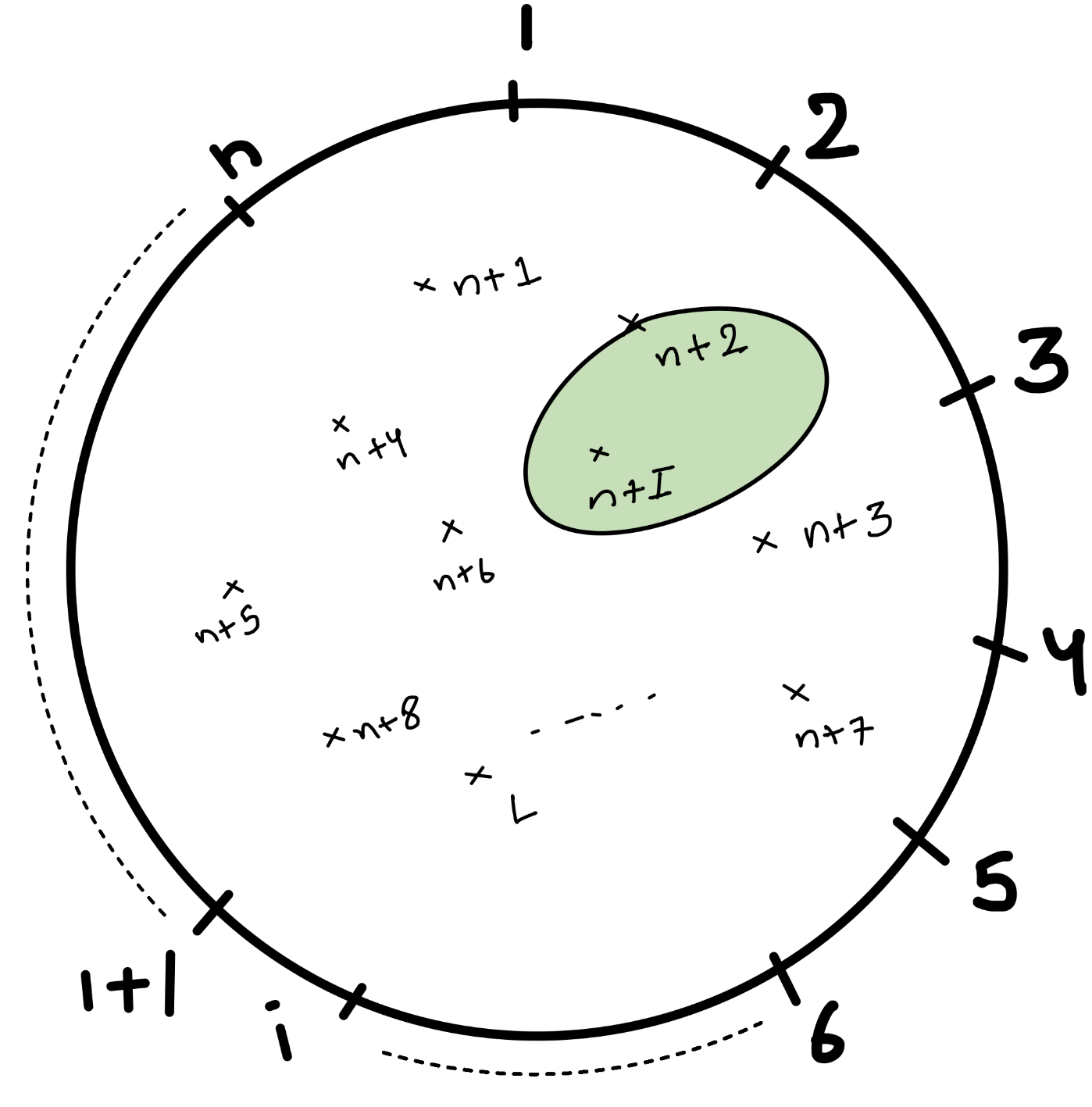}}
	\hfill
	\subcaptionbox{An example of an s-forest with edges $\{\{i+1,n+1\},\{n+5,n+6\},\{5\},\{1,3\},\{5,4\}\}$ enclosing punctures $\{n+2,n+I,n+8,n+4,n+7\}$. \label{fig:Example_of_sf}}
	[0.3\textwidth]{\includegraphics[width=\linewidth]{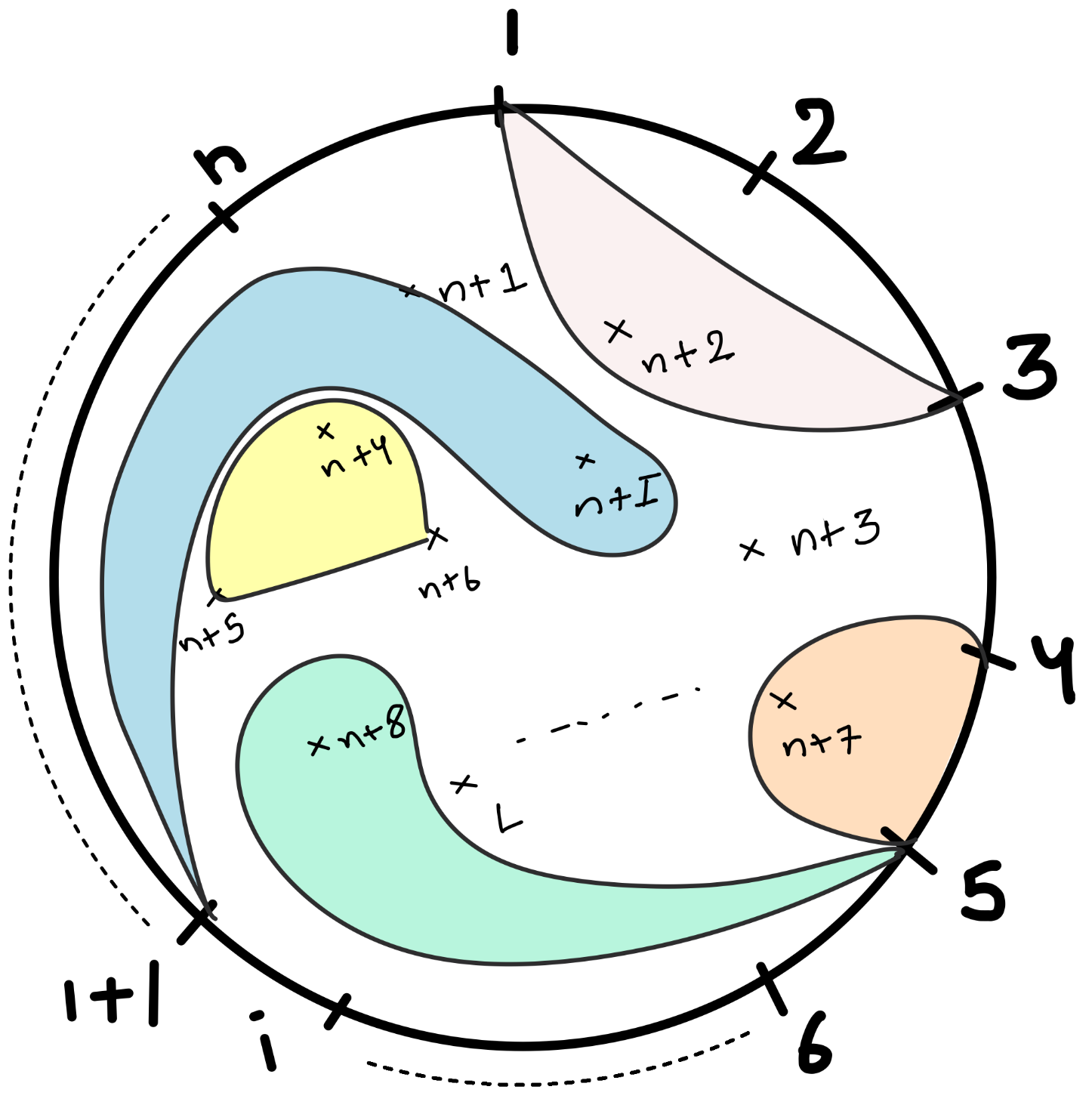}}
	
	\caption{Examples of s-forests.}
	\label{fig:all_s_forests}
\end{figure}

We will denote the boundary of $\sigma(i,I,j)$ as
\begin{align}
	\partial \sigma(i,I,j)\, =\, C_{I}^{+}(v_{i}, v_{j})\, \medcup\, C_{I}^{-}(v_{i}, v_{j}) 
\end{align}

where $C_{I}^{+}(v_{i}, v_{j}),\, C_{I}^{-}(v_{i}, v_{j})  \in\, {\cal C}$ are the two curves such that as we start at $v_{i}$ and circle around the puncture $v_{I}$ then we go clockwise by first traversing $C_{I}^{+}(v_{i}, v_{j})$ and then traversing $C_{I}^{-}(v_{i}, v_{j})$.  

For the case of 1-gon the two points are the same then surface have only one boundary $C^t_{I}(v_{i}, v_{i})$ i.e. the tadpole curve.

We will use the following notation for an s-forest. If it is defined by a collection of 2-gons that bound the punctures $(v_{1}, \dots, v_{K})$ and 1-gons that bound the punctures $(v_{K+1}, \dots , v_{M})$,  (that satisfy definition (\ref{def:surface_forest}) then we will refer to it as, 
\begin{align}\label{sfnotation}
	\textrm{sf}^{(i_{1}, j_{1}), \dots\, (i_{K}, j_{K}) \vert i_{K+1}, \dots, i_{M}}_{v_{1}, \dots, v_{K} \vert v_{K+1}, \dots, v_{M}} 
\end{align}
A few comments are in order regarding some characteristic properties of the s-forests.
\begin{enumerate}
	\item We note that the cardinality of the surface forest (as a set of 2-gons and 1-gons) is the number of punctures that label it. 
	\begin{align}
		\bigg|\, \textrm{sf}^{(i_{1}, j_{1}), \dots\, (i_{K}, j_{K}) \vert i_{K+1}, \dots, i_{M}}_{v_{1}, \dots, v_{K} \vert v_{K+1}, \dots, v_{M}}\, \bigg| \, =\, M
	\end{align}
	\item Given the s-forest in eqn.(\ref{sfnotation}), there is a unique set of curves that hit the punctures $(v_{1}, \dots, v_{K} \vert\, v_{K+1}, \dots, v_{M})$.  This set is
	\begin{align}\label{icsisf}
		{\cal C}^{\textrm{in}}\bigg(\textrm{sf}^{(i_{1}, j_{1}), \dots\, (i_{K}, j_{K}) \vert i_{K+1}, \dots, i_{M}}_{v_{1}, \dots, v_{K} \vert v_{K+1}, \dots, v_{M}}\bigg)\, :=\, \bigg(\bigcup_{\alpha=1}^{K} C(v_{i_{\alpha}}, v_{\alpha}), C(v_{j_{\alpha}}, v_{\alpha}) \bigg)\, \medcup\, \bigg( \bigcup_{\beta=K+1}^{M} C( v_{i_{\beta}}, v_{\beta}) \bigg)
	\end{align}
	\item Given any s-forest, $\textrm{sf}$ we will denote the correspondingly unique set of curves that hit the punctures that label $\textrm{sf}$ as ${\cal C}^{\textrm{in}}(\textrm{sf})$. 
	\item For latter purpose (of introducing the counter-term for curve integrals),  we will also need to introduce a notion of the set of ``boundary curves'' among all the curves that generate 2-gon elements of an $\textrm{sf}$. That is, given the surface forest such as the one defined in eqn.(\ref{sfnotation}), we define the set, 
	\begin{align}\label{cnbsf}
		{\cal C}^{\textrm{b}}(\textrm{sf})\, :=\,  \left\{\, \bigcup_{\epsilon = \pm}\, \bigcup_{a=1}^{K}\, C^{\epsilon}_{i_{a} j_{a}}\, \Big{\vert}\,  \vert j_{a} - i_{a} \vert > 1\,\right\} 
	\end{align}
	\item Similarly if ${\cal C}^{t}(\textrm{sf})$ is the set of all the curves that bound $\sigma(i,I,i)$ for some $(i,\, I)$ then we define 
	\begin{align}\label{cnbtsf} 
		{\cal C}^{\textrm{b} \medcup t}(\textrm{sf})\, :=\, {\cal C}^{\textrm{b}}(\textrm{sf})\, \medcup\, {\cal C}^{\textrm{t}}(\textrm{sf})
	\end{align}
	\item We will denote the set of all  s-forests inside $\Sigma_{L,n}$ as ${\cal SF}(L,n)$. 
\end{enumerate}
An illustrative example for ${\cal C}^{\textrm{in}}_{\textrm{sf}}$ is shown in Fig.~\ref{fig:curves_in_sF1}--\ref{fig:curves_in_tapole_sF}.

We will now show that surface forests can be used to obtain renormalized curve integrals in the present case, i.e. in a theory with no overlapping divergent graphs.  Recall that in the Zimmermann's forest formula for renormalizing individual graphs,  a forest is defined as sets of 1PI sub-divergent graphs such that the subgraphs are 
\[
\gamma_i \medcap \gamma_j = \varnothing \quad \text{or} \quad
\gamma_i \medsubset \gamma_j \quad \text{or} \quad \gamma_j \medsubset \gamma_i.
\]
For each 1PI sub-divergent graph in the forest, we fix a kinematic point for the incoming momenta~\footnote{This implies fixing magnitudes of the momenta $p_i^2=\mu^2$ and fixing angles between the incoming momenta's to a 1PI sub-divergence $\theta=\frac{p_i.p_j}{|p_i||p_j|}$ to some $\theta_0$. }. 

The construction of the s-forest formula follows the same principle with a key difference : Namely, given a  (mutually disjoint) collection of 1 PI divergent subgraphs, which is embedded inside a set of graphs, we write a single counter term for this whole set of such graphs using headlight functions and than sum over all such collections. 

\begin{figure}[h!]
	\centering
	\begin{subfigure}{0.45\textwidth}
		\centering
		\includegraphics[width=0.7\linewidth]{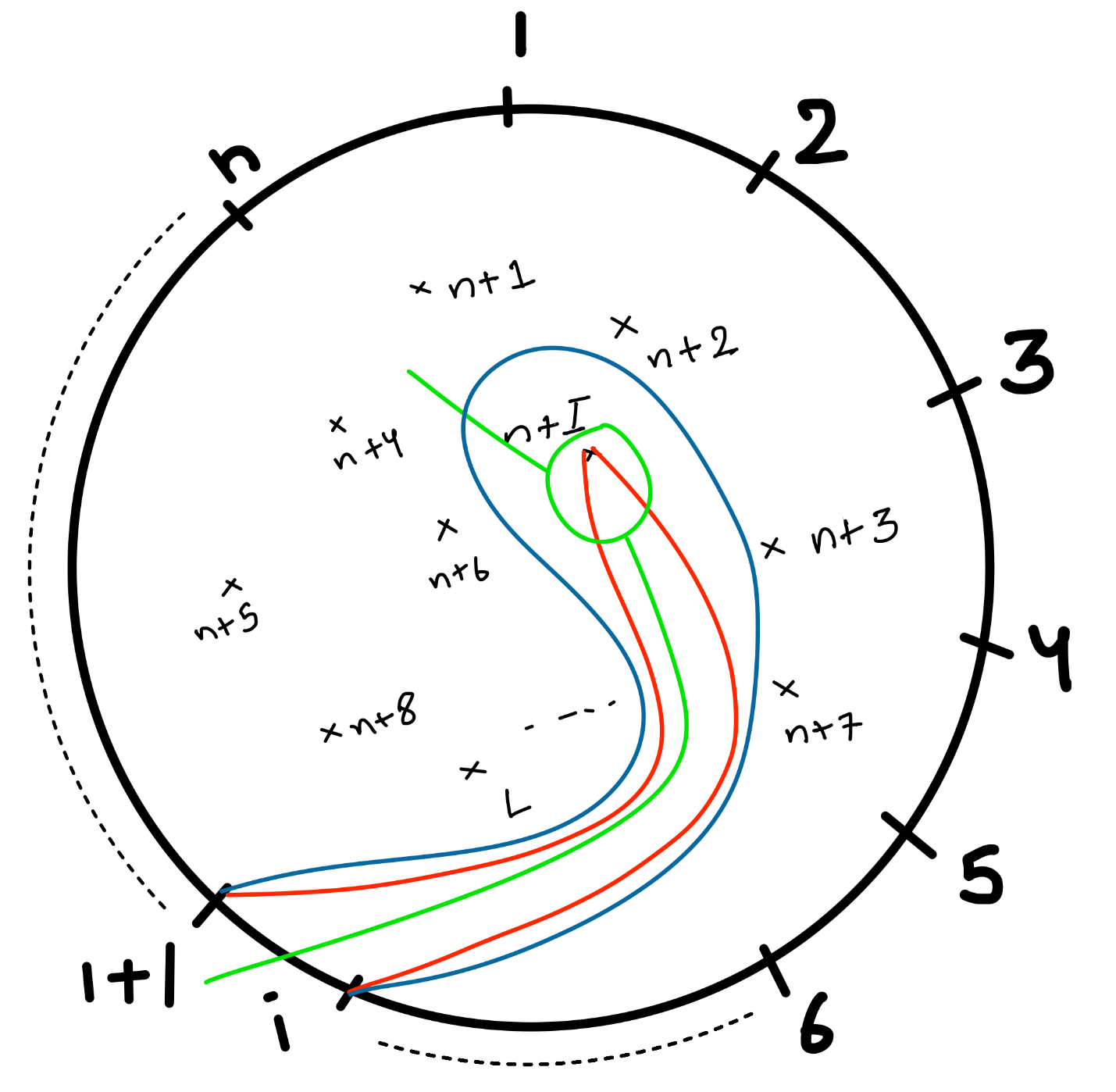}
		\caption{}
		\label{fig:curves_in_sF1}
	\end{subfigure}
	\hfill
	\begin{subfigure}{0.45\textwidth}
		\centering
		\includegraphics[width=0.7\linewidth]{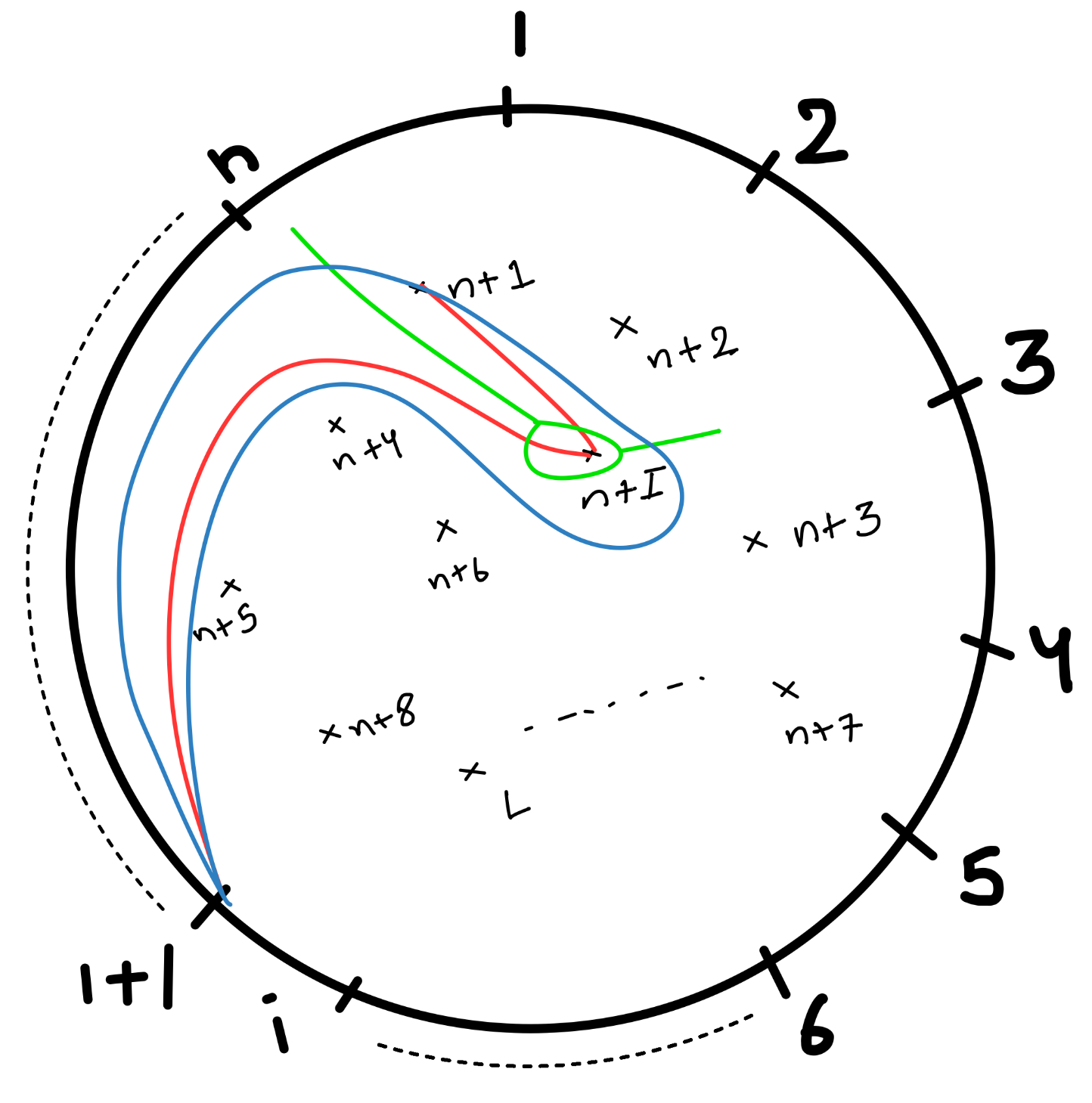}
		\caption{}
		\label{fig:curves_in_sF2}
	\end{subfigure}
	\vskip\baselineskip
	\begin{subfigure}{0.45\textwidth}
		\centering
		\includegraphics[width=0.7\linewidth]{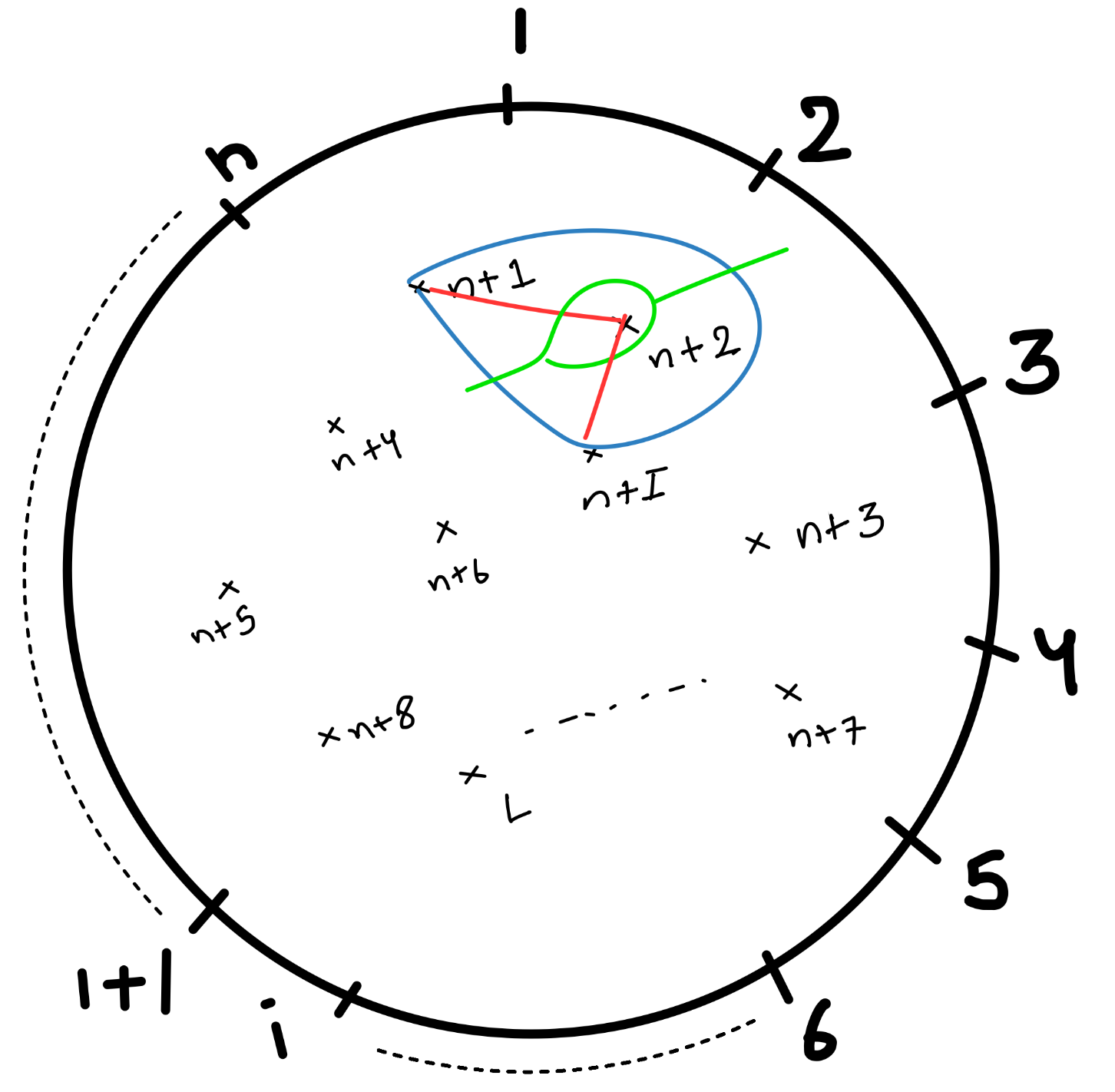}
		\caption{}
		\label{fig:curves_in_sF3}
	\end{subfigure}
	\hfill
	\begin{subfigure}{0.45\textwidth}
		\centering
		\includegraphics[width=0.7\linewidth]{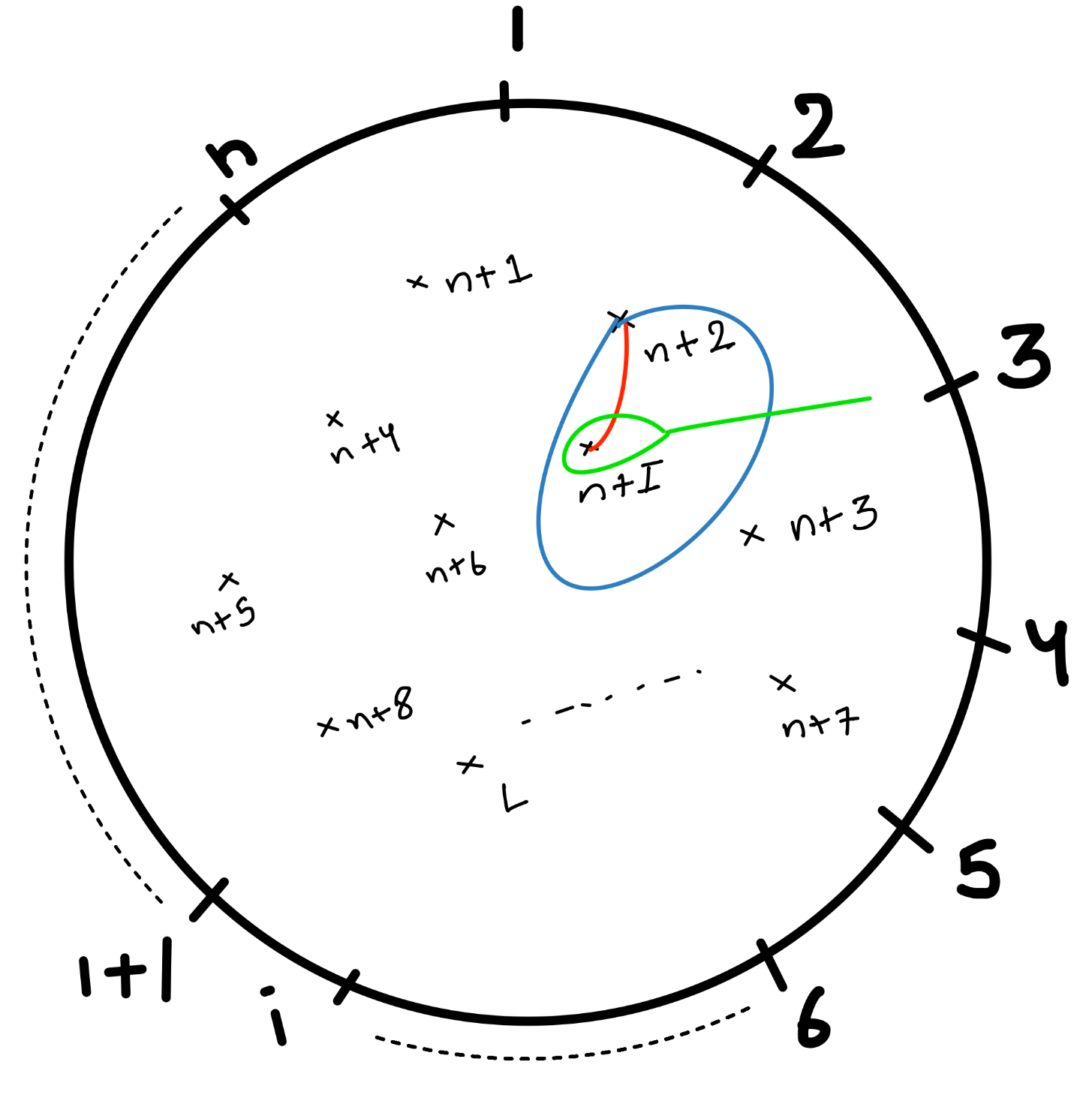}
		\caption{}
		\label{fig:curves_in_tapole_sF}
	\end{subfigure}	
	\caption{In the figure above, the red curves represent the possible non–intersecting curves enclosed within the subsurfaces. The blue curves correspond to the bounding curves $C^{+}_{J},\,C^{-}_{J}$ (for a $2$–gon) or $C^{t}_{K}$ (for a $1$–gon), while the green graph depicts the associated $1$PI subgraph.}
	\label{fig:four_grid}
\end{figure}
Such a counter-term can be written by essentially following the usual Forest formula which was reviewed in section on parametric forest:~\ref{sec:review_parametric_forest}. 

For any s-forest, there is  a unique assignment of momentum labels for each curve in the set ${\cal C}^{\textrm{in}}(\textrm{sf})$ which was defined in eqn.(\ref{icsisf}).  Consider the subset $(v_{1}, \dots, v_{K})$ of the set of punctures, $(v_{1}, \dots, v_{K} \vert\, v_{K+1}, \dots, v_{M})$ which defines an s-forest. We denote the unique momentum assignment to the pair of curves $C(v_{i_{\alpha}}, v_{\alpha}), C(v_{j_{\alpha}}, v_{\alpha})$ incident on the puncture $v_{\alpha}\, \vert \, \alpha\, \in\, \{1, \dots, K\}$ as
\begin{align}
\{\, C(v_{i_{\alpha}}, v_{\alpha}), C(v_{j_{\alpha}}, v_{\alpha})\, \}\, \rightarrow\, ( P^{\mu}_{i_{\alpha} v_{\alpha}},\, P^{\mu}_{j_{\alpha} v_{\alpha}}\, ).
\end{align}
Now let, 
\begin{align}
X_{(i_{\alpha}\, j_{\alpha})\, v_{\alpha}}\, :=\, (\, P_{i_{\alpha} v_{\alpha}} -\, P_{j_{\alpha} v_{\alpha}}\, )^{2}
\end{align}
As we will see, the counter-term in the surface forest formula will be obtained by fixing $X_{i_{\alpha}\, j_{\alpha}\, v_{\alpha}}$ at a renormalization point $S_{0}$. 

Consider a punctured disc $\Sigma_{L,n}$ with a partial dissection given by the s-forest which is defined in eqn. (\ref{sfnotation}).  for brevity of notation, we will denote the s-forest simply as $\textrm{sf}$ in rest of this section.   Let $x_{C}$ be a variable assigned to each curve $C$. A priori these variables are independent of $X_{C}$. Consider the following ``curve integral'' which is obtained by summing over all the dissections of $\Sigma_{n,L}$ in which the partial dissection induced by $\textrm{sf}$ is fixed. 
\begin{align}\label{eq:anrsfpij}
	{\cal A}_{L, n}\bigg|_{\mathrm{sf}}(p_{1}, \dots, p_{n})&:=\nonumber\\
	&\hspace*{-1.0in}\int \frac{d^{\lvert E\rvert} t}{\mathrm{MCG}}\;
	\mathcal{G}_{\mathrm{sf}}
	\left[\, \int \prod_{r=1}^{L} d\ell_{r}\; e^{-\sum_{C'}\, \alpha_{C'} x_{C'}}\, \right]\Bigg|_{\scriptsize \left. \begin{cases} x_{C} = X_{C} \forall\, C\, \notin\, {\cal C}^{t}(\textrm{sf}),\, \\  X_{( i_{\alpha}\, j_{\alpha} )\, v_{\alpha}} =\, S_{0}\, \forall\, \alpha\, \in\, \{1, \dots,\, K\},\, \\  x_{C} = \epsilon^{2}\, \forall\, C\, \in {\cal C}^{t}(\textrm{sf}) \end{cases} \right\} }
\end{align}
where 
\begin{itemize}
	\item ${\cal C}^{\textrm{t}}(\textrm{sf})$ is defined in eqn.(\ref{cnbtsf}) respectively. 
	\item $S_{0}$ is a renormalization scale associated a pre-chosen scale variable $S$. 
	Throughout this paper, our choice of $S$ will be the following. Recall that, 
	\begin{align}
		X_{ji} :=\, (p_{j} +\, \dots\, +\, p_{n}\, +\, p_{1}\, +\, \dots\, + p_{i-1})^{2}\, \ \textrm{for}, \, 1\, \leq\, i < j \leq\, n
	\end{align}
	We then choose $S, S_{0}$ as, 
	\begin{align}\label{schoice}
		S &:=\, \frac{1}{\frac{n(n-3)}{2}}\, \sum_{1 \leq\, i < j\, \leq\, n} X_{ji}\nonumber\\
		S_{0} &= m^{2}.
	\end{align}
	\item As compared to the parametric Forest formula for a fixed graph in a theory without overlapping and overall divergence beyond $n=2$, the surface forest formula looks rather complicated. But we can understand it intuitively as follows. However we will give several examples in section (\ref{sec:sf_reno_examples}) that makes it's interpretation clear.
	\item Note that for any tadpole curve $X_{C} = (\sum_{i=1}^{n} p_{i})^{2}$ is independent of the momenta and as a result  can be chosen to be the same variable when regularizing the bare curve integral. We thus set $X_{C} = \epsilon^{2}, \, \forall\, $ tadpole curves. Our counter-term is defined such that for tadpole curves $x_{C}$ is also set at $\epsilon^{2}$. As we show below, this will imply that the renormalized curve integral obtained via surface-forest formula will have no tadpole contribution. 
	\item Kernel ${\cal G}(\textrm{sf})$ in eqn.(\ref{eq:anrsfpij}) is defined as follows. 
\end{itemize}
Finally, we define the kernel ${\cal G}(\textrm{sf})$ that appears in eqn (\ref{eq:anrsfpij}). 

We first define ${\cal G}$ for the simplest example of the s-forest = $\sigma(i, I, j)$ which is one 2-gon bounding a puncture $v_{I}$. Then, 
$\mathcal{G}_{\sigma(i,I,j)}$ is defined as, 
\begin{align}
	{\cal G}(\sigma(i, I, j)):= \frac{\displaystyle \prod_{C^{\prime}\,\in\,\mathcal{C}^{\mathrm{sf}}}\alpha_{C^{\prime}}}
	{\displaystyle \sum_{C_{I}\,\notin\, \mathcal{C}^{\mathrm{sf}}}\alpha_{C_{I}}
		+ \displaystyle \prod_{C^{\prime}\,\in\,\mathcal{C}^{\mathrm{sf}}}\alpha_{C^{\prime}} } \, .
\end{align}
where \[
{\cal C}^{\mathrm{sf}}
= \left\{{\cal C}^{\mathrm{in}}(\mathrm{sf})
\;\medcup\; {\cal C}^{\textrm{b}\,\cup\,t}(\textrm{sf})\right\} \,\setminus\, C_{ij} \qquad \forall \, C_{ij} \in \partial\Sigma_{L,n},
\]
and which is precisely the set of non-trivial curves that belong to partial dissection whose headlight functions appear in the expressions of the curve integral and $C_{I}\notin \mathcal{C}^{\mathrm{sf}}$ denotes all curves not in $\mathcal{C}^{\text{sf}}$ starting from any point 
$k \in \mathcal{V}\mathcal{P}$ and ending at a puncture $v_{I} \in \mathcal{P}$ of the surface forest.  
The same definition also applies to a surface forest with a $1$-gon.  

Therefore, for a generic surface forest 
\[
\textrm{sf}^{(i_{1}, j_{1}), \dots\, (i_{K}, j_{K}) \vert i_{K+1}, \dots, i_{M}}_{v_{1}, \dots, v_{K} \vert v_{K+1}, \dots, v_{M}}
\]
the function $\mathcal{G}$ can be defined as follows:
\begin{align}
	{\cal G}\left(\textrm{sf}^{(i_{1}, j_{1}), \dots\, (i_{K}, j_{K}) \vert i_{K+1}, \dots, i_{M}}_{v_{1}, \dots, v_{K} \vert v_{K+1}, \dots, v_{M}}\right)=	\frac{\displaystyle \prod_{C^{\prime}\in\;C^{\mathrm{sf}}}\alpha_{C^{\prime}}} {\displaystyle\sum_{I=1}^M\displaystyle\sum_{C_{I}\notin C^{\textrm{sf}}}\alpha_{C_{I}}+ \displaystyle \prod_{C^{\prime}\in\;C^{\mathrm{sf}}}\alpha_{C^{\prime}}} 
\end{align}
${\cal G}$  is defined in such a way that the surface forest formula reduces to the parametric forest formula inside each cone. More in detail, 
if the subgraphs specified by the partial dissection introduced by surface forest are also present within a cone, then the first term in the denominator vanishes (since it contains all headlight functions of curves intersecting the surface-forest curves) and as a result ${\cal G}(\textrm{sf})\, \rightarrow\, 1$ inside this cone. 
On the other hand, if the set of 1PI irreducible subgraphs inside a cone is not same as the s-forest  then ${\cal G}(\textrm{sf})\vert_{\textrm{cone}}\, =\, 0$. 
We will illustrate this phenomena explicitly in the examples of Sec.~\ref{sec:sf_reno_examples}.

After performing the loop integration, the above equation takes the form,
\begin{align}\label{anrsfpijali}
	{\cal A}_{L, n}^{sf(\{x_{c}^{0}, x_{c}\})}\, \displaystyle =\,  \displaystyle \int_{\sum t_{i} \leq\, 0}\, \frac{d^{\vert E\vert}t}{\textrm{MCG}}\ \frac{1}{{\cal U}^{\frac{D}{2}}_{\mathrm{sf}}}\ e^{\frac{{\cal F}^{0}_{\textrm{sf}} }{{\cal U}}-{\cal Z}^{0}_{\textrm{sf}}}\  {\cal G}_{\textrm{sf}}
\end{align}
In the above equation, ${\cal U}_{\mathrm{sf}}$ is the first Symanzik polynomial, notice that it is also depend on the s-forest, this happens because there can be 2-gons $\sigma(i,\, I,\, j)$ in the s-forest that can have one or both of $v_i, v_j\ \in \ \mathcal{P}$ for example see Fig.~\ref{fig:p_l_sF2}, \ref{fig:p_l_sF4} for such surfaces incoming moment have loop momenta which are also fixed. Consequently, after integrating over the loop momenta, the first surface Symanzik polynomial is modified, and therefore depends on the s-forest. ${\cal F}^{0}$ is the surface Symanzik polynomial.

Finally the renormalized amplitude is given as,
\begin{align}
	\mathcal{A}^R_{L, n} =\mathcal{A}_{L, n} -\sum_{\textrm{sf}^{\, i}}(-1)^{|\textrm{sf}^{\, i}|}\mathcal{A}_{L, n} \big{\vert}_{{\textrm{sf}^{\, i}}}
	\label{sf_formula}
\end{align}
We have now defined a s-forest and explained how it can be used to construct counterterms. Let us next turn to some explicit examples.

%%%%%%%%%%%%%%%%%%%%%%%%%%%%%%%%
\subsection{Examples}\label{sec:sf_reno_examples}
This section demonstrates how the preceding statements apply in practice. Since this renormalization scheme does not rely on the properties of the headlight function, one may work with any convenient reference fatgraph.

%%%%%%%%%%%%%%%%%
\subsubsection{One-loop}
At one loop, the surface forest formula takes a particularly simple form, 
yet it is already instructive in illustrating the differences from the 
Zimmermann's forest formula. Recall that $n$ denotes the number of marked points 
on the boundary of the disc (equivalently, the number of external particles 
in the planar amplitude), so that $\mathcal{V} = \{1,2,\dots,n\}$.  

In this case, the surface forests consist of two types of $(k\leq2)$-gons:  
2–gons of the form $\sigma(i, A, j)$ and 1–gons of the form 
$\sigma(i, A, i)$. Both edges of these  $2$-gons lie on boundary points, 
with $v_i, v_j \in \mathcal{V}$ and $v_A \in \mathcal{P}$ similarly for the $1$-gon.  
The full set of surface forests is therefore easily described as  
\begin{align}
	{\cal SF}(L,n)
	= \left\{\, \mathrm{sf}^{(i,j)}_{A} \;\forall\, i<j, \quad 
	\mathrm{sf}^{\,k}_{A} \;\forall\, k \,\right\}.
\end{align}
where $v_i,v_j\ \in \ \mathcal{V}$ and the set of enclosed curves are,
\begin{align}
	&C^{in}(\text{sf}^{i,j}_{A})=\{C_{i0},C_{j0}\}\nonumber\\
	&C^{in}(\text{sf}^{k,k}_{A})\, =\{C_{k0}\}
\end{align}	 
where $C_{i0}$ (starts at some point on the boundary and end up at the loop). The surface forest formula in  eqn.(\ref{sf_formula}) with counter-term defined in  eqn.(\ref{eq:anrsfpij}) simplifies to, 
\begin{align}\label{anrwithN}
	{\cal A}_{1, n}^{R} = {\cal A}_{1, n} (X_{0i},\{ X_{mn},\, X_{nm}\}\vert_{m<n})\, &- \sum_{(i < j)} {\cal A}^{(i,j)}_{1,n}\left( x_{C}\right) \Bigg|_{\scriptsize \left. \begin{cases} x_{C} = X_{C} \forall\, C\, \notin\, {\cal C}^{t}(\textrm{sf}),\, \\  X_{( i\, j)\, v_{\small{A}}} =\, S_{0}\,  \end{cases} \right\} }\nonumber\\
	&- \sum_{k} {\cal A}_{1, n}^{(k)}(x_C) \Bigg|_{\scriptsize \left. \begin{cases} x_{C} = X_{C} \forall\, C\, \notin\, {\cal C}^{t}(\textrm{sf}),\, \\  x_{C_{kk}} = \epsilon^{2} \end{cases} \right\} }
\end{align}
%\begin{align}\label{anrwithN}
%	{\cal A}_{1, n}^{R} = {\cal A}_{1, n} (X_{0i},\{ X_{mn},\, X_{nm}\}\vert_{m<n})\, &- \sum_{(i < j)} {\cal A}^{(i,j)}_{1,n}\left( x_{C}\right)\bigg|_{ \{ x_C=S_0\, , \ \forall  {X_{ij},X_{ji}}\}}\nonumber\\
%	&- \sum_{k} {\cal A}_{1, n}^{(k)}(X^0_{ii}, X_{mn})\bigg|_{\{x_C=\epsilon^2\, , \ \forall {X_{ii}}\}}
%\end{align}
The first counter-term in this equation is given by
\begin{align}
	{\cal A}^{(i,j)}_{1,n}\left( x_{C}\right) \Bigg|_{\scriptsize \left. \begin{cases} x_{C} = X_{C} \forall\, C\, \notin\, {\cal C}^{t}(\textrm{sf}),\, \\  X_{( i\, j)\, v_{\small{A}}} =\, S_{0}\,  \end{cases} \right\} }\, =\, \int_{\sum t_{i} \leq\, 0}\, d^{\vert E \vert}\vec{t}\, \, \frac{1}{{\cal U}^{\frac{D}{2}}}\, e^{-\frac{{\cal F}^0_{i,j}}{{\cal U}}\, -\, {\cal Z}_{i,j}}\, {\cal G},
\end{align}
%\begin{align}
%{\cal A}^{(i,j)}_{1,n}\left( x_{C}\right)\bigg|_{ \{ x_C=S_0 \, \forall {X_{ij},X_{ji}}\}}\, =\, \int_{\sum t_{i} \leq\, 0}\, d^{\vert E \vert}\vec{t}\, \, \frac{1}{{\cal U}^{\frac{D}{2}}}\, e^{-\frac{{\cal F}^0_{i,j}}{{\cal U}}\, -\, {\cal Z}_{i,j}}\, {\cal G},
%\end{align}
with 
\begin{align}
	{\cal F}_{(i,j)} \,&=\, 
	\sum_{(m,n)\, \neq\, (i,j)} \alpha_{m0}\, \alpha_{n0}\, z_{m} \cdot z_{n} 
	+ S_{0}\, \alpha_{i0}\alpha_{j0} \label{eq:Fij}\\[6pt]
	{\cal Z}_{i,j} \,&=\, 
	\sum_{\forall p \,>\, q} \alpha_{pq}\, X_{pq} 
	+ \sum_{i}\alpha_{i0}\, m^{2}, \label{eq:Zij_1l}\\[6pt]
	\mathcal{G}_{i,j} \,&=\, 
	\frac{\alpha_{i0}\;\alpha_{j0}\;\prod_{\small{C_{ij}\;\in \; C^{\text{sf}}_{ij}}}\alpha_{C}}
	{\displaystyle \sum_{m\neq \{i,j\}}\alpha_{m0} \;+\alpha_{i0}\;\alpha_{j0}\;\prod_{\small{C_{ij}\;\in \; C^{\text{sf}}_{ij}}}\alpha_{C}}. 
	\label{eq:G_1l}
\end{align}
where
\begin{align}
	z_{m}^{\mu} \,&=\, \sum_{i=1}^{m-1} p_{i}^{\mu}, \label{eq:zm}
\end{align}
For example, for a surface forest $\mathrm{sf}_p(m^p_1,m^p_2)$ 
\begin{align}
	{\cal F}_{1,2} \,&=\, 
	\sum_{(m,n)\, \neq\, (1,2)} \alpha_{m0}\, \alpha_{n0}\, z_{m} \cdot z_{n} 
	+ S_{0}\, \alpha_{10}\alpha_{20} \label{eq:Fij}\\[6pt]
	{\cal Z}_{1,2} \,&=\, 
	\sum_{\forall p \,>\, q} \alpha_{pq}\, X_{pq} 
	+ \sum_{i}\alpha_{i0}\, m^{2}, \label{eq:Zij_1l}\\[6pt]
	\mathcal{G}_{1,2} \,&=\, 
	\frac{\alpha_{10}\;\alpha_{20}\;\alpha_{21}}
	{\alpha_{30}+\alpha_{10}\;\alpha_{20}\;\alpha_{21}}. 
\end{align}

Now consider the counter term for tadpole surface forest,
\begin{align}
	{\cal A}_{1, n}^{(k)}(X^0_{kk}, X_{mn})\Bigg|_{\scriptsize \left. \begin{cases} x_{C} = X_{C} \forall\, C\, \notin\, {\cal C}^{t}(\textrm{sf}),\, \\  x_{C_{kk}} = \epsilon^{2} \end{cases} \right\} }=
	\, \int_{\sum t_{i} \leq\, 0}\, d^{\vert E \vert}\vec{t}\, \, \frac{1}{{\cal U}^{\frac{D}{2}}}\, e^{-\frac{{\cal F}^0_k}{{\cal U}}\, -\, {\cal Z}}\, {\cal G}
\end{align}
%\begin{align}
%{\cal A}_{1, n}^{(k)}(X^0_{ii}, X_{mn})\bigg\vert_{x_C=\epsilon^2\forall {X_{ii}}}=
%	\, \int_{\sum t_{i} \leq\, 0}\, d^{\vert E \vert}\vec{t}\, \, \frac{1}{{\cal U}^{\frac{D}{2}}}\, e^{-\frac{{\cal F}^0_k}{{\cal U}}\, -\, {\cal Z}}\, {\cal G}
%\end{align}
with 
\begin{align}
	{\cal F}_k \,&=\, 
	\sum_{\substack{m,n \\ n\neq k}} \alpha_{m0}\, \alpha_{n0}\, z_{m} \cdot z_{n} 
	+ \sum_{m}\alpha_{m0}\alpha_{k0}\, z_m \cdot \left(z_k\Big{\vert}_{\,p^2_{C_{kk}}=\,\epsilon^2}\right) , \label{eq:Fi}\\[6pt]
	{\cal Z} \,&=\,  
	\sum_{j} \alpha_{j0}\, m^{2} \;+\; \sum_{C} \alpha_{C}\, X_{C}, \label{eq:Z}\\[6pt]
	\mathcal{G} \,&=\, 
	\frac{\alpha_{kk}\alpha_{k0}}
	{\displaystyle \sum_{m\neq k}\alpha_{m0} \;+\;\alpha_{kk}\;\alpha_{k0}}\, , \label{eq:G}
\end{align}
where
\begin{align}
	z_{m}^{\nu} \,=\, \sum_{k=1}^{m-1} p_{k}^{\nu}, \label{eq:zdef}
\end{align}
Since the momentum carried by a self-curve vanishes,
\begin{align}
	p^{\nu}_{C_{kk}} = 0,
\end{align}
the second term  \eqref{eq:Fi} becomes,
\begin{align}
	{\cal F}_{i}\, =\, \sum_{\substack{m,n \\ n\neq k}} \alpha_{m0}\, \alpha_{n0}\, z_{m} \cdot z_{n}, 
\end{align}
which coincides with the full $\mathcal{F}$ for tree-(fat)graphs with fixed tadpole.

Consequently, the counter-terms $\mathcal{A}^{(i)}_{1,n}$ becomes non zero precisely for the cones that contain a single loop curve $C_{i0}$, i.e. the \emph{tadpole cones}. Equivalently, the $s$-forest $\mathrm{sf}_{A}^k$
removes all tadpoles from the curve integral. This illustrates how the forest formula for curve integrals reproduces the standard renormalization mechanism familiar from Feynman graphs.
Consider the simple example for two-loop two-point. In this case 
\begin{align}
	\mathrm{SF}=\{\mathrm{sf}_{A}^{1,2},\mathrm{sf}_{A}^1,\mathrm{sf}_{A}^2\} \, .
\end{align} 
Now let's write the counter-term for this fatgraph
\begin{align}
	{\cal A}_{1, 2}^{R} = {\cal A}_{1, 2} (\{ X_{10}, X_{20}, X_{11}, X_{22}\})\, &-\sum_{i < j} {\cal A}_{1,2}^{(1,2)}(X_{C_{10}},X_{C_{20}}, X_{11}, X_{22})\Big{\vert}_{p^2=m^2}\nonumber\\
	&- \sum_{i} {\cal A}_{1,2}^{(1)}(X_{C_{10}}\Big{\vert}_{\epsilon^2}, X_{C_{20}},X_{11},X_{22})\nonumber\\
	&- \sum_{i} {\cal A}_{1,2}^{(2)}(X_{C_{10}}, X_{C_{20}}\Big{\vert}_{\epsilon^2},X_{11},X_{22})\, .
\end{align}
The last two counterterms are activated in the case of tadpole cones and serve to cancel them.  
Interested readers may verify this explicitly by working out the contribution from the two tadpole cones.

%%%%%%%%%%%%%%%%%
\subsubsection{Two-loop two-point}
\begin{figure}[h!]
	\centering
	\includegraphics[width=0.8\linewidth]{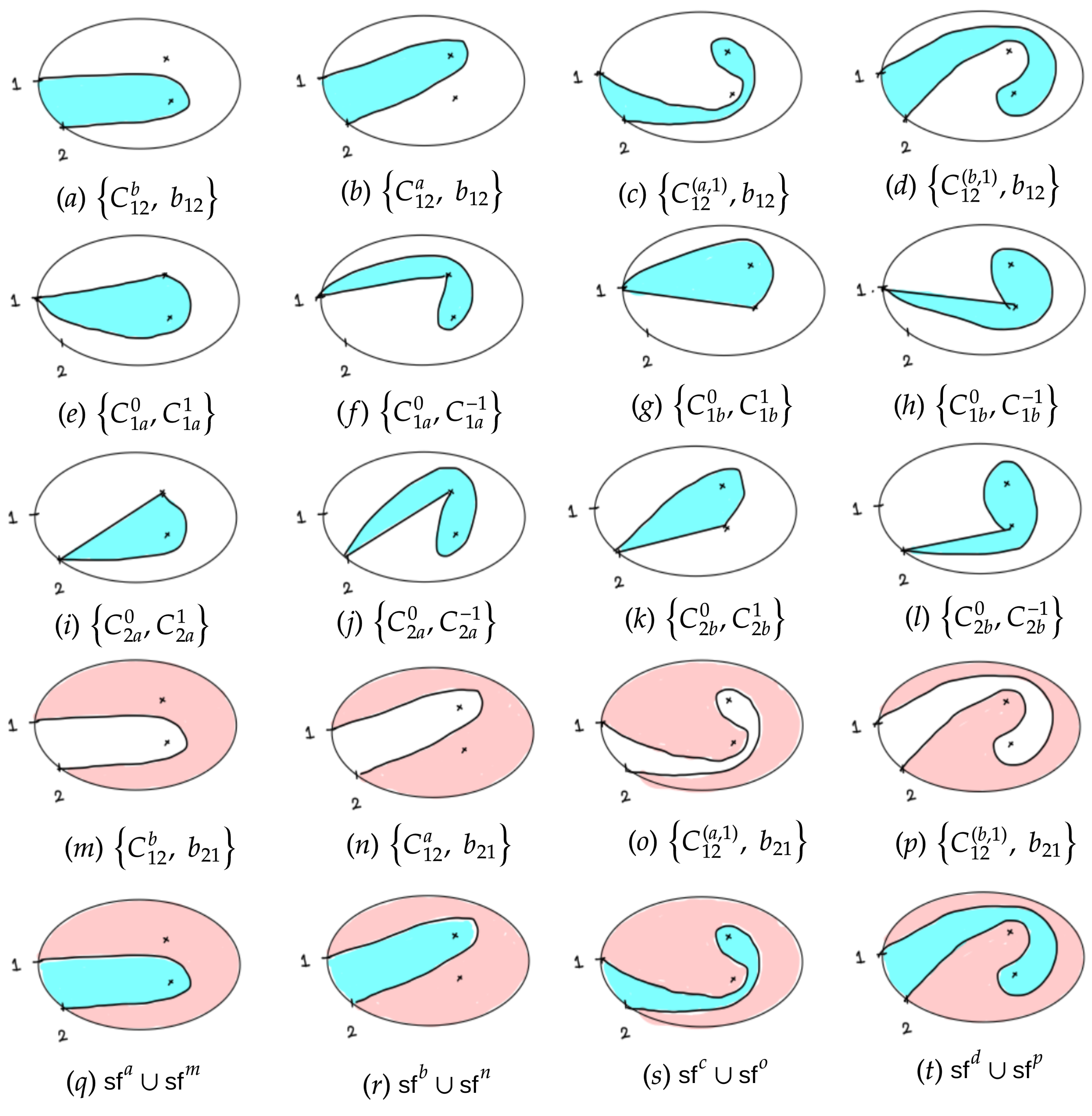}
	\caption{2-gon surface forest of two-loop two-point fatgraph}
	\label{fig:2l_2p_sf}
\end{figure}
Now let us turn to the two-point, two-loop case. The corresponding curve integral can be written as
\begin{align}
	\mathcal{A}_{2,2}
	= \int_{\mathbb{R}^5} 
	dt_x\,dt_w\,dt_z\,dt_y\,dt_1 \;
	\mathcal{K}\;\frac{1}{\mathcal{U}^2}\;
	\exp\!\left[-\frac{\mathcal{F}}{\mathcal{U}} - \mathcal{Z}\right].
\end{align}
Let us restrict attention to surface forests built from $2$–gons.  
In this case, there are $20$ such surface forests (see Fig.~\ref{fig:2l_2p_sf}), 
which renders the computation of the full counterterm for the curve integral 
rather cumbersome. Instead of summing over all surface forest contributions, 
we focus on the contribution from a specific triangulation of the surface and 
demonstrate that it reproduces the result of the forest formula for the 
corresponding Feynman graph (see Fig.~\ref{fig:2l_log^2_triang}). 

\begin{figure}[h!]
	\centering
	\includegraphics[width=0.5\linewidth]{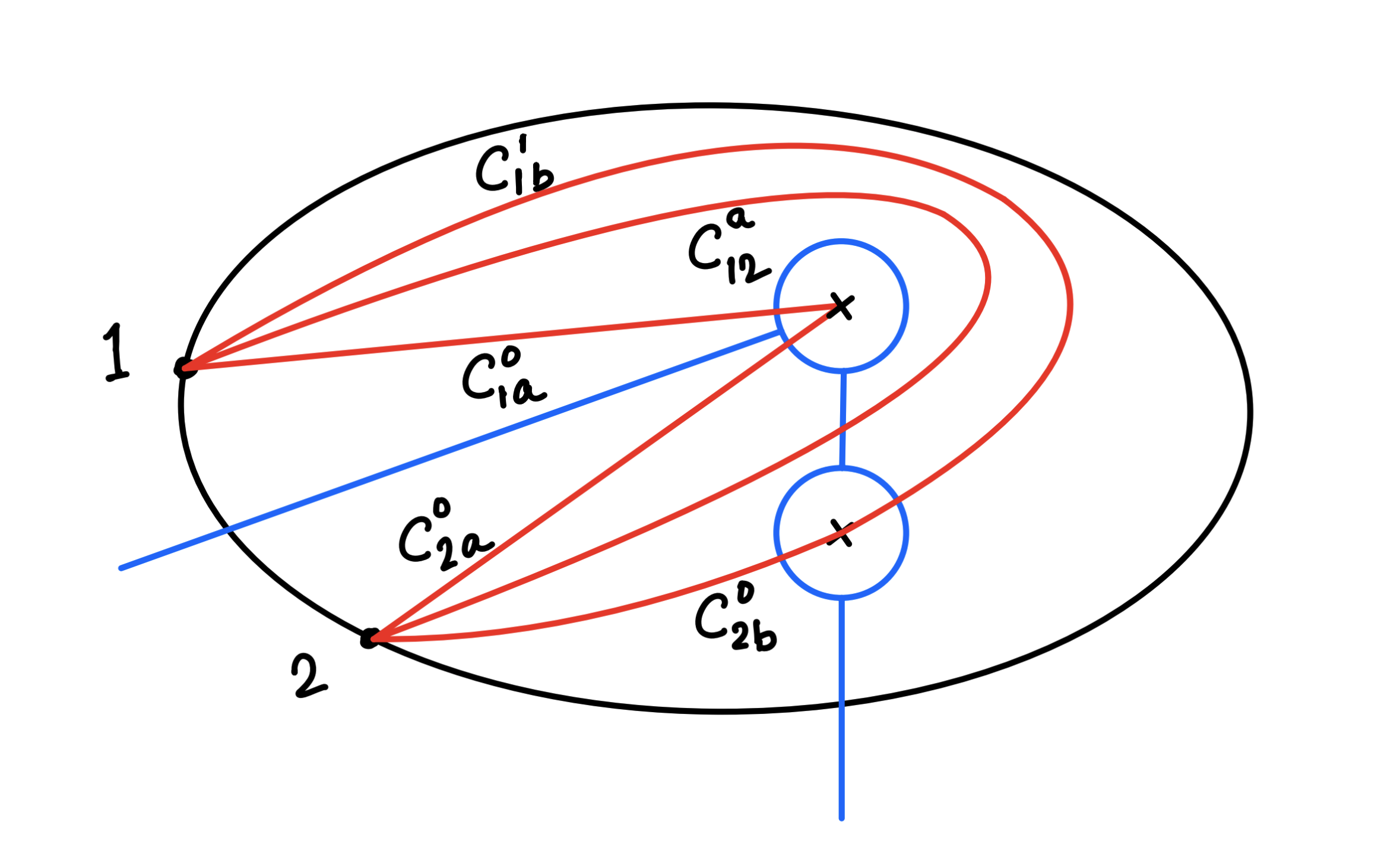}
	\caption{Triangulation $\log^2$ for the two-loop two-point fatgraph.}
	\label{fig:2l_log^2_triang}
\end{figure}

The set of curves making up this triangulation is
\[
\{\,C^0_{1a},\; C^1_{1b},\; C^0_{2a},\; C^0_{2b},\; C^a_{12}\,\}.
\]
Inside this cone, co-ordinates parallel to the corresponding $g$-vectors are more suitable for computations. These co-ordinates can be found using the following transformation
\begin{align}
	\vec{X}=a_1\, g^0_{1a}+a_2\, g^1_{1b}+a_3\, g^0_{2a}+a_4\, g^0_{2b}+a_5\, g^a_{12}, \quad \forall \ a_i>0
\end{align}
where $\vec{X}=\{t_x,t_w,t_z,t_y,t_1\}$ and $g$'s are the associated g-vectors with the curve. In this cone the headlight functions $\alpha_{C}$ and integral in Schwinger variables\footnote{Note that the cone we consider is not exactly the two-loop, two-point Schwinger parametrized amplitude. There is a extra variable $\frac{a_{123}}{a_{1234}}$ that comes due to  Mirzakhani kernel $\mathcal{K}$. We show in appendix Appendix \ref{ricones} that all $log^2$ cones add up to \emph{two} times the $\log^2$ graph.} gets simplified as,
\[
\alpha_{C'} = 
\left\{
\begin{aligned}
	& a_1 && \text{for } C' = C^0_{1a}, \quad
	a_2 && \text{for } C' = C^1_{1b}, \quad
	a_3 && \text{for } C' = C^0_{2a}, \\
	& a_4 && \text{for } C' = C^0_{2b}, \quad
	a_5 && \text{for } C' = C^a_{12}, \quad
	0   && \text{otherwise.}
\end{aligned}
\right.
\]
\begin{align}
	\mathcal{A}_{2,2}(\mathcal{C})=\int\prod_{i=1}^5da_i\ \frac{a_{123}}{a_{1234}} \frac{1}{a_{12}a_{34}} \exp\Big[-p_1^2\frac{a_1a_2}{a_{12}}-p_1^2\frac{a_3a_4}{a_{34}}-m^2a_{1234}+(p_1^2+m^2)a_5\Big],
\end{align}
where $a_{i\dots j} : =\, a_i+\dots+a_j$.
Notice that the curves related to the loop propagators are $\{C^{0}_{1a},C^0_{2a}\}$ and $\{C^0_{2b},C^{1}_{1b}\}$ are enclose in the s-forests $\{\mathrm{sf}^b,\mathrm{sf}^n,\mathrm{sf}^r\}$ (see Fig. \ref{fig:2l_2p_sf} for a check). Thus,
\begin{align}
	\mathcal{G}_{\mathrm{sf}^i}=
	\begin{cases}
		1, & i\in\{b,n,r\},\\
		0, & \text{otherwise}.
	\end{cases}
\end{align}
Hence the counter-terms using s-forest are, 
\begin{align}
	\mathcal{A}^R_{2,2} (\mathcal{C})=	\mathcal{A}_{2,2} (\mathcal{C}) -\mathcal{A}_{2,2}(\mathcal{C})\bigg{\vert}_{\mathrm{sf}^b}-\mathcal{A}_{2,2}(\mathcal{C})\bigg{\vert}_{\mathrm{sf}^n}+\mathcal{A}_{2,2}(\mathcal{C})\bigg{\vert}_{\mathrm{sf}^b\,\cup\,\mathrm{sf}^n}
\end{align}
which in the particular cone of interest evaluates to,
\begin{align}
	\mathcal{A}^R_{2,2}(\mathcal{C})=	\mathcal{A}_{2,2} (\mathcal{C}) &-\int \prod_{i=1}^5da_i \ \frac{a_{123}}{a_{1234}} \frac{1}{a_{12}a_{34}}\ \exp[-\mu^2\frac{a_1a_2}{a_{12}}-p_1^2\frac{a_3a_4}{a_{34}}-m^2a_{1234}+(p_1^2+m^2)a_5]\nonumber\\
	&-\int \prod_{i=1}^5da_i\ \frac{a_{123}}{a_{1234}} \frac{1}{a_{12}a_{34}}\ \exp[-p_1^2\frac{a_1a_2}{a_{12}}-\mu^2\frac{a_3a_4}{a_{34}}-m^2a_{1234}+(p_1^2+m^2)a_5]\nonumber\\
	&+\int \prod_{i=1}^5da_i\ \frac{a_{123}}{a_{1234}} \frac{1}{a_{12}a_{34}}\ \exp[-\mu^2\frac{a_1a_2}{a_{12}}-\mu^2\frac{a_3a_4}{a_{34}}-m^2a_{1234}+(p_1^2+m^2)a_5]\, .
\end{align}

This is precisely the forest formula required for a two-loop $\log^2$ divergent graph in four dimensions. Thus, the naive forest formula indeed captures the correct counter-term structure cone by cone. As in the one-loop case, finite counter-terms also appear for finite graphs (e.g. overlapping divergences in 4D two-point two-loop case). 

We end this section with a few remarks.
\begin{itemize}
	\item Although our focus in this paper is on $\textrm{Tr}(\Phi^{3})$ theory in $D = 4$ dimension, we believe that the surface forest formula can be defined for theories with overlapping divergences. This is because set of all surface forests are in injection with set of all maximal complete forests defined over all 1PI graphs with fixed $n, L$. 
	\item The fact that forest formula cancels UV divergences for any 1PI Graph can be proved using UV factorization property of  the graph Symanzik polynomials, \cite{Brown:2022oix}.  It has been shown by A. Suthar  that the factorization properties of graph Symanzik polyomials can be generalized to surface Symanzik polynomials, \cite{Amit-up}. It will be interesting to see if one can prove the finiteness of surface forest formula using these factorization formulae. 
	\item The fact that tadpole graphs cancel in kinematic renormalization scheme, motivate us to restrict the curve integral to the region of  global Schwinger space which do not include any tadpoles that lead to quadratic divergence. As we show below, these tadpole free regions can be renormalized without recourse to surface forests.
\end{itemize}
We would once again like to emphasize however that this approach to renormalization is not promising from the perspective of the positive geometry program. However it does produce a concrete subtraction scheme for the curve integral and it is thus reassuring that by using essential idea of Zimmermann forest formula, we can renormalize curve integral even without it's explicit decomposition into Feynman graphs.

%%%%%%%%%%%%%%%%%%%%%%%%%%%%%%%%
%%%%%%%%%%%%%%%%%%%%%%%%%%%%%%%%
\section{Decapitating tadpoles from the curve integral.}\label{dtci}
We now come to the main body of the paper and show how curve integrals admit a separation of UV singularities which can then be renormalized by suitable counter terms. 
We first start wih a simple observation. In $\textrm{Tr}(\Phi^{3})$ theory, the amplitude admits  the following perutrbative expansion,   
\begin{align}
{\cal A}_{L, n}\, =\, \sum_{p,q\geq 0 \, \vert \,  p + q  \leq\, L}\, (\Lambda^{2})^{p}\, (\ln\Lambda)^{q}\,  {\cal A}_{L, n}(p,q,\Lambda).
\end{align}
$\Lambda$ is the UV cut-off that regularizes the amplitude. Each term in the summand is associated to a (sum over) Feynman graphs and is classified according to the number of tadpoles or bubbles that it posses as subgraphs. The term corresponding to $p+q = 0$ is the sum over $n$ point $L$-loop Feynman graph which is superficially convergent with no divergent subgraphs (the  divergent subgraphs are shown in  Fig.  \ref{fig:bubble-tadpole}). We illustrate this decomposition in the case of $L=2,\, n=3$ amplitude below (see Fig.  \ref{fig:div-three-point-1loop}).

\begin{figure}
	\centering
	\begin{tikzpicture}
		\begin{scope}[shift={(3,0)}, rotate = 90]
			% Add tadpole in the gap
			\draw[thick, black] (1,-0.2) -- ++(0,0.4); % Perpendicular line
			\draw[black] (-0.2 +1,-0.2) --(0.2 +1,-0.2); % Perpendicular line
			
			\draw[thick, black] (1, 1) --(1,1.4); % Perpendicular line
			\draw[black] (-0.2 +1,1.4) --(0.2 +1, 1.4); % Perpendicular line
			
			\draw[thick, black] (1, 0.6) circle (0.4cm); % Circle on top
		\end{scope}
		
		\begin{scope}[shift={(6,0.6)}]
			% Add tadpole in the gap
			\draw[thick, black] (1,-0.2) -- ++(0,0.4); % Perpendicular line
			\draw[black] (-0.2 +1,-0.2) --(0.2 +1,-0.2); % Perpendicular line
			\draw[thick, black] (1, 0.6) circle (0.4cm); % Circle on top
		\end{scope}
	\end{tikzpicture}
\caption{Only two divergent subgraphs for $\phi^3$-theory in $D=4$ : `bubble' and `tadpole' }
\label{fig:bubble-tadpole}
\end{figure}
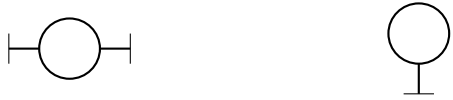

\begin{figure}
	\centering
	\begin{tikzpicture}
		% Diagram 3: Circle on C-A arm (240° rotation) - Now Picture 1
		\begin{scope}[shift={(0,1.75)}, rotate=240]
			% Define the vertices of the triangle
			\coordinate (A) at (0,0);
			\coordinate (B) at (2,0);
			\coordinate (C) at (1,1.732); % Forms an equilateral triangle
			
			% Draw the triangle arms, with a gap on A-B
			\draw[thick, black] (A) -- (0.7,0); % First part of A-B
			\draw[thick, black] (1.3,0) -- (B); % Second part of A-B
			\draw[thick, black] (B) -- (C) -- (A);
			
			% Add small circle in the gap
			\draw[thick, black] (1,0) circle (0.3cm);
			
			% Add small lines extending outward from each vertex
			\draw[thick, black] (A) -- ++(-0.5,-0.35);
			\draw[thick, black] (B) -- ++(0.5,-0.35);
			\draw[thick, black] (C) -- ++(0,0.5);
		\end{scope}
		
		% Diagram 2: Circle on B-C arm (120° rotation) - Remains Picture 2
		\begin{scope}[shift={(6,0)}, rotate=120]
			% Define the vertices of the triangle
			\coordinate (A) at (0,0);
			\coordinate (B) at (2,0);
			\coordinate (C) at (1,1.732); % Forms an equilateral triangle
			
			% Draw the triangle arms, with a gap on B-C (rotated to align with A-B gap)
			\draw[thick, black] (A) -- (0.7,0); % First part of A-B
			\draw[thick, black] (1.3,0) -- (B); % Second part of A-B
			\draw[thick, black] (B) -- (C) -- (A);
			
			% Add small circle in the gap
			\draw[thick, black] (1,0) circle (0.3cm);
			
			% Add small lines extending outward from each vertex
			\draw[thick, black] (A) -- ++(-0.5,-0.35);
			\draw[thick, black] (B) -- ++(0.5,-0.35);
			\draw[thick, black] (C) -- ++(0,0.5);
		\end{scope}
		
		% Diagram 1: Circle on A-B arm (0° rotation) - Now Picture 3
		\begin{scope}[shift={(9,0)}]
			% Define the vertices of the triangle
			\coordinate (A) at (0,0);
			\coordinate (B) at (2,0);
			\coordinate (C) at (1,1.732); % Forms an equilateral triangle
			
			% Draw the triangle arms, with a gap on C-A (rotated to align with A-B gap)
			\draw[thick, black] (A) -- (0.7,0); % First part of A-B
			\draw[thick, black] (1.3,0) -- (B); % Second part of A-B
			\draw[thick, black] (B) -- (C) -- (A);
			
			% Add small circle in the gap
			\draw[thick, black] (1,0) circle (0.3cm);
			
			% Add small lines extending outward from each vertex
			\draw[thick, black] (A) -- ++(-0.5,-0.35);
			\draw[thick, black] (B) -- ++(0.5,-0.35);
			\draw[thick, black] (C) -- ++(0,0.5);
		\end{scope}
	\end{tikzpicture}
	
	\vspace{0.7cm}

	\begin{tikzpicture}
		
		% Diagram 3: Tadpole on C-A arm (240° rotation) - Now Picture 1
		\begin{scope}[shift={(0,1.75)}, rotate=240]
			% Define the vertices of the triangle
			\coordinate (A) at (0,0);
			\coordinate (B) at (2,0);
			\coordinate (C) at (1,1.732); % Forms an equilateral triangle
			
			% Draw the triangle arms, 
			\draw[thick, black] (A) -- (B) -- (C) -- (A);
			
			% Add tadpole in the gap
			\draw[thick, black] (1,0) -- ++(0,- 0.3); % Perpendicular line
			\draw[thick, black] (1,- 0.6) circle (0.3cm); % Circle on top
			
			% Add small lines extending outward from each vertex
			\draw[thick, black] (A) -- ++(-0.5,-0.35);
			\draw[thick, black] (B) -- ++(0.5,-0.35);
			\draw[thick, black] (C) -- ++(0,0.5);
		\end{scope}
		
		% Diagram 2: Tadpole on B-C arm (120° rotation) - Remains Picture 2
		\begin{scope}[shift={(6,0)}, rotate=120]
			% Define the vertices of the triangle
			\coordinate (A) at (0,0);
			\coordinate (B) at (2,0);
			\coordinate (C) at (1,1.732); % Forms an equilateral triangle

			% Draw the triangle arms, 
			\draw[thick, black] (A) -- (B) -- (C) -- (A);
			
			% Add tadpole in the gap
			\draw[thick, black] (1,0) -- ++(0,- 0.3); % Perpendicular line
			\draw[thick, black] (1,- 0.6) circle (0.3cm); % Circle on top

			% Add small lines extending outward from each vertex
			\draw[thick, black] (A) -- ++(-0.5,-0.35);
			\draw[thick, black] (B) -- ++(0.5,-0.35);
			\draw[thick, black] (C) -- ++(0,0.5);
		\end{scope}
		
		% Diagram 1: Tadpole on A-B arm (0° rotation) - Now Picture 3
		\begin{scope}[shift={(9,0)}]
			% Define the vertices of the triangle
			\coordinate (A) at (0,0);
			\coordinate (B) at (2,0);
			\coordinate (C) at (1,1.732); % Forms an equilateral triangle
			
			% Draw the triangle arms, 
			\draw[thick, black] (A) -- (B) -- (C) -- (A);
			
			% Add tadpole in the gap
			\draw[thick, black] (1,0) -- ++(0,- 0.3); % Perpendicular line
			\draw[thick, black] (1,- 0.6) circle (0.3cm); % Circle on top
			
			% Add small lines extending outward from each vertex
			\draw[thick, black] (A) -- ++(-0.5,-0.35);
			\draw[thick, black] (B) -- ++(0.5,-0.35);
			\draw[thick, black] (C) -- ++(0,0.5);
		\end{scope}
	\end{tikzpicture}
	
	\caption{Divergent contributions for $n=3$ at $L=2$ for $\phi^3$-theory  in $D=4$ }
	\label{fig:div-three-point-1loop}
	
\end{figure}
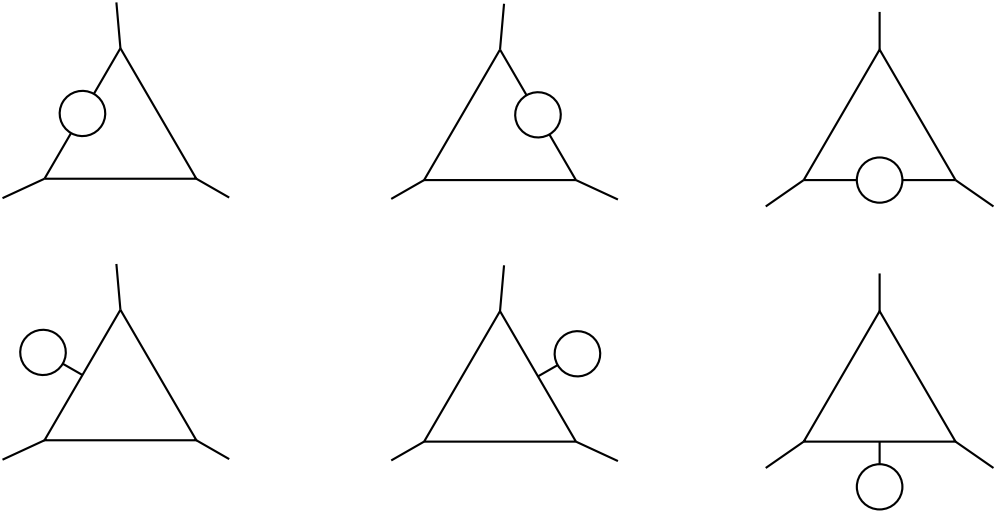

As we are in $D = 4$ dimensions, all the quadratic divergences are associated to $\Phi^{3}$ tadpoles in constituent Feynman diagrams. The physical amplitude should be obtained after decapitating the tadpoles. However as positive geometries obtained from dissection quivers of polygons with $L$ punctures in the interior treat all the planar graphs democratically, it includes all the tadpole contributions as seen in section \ref{sffpa}. Hence, we first define a new curve integral by removing the tadpole contributions. This then leads to a bare amplitude with only log divergences in the momentum space.  In other words, our objects of interest in this paper are renormalization of the tadpole free amplitudes,
\begin{align}
{\cal A}_{L, n}^{\textrm{tf}}(\Lambda)\, =\, \sum_{0 \leq\, q\, \leq\, L}\, {\cal A}^{p = 0,q}_{L, n}\, (\ln\Lambda)^{q} 
\end{align}

The elimination of tadpoles in the curve integral formula can be done by isolating those regions in the global Schwinger space which do not contain cones dual to tadpole graphs. A detailed analysis of this region is at the heart of our approach to renormalize curve integral formula. 

We refer to this region as ${\cal M}_{L, n}^{\textrm{tf}}$  labelled by the number of punctures $L$,  number of external points $n$ and where `$\textrm{tf}$' stands for \emph{``tadpole-free''}. Our renormalized curve integrals are integrable functions defined over ${\cal M}_{L, n}^{\textrm{tf}}$ whose ``projection'' on each Feynman graph leads to parametric forest formula which has been studied in \cite{Appelquist:1969iv, Brown:2011pj}.  \\ %\cite{appelquist, Kreimer 2011}  

\emph{From now on, we will denote ${\cal A}_{1, n}^{\textrm{tf}}$ as ${\cal A}_{1,n}$ throughout this section.}

%%%%%%%%%%%%%%%%%%%%%%%%%%%%%%%%
%%%%%%%%%%%%%%%%%%%%%%%%%%%%%%%%
\section{Parametric renormalization of planar one-loop amplitudes}\label{sec:para_reno_1l_planar}
The surface forest formula derived in the previous section implicitly relies on the Zimmermann's forest formula. To renormalize any $n$-point amplitude with $L$ loops,  it sums over counter-terms indexed by the set of all possible forests which partially triangulate $\Sigma_{L,n}$. Clearly this approach takes us back to the world of Feynman graphs, and such takes life out of the world of positive geometries which generates bare amplitudes that exist regardless of their decomposition into Feynman integrals.

We thus look for a ``generalization'' of the forest formula that directly works on the moduli space of $L$ punctured disc with $n$ marked points without recourse to the additional structures like s-forests. As we prove below, the answer lies in a careful analysis of the UV singularities of the curve integral all of which ``factor through'' the log divergence of a 2-point bubble in this case.  Let us summarize the key results of this section before going into the details that can be found in the subsequent sections.  To help navigate the reader through this section, we outline the flavor of the underlying ideas by highlighting how we renormalize planar one-loop amplitude. 

In a nutshell, this is what we observe. Starting with the reference pseudo-triangulation dual to the tadpole fat graph, we will show that the divergence in the curve integral is a sum of two terms.

%%%%%%%%%%%%%%%%%%%%%%%%%%%%%%%%
\subsection{Analyzing tadpole free region in global Schwinger space for a disc with $L=1$ punctures}

In this section we derive the tadpole free region for $n$-point planar one-loop amplitudes.  As explained in sec. \ref{sec:curve-int-formula}, the headlight functions are computed with respect to the reference fat graph which is obtained by attaching an $n$-point tree-graph to a one-point tadpole.  We now restrict the intergrand of ${\cal A}_{1, n}^{\textrm{planar}}$  to ${\cal M}_{1, n}^{\textrm{tf}}$. As discussed previously, this is the region inside tropical schwinger parameter space where $\alpha_{ii}(t_{1}, \dots, t_{n-1}, t_{x})\, =\, 0, \, \forall\, i$ and contributes a logarithmically divergent term to the curve integral, 
\begin{align}
{\cal A}_{1, n}^{\textrm{tf}}(\Lambda)\bigg|_{{\cal M}_{1, n}^{\textrm{tf}}}\, =\, \sum_{q=0}^{1}\, (\ln\Lambda)^q \,  {\cal A}_{1, n}\, .
\end{align}

Before we start explaining the detailed renormalization procedure, below is a road-map for the readers on how we  approach it. 
\begin{itemize}
\item Inside the global schwinger parameter space ${\bf R}^{n-1}$ , once we can locate the subspace  ${\cal M}_{1, n-1}^{\textrm{tf}}$ such that 	$\alpha_{ii}\, =\, 0,\,  \forall\, i.$, we can recursively construct ${\cal M}_{1, n-1}^{\textrm{tf}} \in {\bf R}^{n}$ which is union of $n-1$ regions, ${\cal M}_{1, n}^{\textrm{tf}} =\, \displaystyle \bigcup_{I=1}^{n-1}\, \mathcal{S}^{\mathit I}$. Each $\mathcal{S}^{\mathit I}$ is union over cones corresponding  to a tadpole-free one-loop planar fat graph.  Evidently only the cones dual to graphs with a bubble sub-graphs need renormalization. 

\item Then we make the following observation. Inside each $S_{I}$ only a certain  headlight functions associated to curves that end at the puncture are non-trivial, 
\begin{align}
\alpha_{_{K0}}\, \neq\, 0 , \ \forall\, 1\, \leq\, K\, \leq\, n-(I-1), \nonumber
\end{align}
and all other such headlight functions  \emph{viz.}	$	\alpha_{_{K0}} $ vanish point-wise inside $\mathcal{S}^{\mathit I}$. Thus each $\mathcal{S}^{\mathit I}$ contains a union of cones where the number of such cones decrease monotonically as $I$ increases from $1$ to $n-1$. We use this fact to \emph{directly} write the renormalized planar amplitude without renormalizing \emph{individual} cones. 

\item Next we propose the following renormalized amplitude with the help of ``tropical counter-term'' ${\cal T}_{1, n}$, 
\begin{align}
{\cal A}_{1, n;\textrm{R}}\, :=\, \lim_{\epsilon \rightarrow\, 0} \bigg[\, {\cal A}_{1,n}(S, \{\Theta\}, \epsilon) - {\cal T}_{1, n}(S, \Theta, S_{0}, \epsilon)\, \bigg]
\end{align}
This can be realized  to be coming from  a ``modified curve integral formula'' (${\cal A}_{1, n}^{R}$) in \emph{kinematic scheme} where the subtracted term is ``curve integral version'' of counter-term in parametric forest formula. 
\item Next we show that the ``tropical counter-term'' ${\cal T}_{1, n}$ depends on two polynomials 
$\Psi_{0}$ and $\Psi_{1}$ of planar kinematic variables and most crucially the exponent of the curve integrand is a linear function of $t_{x}$,
\begin{align}\label{eq:linear-tx}
\bigg(\frac{{\cal F}}{{\cal U}} - {\cal Z}\bigg)= t_{x} \Psi_{1\, I} + \Psi_{0\, I}.
\end{align}

\item Finally using \eqref{eq:linear-tx} and being particularly careful with the  integration domain in Schwinger parameter space we show that  ${\cal A}_{1, n}^{R}$ is  indeed finite. We will demonstrate this procedure in several  lower point examples. 
\end{itemize}

Below we provide the details of the steps summarized above. \\

To compute ${\cal A}_{1, n}$, we need to first determine ${\cal M}_{1, n}^{\textrm{tf}}$.  This is the content of the following lemma. 

\begin{lemma}\label{lemma:tf_region_1l}
Consider a global Schwinger space ${\bf R}^{n-1}$ co-ordinatized by $(t_{n-2}, \dots, t_{1}, t_{x}\, \leq\, 0)$. Let ${\cal M}_{1, n-1}^{\textrm{tf}}$ be a subspace inside this space which is defined as a vanishing loci of the following set of headlight functions.\footnote{The headlight function, $\alpha_{ii}$ can be computed using ancillary file \texttt{`One-loop tadpole headlight functions.nb'}.}\begin{align}
\alpha_{ii}(t_{n-1}, \dots, t_{1}, t_{x})\, =\, 0,\, \quad \forall\, i.
\end{align}
Then, 
${\cal M}_{n}^{\textrm{tf}}$ is a region inside ${\bf R}^{n}$ which has the following stratification. 
\begin{align}
{\cal M}_{1, n}^{\textrm{tf}}\, =\, \bigcup_{I=1}^{n-1}\, {\cal S}^{I}
\label{tf_region}
\end{align} 
That is, ${\cal M}_{1, n}^{\textrm{tf}}$ is union of $n-1$ regions, $\mathcal{S}^{\mathit I}\, \big{|} 1\, \leq\, I\, \leq n-1$. 	%\textcolor{magenta}{\it [Better notation : $S^{I}_{L, n}$]}
\begin{align}
\mathcal{S}^{\mathit 1}\, &=\, (\, 0\, \leq\, t_{n-1}\, \leq\, - t_{x}\, )\nonumber\\
\mathcal{S}^{\mathit 2}\,&=\, \bigg(t_{n-1} + t_{x} > 0\bigg)\, \bigcap\, \bigg(0\, \leq\, t_{n-2} + t_{n-1}\, \leq\, - t_{x} \bigg) \nonumber\\
\mathcal{S}^{\mathit 3}\,&=\, \bigg(t_{n-1} + t_{x} > 0 \bigg) \, \bigcap\, \bigg(t_{n-2} + t_{n-1} + t_{x}\, >\, 0 \bigg) \bigcap\, \bigg( 0 \leq t_{n-3} + t_{n-2} + t_{n-1} \leq - t_{x} \bigg) \nonumber\\
&\vdots\nonumber\\
\mathcal{S}^{\mathit {n-1}}\,&=\, \bigg(\sum_{j=I}^{n-1} t_{j} \geq\, -t_{x}\, \vert 2 \leq\, I\, \leq\, n-1 \bigg) \, \bigcap \bigg(0\, \leq\, \sum_{I=1}^{n-1} t_{I} \leq\, -\, t_{x} \bigg)
\end{align}
\end{lemma}
\begin{proof}
Our proof is via induction and goes in two steps.\\
We start with the reference  $n$ point  tadpole one-loop fat graph $\Gamma_{n}$ with vertices $\{1, \dots, n\}$. The ccorresponding tropical co-ordinates are,  $\{t_{1}, \dots, t_{n-1}, t_{x} < 0\}$. The underlying space is ${\bf R}^{n}\, \times\, {\bf R}_{-}$ that we will denote as ${\bf R}^{n,-}$. 

A direct computation first reveals that, inside the region parametrized by the following set of inequalities, 
\begin{align}
\sum_{j}^{n-1} t_{i} \geq\, 0\, \forall\, i-1\, \leq\, j\, \leq\, n-1
\end{align}
a subset of headlight functions, $\{\alpha_{jj}\}_{j=i}^{n}$ vanish. Clearly this region is then independent of $\{t_{1}, \dots, t_{i-2}\}$ and as a result it is independent of $n$. This is because one could attach a tree-graph of arbitrary valence to the left of $\Gamma_{n}$ and the result will not be affected. 

To now prove that inside any ${\cal S}^{I}$, all the headlight functions vanish. We proceed by induction. So let ${\cal S}^{j}$ for a fixed $1\, \leq\, j\, \leq\, n-1$ be a region in ${\bf R}^{n-1, -}$ such that  $\alpha_{ii}(t_{1}, \dots, t_{n-1}, t_{x}) \forall\, 1 \leq\, i, \leq\, n$ vanish in ${\cal S}^{j}$. Let us now consider $\Gamma_{n+1}$ obtained by adding a three point vertex to the left of $\Gamma_{n}$ such that the resulting fat graph has two new vertices $\{1^{+}, 1^{-}\}$ that replace the vertex labelled as $1$ in $\Gamma_{n}$. Our induction hypothesis then implies that,
\begin{align}
\alpha_{ii}(t_{0}, t_{1}, \dots, t_{n-1}, t_{x})\, =\, 0\, \forall\, 2 \leq\, i\, \leq\, n.
\end{align}
That is, headlight functions associated to all but two of the tadpole curves in $\Gamma_{n+1}$ vanish. It thus remains to show that inside ${\cal S}^{j}$
\begin{align}
\alpha_{mm}\, =\, 0\, \forall\, m\, \in\, \{1^{-}, 1^{+}\}\, 
\end{align}

The space spanned by tropical parameters is now ${\bf  R}^{n+1,-}$ which is co-ordinatized by\\
$\{ t_{0}, t_{1}, \dots, t_{n-1}, t_{x} \leq\, 0\}$. Clearly ${\cal S}^{I}\, \subset\, {\bf R}^{n+1,-}$. We would now like to prove that $\alpha_{1^{+}1^{+}}, \alpha_{1^{-}1^{-}}$ continue to vanish inside ${\cal S}^{I}$.\\ 

We will first analyze $\alpha_{1^{-}1^{-}}$. The matrix associated to the this tadpole curve can be written as,
\begin{align}
M_{1^{-}1^{-}}\, =\, R\, \cdot\, Y_{0}\, \cdot\, L\, \cdot\, R^{-1}\, \cdot\,M_{22}\,\cdot\, L^{-1}\, \cdot\, R\, \cdot\, Y_{0}\, \cdot L
\end{align}
It can be immediately checked that
\begin{align}
R\, \cdot\, Y_{0}\, \cdot\, L\, \cdot\, R^{-1}&=\, \begin{bmatrix}  1 + y_{0} && -1 \\ y_{0} && y_{0} \end{bmatrix}\nonumber\\
L^{-1}\, \cdot\,  R\, \cdot\, Y_{0}\, \cdot\, L&=\, \begin{bmatrix}  1 + y_{0} && y_{0} \\ -1 && 0 \end{bmatrix} 
\end{align}
Hence, if we denote $M_{22}$ as
\begin{align}
M_{22}\, :=\, \begin{bmatrix} a && b \\ c && d \end{bmatrix}
\end{align}
Then, 
\begin{align}
M_{1^{-}1^{-}}\,& = \begin{bmatrix}
a(1 + y_{0})^{2} -  (b + c) (1 + y_{0})\, + d && a(1 + y_{0}) - c \\ (a + c) (1 + y_{0}) - (b + d) && (a + c)
\end{bmatrix}
\end{align}
$\alpha_{1^{-}1^{-}}$ can now be computed by tropicalizing each matrix entry. Let $t(a), t(b), t(c), t(d)$ be the tropical polynomials associated to $a, b, c$ and $d$ respectively. These polynomial satisfy a linear constraint,
\begin{align}\label{tropalpha22}
\alpha_{22} = 0\, \implies\, t(a) + t(d) = t(b) + t(c)
\end{align}
We can now show that 
\begin{align}
\textrm{for}\, \ t_{0} < 0,\ \alpha_{1^{-}1^{-}} = 0\, \quad  \textrm{irrespective of eqn.(\ref{tropalpha22}).}\nonumber\\
\textrm{for}\, \  t_{0} > 0,\ \alpha_{1^{-}1^{-}} = 0\, \quad  \textrm{iff eqn.(\ref{tropalpha22}) is satisfied.}
\end{align}
\end{proof}

It is important to emphasize that the number of the decapitated regions ${\cal S}^{I}$ grow linearly with $n$ as opposed to factorially and hence the number of graphs contained in \emph{at least} a subset of ${\cal S}^{\mathit I}$ must grow factorially with $n$.\footnote{In fact, we expect, based on analysis for $n=3, 4$,  that the number of Feynman graphs decrease monotonically with $\{ \mathcal{S}^{I}\}_{I = 1, \dots, n-1}$.}  That is, separate ${\cal S}^{I}$s do not contain $O(1)$ number of Feynman graphs as we increase $n$. 

It will be convenient for us to use the following co-ordinates on ${\cal M}_{1, n}^{\textrm{tf}}\,$,
\begin{align}\label{tigoingtotitilde}
(\, t_{1}, \dots, t_{n-1}, t_{x})\, \rightarrow\, (\, \widetilde{t}_{1}, \dots, \widetilde{t}_{n-1}, t_{x}).
\end{align}
where 
\begin{align}
\widetilde{t}_{i} = \sum_{j=i}^{n-1}\, t_{j} + t_{x} .
\end{align} 
${\cal M}_{1, n}^{\textrm{tf}}$ is thus defined by the union of the $n-1$ semi-algebraic sets, 
\begin{align}
{\cal M}_{1, n}^{\textrm{tf}}&=\, \bigcup_{I=1}^{n-1}\, \mathcal{S}^{I}\nonumber\\
\mathcal{S}^{I}&=\, \bigg\{\, (\widetilde{t}_{i}\, \geq\, 0\, \vert\,  n-1\, \geq\, i\, \geq\,n-I+1) \, \medcap\, ( t_{x}\, \leq\, \widetilde{t}_{n-I}\, \leq\, 0)\, \bigg\} .
\end{align}
Each $\mathcal{S}^{\mathit I}$ is union over cones such that each cone corresponds  to a tadpole-free one-loop planar fat graph.  Clearly only those cones which are dual to graphs with a bubble subgraph require renormalization. 
Following lemma makes the inclusion of various cones into $\mathcal{S}^{\mathit I}$  rather explicit. 

\begin{figure}[h!]
	\centering
	% Fig 1
	\begin{subfigure}[t]{0.3\textwidth}
		\centering
		\includegraphics[width=1.1\linewidth]{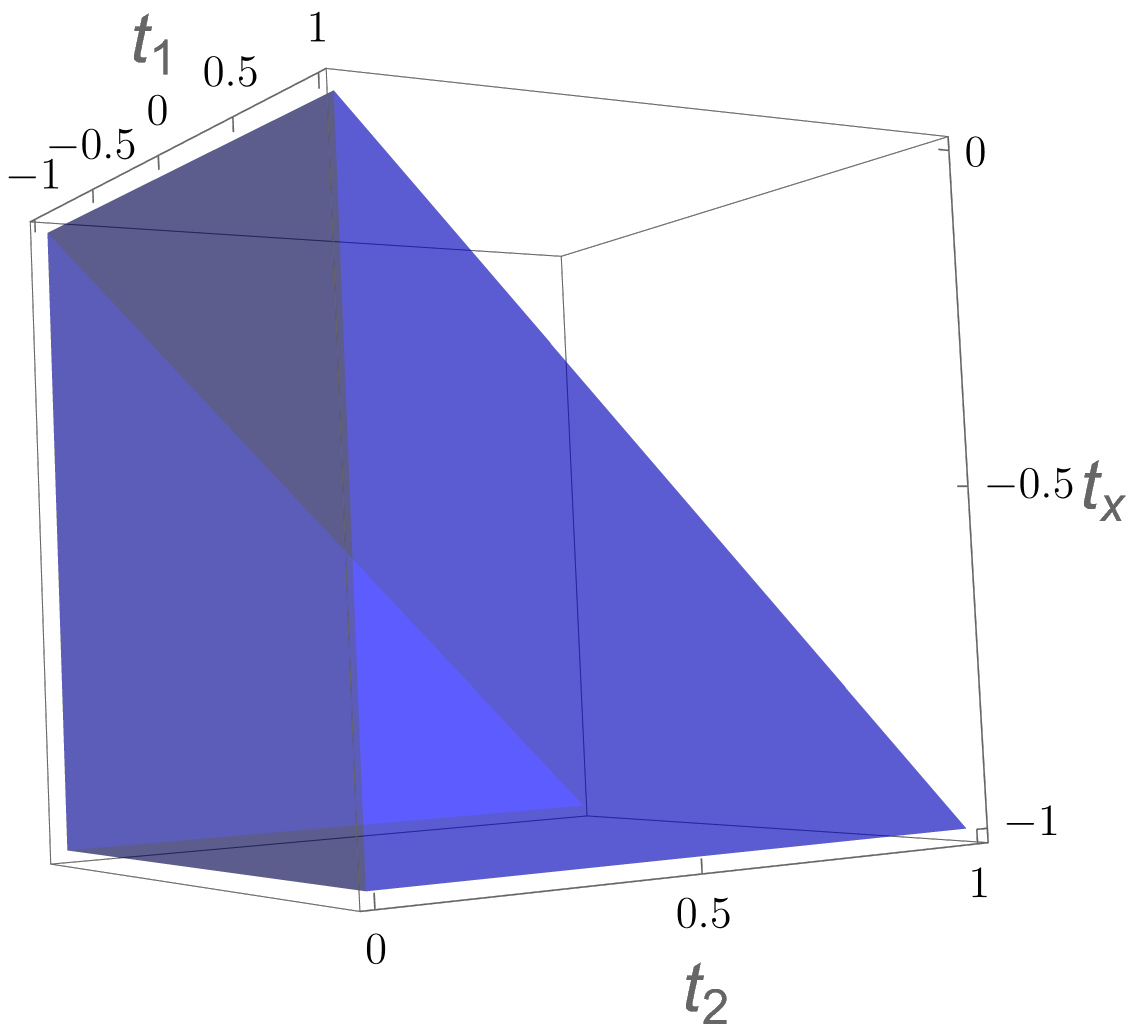}
		\caption{${\cal S}^1$ region for three-point one-loop}
		\label{fig:S1_regoon_3p1l}
	\end{subfigure}
	\hfill
	% Fig 2
	\begin{subfigure}[t]{0.3\textwidth}
		\centering
		\includegraphics[width=1\linewidth]{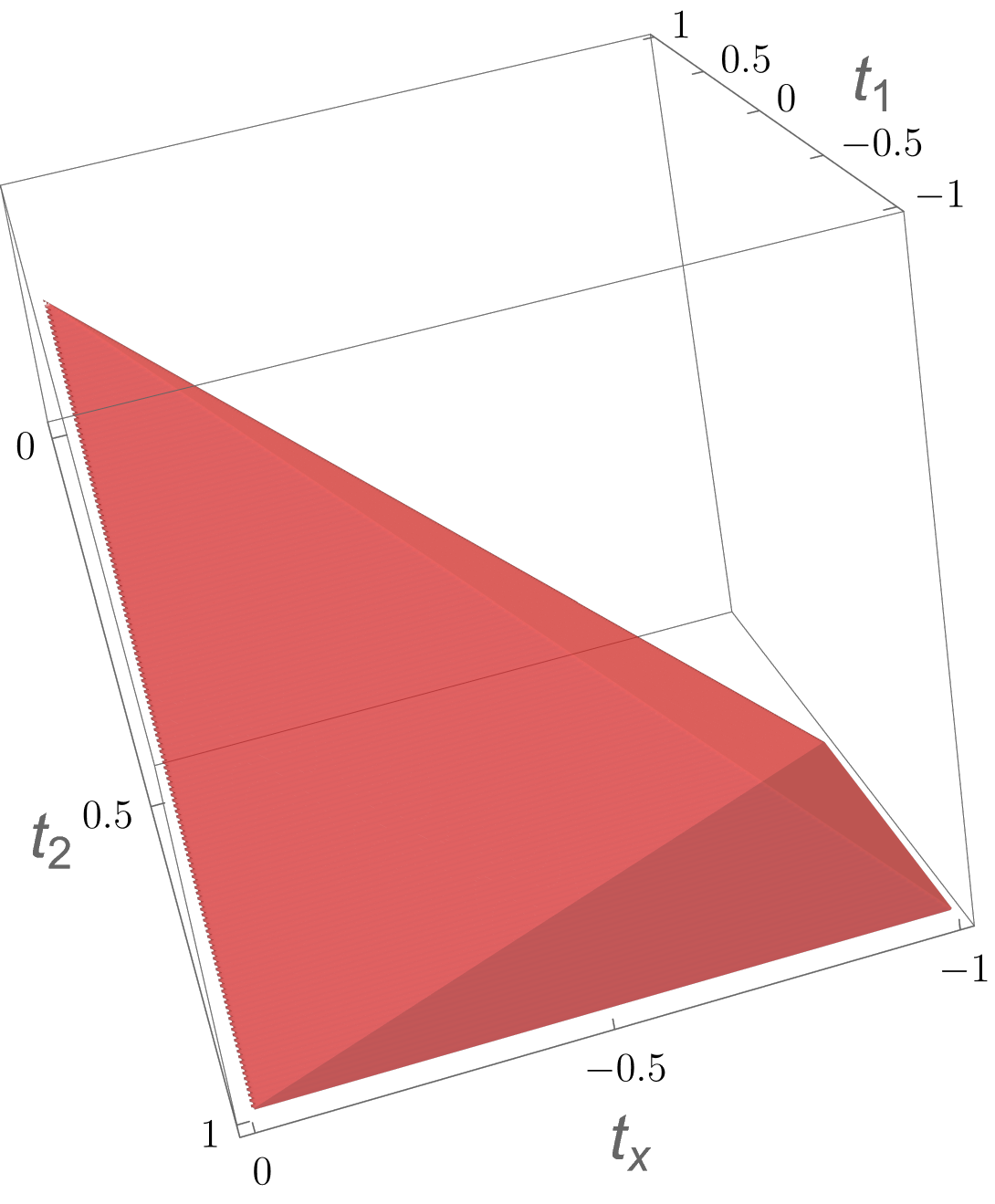}
		\caption{${\cal S}^2$ region for three-point one-loop}
		\label{fig:S2_regoon_3p1l}
	\end{subfigure}
	\hfill
	% Fig 3
	\begin{subfigure}[t]{0.3\textwidth}
		\centering
		\includegraphics[width=1.1\linewidth]{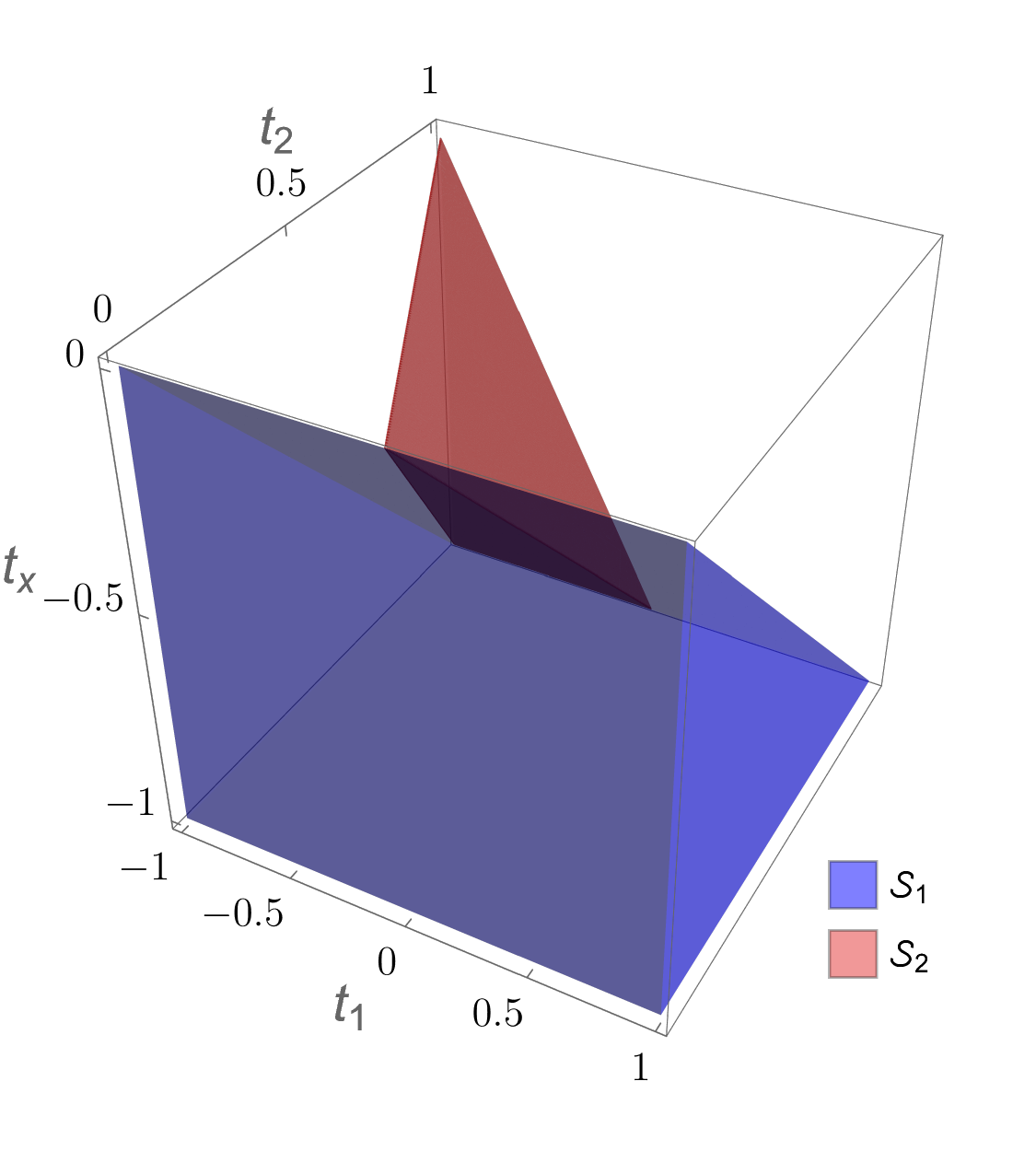}
		\caption{${\cal M}_{1,3}^{\textrm{tf}}$ region for three-point one-loop}
		\label{fig:S1_S2_regoon_3p1l}
	\end{subfigure}
	
	\caption{Tadpole free region for three-point one-loop ${\cal M}_{1,3}^{\textrm{tf}}=\mathcal{S}^1\,\medcup\,\mathcal{S}^2$}
	\label{fig:three_regions}
\end{figure}

\begin{lemma} \label{lemma: vanishing-headlight}
Inside each $\mathcal{S}^{\mathit I}$ the headlight functions associated to curves that end at the puncture are such that, 
\begin{align}
\alpha_{n0} &= \alpha_{n-1,0}\, =\, \dots\, =\, \alpha_{n-(I-2),0}\, =\, 0\nonumber\\
\alpha_{K0}\, &\neq\, 0 , \ \forall\, 1\, \leq\, K\, \leq\, n-(I-1)\nonumber
\end{align}
where `$=0$' implies that the function identically vanishes pointwise in the interior of $\mathcal{S}^{\mathit I}$ and `$\neq\, 0$' implies that the headlight function is non-vanishing on set of positive measure. 
\end{lemma}
\begin{proof}
The proof is simply a matter of chasing the definitions. The headlight functions $\alpha_{i0}, \, \forall\, 1\, \leq\, i\, \leq\, n-1$ can be computed directly and we find that,

\begin{align}\label{1lphfnp}
\alpha_{n0} &= -t_x - \max\left(0, t_{n-1}\right) + \max\left(0, t_x + t_{n-1} \right), \\ ~\nonumber\\ 
\alpha_{i0} &= 
\max\left(
0,\, t_{i},\, t_{i} + t_{i+1},\, \dots,\sum_{k=i}^{n-2} t_k \,,\,\, \sum_{k=i}^{n-1} t_k
\right)\notag\\
&\quad -
\max\left(
0,\, t_{i-1},\, t_{i-1} + t_{i},\, \dots,\sum_{k=i-1}^{n-2} t_k \,,\,\, \sum_{k=i-1}^{n-1} t_k
\right) \notag\\
&\quad -
\max\left(
0,\, t_{i},\, t_{i} + t_{i+1},\, \dots,\sum_{k=i}^{n-2} t_k \,,\, \sum_{k=i}^{n-1} t_k + t_x
\right)\notag\\
&\quad+ \max\left(
0,\, t_{i-1},\, t_{i-1} + t_{i},\, \dots,\sum_{k=i-1}^{n-2} t_k \,,\,\right.  \left. \sum_{k=i-1}^{n-1} t_k + t_x
\right),\quad 2 \leq i \leq n-1\\   ~\nonumber\\
\alpha_{10} &= 
\max\left(
0,\, t_1,\, t_1 + t_2,\, \dots,\sum_{k=1}^{n-2} t_k,\, \sum_{k=1}^{n-1} t_k
\right)
-
\max\left(
0,\, t_1,\, t_1 + t_2,\, \dots, \sum_{k=1}^{n-2} t_k \,,\, \sum_{k=1}^{n-1} t_k + t_x
\right).
\end{align}
Let $(\vec{t}, t_{x})\, \in\, \mathcal{S}^{\mathit I}$. Then $\forall n \geq\, i\, \geq\, n- I + 2$, 
\begin{align}
\sum_{k=j}^{n-1} t_{k} > - t_{x} \geq\, 0,\, \quad \forall\, n-1 \geq\, j\, \geq\, n-I+1
\end{align}
and hence, 
\begin{align}
\alpha_{i0}\ =\, 0, \quad \, \forall\, 1\, \leq\, i\, \leq\,  n-I+2
\end{align}

As $t_{I-1}\, , t_{x}\, > 0$, the reader can easily convince themselves that remaining headlights do not vanish. 
\end{proof}
This hierarchy implies that all the cones which are dual to Feynman graphs containing a curve $C_{n0}$ are in $\mathcal{S}^{\mathit 1}$ , all the cones which are dual to Feynman graphs not containing $C_{n0}$ but necessarily containing $C_{n-1,0}$ are in $\mathcal{S}^{\mathit 2}$  and so on. Finally, all the cones dual to graphs containing the curves $C_{10}, C_{20}$ are in $\mathcal{S}^{\mathit {n-1}}$ .  Thus each $\mathcal{S}^{\mathit I}$  contains a union of cones where the number of such cones decrease monotonically as $\mathit{I}$ increases from 1 to $n-1$. 

The above observation will be used in the next section to write the renormalized planar amplitude without taking recourse to renormalizing individual cones.  

%%%%%%%%%%%%%%%%%%%%%%%%%%%%%%%%
\subsection{Renormalized amplitude}
We finally have all the tools at our disposal to define a renormalized amplitude whose evaluation inside each cone is a period \cite{Kontsevich:2001}.  Taking inspiration from the parametric renormalization, we will now prove that \begin{align}
{\cal A}_{1, n}^{\textrm{R}} ( \{X_{ij}\}, \{X_{ji}\} \vert 1 \leq\, i < j \leq\, n; S_{0} )\, :=\, \lim_{\epsilon \rightarrow\, 0} \bigg[\, {\cal A}_{1, n}(S, \{\Theta\}, \epsilon) - {\cal T}_{1, n}(S, \Theta, S_{0}, \epsilon)\, \bigg]
\end{align}
is UV finite whose restriction to each cone is the parametrized forest formula for corresponding graph. We will refer to ${\cal T}_{1, n}$ as a \emph{tropical counter-term} which depends on the external momenta as well as renormlisation scale $S = S_{0}$ defined in eqn. (\ref{schoice}).  In the above equation, the dependence of the renormalized amplitude on the ``generalized'' Kinematic variables is shown explicitly. In the subsequent equations, we will hide this dependence in the interest of brevity. 

We now define a ``modified curve integral formula'' ${\cal A}_{1, n}^{R}$ and prove that it is UV finite and consistent with the so-called \emph{kinematic scheme} in which the counter-term for UV finite graphs is zero. As we will see below, the  modified curve integral formula involves a subtraction term which depends on a renormalization scale $S_{0}$ and is ``curve-integral'' version of the counter-term in the parametric forest formula at one-loop. We first note that, as the second surface Symanzik polynomial is a homogeneous polynomial in $X_{ij}$ we assert that 
\begin{align}
\frac{{\cal F}}{{\cal U}}\, -\, {\cal Z}\, =\, S\, (t_{x}\, \Psi_{1}(\vec{\Theta})\,  +\, \Psi_{0}(\vec{\Theta})\, )
\end{align}
where, 
\begin{align}
\vec{\Theta}\, &:=\, (\, \{\Theta\}_{ij},\, \Theta_{m}) \\
\Theta_{ij} &:=\, \frac{X_{ij}}{S},\, \ \textrm{and,} \  \Theta_{m}\, =\, \frac{m^{2}}{S} 
\end{align}
We assert  that $\frac{{\cal F}}{{\cal U}} - {\cal Z}$ can be written in the above form, and $\Psi_{1}, \Psi_{0}$ are the degree $(L-1)$ and degree $L$ polynomials in $(\widetilde{t}_{1}, \dots, \widetilde{t}_{n-1})$ respectively. We back this assertion with several examples below (see sec. \ref{subsec:examples-one-loop}). %\textcolor{blue}{See lemma {\bf cite xxxx}) for details.} 
\begin{align}\label{cfraat1l}
\tcboxmath{
{\cal A}_{1, n}^{R}\, :=\, {\cal A}_{1, n} - \sum_{I=1}^{n-1}\, \int_{-\infty}^{0} \D t_{x}\, \prod_{K= I+1}^{n-1}\, \int_{0}^{\infty}\, \D \widetilde{t}_{K}\, \int_{t_{x}}^{0} \D\widetilde{t}_{I}\, \prod_{J=1}^{I-1} \cancel{\int} \D\widetilde{t}_{j}\, \frac{1}{{\cal U}^{2}}\, \exp^{- t_{x}S_{0}\, \Psi_{1}(\vec{\Theta_{0}})} e^{- S\, \Psi_{0}(\vec{\Theta})}}
\end{align}
where $\Theta^{0}_{ij}\, :=\, \frac{X_{ij}}{S_{0}}$ and $\Theta_{m}^{0}\, =\, \frac{m^{2}}{S_{0}}$.  We refer to the subtracted term as a \emph{tropical counter-term}.  In the above equation,  $\cancel{\int}$ is integral over ${\bf R}\, -\, [0, t_{x} ]$. Before going into the details of this formula and its derivation,  let us make some anticipatory remarks which will help clarify its  key features.
\begin{enumerate}
\item The renormalization scheme chosen via tropical counter-term only cancels logarithmic divergences arising from bubble (sub)graphs. Counter-term in the tropical region which corresponds to the contribution due to a UV finite graph  is zero in the chosen scheme. 
\item The dependence on the scheme is through the restriction of the domain of integration in $\mathcal{S}^{\mathit I}$ . It is clear that given the fact that $\widetilde{t}_{I} = t_{x}\tau_{I}$ with $\tau_{I}\, \in\, [0,1]$, if any of the remaining $\widetilde{t}_{j}$ is integrated over a domain $[0, \pm t_{x}]$ then such an integral will be convergent and polynomial in $S$. As a result we naively expect that $\cancel{\int}$ should apply to all the $\widetilde{t}_{j}\, \vert\, j\, >\,  I$ variables. That, there is no such restriction on $\widetilde{t}_{j} \vert j < I$ is due to lemma \ref{lfrfc}.
\item  The renormalized amplitude is an on-shell object and is a function of generalized kinematic variables that include $X_{ji} \vert 1 \leq\, i < j \leq\, n$. The physical amplitude is obtained by taking the limit 
\begin{align}
\prod_{j > i}\, \lim_{X_{ji}\, \rightarrow\, m^{2}}{\cal A}_{1, n}^{R} =: {\cal A}_{1,n}^{\textrm{physical}}
\end{align}
Our renormalization scheme is chosen in such a way that this limit is well defined as all the cones that are dual to fat graphs with radiative corrections on any of the external legs vanish.
\item In section \ref{subsec:examples-one-loop}, we will explicitly work out several lower point amplitudes to verify eqn. (\ref{cfraat1l}).
\end{enumerate}
We now analyze the formula in greater detail.

Right hand side of eqn.(\ref{cfraat1l}) can be re-written as follows.
\begin{align}
{\cal A}_{1, n}^{R}\, :=\, {\cal A}_{1, n} -\, \sum_{I=1}^{n-1}\,  \int_{-\infty}^{0} \frac{1}{- t_{x}}\, \D t_{x} \int\, \D\mu_{I}(\widetilde{t}_{1}, \dots\, \widetilde{\tau}_{I}\, \dots,\, \widetilde{t}_{n-1})\, \exp^{t_{x} S_{0} \Psi_{1}(\vec{\Theta_{0}})} e^{S \Psi_{0} (\vec{\Theta})}
\end{align}
where the measure $d\mu_{I}$ in the tropical Schwinger parameter space involves  precisely one integral from over $\tau_{I} \in [0,1]$, with all the integrals over $\widetilde{t}_{J} \in\, {\bf R}^{+}\,\forall\, n-1\, \geq\, J\, >\, I$ and all the remaining integrals being over $[t_{x}, 0]^{c}$. 

Taking inspiration from the decomposed Feynman rules for convergent graphs with log sub-divergences, we can write the renormalized amplitude as
\begin{align}\label{dfrfcif}
\tcboxmath{{\cal A}_{1, n}^{R}\, =\, I_{1, n}\, +\, \sum_{I=1} \int \D\mu_{I}(\widetilde{t}_{1}, \dots, \tau_{I}, \dots, \widetilde{t}_{n-1})\,  e^{S\, \Psi_{0}(\Theta)}\, \bigg[\, \ln(\frac{S}{S_{0}})\,  + \ln(\frac{\Psi_{1}(\vec{\Theta})}{\Psi_{1}(\vec{\Theta}_0)})\, \bigg]}
\end{align}
where, 
\begin{align}
I_{1, n}\, =\, {\cal A}_{1, n}\, -\, \sum_{I=1}^{n-1}\, \displaystyle \int_{-\infty}^{0} \frac{1}{- t_{x}}\, \D t_{x} \int\, \D\mu_{I}(\widetilde{t}_{1}, \dots\, \widetilde{\tau}_{I}\, \dots,\, \widetilde{t}_{n-1})\, e^{t_{x} S\Psi_{1}(\vec{\Theta})}\, e^{S \Psi_{0} (\vec{\Theta})},
\end{align}
is a UV-finite  polynomial function of the Mandelstam variables with no branch-cuts. 

We refer to Eqn.(\ref{dfrfcif}) as \emph{the decomposed curve-integral rule at one-loop}  in it's analogy with \emph{decomposed Feynman rule} derived in section \ref{sec:forest-formula-review}.

The tropical counter-term depends on two  polynomials $\Psi_{1}, \Psi_{0}$ that are defined based on the following assertion.
\begin{lemma}\label{hftxzero}
Consider any region $\mathcal{S}^{\mathit I}\bigg|_{I\, \in\, \{1, \dots, n-1\}}$ in ${\cal M}_{1, n}^{\textrm{tf}}$.  Let
\begin{align}
\overline{\mathcal{S}}^{I}\, =\, \mathcal{S}^{\mathit I}\, \bigcap_{j < i(I)} (\widetilde{t}_{j}\, <\, t_{x}\, \medcup\, \widetilde{t}_{j} > 0\, ).
\end{align}
Then inside $\overline{\mathcal{S}}^{I}$, any non-trivial quadratic monomial in the headlight functions $\alpha_{i0}$ vanish in the limit $t_{x}\, \rightarrow\, 0$.\footnote{$\alpha_{i0}\, \alpha_{j0}$ is non-trivial inside a region in ${\cal M}_{n}^{\textrm{tf}}$ if it does not vanish pointwise inside that region.} 
\end{lemma}
Although we do not have a proof of this assertion, we have been able to verify it in a number of examples as illustrated below.  However before checking it's validity via examples, we provide a heuristic argument in its support.
\begin{proof}
Every logarithmically divergent cone contains at most one variable $\widetilde{t}_i$ that is bounded between $(t_x,0)$. Upon scaling this variable, the structure must be such that
\[
\left(\frac{{\cal F}}{{\cal U}} - {\cal Z}\right)
(t_x,\widetilde{t}_1,\dots,\widetilde{t}_{i-1},\tau_i,\widetilde{t}_{i+1},\dots,\widetilde{t}_{n-1})
= t_x \,\Psi_{1}
+ \Psi_{0},
\]
in order to reproduce the correct logarithmic divergence, where $\Psi_{1}$ is dependent on $(t_x,\widetilde{t}_1,\break\dots,\widetilde{t}_{n-1},\tau_i,\widetilde{t}_{n+1},\dots,\widetilde{t}_{n-1})$  and   $\Psi_{0}$ is dependent on $(t_x,\widetilde{t}_1,\dots,,\widetilde{t}_{i-1},\widetilde{t}_{i+1},\dots,\widetilde{t}_{n-1})$.

When we remove the convergent regions using the $\cancel{\int}$ notation from the domains $\mathcal{S}^{I}_{1,n}$ , the remaining region 
\[
\mathcal{S}^{I}_{1,n} \setminus \text{\{convergent regions\}}
\]
is a union of logarithmically divergent cones, all of which share the property that the same coordinate $\widetilde{t}_{n-I}$ is bounded between $(t_x,0)$. Consequently, within such a region one must have
\[
\left(\frac{{\cal F}}{{\cal U}} - {\cal Z}\right)
(t_{1}, \dots, t_{I-1}, \tau_{I}, \widetilde{t}_{I+1}, \dots, \widetilde{t}_{n-1})
= t_x \,\Psi_{1}
+ \Psi_{0},
\]
under the scaling of $\widetilde{t}_i$,where $\Psi_{1}$ is dependent on $(t_{1}, \dots, t_{I-1}, \tau_{I}, \widetilde{t}_{I+1}, \dots, \widetilde{t}_{n-1}) $ and $\Psi_{0}$ is dependent on $(t_{1}, \dots, t_{I-1}, \widetilde{t}_{I+1}, \dots, \widetilde{t}_{n-1}) $.

For readers not satisfied with the heuristic argument, we now support it by computing ${\cal F}$ for $n = 4$ and show why it is consistent with our assertion. 

Let us consider the $n=4$ case. As an illustration, we evaluate the limit in the region $\mathcal{S}^{\mathit 2}$ where
\begin{align}
\widetilde{t}_{3}\, \geq\, 0,\, t_{x}\, \leq\, \widetilde{t}_{2}\, \leq\, 0, -\, \infty\, < \widetilde{t}_{1} < \infty
\end{align}
In $\mathcal{S}^{\mathit 2}$,  $\alpha_{40}\,=\, 0$ and it can be checked that
\begin{align}
\alpha_{30}&=\, -\, \widetilde{t}_{2}\nonumber\\
\alpha_{20}&=\, \textrm{max}(0,\widetilde{t}_{1} - \widetilde{t}_{2}) - \textrm{max}(0,\widetilde{t}_{1} - t_{x})\, +\, (\widetilde{t}_{2} - t_{x})\nonumber\\
\alpha_{10}&=\, \textrm{max}(0,\widetilde{t}_{1} - t_{x})\, -\, \textrm{max}(0,\widetilde{t}_{1} - \widetilde{t}_{2})\, .
\end{align} 
Hence
\begin{align}
\lim_{t_{x}\, \rightarrow\, 0}\, t_{x}\, \frac{{\cal F}}{{\cal U}}\, =\, -\, \lim_{t_{x}\, \rightarrow\, 0}\, \sum_{1 \leq\, i < j \leq\, 3}\, z_{ij}\, \alpha_{i0} \alpha_{j0}\, .
\end{align}
As $\widetilde{t}_{2} = t_{x} \tau_{2}$, both the terms involving $\alpha_{30}$ atleast scale as $t_{x}$ and hence vanish as $t_{x}\, \rightarrow\, 0$.  Hence in order to show the desired limit, we have to evaluate, 
\begin{align}
\lim_{t_{x}\, \rightarrow\, 0}\, \alpha_{10}\, \alpha_{20}&=\, -\, \lim_{t_{x}\, \rightarrow\, 0}\, (\textrm{max}(0, \widetilde{t}_{1} - \widetilde{t}_{2})\, -\, \textrm{max}(0, \widetilde{t}_{1} - t_{x})\, )^{2}\nonumber\\
&=\, 0\, \qquad  \textrm{if} -\infty < \widetilde{t}_{1} < t_{x}\nonumber\\
&=\, (t_{x} - \widetilde{t}_{2})^{2}\, ( =\, t_{x}^{2} (1 - \tau_{2})^{2}) \  \qquad \textrm{for}\ \widetilde{t}_{1} > 0\nonumber\\
&=\, [\, \textrm{max}(0, \widetilde{t}_{1} - \widetilde{t}_{2}) - (\widetilde{t}_{1} - t_{x})\, ]^{2}\, =\, O(t_{x}^{2})\, \qquad  \textrm{for}\ t_{x} \leq\, \widetilde{t}_{1} \leq\, 0.
\end{align}
The final equality in the last line is due to the fact that when integrating over the interval, $t_{x}\, \leq\, \widetilde{t}_{1}\, \leq\, 0$ we can change $\widetilde{t}_{1}\, \rightarrow\, t_{x}\tau_{1}$. 

This implies that in $\mathcal{S}^{\mathit 2}$, $ \displaystyle \lim_{t_{x}\, \rightarrow\, 0}\, t_{x}\, \frac{{\cal F}}{{\cal U}}$ vanishes for the entire range of $\widetilde{t}_{1}$. The analysis in the region $\mathcal{S}^{\mathit 1}$ is similar. 
\end{proof}
A few comments are in order.
\begin{enumerate}
\item Lemma \ref{hftxzero} implies that the exponent of the curve integrand is a linear function of $t_{x}$.
\begin{align}\label{fbyuminuszinsi}
\bigg(\frac{{\cal F}}{{\cal U}} - {\cal Z}\bigg)(t_{1}, \dots, t_{I-1}, \tau_{I}, \widetilde{t}_{I+1}, \dots, \widetilde{t}_{n-1})\, =: t_{x} \Psi_{1\, I} + \Psi_{0\, I},
\end{align}
where $\Psi_{1}, \Psi_{0}$ are functions of the planar kinematic variables.
\item The tropical counter term is not an integral over the entire region but excludes certain domain. In particular inside each $\mathcal{S}^{\mathit I}$ the measure is
\begin{align}\nonumber
\displaystyle \int_{-\infty}^{0} \D t_{x} \int_{t_{x}}^{0} \D \widetilde{t}_{I}\, \prod_{J=I+1}^{n-1} \cancel{\int} \D \widetilde{t}_{j} \prod_{K=1}^{I-1}\int_{0}^{\infty}\, \D \widetilde{t}_{K}.
\end{align}
The reason for this exclusion is as follows.
Consider an integral of the form,\footnote{Where $f$ is such that the integral is convergent at large $\vert t_{x} \vert$.}
\begin{align}
I = \int_{-\infty}^{0} \D t_{x} \int_{t_{x}}^{0} \D s_{_1}\, \int_{t_{x}}^{-t_{x}} \D s_{_2}\, \frac{1}{t_{x}^{2}}\, e^{ f(t_{x}, t_{_1}, t_{_2})}. 
\end{align}
Such an integral is  convergent near $t_{x} = 0$, so long as at small $t_{x}$, $f$ admits a Taylor expansion. This is because,
\begin{align}
I\, =\, \int_{-t_{x}}^{0}\, \D t_{x}\, \int_{1}^{0}\, t_{x} \D \tau_{_1}\, \int_{1}^{-1}\, t_{x} \, \D \tau_{_2}\, \frac{1}{t_{x}^{2}}\, e^{f(t_{x}, \tau_{_1}, \tau_{_2})}.
\end{align}
Clearly this integral is convergent near $t_{x} = 0$ since the singularity is cancelled by the fact that $s_{1}, s_{2}$ integration range scale with $\vert t_{x} \vert$. This simple example reveals that once $t_{n-I+1}$ is bounded between $[t_{x}, 0]$ inside $\mathcal{S}^{\mathit I}$, the divergence can only arise if all $\widetilde{t}_{K},\, 1\leq\, K < n-I+1$. As a result we expect inside each $\mathcal{S}^{\mathit I}$,  the integrand inside tropical counter-term is integrable so long as the region of integration $\forall\, \widetilde{t}_{m}\,  \vert m\, \neq\, n - I$ do not scale with $t_{x}$. It turns out that we can make this expectation more precise. 
\end{enumerate}
\begin{lemma}\label{lfrfc}
Let $C$ be a cone which lies inside $\mathcal{S}^{\mathit I}$. $C$ is dual to a UV finite graph if and only if,  the range of at least one of the $\widetilde{t}_{j}$ for $j < I$ lies between $t_{x}$ and $0$.
\end{lemma}
\begin{proof}
Although the detailed proof requires us  to analyze all possible regions of this type, we will show explicitly how this statement is true for a specific choice of $j$. The analysis of any such region should be similar. 

We analyze the cones that sit inside $\mathcal{S}^{\mathit 2}\, \medsubset\, {\cal M}_{1, n}^{\textrm{tf}}$.  We recall that $\mathcal{S}^{\mathit 2}$ is the region in global Schwinger space bounded by the following inequalities.
\begin{align}
\widetilde{t}_{n-1}\, >\,  0\, \medcap\, t_{x}\, < \widetilde{t}_{n-2}\, <\, 0\, .
\end{align}
We know from Lemma \ref{lemma: vanishing-headlight} that
\begin{align}
\alpha_{n, 0} = 0,\, \alpha_{n-1,0}\, =\, -\, \widetilde{t}_{n-2} \, .
\end{align}
Using eqn. (\ref{1lphfnp}), the reader can verify that the region where $\alpha_{n-2,0}\, \neq\, 0,\, \alpha_{n-3,0} \neq\, 0$ lies inside the following sub-region $\mathcal{S}^{\mathit 2}$, 
\begin{align}
R(n-3,n-2,n-1)\, :=\, \mathcal{S}^{\mathit 2}\, \medcap\, (\, t_{x}\, \leq\,\widetilde{t}_{n-3}\, \leq\, 0\, )\,  \bigcap\, \bigg[ \bigcap_{k\ \neq\, n-1,n-2,n-3}\, \widetilde{t}_{k}\, \in\, {\bf R} - [t_{x}, 0] \bigg].
\end{align}
Hence the cone associated to the triangle graph spanned by the three curves $C_{n-1,0}, C_{n-2,0}, C_{n-3,0}$ lies inside $R(n-3,n-2,n-1)$. 
\end{proof}

In section \ref{argforarbitn}, we prove that (a) ${\cal A}_{1, n}^{R}$ is  finite and (b) it leads to renormalized Feynman graph inside each cone. But before giving the general argument, in section (\ref{subsec:examples-one-loop}) we work out the tropical counter term in several examples and show explicitly that  how their subtraction indeed leads to renormalized amplitude at one-loop in the kinematic renormalization scheme,\cite{Brown:2011pj} . 

%%%%%%%%%%%%%%%%%%%%%%%%%%%%%%%%
\subsection{Some examples for ${\cal A}_{1, n}^{R}$} \label{subsec:examples-one-loop}
In this section we explicitly evaluate the renormalized amplitude for $n=2,3$ and $n=4$ case and show how renormalized amplitudes are period integrals just as individual renormalized Feynman graphs are. 

%%%%%%%%%%%%%%%%%
\subsubsection*{${\cal A}_{1, 2}^{R}$}
For the two--point case, the tadpole--free region is given by
\begin{align}
\mathcal{S}_{1,2}^1 \;=\; \{\, t_x < \widetilde{t}_1 < 0 \,\}.
\end{align}
This region corresponds to a single graph: the one-loop correction to the propagator.  
We consider the renormalized integral,
\begin{align}\label{countertermforne3le1}
{\cal A}_{1,2}^{R} 
= {\cal A}_{1,2}  
- \int_{-\infty}^{0} \D t_{x} 
\int_{t_x}^{0} \D \widetilde{t}_{1} \, 
\frac{1}{t_x^{2}} \,
e^{t_{x} S_{0} \Psi_{1,I}(\vec{\Theta_{0}})} \,
e^{S \Psi_{0,I}(\vec{\Theta})},
\end{align}
where 
\[
\Psi_{1,1} = -\Theta_e + \tau_1 - \tau_1^2, 
\qquad 
\Psi_{0,1} = 0,
\qquad 
S = p_1^2,
\qquad 
\Theta_e = \frac{m^2}{S}.
\]

After rescaling $\widetilde{t}_1=t_x \tau_1$ the integral, this can be rewritten as
\begin{align}\label{countertermforne2le1}
{\cal A}_{1,2}^{R} 
&= \int_{-\infty}^{0} \D t_{x} 
\int_{1}^{0} \D \tau_{1}\, 
\frac{1}{t_x} \, e^{t_{x} S \Psi_{1,I}(\vec{\Theta_{0}})}
- \int_{-\infty}^{0} \D t_{x} 
\int_{1}^{0} \D \tau_{1} \, 
\frac{1}{t_x}\, e^{t_{x} S_{0} \Psi_{1,I}(\vec{\Theta_{0}})} \\[6pt]
&= \int_{1}^{0} \D \tau_{1}\, 
\bigg[\, \ln(\frac{S}{S_{0}})\,  + \ln(\frac{\Psi_{1,1}(\vec{\Theta})}{\Psi_{1,1}(\vec{\Theta}_0)})\, \bigg],
\end{align}
where $\Psi_{1,1}^{(0)}$ denotes that the angle is fixed, with $\Theta_e^{(0)} \;=\; \frac{m^2}{S_0},$ in the expression for $\Psi_{1,1}$.

%%%%%%%%%%%%%%%%%
\subsubsection*{${\cal A}_{1,3}^{R}$}
We now consider  renormalized  three-point amplitude.  
\begin{align}\label{countertermforne3le1}
{\cal A}_{1, 3}^{R} = {\cal A}_{1, 3}  - \sum_{I=1}^{2} \int_{-\infty}^{0} \D t_{x} \int_{0}^{\infty} \prod_{J=I+1}^{3} \D \widetilde{t}_{J} \int_{t_{x}}^{0} \D\widetilde{t}_{I} \cancel{\int} \prod_{K=I-1}^{1} \D \widetilde{t}_{K}\, \frac{1}{t_x^2} \, e^{t_{x} S_{0} \Psi_{1,I}(\vec{\Theta_{0}})} e^{S \Psi_{0,I} (\vec{\Theta})}\, .
\end{align}
For three-point one-loop we can find the $\frac{\mathcal{F}}{\mathcal{U}}-\mathcal{Z}$ using the properties of surface Symanzik polynomials (see appendix of \cite{Arkani-Hamed:2023mvg}) we can write,
\begin{align}
&\mathcal{F}=p_2^2\;\alpha_{10}\,\alpha_{20}+p_3^2\;\alpha_{20}\,\alpha_{30}+p_1^2\;\alpha_{30}\,\alpha_{10}\\
&\mathcal{Z}=\sum_{i=1}^{3}\alpha_{i0}\; m^2+\alpha_{12}(p_2^2+m^2)+\alpha_{23}(p_3^2+m^2)+\alpha_{31}(p_1^2+m^2)\, .
\end{align}
Where in this case kinematic variables are $S=(p_1+p_2)^2,\Theta_{i,j}=\frac{p_i.p_j}{S}, \Theta_e=\frac{m^2}{S}$ where S is scale and $\Theta$ are angles and $S_0, \Theta^0_{i,j},\Theta^0_e$, denote some fix kinematic points or fix angle and scales. The tadpole free region is made by union of two regions  $ \mathcal{S}^1_{1,3} \medcup \mathcal{S}^2_{1,3}$, hence counter-term also have two parts where 
\begin{align}
&\mathcal{S}^1_{1,3}=t_x\leq\widetilde{t}_2\leq0\\
&\mathcal{S}^2_{1,3}=(\widetilde{t}_2\geq0)\ \medcap\ (t_x \leq \widetilde{t}_1\leq0) \, .
\end{align}
Let's look at region  $\mathcal{S}^2_{1,3}$ for the form of counterterm. In this region, 
\begin{align}
\left(	\frac{\mathcal{F}}{\mathcal{U}}-\mathcal{Z}\right)\left(t_x,\tau_1,\widetilde{t}_2\right)=
t_{x}\,S_0\Psi^0_{1,2}(\tau_{2},\Theta^0)+S\Psi_{0,2}\, ,
\end{align}
where $\Psi_{1,2}=( \Theta_e -
\Theta_{2,2}\,\tau_{1} + \Theta_{2,2}\,\tau_{1}^{2} )$ and $\Psi_{2,2}=-\Theta_e\,\widetilde{t}_{2}
-\Theta_{2,2}\,\widetilde{t}_{2}$. 

For region $\mathcal{S}^1_{1,3}$ we remove the region $(0>\widetilde{t}_1>t_x)$ in $\mathcal{S}^1_{1,3}$ to remove finite counter-term contributions as explained in lemma~\ref{lfrfc}. This is denoted by $\cancel\int$ sign. As a result, in the region $\mathcal{S}_1 \setminus (0>\widetilde{t}_1>t_x)$,
\begin{align}
\left(\frac{\mathcal{F}}{\mathcal{U}}-\mathcal{Z}\right)\left(t_x,\widetilde{t}_1,\tau_2\right)
=& S \,\widetilde{t}_1 
+ S \, \widetilde{t}_1 \Theta_e  - S \max\left(0, \widetilde{t}_1 - t_x \tau_2\right) 
\left(\tau_2 + \Theta_e - \left( \tau_2-1\right)\Theta_{1,1} + \left( \tau_2-1\right)\Theta_{2,2}\right) \nonumber \\
& - S \, \max\left(0, \widetilde{t}_1 - t_x\right) 
\left(\Theta_e + \tau_2 \Theta_{1,1} - \left( \tau_2-1\right)\left(1 + \Theta_{2,2}\right)\right) \nonumber \\
& + S \, t_x \left( 
-1
- \left( \tau_2-1\right)\tau_2 
- \left(  \tau_2-1\right)^2 \Theta_{2,2}\,\mathrm{H}\!\left(\widetilde{t}_1\right) \right).
\end{align}
where $\mathrm{H}$ is the Heaviside function $H(x)$ (i.e. $1$ for $x>0$, else $0$), and $\Psi_{1,1}$ and $\Psi_{0,1}$ can be obtained by reducing $\left(\frac{\mathcal{F}}{\mathcal{U}}-\mathcal{Z}\right)$ in piecewise form. We remind the reader that,   $\Psi^0_{1,I}$ is the evaluation of $\Psi_{1,I}$ at $\Theta^{0}$. The reader can check that this result can also be written in the form given in eqn. \ref{dfrfcif}. 

%%%%%%%%%%%%%%%%%
\subsubsection*{${\cal A}^{R}_{1,4}$}
In the case of $n=4$ we can explicitly verify this formula as follows. ${\cal M}_{1, 4}^{\textrm{tf}}$ is a union of three sub-regions 
\begin{align}
\mathcal{S}^{\mathit 1}_{1, 4}\, &=\, t_{x} \leq\, \widetilde{t}_{3} \leq\, 0\nonumber\\
\mathcal{S}^{\mathit 2}_{1, 4}\,  &=\, \big(\widetilde{t}_{3} \geq\, 0\big)\, \medcap\, \big(t_{x} \leq\, \widetilde{t}_{2} \leq\, 0 \big)\nonumber\\
\mathcal{S}^{\mathit 3}_{1, 4}\,  &=\, \big(\widetilde{t}_{3} \geq\, 0\big)\, \medcap \big(\widetilde{t}_{2} \geq\, 0\big)\, \medcap \big(t_{x} \leq\, \widetilde{t}_{1} \leq\, 0 \big)\, .
\end{align}
We can now explicitly write down the proposed formula for renormalized amplitude as %(\textcolor{red}{Write the integrand completely})
\begin{align}\label{countertermforne4le1}
{\cal A}_{1, 4}^{R} = {\cal A}_{1, 4}  - \sum_{I=1}^{3} \int_{-\infty}^{0} \D t_{x} \int_{0}^{\infty} \prod_{J=I+1}^{4} \D \widetilde{t}_{J} \int_{t_{x}}^{0} \D \widetilde{t}_{I} \cancel{\int} \prod_{K=I-1}^{1} \D \widetilde{t}_{K} \,\frac{1}{t_x^2}\, e^{t_{x} S_{0} \Psi_{1,I}(\vec{\Theta_{0}})} e^{S \Psi_{0,I} (\vec{\Theta})}\, .
\end{align}
The fact that this formula satisfies the kinematic scheme can now be explicitly verified. In $\mathcal{S}^{\mathit 1}$, one has $\alpha_{4,0} = - \widetilde{t}_{3} \neq\, 0$. As a result, this region contains all the UV finite cones dual to graphs that satisfy $\alpha_{i 0} \neq\, 0\, \forall\, i$ and $\alpha_{40} , \alpha_{i0} , \alpha_{j0} \neq\, 0$ for $ (i,j)\ \in\, \{1,2,3\}$.  We can now explicitly check that the former corresponds to the region where $t_{x} \leq\, \widetilde{t}_{i} \leq\, 0\, \ \forall\, i$ and the latter corresponds to $(t_{x} \leq\, \widetilde{t}_{3} \leq\, 0\,) \medcap\, (t_{x} \leq\, \widetilde{t}_{i} \leq\, 0) \, \vert\, i\, \in\, \{1, 2\}$. As a result the domain over which we integrate in eqn. (\ref{countertermforne4le1}) excludes precisely these cones and the resulting formula for ${\cal A}_{1, 4}^{R}$ is UV finite as per the kinematic scheme employed in traditional forest formula \cite{Brown:2011pj}.

%%%%%%%%%%%%%%%%%
\subsubsection{An argument for arbitrary $n$}\label{argforarbitn}
In this section we prove how the tropical counter-term renormalizes planar one-loop amplitudes $\forall\, n$. Essence of this argument is based on the following simple observation.

Inside each $\mathcal{S}^{\mathit I}$, thanks to eqn.(\ref{fbyuminuszinsi}), the only integral of the form $\int_{-\infty}^{0}\, \frac{1}{t_{x}}\, e^{A t_{x}}$ diverges at $t_{x} = 0$ with the divergence given by
\begin{align}
\int_{-\infty}^{\delta}\, \frac{\D t_x}{t_{x}}\, e^{A \, t_{x}} = \ln(\delta A)\, +\, O(\delta^{0}) \, .
\end{align}
The cancellation of the logarithmic divergence inside each ${\cal S}_{I}$ than simply follows from the identity,
\begin{align}
\int_{-\infty}^{0}\, \frac{\D x}{x} \bigg(\, e^{- a x}\, -\, e^{- a_{0} x}\, \bigg)\, =\, \ln(\frac{a}{a_{0}})\, .
\end{align}
The proof for why ${\cal A}_{1, n}^{R}$ is finite follows simply from the fact that all the divergence in ${\cal A}_{1, n}$ arise when precisely one of the tropical co-ordinates (upto transformation with unit Jacobian) is bounded between 0 and $\vert t_{x}\vert$. All such singularities are precisely cancelled by the counter-term. 

We will now sketch the proof for our claims (a) and (b). 

All the cones in the tropical space which contribute logarithmically divergent terms to the bare amplitude are domains where $\alpha_{i0} = 0, \, \ \forall i\, \in\, \{1, \dots, \cancel{i_{1}}, \dots, \cancel{i_{2}}, \dots, n\}$. Without loss of generality let us consider one such cone where $\alpha_{i0} = 0, \, \ \forall i\, \in\, \{1, \dots, n-2\}$. This cone resides in $\mathcal{S}^{\mathit 1}$ and is the parametrized by the following inequalities.  
\begin{align}
(t_{x} \leq\, \widetilde{t}_{n-1} \leq\, 0)\, \medcap\, (\widetilde{t}_{i} < t_{x}), \, \quad \forall\, i\, \in\, (1, \dots, n-2).
\end{align}
%\textcolor{blue}{Check this}.  
Clearly the counter-term cancels the log divergence  in the bare amplitude that arises from this cone. This essentially proves point (b). To prove point (a), consider a cone in $\mathcal{S}^{\mathit 1}$ where $\alpha_{i0} = 0\, \forall i \in\, \{1, \dots, n-3\}$. It can now be shown that this region is parametrized by the inequalities
\begin{align}
(t_{x} \leq\, \widetilde{t}_{n-1} \leq\, 0)\, \medcap\, (t_{x} \leq\, \widetilde{t}_{n-2} \leq\, 0) \, \  \textrm{and}\, \ \widetilde{t}_{i} < t_{x}, \quad \forall\, i\, \in\, \{1,\, \dots, n-3\}.
\end{align}
This cone is excluded from the domain over which counter-term is integrated and hence the corresponding triangle loop has no counter-term in our renormalization scheme. All the cones can be analysed in analogous ways and a reader can verify this at their perusal. This completes the proof.

%%%%%%%%%%%%%%%%%%%%%%%%%%%%%%%%
%%%%%%%%%%%%%%%%%%%%%%%%%%%%%%%%
\section{Cones and Feynman diagrams for planar two-loop curve integral formula}\label{sec:reno_beyond_1loop}
In this section we study renormalization of planar amplitudes at two loops in $D=4$ dimensions. If we exclude the contribution of the quadratically divergent tadpoles, then the regularized amplitude, ${\cal A}_{2, n}(\Lambda)$ must then take the following form in angle and scale variables, 
\begin{align}
{\cal A}_{2, n}(\{X_{13}, X_{24}, \dots\, \},\, \Lambda)\, =\, \sum_{m=0}^{2} \bigg[\ln(\frac{S}{\Lambda})\bigg]^{m}\, a_{m}(\{\, \Theta_{13}, \Theta_{24},\, \dots\, \}),
\end{align}
where 
\begin{align}
S = \sum_{(i,j)} \lambda_{ij} X_{ij}\, \vert\, \lambda_{ij}\, \in\, {\bf R}^{+}\, , \nonumber
\end{align}
is fixed once and for all and 
\begin{align}
\Theta_{ij} :=\, \frac{X_{ij}}{S} \, .\nonumber
\end{align}
From the decomposed Feynman rules for graphs, we know that the renormalized amplitude should be of the form
\begin{align}
{\cal A}_{2, n}^{R}(S, S_{0}, \Theta, \Theta_{0})\, =\,  \ln(\frac{S}{S_{0}})^{2}\, a_{2}(S,  \vec{\Theta})\, +\,  \ln(\frac{S}{S_{0}})\, a_{1}(S,\, \vec{\Theta},\, \vec{\Theta^{0}})\, +\, a_{0}(S,\, \vec{\Theta},\, \vec{\Theta^{0}})\, ,
\end{align}
where $a_{0}, a_{1}, a_{2}$ are polynomials in $S$.  In this section, our goal is to prove this expansion by renormalizing the two-loop curve integral. 

Even though the essential idea behind parametric renormalization for $L > 2$ is the same as that in the planar one-loop case, ( which is to isolate and parametrically renormalize the divergences in global Schwinger space by computing residue of the integrand at the zeros of ${\cal U}$),  there are qualitative differences between one-loop and higher loop integrands. These differences are in the structure of the tadpole free region and the fact that the first surface Symanzik polynomial ${\cal U}$ is only piecewise smooth in tadpole free region. As a result, the difference between renormalizing two and higher loop planar amplitudes are simply quantitative\footnote{As the headlight functions as well as tropical Mirzakhani kernel are harder to compute at increasing loop order.}, and we will only discuss two-loop renormalization in this section. \\

\noindent We begin with some obervations which will assist us later in our analysis. 
\begin{enumerate}
\item We begin by isolating the subregion ${\cal M}_{2, n}^{\textrm{qtf}}$  inside ${\cal K}\, \neq\, 0$ region (that is where the Mirzakhani kernel is nowhere vanishing).  ${\cal M}_{2, n}^{\textrm{qtf}}$ is defined such that all the cones dual to graphs which diverge faster than $(\ln\Lambda)^{2}$ are decapitated. However unlike in the case of planar one-loop amplitudes, such a decapitated region does not exhaust all the graphs that contain tadpoles. In particular, for $L=2$ there are tadpole graphs which are logarithmically divergent as well as  UV finite. Although we can restrict the analysis presented below to tadpole free region ${\cal M}_{2,n}^{\textrm{tf}}$, we will be slightly more general in what follows. That is, we consider a region denoted as  ${\cal M}_{2, n}^{\textrm{qtf}}$ which exclude any cones dual to tadpole graphs which are at least quadratically divergent. Thus to summarize, our region of integration is
\begin{align}
{\cal M}_{2, n}^{qtf}\, \subset\, \textrm{Ker}({\cal K})^{c}\, ,
\end{align}
where $\textrm{Ker}({\cal K})^{c}$ is the compliment of the region in which the Tropical Mirzakhani kernel vanishes. 
\item Our strategy of finding ${\cal M}_{2,n}^{\textrm{qtf}}$ slightly differs from simply finding the intersection region $\bigcap_{i}\, \alpha_{ii} = 0$ that we adapted for planar one-loop amplitudes. This is because of the complexity involved in computing the headlight functions. While we do believe that the latter strategy can be generalised for any $\Sigma_{L,n}$, it requires the level of automation which is beyond the scope of the present authors. As a result our approach to find these regions is as follows. 
\noindent (1) We first find ${\cal M}_{2,2}^{\textrm{qft}}$ by simply analyzing all the regions that are orthogonal to the $g$ vectors associated to the tadpole curves. This amounts to simply identifying all the Feynman graphs for planar two-loop amplitudes for $n=2$, excluding all the cones dual to (quadratically divergent) tadpole fat graphs and then gluing all the cones together in the tropical co-ordinates. We explain this strategy in complete detail in appendix~\ref{sec:Cones_for_2-point_2-loop}.
\item Rather remarkably the structure of ${\cal M}_{2,2}^{\textrm{qtf}}$ then aids us in identifying decapitated tadpole regions for $n > 2$. We come back to this point in section~\ref{sec:Renormalization_of_planar_A_n>2_l=2}.
\item One of the remarkable simplifications that arose at one-loop was thanks to the ``tree-loop'' factorization property that is inherent in the choice of the reference graph. This choice results in the basis of $g$-vectors such that 
\begin{align}
{\cal U}_{\textrm{1-loop}}\, =\, -\, t_{x}\, .
\end{align}
As a result, all the UV divergences were localized at $t_{x} = 0$, with precisely one additional $\widetilde{t}_{i}\, \sim\, O(t_{x})$.\\ 
In the two-loop case, to the best of our understanding, there is no reference graph for which ${\cal U}$ is a globally defined smooth quadratic function. However thanks to tree-loop factorization, it suffices to compute ${\cal U}$ for $n=2$ as when the reference fat graph is a two-loop tadpole, ${\cal U}$ is independent of $n$. In appendix \ref{sec:Decomposition_of_U},  we compute ${\cal U}$ and show that ${\cal M}_{2, 2}^{qtf}$ can be decomposed into \emph{five} subregions so that, 
\begin{align}
{\cal U}\bigg|_{{\cal M}_{2,2}^{qtf}}\, =\, \sum_{{\cal R}_{I}}\, \kappa_{I}\, {\cal U}_{I}(t_{1}, \dots, t_{n-1}, t_{w}, t_{x}, t_{z}, t_{y})\, ,
\end{align}
where ${\cal R}_{I}\vert I\, \in\, \{1, \dots, 5\}$ are \emph{five} open domains in ${\bf R}^{n+3}$ with $\kappa_{I}$ being the bump function inside each ${\cal R}_{I}$. 
\item As shown in appendix \ref{sec:Cones_for_2-point_2-loop}, ${\cal M}_{2, 2}^{\textrm{qtf}}$  is a union over 24 cones. Each cone is dual to a $n=2$ two-loop planar fat graph such that for every fat graph there are four cones and as a result contribution of every graph is replicated four times. The tropical Mirzakhani kernel in fact ensures that the final result after summing over all the cones is precisely twice the 2 point amplitude. It is important to note that unlike in the case of $L=1$, the contribution of each cone does not simply equal the contribution of a specific graph as ${\cal K}$ is not constant inside each cone. However given a planar fat graph $\Gamma$ that is dual to the cones $C_{1}, \dots, C_{4}$, the net contribution of these cones to the curve integral is twice the contribution of $\Gamma$. We refer the reader to appendix \ref{ricones} for details. 
\end{enumerate}

%%%%%%%%%%%%%%%%%%%%%%%%%%%%%%%%
\subsection{Dissecting ${\cal M}_{2, 2}^{\textrm{qtf}}$ along discontinuities of ${\cal U}$.}
We will now describe the regions   ${\cal R}_{a}$.  These regions are defined by a set of inequalities which can be conveniently phrased in terms of the following co-ordinates in global Schwinger space,
\begin{align}\label{newtvariablesforR1}
\widetilde{t}_{w} &= t_{w} + t_{x}, &\quad
\widetilde{t}_{z} &= t_{z} + t_{y}, \nonumber\\[6pt]
t_{1}^{x}     &= \widetilde{t}_{w}, 
&\quad t_{1}^{y} &= t_{1} + \widetilde{t}_{z}, \nonumber\\
t_{2}^{x}     &= t_{1} + t_{w} + t_{z} + \widetilde{t}_{z} + t_{x}, 
&\quad t_{2}^{y} &= \widetilde{t}_{z}, \nonumber\\
t_{3}^{x}     &= t_{1} + t_{w} + t_{x}, 
&\quad t_{3}^{y} &= \widetilde{t}_{z}, \nonumber\\
t_{4}^{x}     &= \widetilde{t}_{w}, 
&\quad t_{4}^{y} &= t_{1} + \widetilde{t}_{w} + t_{w} + t_{z} + t_{y}, \nonumber\\[6pt]
s_{1,1}       &= t_{z}, &\quad
s_{1,2}       &= t_{w} + t_{z}, \nonumber\\
s_{1,3}       &= -\,\widetilde{t}_{w}, &\quad
s_{1,4}       &= -\,(\widetilde{t}_{w} + \widetilde{t}_{z}), \nonumber\\
s_{4,1}       &= \widetilde{t}_{w} + t_{z}\, . &\quad
%s_{5,1}       &= t_{w} + \widetilde{t}_{z}.
\end{align}

As shown in appendix   \ref{sec:Decomposition_of_U}, ${\cal R}_{a}$ are then defined by the following set of inequalities in the region $t_{x} < 0,\, t_{y} < 0$ 
\begin{align}\label{r1tr5innc}
{\cal R}_{1}\, &= \bigcup_{\alpha = 1}^{4} {\cal R}_{1,\alpha}\nonumber\\
{\cal R}_{1,\alpha} &:= \{\, (t_{x} < 0) \,\medcap\, (t_{y} < 0) \,\medcap\, (t_{\alpha}^{x} \in (t_{x},0)) \,\medcap\, (t_{\alpha}^{y} \in (t_{y},0)) \,\medcap\, (s_{1,\alpha} < 0)\,\}, 
\quad \forall\, \alpha \in \{1,\dots,4\}\nonumber\\
{\cal R}_{2}\, &=\, 
\{\, (s_{4,1} = t_{w}+t_{x}+t_{z} < 0),\ (t_{y} < 0),\ (t_{y} < \widetilde{t}_{z} < 0),\ (t_{w} > 0),\ (t_{1} \in \mathbb{R})\,\}\nonumber\\
{\cal R}_{3}\, &=\, 
\{\, t_{y} < t_{x} < 0, ( t_{x} < \widetilde{t}_{w} < 0 ), ( 0\, < s_{4,1} < \widetilde{t}_{w} - t_{y} ),\,  (t_{1} \in \mathbb{R})\,\}\nonumber\\
{\cal R}_{4}\, &=\, 
\{\, (0 < s_{4,1} < -t_{y}),\ (t_{y} < 0),\ (t_{w} > 0),\ (\widetilde{t}_{z} > 0),\ (t_{1} \in \mathbb{R})\,\}\nonumber\\
{\cal R}_{5}\, &=\, 
\{\, (t_{y} < t_{x} < 0),\, ( s_{4,1} - t_{y} < \widetilde{t}_{w} < t_{x}) , (  t_{x} < s_{4,1} < 0), \,  t_{1} \in {\bf R} \}\, \, .
\end{align}
%Note that this implies that $t_{y} < t_{x}$ in R_{3} which makes sense.
As discussed previously, this decomposition of ${\cal M}_{2,2}^{\textrm{qtf}}$ in ${\cal R}_{1}, \dots,\, {\cal R}_{5}$ is premised on smoothness of ${\cal U}$ inside each such region.  ${\cal U}$ can be computed directly in ${\cal R}_{a}\, \vert\, a\, \in\, \{1, \dots, 5\}$ and we find that, 
\begin{align}
\mathcal{U}\big|_{\mathcal{R}_1} &= t_{x}\, t_{y}, \\[6pt]
\mathcal{U}\big|_{\mathcal{R}_2} &= t_{x}\, t_{y} - t_{z}^{2}\nonumber\\
&= t_{y}\,\left(\,-\frac{\widetilde{t}_{z}^{\;2}}{t_{y}} + s_{4,1} + \widetilde{t}_{z} - t_{w}\,\right) , \\[6pt]
\mathcal{U}\big|_{\mathcal{R}_3} &= -\,t_{w}^{2} - 2 t_{w} t_{x} + t_{x}(t_{y} - t_{x}) \nonumber\\
&= t_{x} t_{y} - (t_{w} + t_{x})^{2}\nonumber\\
\mathcal{U}\big|_{\mathcal{R}_4} &= (t_{w} + t_{x} + t_{z})(t_{w} + t_{y} + t_{z})  - (t_{w} + t_{x} + t_{y} + t_{z})(2 t_{w} + t_{x} + t_{y} + 2 t_{z}) \nonumber\\
&= - (t_{w} + t_{z} + t_{x})^{2} - t_{y}\,(2 t_{w} + 2 t_{z} + t_{x} + t_{y}) \nonumber\\
&= - t_{y}\,\left( \frac{s_{4,1}^{2}}{t_{y}} + s_{4,1} + \widetilde{t}_{z} + t_{w} \right)
, \\[6pt]
\mathcal{U}\big|_{\mathcal{R}_5} &= -\,t_{w}^{2} - 2 t_{w} t_{z} + t_{x} t_{y} - t_{z}^{2} \nonumber\\
&= t_{x} t_{y} - (t_{w} + t_{z})^{2} \, . \nonumber\\
%&= - t_{x}\,\left( \frac{s_{4,1}^{2}}{t_{x}} - s_{4,1} - t_{w} - \widetilde{t}_{z} + t_{x} \right)
\end{align}
Several comments are in order at this stage.
\begin{enumerate}
\item As shown in appendix \ref{sec:Decomposition_of_U}, ${\cal R}_{1}$ has all the $\ln^{2}$ diverrgent cones. However this fact is also contained in the structure of ${\cal R}_{1}$ in that  there exists a sub-region inside ${\cal R}_{1}$ where in the vicinity of $t_{x}, t_{y} = 0$, the  integrand scales as $\frac{1}{t_{x}t_{y}}\, \exp^{ t_{x} \Psi_{1, x} + t_{y}\Psi_{1,y}}$  where $\Psi_{1,x/y}$ are independent of $t_{x}$ and $t_{y}$. By isolating this contribution, we can directly write down the tropical counter-term. 
\item Interestingly, we note that, 
\begin{align}
{\cal R}_{2}\, \medcup\, {\cal R}_{4}\, =\,  \big(\, s_{4,1} < - t_{y}\, \medcap\, t_{y} < \widetilde{t}_{z} \medcap\, t_{w} > 0  \big) \, .
\end{align}
This motivates us to define a sub-region in ${\cal R}_{24}\, \subset\, {\cal R}_{2} \medcup\, {\cal R}_{4}$ 
\begin{align}
{\cal R}_{24}&:=\nonumber\\
&\hspace*{-0.5in}\big(\, (\, t_{y} < 0\, \medcap\, t_{1}\, \in\, {\bf R}\, \medcap\, t_{w} > 0 \medcap\,  \big[\, ( 0 < s_{4,1} < - t_{y})\, \medcap\, \widetilde{t}_{z} > t_{y}\, )\, \medcup\, ( s_{4,1} <  t_{y}\, \medcap\, t_{y} < \widetilde{t}_{z} < 0\, )\, \big]\, \big) \, ,
\end{align}
where $t_{x} = s_{4,1} - t_{w} - (\widetilde{t}_{z} - t_{y})$.\\
As we will see, ${\cal R}_{24}$ contributes to the tropical counter term which cancels the log divergence arising from $t_{y}\, \rightarrow\, 0^{-}$.
\item Similarly, we see that
\begin{align}
{\cal R}_{3}\, \medcup\, {\cal R}_{5}\, =\, 	 (\, t_{y} < t_{x} < 0,\, ( s_{4,1} - t_{y} < \widetilde{t}_{w} < 0) , (  t_{x} < s_{4,1} <  \widetilde{t}_{w} - t_{y} )\,  t_{1} \in {\bf R})\, ),
\end{align}
with a sub-region ${\cal R}_{35} \subset {\cal R}_{3}\, \medcup\, {\cal R}_{5}$ defined as,
\begin{align}
{\cal R}_{35}\, &=\nonumber\\
&\hspace*{-0.7in}\{\, t_{y} < t_{x} < 0,\,  t_{1}\, \in\, {\bf R},\,  \big[\, \big(\, (t_{x} < \widetilde{t}_{w} < 0)\, \medcap\, s_{4,1} > 0\, \big)\, \medcup\, \big(\, ( s_{4,1} - t_{y} < \widetilde{t}_{w} < t_{x} )\, \medcap\, t_{x} < \widetilde{s}_{4,1} < 0\, \big)\, \big]\, \}  \, .
\end{align}
${\cal R}_{35}$ contributes a tropical counter-term which will cancel log divergence arising from $t_{x}\, \rightarrow\, 0^{-}$ as shown in the next section.
\item \emph{It can be checked that in the neighbourhood of UV singularity, , the first Surface Symanzik  ${\cal U}\, \approx\, t_{x} t_{y}\, \in\, \forall\, {\cal R}_{a}$. We will use this universal property of ${\cal U}$ in deriving the tropical counter terms at two-loop.}  As an example, let us analyze,  ${\cal U}\vert_{{\cal R}_{2}}$ where, 
\begin{align}
{\cal U} = t_{y}\, \bigg(\, t_{x} - \frac{t_{z}}{t_{y}}\, t_{z}\, \bigg) \, .
\end{align}
The term inside the bracket is non vanishing  as $t_{x} < - t_{w} - t_{z}$ and $t_{w} > - t_{y}$ in ${\cal R}_{2}$. Hence the only singularity is at $t_{y} = 0$ which leads to \emph{log divergence} as $t_{y}\, \rightarrow\, 0$, with  $t_{x}$ being bounded from above. In this limit, 
\begin{align}
{\cal U}\, \approx\, t_{x}\, t_{y} .
\end{align}
\item Similarly in the region ${\cal R}_{3}$ after some meditation over the inequalities we find that they imply,
\begin{align}
\widetilde{t}_{w} < - t_{z} < t_{y}.\nonumber\\
\end{align}
Hence ${\cal U}\vert_{{\cal R}_{3}}\, = t_{x}\, (\frac{\widetilde{t}_{w}^{2}}{t_{x}} - t_{y})$ only admits a zero from the parenthesis when $t_{w}, t_{z}, t_{y}\, \rightarrow\, 0$. Such a zero can not lead to any divergence as in that case, 
\begin{align}
\frac{1}{t_{y}^{2}}\, \mathrm{d} t_{w} \mathrm{d} t_{z} \mathrm{d} t_{y}\, \rightarrow\, t_{y}^{2} \, \mathrm{d}\bigg( \frac{t_{z}}{t_{y}} \bigg)\, \mathrm{d} \bigg( \frac{t_{w}}{t_{y}} \bigg) \, .
\end{align}
One can easily verify that this property continues to hold in ${\cal R}_{3}, {\cal R}_{4}$ and ${\cal R}_{5}$ regions as well.
\item The main upshot of this discussion is then the following : In the case of  $n=2$ and $L=1$,  ${\cal M}_{1, 2}^{\textrm{tf}}$ was simply one ``smooth'' region $S_{1}$. At $L = 2$, $S_{1}$ is generalized to a union over \emph{five} regions and we can now proceed with renormalizing the curve integrals in each of these regions  by adding appropriate tropical counter-terms. 
\end{enumerate}
We would like to emphasize that although this strategy looks needlessly complicated for $n=2$, the fact that $\displaystyle \bigcup_{a}\, \overline{{\cal R}}_{a}$ will continue to be a subregion inside ${\cal M}_{2, n}^{\textrm{qtf}}$ implies that the $n=2$ analysis will directly renormalize a large number of Feynman graphs at $n$ points.\footnote{As in the one-loop case, the complete ${\cal M}_{n}^{\textrm{qtf}}$ contains  $\displaystyle \bigcup_{a}\, \overline{{\cal R}}_{a}$ via a strict inclusions. Hence renormalization of the new regions has to be performed to obtain complete renormalized amplitude at two-loop. We will come back to this issue in sec. \ref{sec:2-loop-2-point-renormalization}.}

\emph{This will be the principal strategy we will adopt in the following sections. We believe that this strategy will continue to apply at higher loops as well. However, the complexity involved in finding the region ${\cal M}_{L, n}^{\textrm{qtf}}$ and decomposing it into subregions over each of which ${\cal U}$ is smooth will become harder and perhaps requires more sophisticated automation. In this paper, we are simply scratching a surface of the structure of renormalized curve integrals for $L > 1$.}

%%%%%%%%%%%%%%%%%%%%%%%%%%%%%%%%
%%%%%%%%%%%%%%%%%%%%%%%%%%%%%%%%
\section{Renormalization of ${\cal A}_{2, 2}$} \label{sec:2-loop-2-point-renormalization}
We now turn to the renormalization of planar two-loop amplitude for $n=2$. Thanks to the ``tree-loop factorization'' property of planar amplitudes, the generalization of the results at $L=2$ for arbitrary $n$ will only depend on the enhancement of ${\cal R}_{a}$ regions as we increase $n$. We come back to this point in the next section \ref{sec:Renormalization_of_planar_A_n>2_l=2}.

As the first surface Symanzik, ${\cal U}$ is discontinuous across the boundaries between different ${\cal R}$ regions, we are compelled to compute the tropical counter-terms separately in each region. Taking a cue from the tropical counter-term at one-loop and the forest formula, we can directly write down the counter-terms at two-loops.  Before going into the detailed form of the counter-terms, we first summarize the main result of this section. \\ 

\noindent Let ${\cal R}_{I(a), a}\big|_{a=1}^{5}$ be the subregions inside the global Schwinger space divided using following two properties.
\begin{enumerate}
\item In each ${\cal R}_{a}$ for a fixed $a,\,$   ${\cal U}$ is a smooth polynomial. But ${\cal U}$ is necessarily non-differentiable on $\displaystyle \bigcup_{a}\, ( \overline{{\cal R}}_{a}  - {\cal R}_{Ia} )$, where $\overline{{\cal R}}_{a}$ denotes the closure of ${\cal R}_{a}$.
\item ${\cal R}_{a}$ are specified by different inequalities on the tropical variables.\footnote{We note that although the different subregions ${\cal R}_{I(a), a}$ inside ${\cal R}_{a}$ simply correspond to cones in the case of $n=2$, these regions will continue to be tadpole free thanks to tree-loop factorisation and each of them will be a union of large number of cones. As a result our formula does not rely on existence of forests based on 1PI divergent subgraphs.} 
\item The defining inequalities for each such region ${\cal R}_{a}$ naturally lead us to the derivation of tropical counter terms. This is because in each such region, at least one parameter always scales with either $t_{x}$ or $t_{y}$ and hence one can isolate the divergent chambers such that the curve integrals over their compliment (in ${\cal R}_{a}$) are convergent. 
\item Within each $\overline{{\cal R}}_{a}$ we will label these ``divergent chambers'' as ${\cal R}_{a}^{\prime}$. Just as in the $L=1$ case, our definition of tropical counter-term as integral over $\overline{{\cal R}}_{a}^{\prime}$ as opposed to $\overline{{\cal R}}_{a}$ encodes the scheme dependence of the renormalized amplitude.
\end{enumerate}
Hence we claim that the renormalized 2 point planar amplitude at two loops has the following form, 
\begin{align}\label{a22rclaim}
{\cal A}_{2, 2}^{R}&=\, {\cal A}_{2, 2}  + \sum_{a=1}^{5}\, {\cal C}_{a},
\end{align}
where every tropical counter-terms ${\cal C}_{a}$ can be read off by simply analyzing the regions ${\cal R}^{\prime}_{a}$ in co-ordinates where at most one co-ordinate lies between $[0, t_{x}]$ and/or $[0, t_{y}]$. 

As in the one-loop case, the integrand associated to each counter term can be written simply by inspecting the zero's of ${\cal U}$ in ${\cal R}_{a}\, \forall\, a$ and using the Taylor expansion of $\frac{{\cal F}}{{\cal U}} - {\cal Z}$ around ${\cal U} = 0$.  We now explicitly write  the five tropical counter-terms. 

%%%%%%%%%%%%%%%%% 
\subsubsection*{Tropical counter term in ${\cal R}_{1}$}
The first counter-term,  ${\cal C}_{1}$ is an integral over the following regions inside ${\cal R}_{1}$.
\begin{align}
{\cal R}_{1\, \alpha}^{\prime}&=\, \{  t_{x}^{\alpha}\, \in\, (0,  t_{x})\,  \medcap\, t_{y}^{\alpha} \in\, (0, t_{y})\, \medcap\, s_{\alpha} <  \textrm{min}(t_{x},\, t_{y})\, \}\, \nonumber\\
{\cal R}_{1\, \alpha}^{\prime x}&=\, \{\, t_{x}^{\alpha}\, \in\, (0, t_{x})\,  \medcap\, t_{y}^{\alpha} \in\, (0, t_{y})\, \medcap\, \big(\, ( t_{x} < s_{\alpha}\, <\, 0\, )\, \medcap\, (s_{\alpha} < t_{y} )\, \big)\}\, \nonumber\\
{\cal R}_{1\, \alpha}^{\prime y}&=\, \{ ( t_{x}^{\alpha}\, \in\, (0,  t_{x})\,  \medcap\, t_{y}^{\alpha} \in\, (0, t_{y})\, \medcap\, \big(\, ( s_{\alpha}\, <\, t_{x}\, )\, \medcap\, ( t_{y} < s_{\alpha} < 0)\, \big)\},
\end{align}
for all $\alpha$ and where $t_{x}, t_{y} < 0$. 
\begin{align}\label{c1forn2l2}
{\cal C}_{1}\, &=\, \sum_{\alpha = 1}^{4}\, \left[\, \int_{\overline{{\cal R}}^{\prime}_{1,\alpha}}\, \frac{1}{t_{x} t_{y}}\, e^{t_{x} \Psi_{\alpha, x}(\tau_{1,\alpha}, \tau_{2,\alpha}, s_{1 \alpha}; S_{0}, \vec{\Theta}_{0})}\, e^{t_{y}\Psi_{\alpha, y}(\tau_{1 \alpha}, \tau_{2 \alpha}, s_{1 \alpha},\, S_{0}, \vec{\Theta}_{0})}\, e^{\Psi_{0}(\tau_{1 \alpha}, \tau_{2 \alpha}, s_{1 \alpha},\, S, \vec{\Theta})}\right.\nonumber\\
&\hspace*{0.5in}\left. -\, \int_{\overline{{\cal R}}^{\prime x}_{1,\alpha}}\, \frac{1}{t_{x}}\, e^{t_{x} \Psi^{x}_{\alpha}(\tau_{1,\alpha}, \tau_{2,\alpha}, \frac{s_{1 \alpha}}{t_{y}}; t_{y}, S_{0}, \vec{\Theta}_{0})}\,  e^{\Psi_{0}(\tau_{1 \alpha}, \tau_{2 \alpha}, \frac{s_{1 \alpha}}{t_{y}}\, S, \vec{\Theta})}\right.\nonumber\\
&\hspace*{0.5in} \left. -\, \int_{\overline{{\cal R}}^{\prime y}_{1,\alpha}}\, \frac{1}{t_{y}}\, e^{t_{y} \Psi^{y}_{\alpha}(\tau_{1,\alpha}, \tau_{2,\alpha}, \frac{s_{1 \alpha}}{t_{x}}; t_{x}, S_{0}, \vec{\Theta}_{0})}\,  e^{\Psi_{0}(\tau_{1 \alpha}, \tau_{2 \alpha}, \frac{s_{1 \alpha}}{t_{x}}\, S, \vec{\Theta})}\, \right]\, ,
\end{align}
where $\tau_{1\, \alpha}\, =\, \frac{t_{x}^{\alpha}}{t_{x}},\, \tau_{2, \alpha}\, =\, \frac{t_{y}^{\alpha}}{t_{y}}$, and 
$\{\Psi_{\alpha, x}, \Psi_{\alpha, y},\, \Psi_{\alpha}^{x}, \Psi_{\alpha}^{y}\}$ are obtained by Taylor expanding $\frac{{\cal F}}{{\cal U}} - {\cal Z}$ in the corresponding regions and expanding the co-efficients.  We note that ${\cal C}_{1}$ has a term that diverges as $\ln^{2}$ and two terms that diverge logarithmically. These terms occur with opposite signs to ensure that sub-divergences in ${\cal R}_{1}$  cancel consistently.  We thus see that the tropical counter-terms structurally resemble the Forest formula in the sense that at $L > 1$ loops, the formula involves subtraction terms with alternating signs.\footnote{Co-efficients of Taylor expansion in $t_{x}, t_{y}$  can be computed using the mathematica file attached with this draft. Using the codes given there, we have verified that in the re-scaled variables $\tau$,  $\frac{{\cal F}}{{\cal U}} - {\cal Z}$ admit a taylor expansion in $t_{x}, t_{y}$.}

%%%%%%%%%%%%%%%%%
\subsubsection*{Tropical counter-term in ${\cal R}_{24}$}
As in the case of ${\cal R}_{1}$, we claim that the tropical counter-term is an integral over the following region inside ${\cal R}_{2}$.
\begin{align}
{\cal R}_{24}^{\prime}&:=\nonumber\\
&\hspace*{-0.3in}\big(\, (\, t_{y} < 0\, \medcap\, t_{1}\, \in\, {\bf R} - [ t_{y}, -t_{y} ]\, \medcap\, t_{w} > 0\, \medcap\,  ( 0 < s_{4,1} < - t_{y})\, \medcap\, \widetilde{t}_{z} > t_{y}\, )\, \medcup\, ( s_{4,1} <  t_{y}\, \medcap\, t_{y} < \widetilde{t}_{z} < 0\, )\, \big)
\end{align}
\begin{align}
{\cal C}_{24}\, =\,  \int_{\overline{{\cal R}}_{24}^{\prime}}\,  \frac{1}{t_{y}}\, \frac{1}{t_{x}^{2}}\, {\cal K}(s_{4,1},  t_{y}, \tau_{z}, t_{w}, t_{1})\, e^{\frac{1}{t_{x}}\, t_{y}\, \Psi_{1, {\cal R}_{24}}(S_{0},\, \vec{\Theta}_{0})}\, e^{\frac{1}{t_{x}}\, \Psi_{0, {\cal R}_{24}}(S, \vec{\Theta})}\, ,
\label{C24_2l2p}
\end{align}
with $t_{x} = s_{4,1} - t_{w} - (\widetilde{t}_{z} - t_{y})$.

%%%%%%%%%%%%%%%%%
\subsubsection*{Tropical Counter-term in ${\cal R}_{35}$.}
Let, 
\begin{align}
{\cal R}_{35}^{\prime}(t_{x}, t_{y}, \widetilde{t}_{w}, s_{4,1}, t_{1})&:=\nonumber\\
&\hspace*{-1.3in}\{\, t_{y} < t_{x} < 0,\,  t_{1}\, \in\, {\bf R} - [\, t_{x}, - t_{x}\, ],\,  \big[\, \big(\, (t_{x} < \widetilde{t}_{w} < 0)\, \medcap\, s_{4,1} > 0\, \big)\, \medcup\, \big(\, ( s_{4,1} - t_{y} < \widetilde{t}_{w} < t_{x} )\, \medcap\, t_{x} < \widetilde{s}_{4,1} < 0\, \big)\, \big]\, \}  
\end{align}
\begin{align}
{\cal C}_{35}\, =\,  \int_{\overline{{\cal R}}_{35}^{\prime}}\,  \frac{1}{t_{x}}\, \frac{1}{t_{y}^{2}}\, {\cal K}(t_{x}, t_{y}, \tau_{w},  s_{4,1}, t_{1})\, e^{\frac{1}{t_{y}}\, t_{x}\, \Psi_{1, {\cal R}_{3}}(S_{0}, \vec{\Theta}_{0})}\, e^{\frac{1}{t_{y}}\, \Psi_{0, {\cal R}_{3}}(S, \vec{\Theta})}\, .
\label{C35_2l2p}
\end{align}
 $\Psi$ polynomials can be computed using the ancillary file {\bf \texttt{`Two-loop two-point F, U, K, and Z for counter terms.nb'}}. 

%%%%%%%%%%%%%%%%%%%%%%%%%%%%%%%%
\subsection{Renormalization of planar  ${\cal A}_{2, n > 2}$}\label{sec:Renormalization_of_planar_A_n>2_l=2}
As noticed in the previous section, the renormalization of ${\cal A}_{2,2}$ appears needlessly complicated for an amplitude which is simply a sum over 6 graphs. However the complexity came from the analysis of decapitated region ${\cal M}^{\textrm{qtf}}_{2,2}$ and the fact that inequivalent pseudo-triangulations of $\Sigma_{2,2}$ mandates a Four fold replication of each Fat graph with the redundancy eventuall removed thanks to the presence of ${\cal K}$ inside the curve integrand. We will now show that the fruits of our labour in deriving ${\cal M}_{2,2}^{\textrm{qtf}}$ immediately lead us to the derivation of ${\cal M}_{2,n}^{\textrm{qtf}}$.  

We already know that for any $n$ the tadpole free region can always be decomposed into union of 5 closed regions $\overline{{\cal R}}_{a}\vert_{a=1}^{5}$ such that inside each ${\cal R}_{a}$ ${\cal U}$ is a smooth polynomial of degree $L$ in $(t_{1}, \dots, t_{n-1}, t_{w}, t_{z}, t_{x}, t_{y})$. 
\begin{align}
{\cal M}_{n}^{\textrm{qtf}}\, =\, \bigcup_{a=1}^{5}\, \overline{{\cal R}}_{a}.
\end{align}
The regions ${\cal R}_{2}, \dots, {\cal R}_{5}$ whose contribution to the curve integral for $n=2$ is logarithmically divergent do not depend on $n$. This immediately follows from the fact that in the case of $n=2$ they were completely determined by inequalities which only depends on $(t_{w}, t_{z}, t_{x}, t_{y})$. Hence
\begin{align}
{\cal R}_{a}(n) = {\cal R}_{a}\, \forall\, a\, \in\, \{2, \dots, 5\} \, .
\end{align}
This immediately implies that the tropical counter-terms over ${\cal R}_{2}, \dots,\, {\cal R}_{5}$ are direct generalisation of the two counter-terms written in $n=2$ case in eqns (\ref{C24_2l2p}, \ref{C35_2l2p}), with the only difference being that the counter-terms are only over sub-regions inside each of the ${\cal R}_{a}$ in which 
\begin{align}
t_{I}\, \notin\,  [ -t_{x}, t_{x}]\, \forall\, 1 \leq\, I\, \leq\, n-1\, \textrm{or}\, t_{I}\, \notin\, [ - t_{y}, t_{y} ]\, \forall\, 1\, \leq\, I\, \leq\, n-1 \, .
\end{align}
It then only remains to find the region ${\cal R}_{1}(n)$ that has no cones dual to fat graphs containing tadpoles with 
\begin{align}
{\cal U}\vert_{{\cal R}_{1}(n)}\, =\, t_{x}\, t_{y}.
\end{align}
Based on the analysis for $n=2$ we expect ${\cal R}_{1,\, \alpha}(n)\, \subset\, {\cal R}_{1,\, \alpha}(n)$ so that, 
\begin{align}
{\cal R}_{1}(n)\, =\, \sum_{\alpha = 1}^{4}\, {\cal R}_{1,\, \alpha}(n)\nonumber\\
{\cal R}_{1,\, \alpha}(n)\, :=\, \medcup_{I=1}^{n-1}\, {\cal R}_{1,\, \alpha,\, I}(n)\, ,
\end{align}
where
\begin{align}
{\cal R}_{1,\, \alpha,\, 1}(n)\, =\, {\cal R}_{1,\, \alpha}\, ,
\end{align}
which is defined in eqn.(\ref{r1tr5innc}). 

We now conjecture that for any $n$, ${\cal R}_{1}^{\alpha}(n)$ is parametrized in terms of following set of inequalities, 
\begin{align}\label{qtf_region_2lnp}
\mathcal{R}_{1\, I, \alpha=1}(n)\, &= (t_y < 0) \ \medcap\ (t_z < 0) \ \medcap\ (t_w > 0) \ \medcap\ (t_x < -t_w)\nonumber\\
&\hspace*{0.5in} \bigcap\, \Big(\, (\, \medcap_{J=I+1}^{n-1}\, (t_{K}\, +\, \dots\, +\, t_{n-1} ) > - t_{y} - t_{z}\, )\,  \medcap\, ( - t_{z} < \sum_{M=I}^{n-1} t_{M} < - t_{y} - t_{z}\, ) \Big)\nonumber\\
{\cal R}_{1,\, I,\, \alpha=2}(n)\, &=\,  (t_y < 0) \ \medcap\ (0 < t_z < -t_y) \ \medcap\ (t_w < -t_z) \ \medcap\ (t_x < 0)\nonumber\\
&\hspace*{1.0in} \bigcap\, \Big(\, (\, \medcap_{J=I+1}^{n-1}\, (t_{K}\, +\, \dots\, +\, t_{n-1} ) > - t_{w} - t_{x} - t_{y} - 2 t_{z}\, )\nonumber\\
&\hspace*{1.2in}\medcap\, ( -t_{w} - t_{y} - 2 t_{z} < \sum_{M=I}^{n-1} t_{M} < - t_{y} - t_{z}\, <\, -t_w - t_x - t_y - 2t_z )\, \Big)\nonumber\\
{\cal R}_{1\, I,\, \alpha=3}(n)&=\, (t_{y}\, <\, 0)\, \medcap\, (0 < t_{z} < - t_{y})\, \medcap\, (t_{w} > 0)\, \medcap\,( - t_{w} <\, t_{x}\, <\, 0)\nonumber\\
& \bigcap\, \Big(\, (\, \medcap_{J=I+1}^{n-1}\, (t_{K}\, +\, \dots\, +\, t_{n-1} ) > - t_{w} - t_{x} )\, \medcap\, (\, - t_{w}\, <\, \sum_{M=I}^{n-1} t_{M} < -t_w - t_x\, )\, \Big)\nonumber\\
{\cal R}_{1\, I,\, \alpha=4}(n)&=\, ( t_y < 0 ) \ \medcap\ (t_z > -t_y) \ \medcap\ (t_w > 0) \ \medcap\ 
(-t_w - t_y - t_z < t_x < -t_w)\nonumber\\
&\bigcap\, \Big(\, (\, \medcap_{J=I+1}^{n-1}\, (t_{K}\, +\, \dots\, +\, t_{n-1} ) > -2t_w - t_x - t_y - t_z )\, \nonumber\\
&\quad\quad\medcap\, (\, -2t_w - t_x - t_z\, <\, \sum_{M=I}^{n-1} t_{M} < -2t_w - t_x - t_y - t_z\, )\, \Big) \, .
\end{align}
Even though we do not have a complete proof of this conjecture, we have verified it for $ n = 3$ and based on our intuition from tree-loop factorization that is manifest in reference tadpole fat graph, we expect the conjecture to be satisfied for any $n$.\footnote{Having said that, a clear analytic proof of this statement is highly desirable.} 

Once we have found ${\cal R}_{1,\, \alpha}(n)$, we can isolate the divergent chambers ${\cal R}_{1,\alpha}(n)^{\prime}$ simply by inspection.  Finally, the structural form of the counter term is essentially the same as that in eqn. (\ref{c1forn2l2}). The real complexity comes in expansion of $\frac{{\cal F}}{{\cal U}} - {\cal Z}$ around $(t_{x}, t_{y}) = (0,0)$ for which better codes will need to be developed for pratical computations.

%%%%%%%%%%%%%%%%%%%%%%%%%%%%%%%%
%%%%%%%%%%%%%%%%%%%%%%%%%%%%%%%%
\section{Conclusion and Outlook} 

Curve integrals unify perturbative QFT and string amplitudes into the same paradigm.  String amplitudes are integrals of certain top forms over moduli space of hyperbolic geometries on a punctured Riemann surface, and the QFT amplitudes for colored theories can be understood as the limit where the so called surface variables $u_{C}$ (that generalize Plucker co-ordinates to moduli space with punctures) tropicalize.  The tropical curve integrals associated to massive $\textrm{Tr}(\Phi^{3})$ theory are UV divergent. Our goal in this paper has been to show that one can develop subtraction scheme for the bare curve integrals which do not rely on decomposition of the curve integrals into Feynman cones.  We believe that this reinforces our belief that positive geometries are indeed fundamental structures that can generate physical amplitudes in QFT. 

For massive $\textrm{Tr}(\Phi^{3})$  amplitudes in $D = 4$ dimensions where  tadpole free amplitudes only have logarithmic divergences, we proved that there are two complementary approaches to renormalizing the curve integral. In the first approach, we used the fact that 1PI divergent subgraphs as well as forests generated by them are common to a family of larger graphs (all of which are related by mutation ``away from the sites of the forest'').  The surface forest formula consisted of subtracting a single counter-term for a given forest regardless of its embedding. Although this formula can be generalized to  more complicated theories such as $\phi^{3}$ theory in $D = 6$ dimensions, it violates the fundamental premise of positive geometry program by relying on partial dissections of $\Sigma_{L, n, h}$ which are dual to forests. 

Starting with the tadpole fat graph as the reference graph which fixes a basis for the $g$-vectors, we proved that the one-loop planar amplitude ${\cal A}_{1, n}(p_{1}, \dots, p_{n})$ could be renormalized without appealing to surface forests. This was done in two steps. We first isolated the region in global Schwinger space which consisted only of tadpole graphs and by decapitating such a region we obtained a curve integral with purely logarithmic divergence. The decapitated region in the global Schwinger space has a natural decomposition into $n-1$  regions such that the contribution of each of these regions to the curve integral has a log divergence that arises when precisely two parameters, $t_{x}$ and $\displaystyle \sum_{j=i}^{n-1} t_{j} + t_{x} \,  \vert \ 1 \leq\, i\, \leq\, n-1$ approach \emph{zero}. Thanks to kinematic renormalization of Brown and Kreimer, we could then immediately define a tropical counter-term for each of these regions.

We then showed how one could derive a tropical counter term at two loops at leading order in 't Hooft expansion.  Although proving the existence of, and deriving tropical counter term for a generic $n$-point two-loop amplitude is considerably more complicated and nuanced than in the case of one-loop, we could prove the existence of tropical counter terms as follows. Starting with the two-loop tadpole fat graph as a reference, we began by isolating the region which contained no quadratically divergent tadpoles. We then divided this region ${\cal M}_{n}^{\textrm{qtf}}$  into sub-regions, such that in each of  these sub-regions  ${\cal U}$ is smooth. Once we found all of such sub-regions, then by meditating over the curve integral into each of these sub-regions we could immediately write down the corresponding tropical counter term. In the end, a signed sum over all the tropical counter terms subtracts all the UV divergences from the bare curve integral and produces renormalized amplitude in kinematic renormalization scheme. We showed how our strategy worked explicitly for the case of ${\cal A}_{L=2, n =2}$. 

We expect that our strategy can be generalized at higher loop when $h = 0$. This is because if we choose the reference graph as tadpole $L$-loop graph, then due to the so called telescopic property, the first surface Symanzik ${\cal U}$ only depends on the global Schwinger parameters of $\Gamma_{\textrm{ref}}(L,1)$  and any graph obtained by adjoining a $n-1$ point tree graph to $\Gamma_{\textrm{ref}}(L, 1)$ keeps ${\cal U}$ invariant.  Then, by decapitating quadratically divergent tadpoles, and decomposing the resulting regions into sub-region in each of which ${\cal U}$ is smooth, we expect to write the tropical counter term for each of the sub-region in which ${\cal U}$ is smooth by simply isolating the divergence. 
 
 Our analysis also indicates that this framework extends to  non-planar amplitudes. Although a detailed analysis will be published elsewhere, here we very briefly sketch renormalization of one-loop amplitudes for $ n = 2, 3$ for $h = 1$. 
\begin{itemize}
\item \underline{Renormalization of ${\cal A}_{L=1, n=2, h = 1}$} :

We start with the tadpole graph as a reference (see the one loop non-planar example in \cite{Arkani-Hamed:2023mvg} for the depiction of such a fat graph). 
It can be immediately checked that the tadpole free regions defined by, $\alpha_{11} = \alpha_{22} = 0$ imply that,  $t_{2}\, \geq\, 0$. It was shown in \cite{Arkani-Hamed:2023mvg} that the non-triviality of Mirzakhani kernel implies that $ t_{1}\, \geq\, 0$. Hence the two point amplitude is 
\begin{align}
&\int_{0}^{\infty} \frac{t_{1}}{(t_{1} + t_{2})^{2}}\, e^{\frac{{\cal F}}{{\cal U}} - {\cal Z}}\nonumber\\
=\, &\int_{P^{1}_{+}}\, \tau_{1} \D\tau_{2} - \tau_{2} \D\tau_{1} \int_{0}^{\infty}\, \D\lambda\, \frac{1}{\lambda}\, e^{- \lambda\, \big[\frac{{\cal \widetilde{F}}}{{\cal \widetilde{U}}}\big]}\, ,
\end{align}
where $t_{i} = \lambda \tau_{i}$ with $\sum_{i} \tau_{i}^{2} = 1$. 
Thus,  the renormalized non-planar one-loop amplitude is
\begin{align}
{\cal A}_{2, R}^{\textrm{one-loop, np}}\, =\, {\cal A}_{2}^{ \textrm{one-loop, np}}\, -\, \int_{P^{1}_{+}}\, \tau_{1} \D\tau_{2} - \tau_{2} \D\tau_{1} \int_{0}^{\infty}\, \D\lambda\, \frac{1}{\lambda}\, e^{- \lambda\, \big[\frac{{\cal \widetilde{F}}(S_{0})}{{\cal \widetilde{U}}}\big]}\, .
\end{align}

\item \underline{Renormalization of ${\cal A}_{L=1,n=3, h=1}$ }:

 We consider the reference graph where $p_{2}, p_{3}$ legs are incident on the loop from ``inside'' and the state with momentum $p_{1}$ is incident on the loop ``from outside''. (See figure in \cite{Arkani-Hamed:2023mvg})
Mirzakhani kernel is non-trivial if and only if $t_{x}\, \geq\, 0$. Inside this region,  we now want to identify ${\cal M}_{n=3}^{\textrm{tf}}$.\\
It can then be immediately verified that the corresponding headlight functions are given by the following formulae. 
\begin{align}
\alpha_{11} &= t_{x} + t_{y} - \textrm{max}(0, t_{x} + t_{y})\nonumber\\
\alpha_{22} &=\, t_{x} + t_{y} - 2\textrm{max}(t_{x} + t_{y}, t_{x} + t_{y} + t_{1})\, + \textrm{max}(t_{x} + t_{y}, t_{x} + t_{y} + t_{1}, t_{x} + t_{y} + 2 t_{1})\nonumber\\
\alpha_{33} &=\, - t_{x} - \textrm{max}(0, t_{1})\, - \textrm{max}(t_{y}, t_{1} + t_{y})\, +\, \textrm{max}(0, t_{x} + t_{y}, t_{x} + t_{y} + t_{1}, t_{x} + t_{y} + 2 t_{1})\, .
\end{align}
We then see that for $t_{x}, t_{y}\, \geq\, 0$, all the three headlight functions vanish. As a result the bare curve integral formula in the tadpole free region takes the form,  
\begin{align}
{\cal A}_{3}^{\textrm{one-loop, np}}\, =\, \int_{-\infty}^{\infty} \D t_{1}\, \int_{0}^{\infty} \D t_{x}\, \D t_{y}\, \frac{t_{x}}{(t_{x} + t_{y})^{3}}\, e^{\frac{{\cal F}}{{\cal U}} - {\cal Z}}\, .
\end{align}
The UV divergence arises simply from the ray $(t_{x}, t_{y})\, =\, \lambda(\tau_{x}, \tau_{y})$ as $\lambda\, \rightarrow\, 0^{+}$ and hence the renormalized amplitude is given by, 
\begin{align}
{\cal A}_{3, R}^{\textrm{one-loop, np}}\, =\, {\cal A}_{3}^{\textrm{one-loop, np}}\, -\, \int_{-\infty}^{\infty}\, \D t_{1}\, \int_{P^{2}_{+}}\, \D\mu(\tau_{x}, \tau_{y})\, \int_{0}^{\infty}\, \frac{1}{\lambda}\, \frac{\tau_{x}}{(\tau_{x} + \tau_{y})^{3}}\, e^{\Psi_{0}}\, e^{-\lambda\, \Psi_{1}}\, .
\end{align}
where $\Psi_{0}, \Psi_{1}$ defined precisely as in the case of ${\cal A}_{L=1, n, h=0}$. 
\end{itemize}
It wil be interesting to see if the idea of deriving tropical counter-terms by examining the singularities of the curve integrals in a specific $g$-vector basis can be generalized to $\textrm{Tr}(\Phi^{3})$ theory in $D = 6$ dimensions. This will be an interesting case study overlapping as well as quadratic divergences are present even in the tadpole free regions. Thus finding appropriate tropical counter terms for this theory, without taking recourse to the forests will be a non-trivial test for the ideas presented in this paper.   
%%%%%%%%%%%%%%%%%%%%%%%%%%%%%%%%
%%%%%%%%%%%%%%%%%%%%%%%%%%%%%%%%
\section*{Acknowledgements}
We would like to thank Amit Suthar for collaboration in the initial stages of this project. We thank Michael Douglas, Rajesh Gopakumar, Thomas Grimm, Dirk Kreimer, Shiraz Minwalla, Prashanth Raman, Kirubakaran S., Giulio Salvatori, Tanmoy Sengupta, Ashoke Sen, O. Shetye, Aninda Sinha and Amit Suthar for discussions and clarifications at various stages of this project. We are especially indebted to Nima-Arkani Hamed for several insightful discussions at crucial stages of this project. The work of PB is supported by the Ramanujan fellowship grant RJF/2023/000007 from the Anusandhan National Research Foundation (ANRF), India and he would also like to acknowledge the hospitality at Chennai Mathematical Institute and ICTS-TIFR, Bangalore  where part of the work was done. H would like to thank A. Virmani for constant encouragement and support.  AL thanks ICTS-TIFR, Bangalore for hospitality during the course of this work. 

%%%%%%%%%%%%%%%%%%%%%%%%%%%%%%%%
%%%%%%%%%%%%%%%%%%%%%%%%%%%%%%%%
%%%%%%%%%%%%%%%%%%%%%%%%%%%%%%%%
\appendix

\newpage

%%%%%%%%%%%%%%%%%%%%%%%%%%%%%%%%
%%%%%%%%%%%%%%%%%%%%%%%%%%%%%%%%
\section{Gentle blossomed dissection quivers for a disc with $1 \leq\, L\, \leq\, 2$ punctures : pseudo-triangulation approach.} \label{sec:pseudo triangulation}
In this appendix we review the computation of $g$-vectors associated with the dissection curves of a disc $\Sigma_{L,n}$ with $n$  marked points on the boundary and a set  $\{p_{1}, \dots, p_{L}\}$ of $L$ punctures. Any reference fat graph $\Gamma_{L,n}$ which can be embedded in $\Sigma_{L,n}$,\footnote{Although $g$-vectors can be computed for any fat graph, in this appendix we will restrict our attention to planar graphs as in the rest of the paper.}  generates a fan in ${\bf R}^{d}$ where the rays of the fan are determined by the $g$-vectors.\footnote{A $d$ dimensional fan is a union of $d$ dimensional cones in $R^{n}$ such that two cones can intersect only in subspace of co-dimension $1$ and the union over all the cones span $R^{d}$. For a precise characterization, we refer the reader to \cite{PPPP}.}  More in detail, to any curve inside the fat graph $C$, we can associate a vector $g_{C}\, \in\, {\bf R}^{d}$ where $d$ is the number of curves required to completely triangulate (strictly speaking pseudo-triangulate) $\Sigma_{n,L}$. A beautiful and direct way to compute $g_{C}\, \  \forall\, C\, \in\, \Sigma_{L,n}$ with respect to any fat graph $\Gamma_{L,n}$ was given in \cite{Arkani-Hamed:2023lbd}. These $g$-vectors satisfy the beautiful condition
\begin{align}
\alpha_{C^{\prime}}(g_{C})\, =\, \delta_{C,C^{\prime}}\, ,
\end{align}
where $\alpha_{C}$ are the headlight functions.\footnote{Headlight functions are also called $c$-vectors in mathematics literature, e.g. \cite{PPPP} .} As an explicit example, we list down the set of all $g$-vectors for walks in $\Sigma_{1,3}$ with respect to the tadpole fat graph (see Fig. \ref{fig:3p1l_tadpole_ref}). 
\begin{figure}[h!]
\centering
\includegraphics[width=0.5\linewidth]{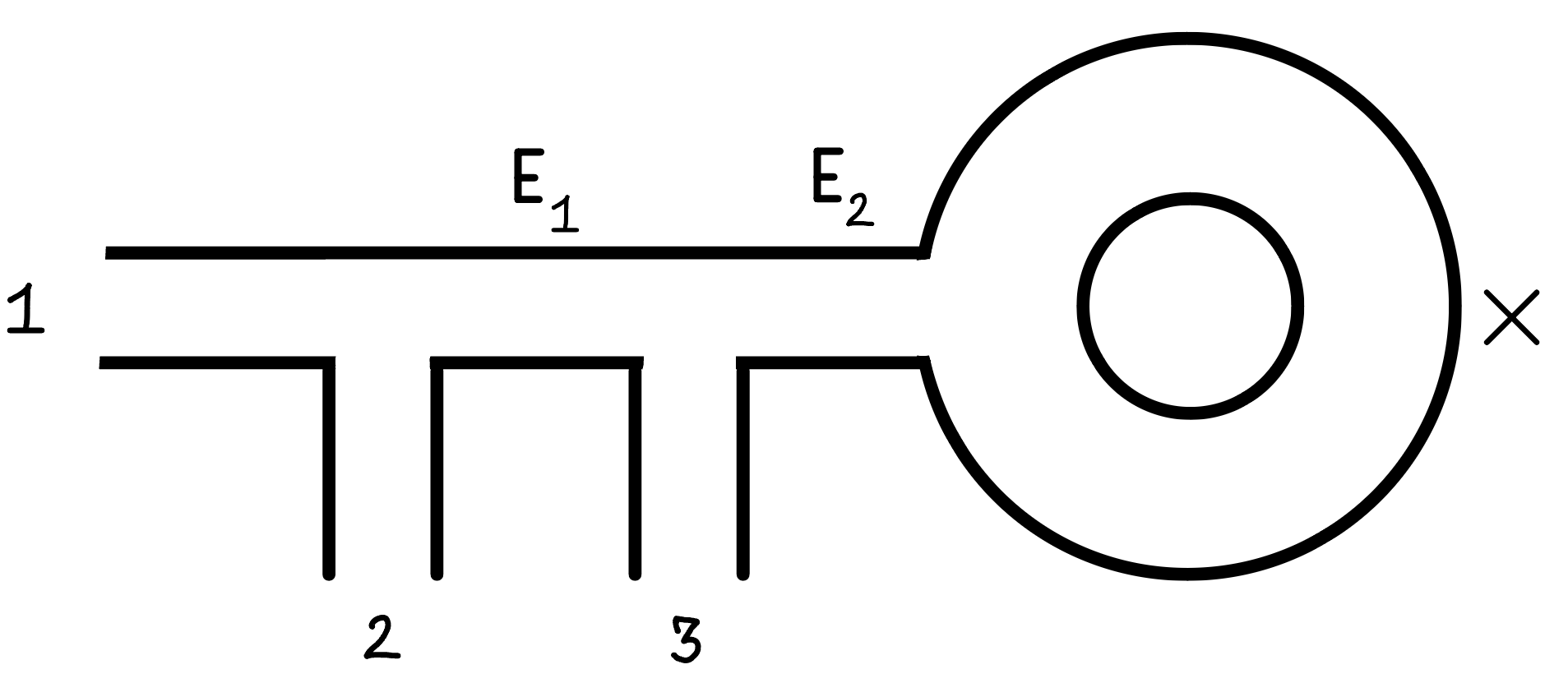}
\caption{Three-point one-loop tadpole reference fatgraph}
\label{fig:3p1l_tadpole_ref}
\end{figure}

In the table\, \ref{tab:g-vectors_matching_PT}, the unit vectors along which the $g$-vectors are evaluated correspond to the walks $\{ (1,3),\, (1, 1),\, (10)_{\textrm{cc}}\, \}$ where `$\textrm{c} (\textrm{cc})$' stands for clockwise (counter-clockwise) winding.
\begin{table}[h!]
\renewcommand{\arraystretch}{1.4}
\centering
\begin{tabular}{|>{\raggedright\arraybackslash}m{3cm}|
>{\raggedright\arraybackslash}m{5cm}|
>{\raggedright\arraybackslash}m{4cm}|}
\hline
\textbf{g-vectors} & \textbf{Basis : $( e_{E_{1}},\, e_{E_{2}},\, e_{X} )$} & \textbf{$g_{\textrm{proper walk}}$ with basis $\{\, e_{13},\, e_{1\widetilde{1}},\, e_{10}\,\}$} \\
\hline
$g_{10_{cc}}$ & $\{0, 1, -1\}$ & $g^{-}_{1^{\circ}0}$ \\
$g_{20_{cc}}$ & $\{-1, 1, -1\}$ & $g^{-}_{2^{\circ}0}$ \\
$g_{30_{cc}}$ & $\{0, 0, -1\}$ & $g^{-}_{3^{\circ}0}$ \\
$g_{21}$ & $\{-1, 1, 0\}$ & $g_{2\widetilde{1}}$ \\
$g_{32}$ & $\{-1, 0, 0\}$ & $g_{3\widetilde{2}}$ \\
$g_{13}$ & $\{1, 0, 0\}$ & $g_{13}$ \\
$g_{11}$ & $\{0, 1, 0\}$ & $g^{+}_{1^{\circ}\widetilde{1}^{\circ}}$ \\
$g_{22}$ & $\{-2, 1, 0\}$ & $g_{(2^{\circ}\widetilde{2}^{\circ})^{+}}$ \\
$g_{33}$ & $\{0, -1, 0\}$ & $g_{(3^{\circ}\widetilde{3}^{\circ})^{+}}$ \\
\hline
\end{tabular}
\caption{ comparison between g-vector with basis $( e_{E_{1}}, e_{E_{2}}, e_{X} )$ and $g_{\text{proper walk}}$ in pseudo-triangulation with basis $\{\, e_{13},\, e_{1\widetilde{1}},\, e_{10}\,\}$.}
\label{tab:g-vectors_matching_PT}
\end{table}

%%%%%%%%%%%%%%%%%%%%%%%%%%%%%%%%
\subsection*{Pseudo-triangulation model for $L=1$ and computation of $g$-vectors}
For $\Sigma_{1, n}$ the $g$-vectors can also be computed by realising that any $\Gamma_{L,n}$ defines a so-called \emph{gentle blossomed dissection quiver} which can be used to compute the set of all $g$-vectors. On a disc without punctures, this method was explicitly developed in a series of beautiful papers  \cite{PPPP, palu2019non}. 
In the case of $L=1$,   $g$-vector fan associated to curves on a planar tadpole fat graph can also be computed using the so called pseudo-triangulation model which was introduced by Ceballlos and Pilaud to define cluster $D$-polytopes.\footnote{Strictly speaking, the positive geometry associated to planar one-loop amplitude is known has $\widehat{D}$-polytope and was discovered by Arkani-Hamed, Frost, Salvatori and Plamondon in \cite{afpst}. As we are not using positive geometries directly, we refer the reader to \cite{Jagadale:2020qfa} for a review of the definition of $\widehat{D}$ polytope.} In this appendix, we propose an extension of the pseudo-triangulation model for $L = 2$ which can be used to (a) define the combinatorial positive geometry whose vertices are pseudo-triangulations of $\Sigma_{2,n}$ and (b) to compute $g$ vectors for all the curves in such a disc, given a reference triangulation.

We now review the $g$-vector computation given in \cite{Jagadale:2022rbl} for  $\Sigma_{1,3}$. We start with a reference dissection of the $6$-gon which has a puncture in the middle. The vertices of the 6-gon are labelled as $\{ 1, \dots, 3,\, \widetilde{1}, \dots, \widetilde{3}\, \}$ and we choose the specific reference pseudo-triangulation, $\textrm{PT}_{R}$,  that is dual to the tadpole planar fat graph. The chords of these dissections are
\begin{align}
\textrm{PT}_{R}\, =\, \{\, 1 3,\, (1\widetilde{1})_{L},\, 10, \widetilde{1}0,\,  (1 \widetilde{1})_{R,}, \widetilde{1} \widetilde{3}\, \}.
\end{align}
as shown in the Fig. \ref{fig:pseudo_triangulation_1l3p}.
\begin{figure}[h!]
\centering
\includegraphics[width=0.3\linewidth]{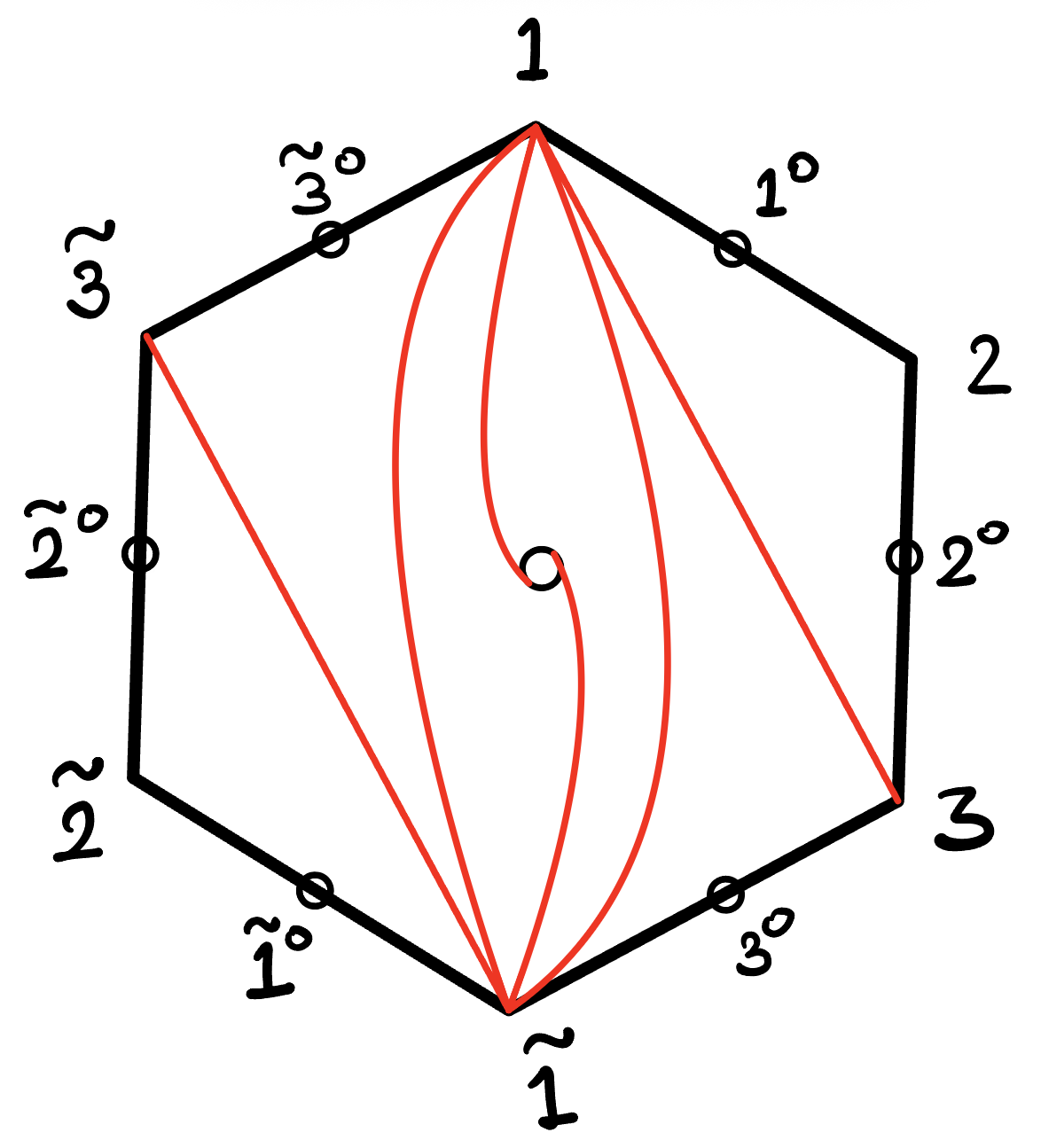}
\caption{Pseudo-triangulation dual to three-point one-loop fatgraph.}
\label{fig:pseudo_triangulation_1l3p}
\end{figure}

\begin{figure}[h!]
\centering
\begin{subfigure}[t]{0.32\textwidth}
\centering
\includegraphics[width=0.9\linewidth]{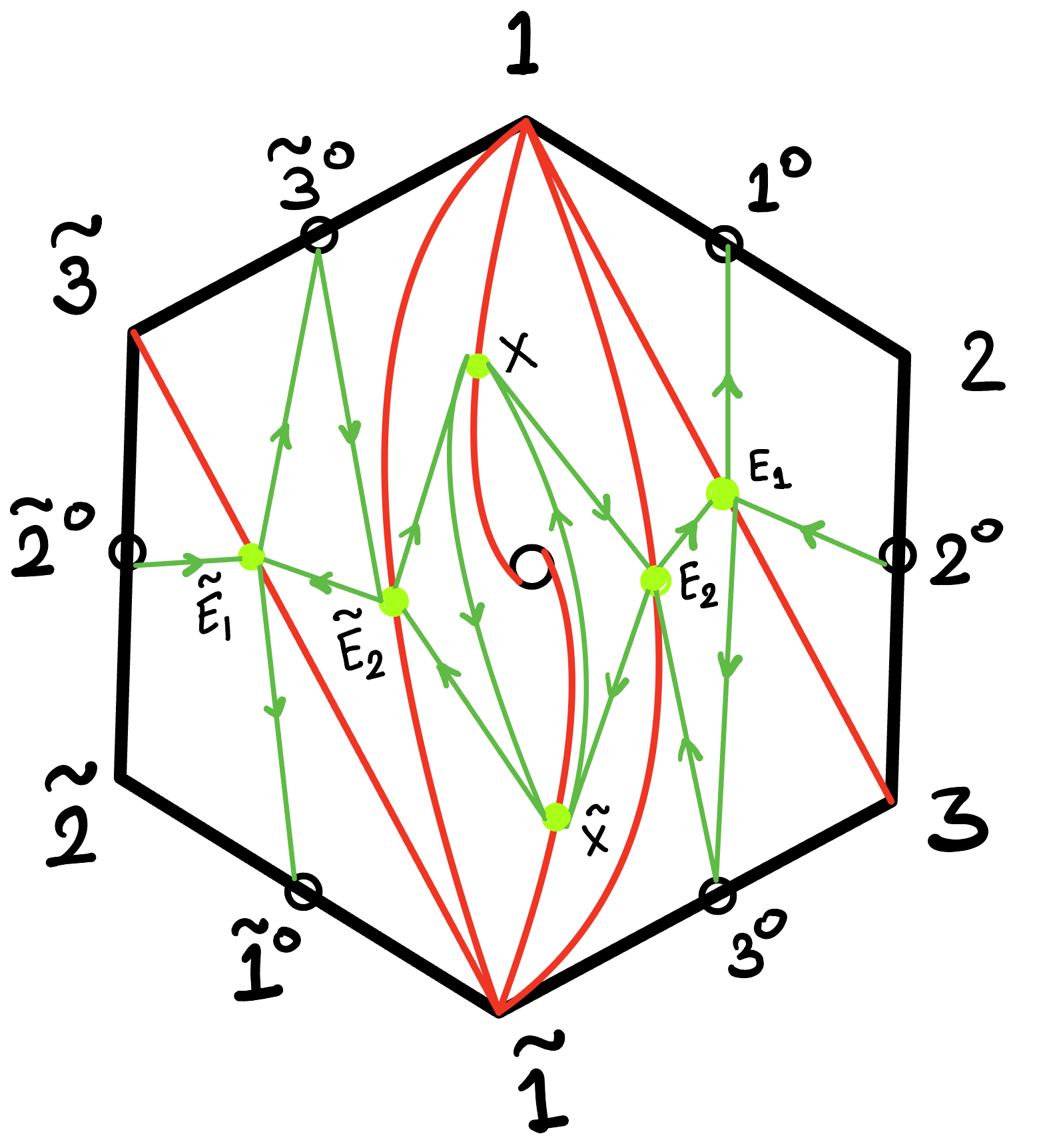}
\caption{}
\label{fig:quiver_for_tadpole_ref_3p1l}
\end{subfigure}
\hfill
\begin{subfigure}[t]{0.32\textwidth}
\centering
\includegraphics[width=0.9\linewidth]{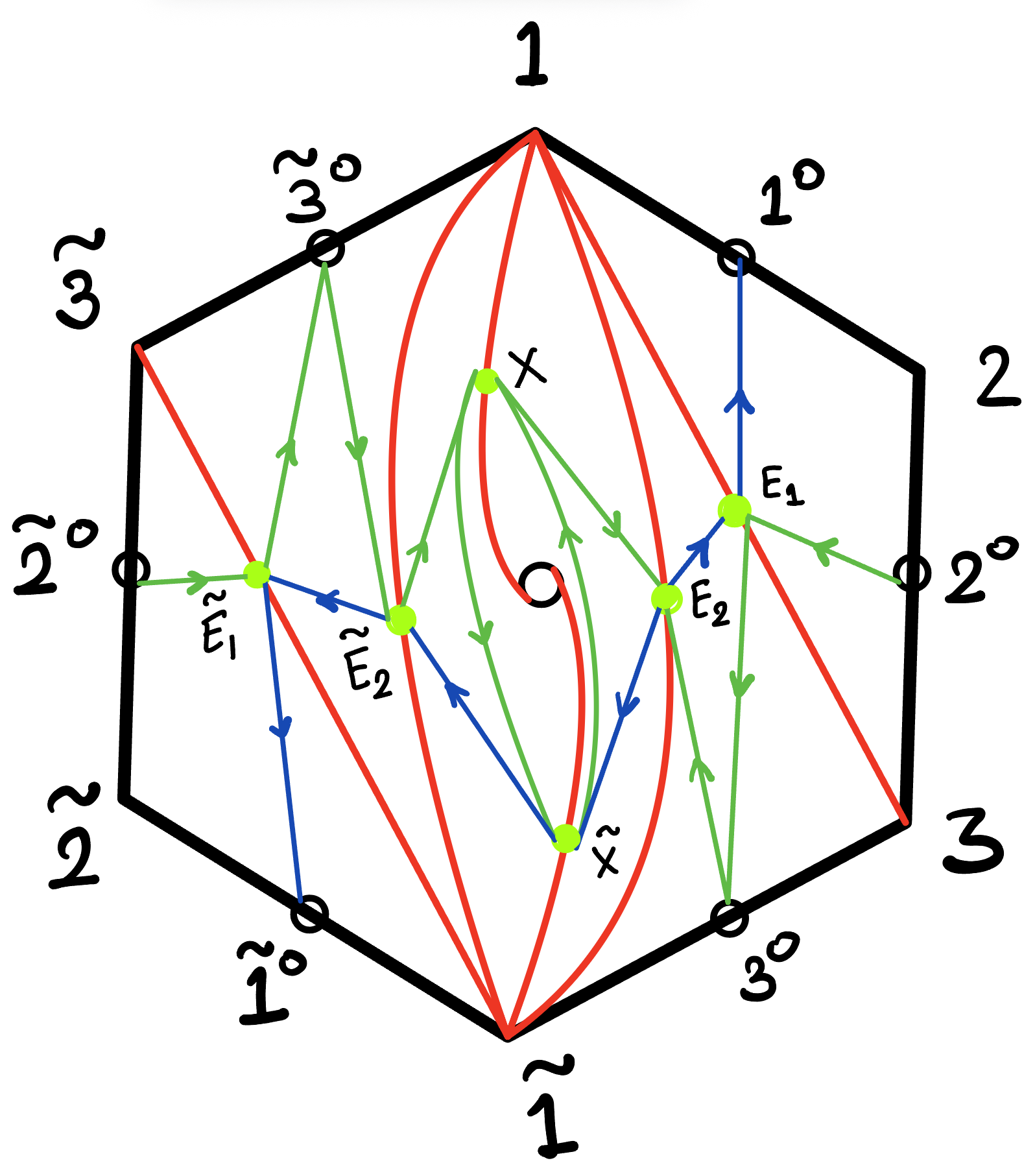}
\caption{}
\label{fig:proper_walk_example}

\end{subfigure}
\hfill
\begin{subfigure}[t]{0.32\textwidth}
\centering
\includegraphics[width=0.9\linewidth]{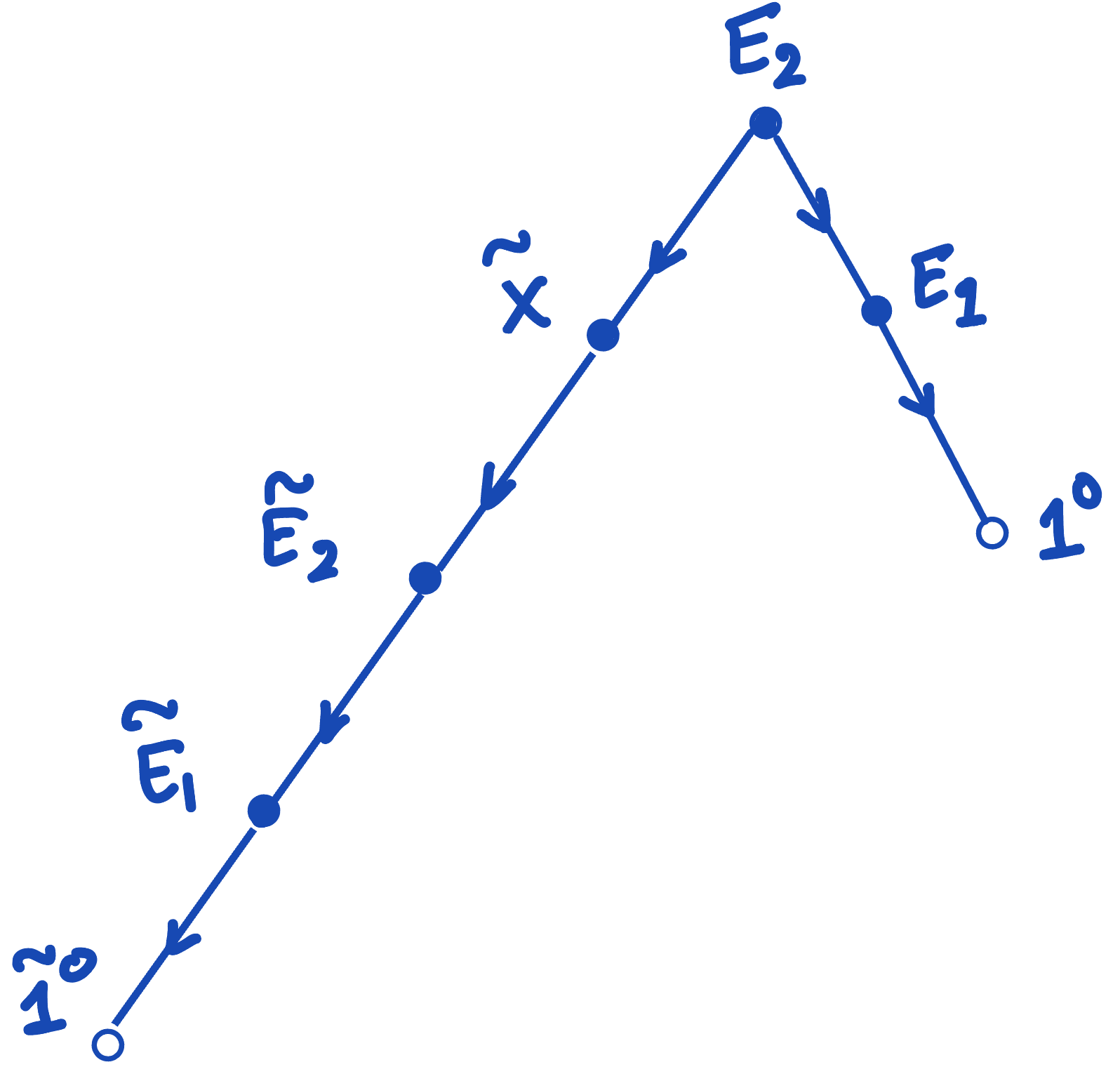}
\caption{}
\label{fig:proper_walk_peak_and_dip}

\end{subfigure}

\end{figure}

The  gentle blossomed quiver associated to $\textrm{PT}_{R}$ is drawn in green in Fig. \ref{fig:quiver_for_tadpole_ref_3p1l}. We note that it has 6 blossoms which are denoted as 
\begin{align}
V^{\circ}\, :=\, \{\, 1^{\circ}, \dots, \widetilde{3}^{\circ}\, \}
\end{align}
Starting with a vertex in $V^{\circ}$, we can now compute a $g$-vector $g_{w}$ associated to any  proper walk $w$. A proper walk is a map, 
\begin{align}
V^{\circ}\,  \rightarrow\, V^{\circ}\, \medcup\, 0, 
\end{align}
which never traverses two arrows in the same cell. In Fig. (\ref{fig:proper_walk_example},\ref{fig:proper_walk_peak_and_dip}), we indicate examples of proper walks.  For any such $w$, $g_{w}$ is computed by the following algorithm \cite{Jagadale:2022rbl}.
\begin{enumerate}
\item  We first assign a unit vector to each of the chords in $\textrm{PT}_{R}$ and thus obtain the cartesian space ${\bf R}^{\vert \textrm{PT}_{R}\vert}$.
\item If a proper walk $w$ has a peak (dip) at the location of a chord $c$ then we multiply the corresponding unit vector $e_{c}$ with $+1 (-1)$. If a walk either does not intersect a $c\, \in\, \textrm{PT}_{R}$ or does not have either a peak or a dip than $g_{w}$ is orthogonal to $e_{c}$.
\item We finally then identify $e_{(1\widetilde{1})_{R}}\, =\, e_{(1\widetilde{1})_{L}},\, e_{13}\, =\, e_{\widetilde{1}\widetilde{3}}$.
\end{enumerate}

\begin{itemize}
\item In the table above, we have dropped $g^{-}_{i\widetilde{i}}$ as they equal $g^{+}_{i^{\circ}\widetilde{i}^{\circ}}$ under the identification of the reference chords on the two sides of the puncture.
\item $g^{+}_{\widetilde{i} 0}\, \neq\, g^{-}_{\widetilde{i} 0}$ but they all projections on $e_{10}$ is positive and hence we do not include them here they do not have support over $t_{x} \leq\, 0$ hyper-plane.
\item This algorithm can be readily applied to the pseudo-triangulation model over $2n$-gon. 
\end{itemize}
In table \ref{tab:g-vectors_matching_PT}, we see that the two methods are equivalent under suitable identifications.

%%%%%%%%%%%%%%%%%%%%%%%%%%%%%%%%
\subsection{Pseudo-triangulation model for two-loop case}
In this section, we show that there exists a pseudo-triangulation model along with corresponding gentle-blossomed quiver,  which can then be used to compute $g$-vectors associated to the pseudo-triangulations of $\Sigma_{2,2}$. This model is a direct generalization of the pseudo-triangulation model proposed by Ceballos and Pilaud in their seminal work \cite{Ceballos_2015} where the flips in pseudo-triangulations were  shown to be equivalent to quiver mutations of cluster $D_{n}$ polytope. To the best of our knoweldge, a pseudo-triangulation model for so called surfacehedra \cite{Surfaceohedra} does not exist in the literature. We propose such a model and use it to compute $g$-vector fan for dissections of $\Sigma_{2,2}$. The justification of our proposed model lies in the fact that the set of $g$ vectors for all but two of the walks ( $C_{aa}, C_{bb}$ which are self-intersecting) precisely match known results derived in \cite{Arkani-Hamed:2023lbd} after appropriately identifying chords of pseudo-triangulation with curves on a planar fat graph with two loops. For concreteness, we will stick to $n=2$, although as before, the algorithm can be used to obtain $g$-vector fan for all $\Sigma_{2,n}$ and any reference fat graph $\Gamma_{L,n}$.   

We thus suspect that it may be plausible to have a definition of combinatorial positive geometries such as surfacehedra  through flips of pseudo-triangulations just as in the case of assocaihedron, cluster $D$ polytope as well as $\widehat{D}_{n}$  polytope.\footnote{Further explorations of surfacehedra, especially the role of two closed walks which do not appear to have a natural place in the pseudo-triangulation model, is beyond the scope of the present paper and we hope to come back to a more extensive analysis of our model in the future.} 

In table \ref{tfgcuci}, entries of first column describes the $g$-vectors for various chords in $\Sigma_{2,2}$ using the formulae derived in \cite{Arkani-Hamed:2023lbd}.  Our convention for  labels of various dissection chords in $\Sigma_{2,2}$ is given in table \ref{tab:2l2p-g-vectors_and_curves}. The third column in table (\ref{tfgcuci}) is based on the $g$-vectors computed using pseudo-triangulation model described in detail below. 

\begin{table}[h!]
\renewcommand{\arraystretch}{1.4}
\centering
\begin{tabular}{|>{\raggedright\arraybackslash}m{2.0cm}|
>{\raggedright\arraybackslash}m{8cm}|
>{\raggedright\arraybackslash}m{1.2cm}|
>{\raggedright\arraybackslash}m{6.5cm}|}
\hline
 \textbf{Curves} 
& \textbf{Words associated with curves} \\
\hline
 $C^{0}_{1a}$ & $1LY_1L w R (x L)^\infty$ \\
 $C^{0}_{2a}$ & $1RY_1L w R (x L)^\infty$ \\
 $C^{0}_{1b}$ & $1LY_1 R z R (y L)^\infty$ \\
 $C^{0}_{2b}$ & $2RY_1 R z R (y L)^\infty$ \\
 $C^a_{11}$ & $1 L Y_1 L w L x L w R Y_1 R 1$ \\
 $C^a_{12}$ & $1 L Y_1 L w L x L w R Y_1R2$ \\
 $C^a_{22}$ & $2RY_1L w L x L w R Y_1L2$ \\
$C^b_{11}$ & $1 L Y_1 R z L y L z LY_1 R 1$ \\
 $C^b_{12}$ & $1 L Y_1 R z L y L z L Y_1L2$ \\
 $C^b_{22}$ & $2RY_1R z L y L z L Y_1L2$ \\
$C_{ab}$ & $(R x)^\infty L w L z R (y L)^\infty$ \\
$C^{1}_{1a}$ & $1LY_1 R z R y R z R w R (x L)^\infty$ \\
$C^{1}_{2a}$ & $2RY_1 R z R y R z R w R (x L)^\infty$ \\
$C^{1}_{1b}$ & $1LY_1 L w L x L w L z R (y L)^\infty$ \\
$C^{1}_{2b}$ & $2RY_1 L w L x L w L z R (y L)^\infty$ \\
$C_{aa}$ & $(R x)^\infty L w L z L y L z R w R (x L)^\infty$ \\
$C^{bb}$ & $(R y)^\infty L z R w L x L w L z R (y L)^\infty$ \\
$C^{-1}_{1a}$ & $1LY_1L w L x L w L z L y L z R w R (x L)^\infty$ \\
$C^{-1}_{2a}$ & $2RY_1L w L x L w L z L y L z R w R (x L)^\infty$ \\
$C^{-1}_{1b}$ & $1LY_1 R z R y R z R w R x R w L z R (y L)^\infty$ \\
$C^{-1}_{2b}$ & $2RY_1R z R y R z R w R x R w L z R (y L)^\infty$ \\
$C^{a,1}_{11}$ & $1LY_1R z R y R z R w L x L w L z L y L z L Y_1R1$ \\
$C^{a,1}_{12}$ & $1LY_1R z R y R z R w L x L w L z L y L z L Y_1L2$ \\
$C^{a,1}_{22}$ & $2RY_1R z R y R z R w L x L w L z L y L z L Y_1L2$ \\
$C^{b,1}_{11}$ & $1LY_1L w L x L w L z L y L z R w R x R w R Y_1R1$ \\
$C^{b,1}_{12}$ & $1LY_1L w L x L w L z L y L z R w R x R w RY_1L2$ \\
$C^{b,1}_{22}$ & $2RY_1L w L x L w L z L y L z R w R x R w R Y_1L2$ \\
\hline
\end{tabular}
\caption{Curve in $\Sigma_{2,2}$. A generic curve is denoted by 
$C^{\alpha}_{I\,F}$, where $I$ is the starting point, $F$ is the final point, 
and $\alpha$ encodes the winding information.
For example, $C^{(a,1)}_{12}$ denotes a curve from $1$ to $2$ that encloses 
puncture $a$ between itself and the boundary, while also making a turn 
around another puncture $b$ without enclosing it between itself and the boundary 
(see Fig.~\ref{fig:2loop_2point_curves}).  $C^{-1}_{1a}$ denotes a curve 
starting at $1$ and ending by spiraling into puncture $a$. The superscript $-1$ 
specifies that before spiraling into $a$, the curve winds once around $a$ and then 
around $b$ (see Fig.~\ref{fig:2loop_2point_curves}). By contrast, 
$C^{1}_{1a}$ denotes a curve starting at $1$ and ending at puncture $a$, but this 
time winding once around $b$ before spiraling into $a$.}

\label{tab:2l2p-g-vectors_and_curves}
\end{table}

\begin{table}[h!]
\renewcommand{\arraystretch}{1.4}
\centering
\begin{tabular} {|>{\raggedright\arraybackslash}m{4cm}|
>{\raggedright\arraybackslash}m{6cm}|
>{\raggedright\arraybackslash}m{4cm}|}
\hline
\textbf{g-vector name(Anti-Clockwise)} & \textbf{g-vector as dual to headlight functions in basis : $(\, e_{X}, e_{W}, e_{Z}, e_{Y},\, e_{E_{1}}\, )$} & \textbf{g-vectors from pseudo-triangulation model in the same basis after certain identification of dissection chords. .} \\
\hline
$g_{1a}^{0}$ & $\{-1, 1, 0, 0, 0\}$ & $g_{\widetilde{1}^{\circ} 0_{1}}^{\textrm{cc}}$\\
$g_{2a}^{0}$ & $\{-1, 1, 0, 0, -1\}$ & $g_{\widetilde{2}^{\circ} 0_{1}}^{\textrm{cc}}$\\
$g_{1b}^{0}$ & $\{0, 0, 0, -1, 1\}$ & $g_{\widetilde{1}^{\circ} 0_{2}}^{\textrm{c}}$\\
$g_{2b}^{0}$ & $\{0, 0, 0, -1, 0\}$ & $g_{\widetilde{2}^{\circ} 0_{2}}^{\textrm{c}}$\\
$g_{11}^{a}$ & $\{0, 1, 0, 0, 0\}$ &  $g_{\widetilde{1}^{\circ} 1^{\circ}}^{\textrm{cc}}$\\
$g_{12}^{a}$ & $\{0, 1, 0, 0, -1\}$ &  $g_{\widetilde{1}^{\circ} 2^{\circ}}^{\textrm{c}}$\\
$g_{22}^{a}$ & $\{0, 1, 0, 0, -2\}$ &  $g_{\widetilde{2}^{\circ} 2^{\circ}}^{\textrm{c}}$\\
$g_{11}^{b}$ & $\{0, 0, -1, 0, 2\}$ &  $g_{\widetilde{1}^{\circ} \widetilde{\widetilde{1}}^{\circ}}^{\textrm{c}}$\\
$g_{12}^{b}$ & $\{0, 0, -1, 0, 1\}$ &  $g_{\widetilde{1}^{\circ}\widetilde{\widetilde{2}}^{\circ}}^{\textrm{cc}}$\\
$g_{22}^{b}$ & $\{0, 0, -1, 0, 0\}$ &  $g_{\widetilde{2}^{\circ}\widetilde{\widetilde{2}}^{\circ}}^{\textrm{cc}}$\\
$g_{ab}$ & $\{-1, 0, 1, -1, 0\}$ &    $g_{0_{1} 0_{2}}^{\textrm{cc}\vert \textrm{c}}$\\
$g_{1a}^{1}$ & $\{-1, 0, 0, 0, 1\}$ &  $g_{1^{\circ} 0_{1}}^{\textrm{c}}$\\
$g_{2a}^{1}$ & $\{-1, 0, 0, 0, 0\}$ &  $g_{2^{\circ} 0_{1}}^{\textrm{c}}$\\
$g_{1b}^{1}$ & $\{0, 0, 1, -1, 0\}$ &   $g_{1^{\circ} 0_{2}}^{\textrm{cc}}$\\
$g_{2b}^{1}$ & $\{0, 0, 1, -1, -1\}$ &  $g_{2^{\circ} 0_{2}}^{\textrm{cc}}$\\
$g_{aa}$ & $\{-2, 0, 0, 0, 0\}$ &   $\cross$\\
$g_{bb}$ & $\{0, -1, 2, -2, 0\}$ &  $\cross$\\
$g_{1a}^{-1}$ & $\{-1, 0, 1, 0, 0\}$ & $g_{\widetilde{\widetilde{1}}^{\circ} 0_{2} 0_{1}}^{\textrm{c} \vert \textrm{cc}}$ \\
$g_{2a}^{-1}$ & $\{-1, 0, 1, 0, -1\}$ &  $g_{\widetilde{\widetilde{2}}^{\circ} 0_{2} 0_{1}}^{\textrm{c}\vert \textrm{cc}}$\\
$g_{1b}^{-1}$ & $\{0, -1, 1, -1, 1\}$ &  $g_{\widetilde{\widetilde{1}}^{\circ} 0_{2} 0_{1}}^{\textrm{cc} \vert \textrm{c}}$\\
$g_{2b}^{-1}$ & $\{0, -1, 1, -1, 0\}$ &   $g_{\widetilde{\widetilde{2}}^{\circ} 0_{2} 0_{1}}^{\textrm{cc} \vert \textrm{c}}$\\
$g_{11}^{(a,1)}$ & $\{0, -1, 0, 0, 2\}$ &  $g_{\widetilde{\widetilde{1}}^{\circ} 0_{2} 0_{1} 1^{\circ}}^{\textrm{c} \vert \textrm{c}}$\\
$g_{12}^{(a,1)}$ & $\{0, -1, 0, 0, 1\}$ &  $g_{\widetilde{\widetilde{1}}^{\circ} 0_{2} 0_{1} 2^{\circ}}^{\textrm{c} \vert \textrm{cc}}$\\
$g_{22}^{(a,1)}$ & $\{0, -1, 0, 0, 0\}$ &  $g_{\widetilde{\widetilde{2}}^{\circ} 0_{2} 0_{1} 2^{\circ}}^{\textrm{c} \vert \textrm{c}}$\\
$g_{11}^{(b,1)}$ & $\{0, 0, 1, 0, 0\}$ &  $g_{\widetilde{\widetilde{1}}^{\circ} 0_{2} 0_{1} 1^{\circ}}^{\textrm{cc} \vert \textrm{c}}$\\
$g_{12}^{(b,1)}$ & $\{0, 0, 1, 0, -1\}$ & $g_{\widetilde{1}^{\circ} 2^{\circ}}^{\textrm{cc}}$\\
$g_{22}^{(b,1)}$ & $\{0, 0, 1, 0, -2\}$ &  $g_{\widetilde{\widetilde{2}}^{\circ} 0_{2} 0_{1} 2^{\circ}}^{\textrm{cc} \vert \textrm{c}}$\\
\hline
\end{tabular}
\caption{Comparison of $g$-vectors using words in fatgraph with Pseudo Triangulation.}
\label{tfgcuci}    
\end{table}

%%%%%%%%%%%%%%%%%%%%%%%%%%%%%%%%
\subsection*{Details on the new pseudo-triangulation model for dissections quivers in $\Sigma_{2,2}$.}
The pseudo-triangulation model we propose  starts with a three fold replica of a disc with $n$ points on the boundary and  2 punctures. That is, we consider a $3n$ gon where the vertices $v_{1},\, \dots,\, v_{n}$ are repeated twice more in clockwise fashion.  This leads to a pseudo-triangulation model for path algebras defined on surfaces with two punctures. In the $n=2, L=2$ case as shown in the Fig.  \ref{fig:pseudo-triangulation_2point_2loop}.
\begin{figure}[h]
\centering
\begin{subfigure}{.5\textwidth}
\centering
\includegraphics[width=.6\linewidth]{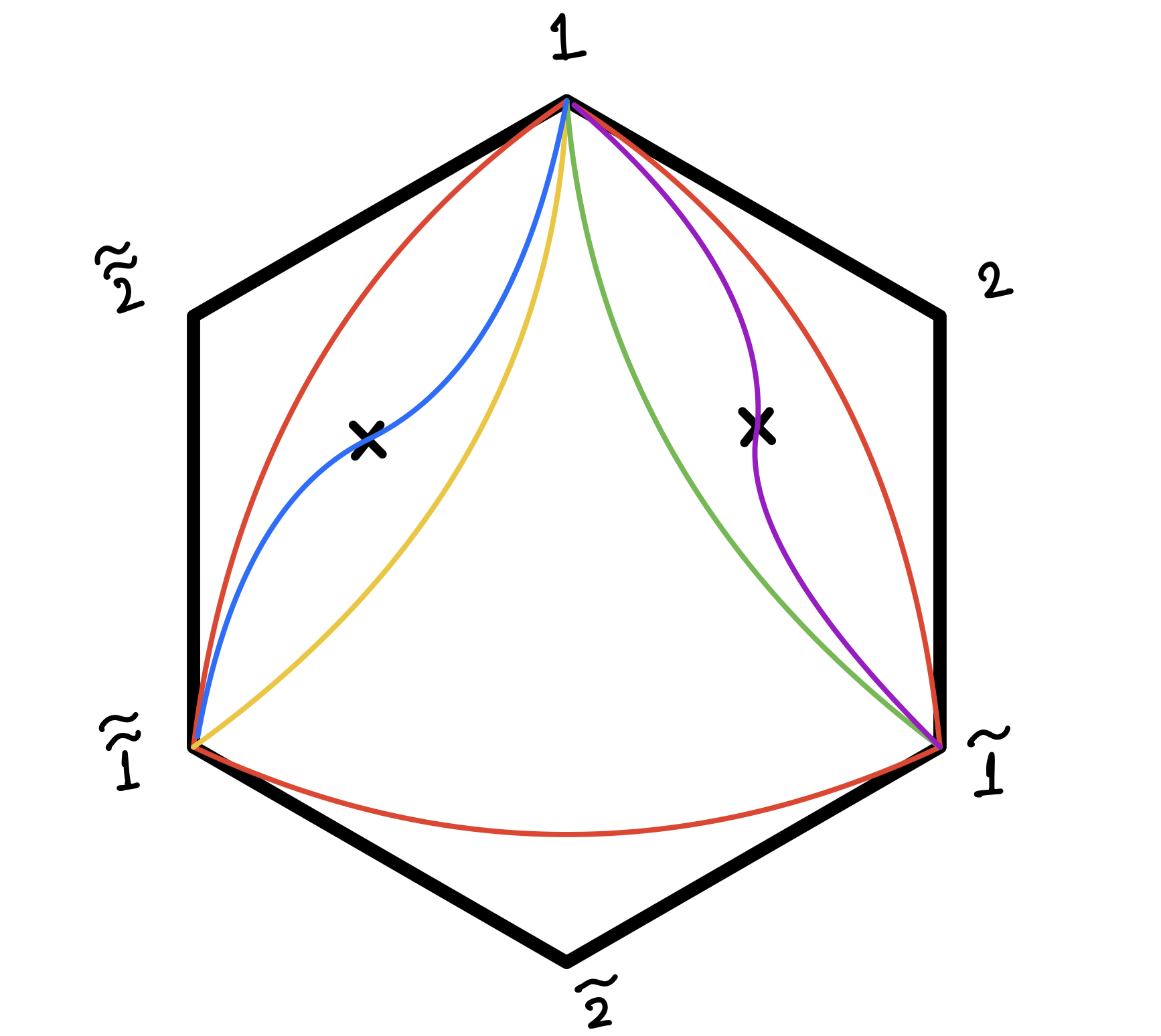}
%		\caption{}
%		\label{}
\end{subfigure}%
\begin{subfigure}{.5\textwidth}
\centering
\includegraphics[width=\linewidth]{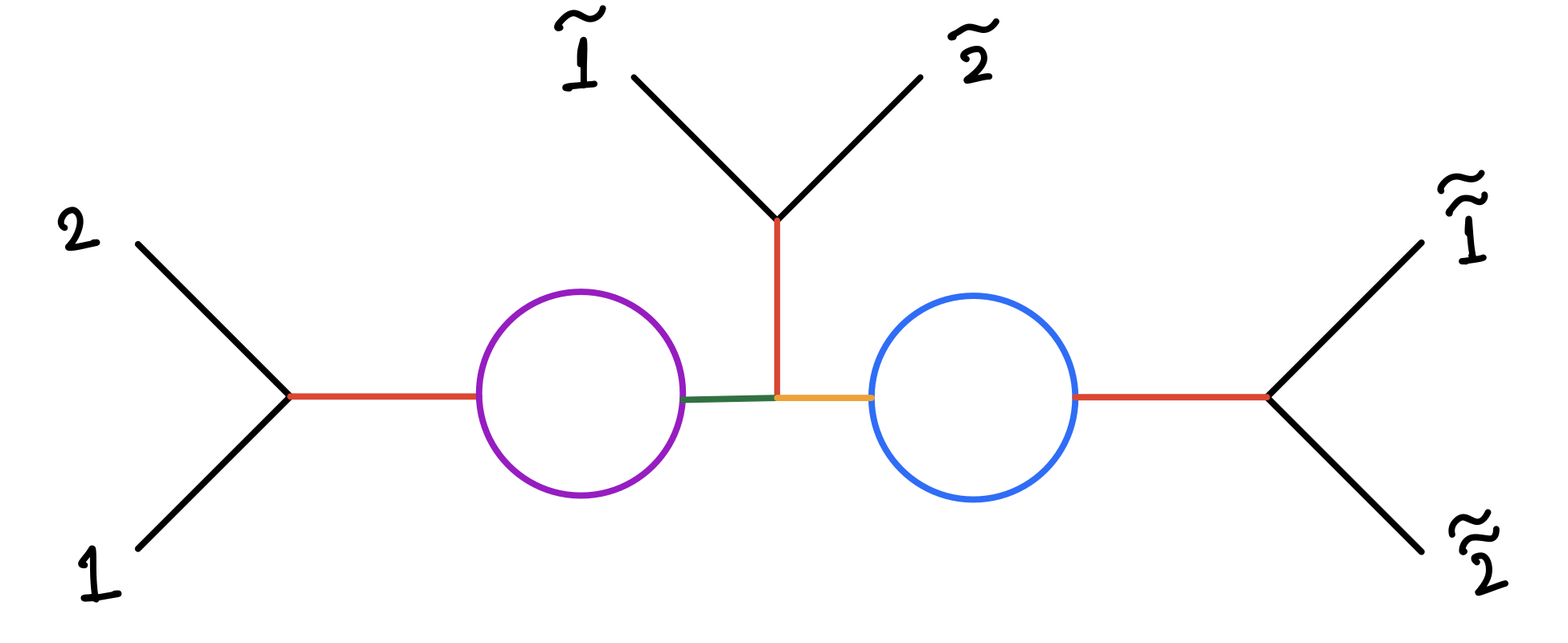}
%		\caption{}
%		\label{}
\end{subfigure}
\caption{pseudo-triangulation model for g-vector fans associated to two-loop surfacehedra}
\label{fig:pseudo-triangulation_2point_2loop}
\end{figure}

$\textrm{PT}_{r}$ consists of 9 dissections chords,
\begin{align}\label{ptrn2l2}
\textrm{PT}_{r}\, =\, \{\, (\widetilde{1}\widetilde{\widetilde{1}})_{R},\, (\widetilde{1}\widetilde{\widetilde{1}})_{L},\, (\widetilde{\widetilde{1}} 1),\, (1\widetilde{1})_{R}, (1\widetilde{1})_{L},\,  (\widetilde{1} 0_{1}), (1 0_{2}),\, \}
\end{align}
We propose that the gently blossomed dissection quiver which corresponds to triangulation of a 2 punctured disc with 2 marked points is obtained under the following identification of the chords in eqn.(\ref{ptrn2l2}).
\begin{align}
(1\widetilde{1})_{R}\, \sim\, (\widetilde{1}\widetilde{\widetilde{1}})\, \sim\, (\widetilde{\widetilde{1}} 1)_{L}\, =: E1
\end{align}
We will also denote the remaining chords that begin and end at the marked points as, \footnote{For readers who are aaquainted with the formalism given in \cite{Arkani-Hamed:2023lbd,Arkani-Hamed:2023mvg}, The three chords which are identified are associated with Tropical parameter $t_{1}$, the chord $W$ is identified with $t_{w}$ and chord $Z$ is identified with $t_{z}$.} 
\begin{align}
(\widetilde{1}\widetilde{\widetilde{1}})_{L} &=: W,\, (1\widetilde{1})_{R}\, \sim\, Z\nonumber\\
(\widetilde{1} 0_{1}) &=: X,\, ( \widetilde{1} 0_{2} ) =: Y
\end{align}
After such an identification, we can directly compute $g$-vector for any proper walk by assigning indices $\pm 1,\, 0$ with respect to the basis vectors, \footnote{We remind the reader that for walk $w$, $\langle\, g_{W}, g_{E_{1}}\, \rangle\, :=\, \sum_{I=1}^{3} g_{W}^{I}$, where the sum is over the three dissection chords which are identified as $E_{1}$.}
\begin{align}
\textrm{Set of basis-vectors after identification}\, =\, \{\, e_{E1},\, e_{W}, e_{Z}, e_{X}, e_{Y}\, \}.
\end{align}
The final results can be summarized as follows : In the terminology of \cite{PPPP}, we can associate to any proper walk (non self-intersecting), a $g$-vector with respect to the basis defined by a reference triangulation. All such walks generate a proper subset of the set $S$ which is the set of curves on $\Sigma_{2,2}$ which are not in the kernel of tropical Mirzakhani kernel. Thus the pseudo-triangiulation model presents an efficient way to generate a large class of curves which are in the compliment of $\textrm{Ker}({\cal K})$.  Precisely two element of $S$ which are closed loops are left out of the set of proper walks that we can define via this model. A more detailed analysis of how to parametrize Teichmuller space of punctured disc with marked points on the boundary in terms of such triangulations will be pursued elsewhere. 

As the table \ref{tfgcuci}  shows,  shows that the set of $g$ vectors for all the curves in the set ${\cal S}$ can be computed either using pseudo-triangulation model we have proposed here, or the formulae derived in \cite{Arkani-Hamed:2023lbd}.  We first set up some notations which makes the comparison with $g$-vectors computed in  \ref{tfgcuci} manifest.
\begin{enumerate}
\item Consider a walk $w$ that winds around puncture $0_{1}$ without winding around $0_{2}$. As there are two homotopically inequivalent curves, we denote the corresponding $g$-vectors as $g_{w}^{\epsilon}$ where $\epsilon = \textrm{cc} (\textrm{c})$  if the walk winds around $0_{1}$ counter-clockwise (clockwise). 
\item We similarly have two sets of $g$-vectors $g_{\widetilde{\widetilde{i}}^{\circ} i^{\circ}}^{\epsilon}$ depending on how the walk winds around $0_{2}$. 
\item For the set of walks that wind around both the punctures we have $g_{w}^{\epsilon_{1}\epsilon_{2}}$. 
\end{enumerate}
The set of all walks whose $g$ vectors then match with those in the set $S$ are listed in the third column of the table above. This motivates us to speculate that mutation relations of the walks defined via the pseudo-triangulation model will generate for us the surfacehedra for two-loop amplitudes, \cite{Surfaceohedra}. 
%%%%%%
\begin{figure}[h]
\centering
\includegraphics[width=1.1\linewidth]{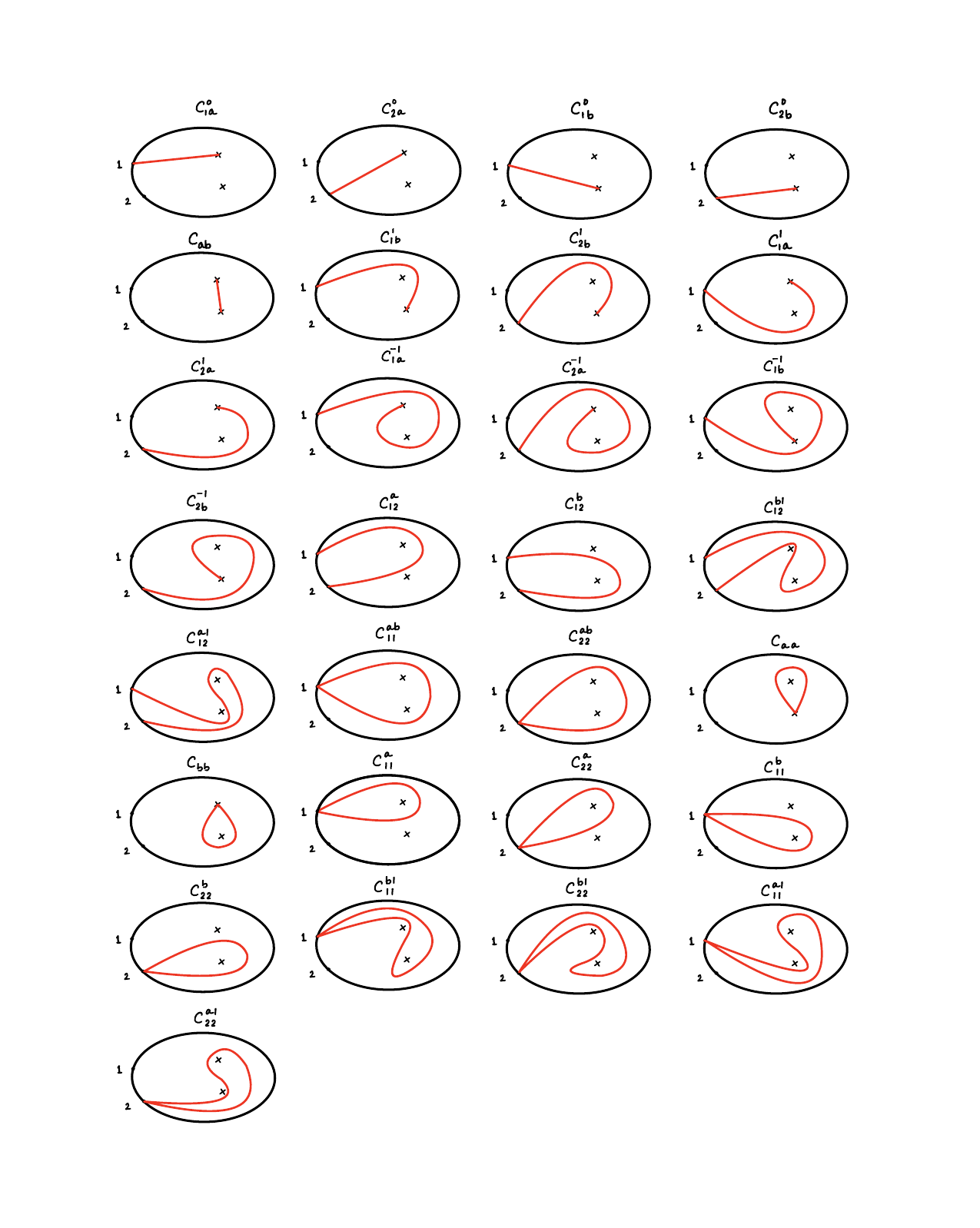}
\caption{Two-loop two-point curves}
\label{fig:2loop_2point_curves}
\end{figure}

%%%%%%%%%%%%%%%%%%%%%%%%%%%%%%%%
%%%%%%%%%%%%%%%%%%%%%%%%%%%%%%%%
\section{Cones for two-point two-loop}\label{sec:Cones_for_2-point_2-loop}

The construction of the quadratically divergent tadpole–free region 
${\cal M}_{2,2}^{\textrm{qtf}}$ is computationally challenging.  
We therefore proceed by first determining the contributing regions 
associated with all top–dimensional cones. Once these cones are 
constructed, we find that the full region ${\cal M}_{2,2}^{\textrm{qtf}}$ 
decomposes into five distinct subregions, within which the first 
Symanzik polynomial $\mathcal{U}$ takes different forms.  

To begin, consider the set $S_C$ consisting of the $g$–vectors associated 
with all non–tadpole curves, together with the tadpole curves 
$C^{(a,b)}_{11}$ and $C^{(a,b)}_{22}$. These latter curves contribute to 
the logarithmically divergent cones of the two–point, two–loop fatgraph.  
From $S_C$, we construct all possible sets of five linearly independent 
$g$–vectors corresponding to non–intersecting curves. Explicitly, one 
selects
\[
M = \{g_1,\, g_2,\, g_3,\, g_4,\, g_5\},
\]
which is sufficient since a two–loop diagram involves precisely five 
independent Schwinger parameters.  

Each cone is defined as the convex region spanned by the chosen $g$–vectors,
\begin{align}\label{eq:cones_eqn}
	\vec{X} = a_i \,\vec{g}_i, 
	\qquad a_i > 0,
\end{align}
where $\vec{X}$ is expressed in the basis 
$(e_{X},\, e_{W},\, e_{Z},\, e_{Y},\, e_{E_1})$, with 
$\{X, W, Z, Y, E_1\}$ denoting the global Schwinger variables. By inverting Eq.~\eqref{eq:cones_eqn} and imposing the conditions $a_i > 0$, one obtains the set of inequalities that bound 
each cone:
\begin{align}
	a_i = f_i(t_x, t_w, t_z, t_y, t_1) > 0.
\end{align}
Moreover, Eq.~\eqref{eq:cones_eqn} can be used to express the global 
coordinates $\vec{X}$ in terms of the cone parameters $(a_1,\dots,a_5)$, 
thus yielding the explicit Schwinger representation of the integral 
associated with the Feynman diagram corresponding to that cone.

From this analysis we find a total of $24$ top–dimensional cones. The form of the surface Symanzik polynomial $\mathcal{U}$ in each cone 
is obtained by simplifying $\mathcal{U}$ subject to the inequalities 
that characterize the cone. Furthermore, Eq.~\eqref{eq:cones_eqn} 
provides a direct map to the Schwinger parametrization of the integral 
for the corresponding Feynman diagram. We denote the first Symanzik 
polynomial of the associated Feynman diagram by $\mathcal{U}_{\text{FD}}$.

The correspondence between cones and Feynman diagrams can then be identified 
by counting the number of terms that appear in $\mathcal{U}_{\text{FD}}$ for 
each cone. Concretely:  
\begin{itemize}
\item Cones with $4$ terms in $\mathcal{U}_{\text{FD}}$ correspond to 
$\log^2$–divergent graphs.  
\item Cones with $7$ terms in $\mathcal{U}_{\text{FD}}$ correspond to ring graphs, 
i.e.\ $\log$–divergent graphs.  
\item Cones with $8$ terms in $\mathcal{U}_{\text{FD}}$ correspond to overlapping 
graphs, which are finite (non–divergent).  
\item Cones with $5$ terms in $\mathcal{U}_{\text{FD}}$ correspond to 
$\{ab\}$–tadpole graphs, which are $\log$–divergent.  
\end{itemize}
\begin{figure}[h!]
\centering
% Row 1
\begin{subfigure}[t]{0.5\textwidth}
\centering
\includegraphics[width=\linewidth]{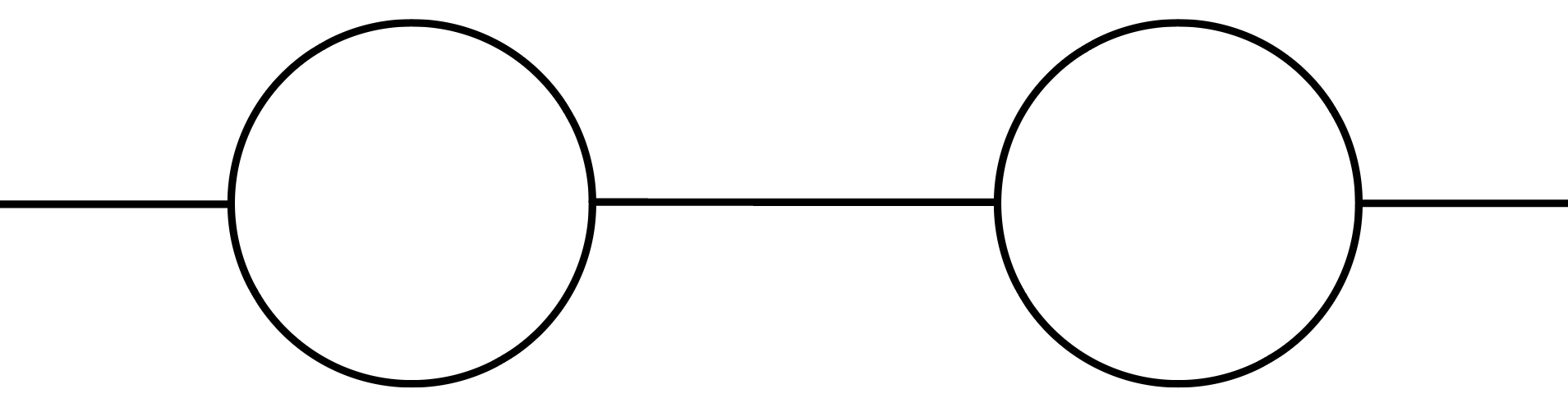}
\caption{$\log^2$–divergent graphs}
\label{fig:2l2p_log^2_graph}
\end{subfigure}
\hfill
\begin{subfigure}[t]{0.3\textwidth}
\centering
\includegraphics[width=\linewidth]{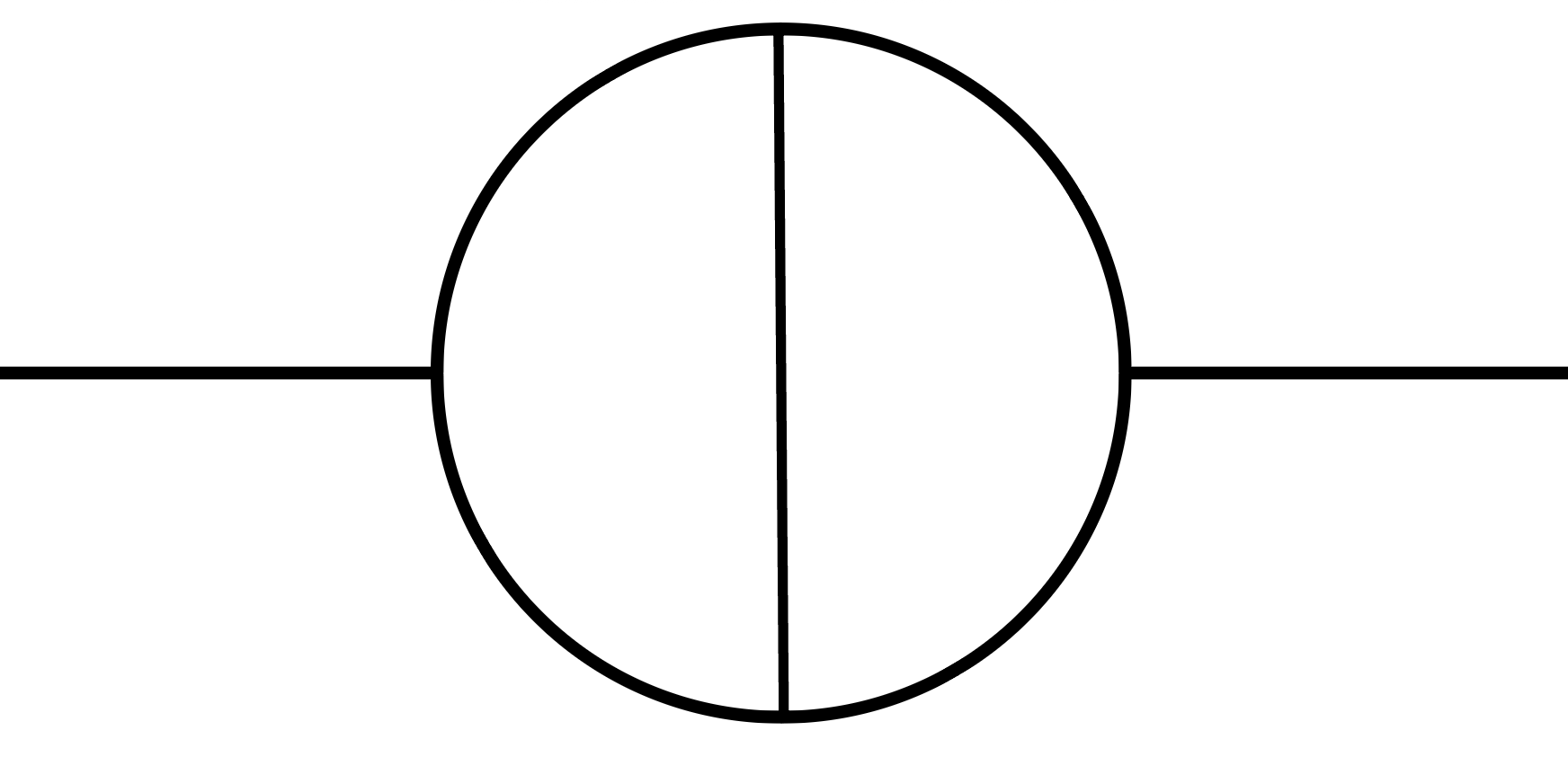}
\caption{Overlapping loops non-divergent graph}
\label{fig:2l2p_overlaping_graph}
\end{subfigure}

% Row 2
\vskip\baselineskip
\qquad\quad\quad
\begin{subfigure}[t]{0.3\textwidth}
\centering
\includegraphics[width=\linewidth]{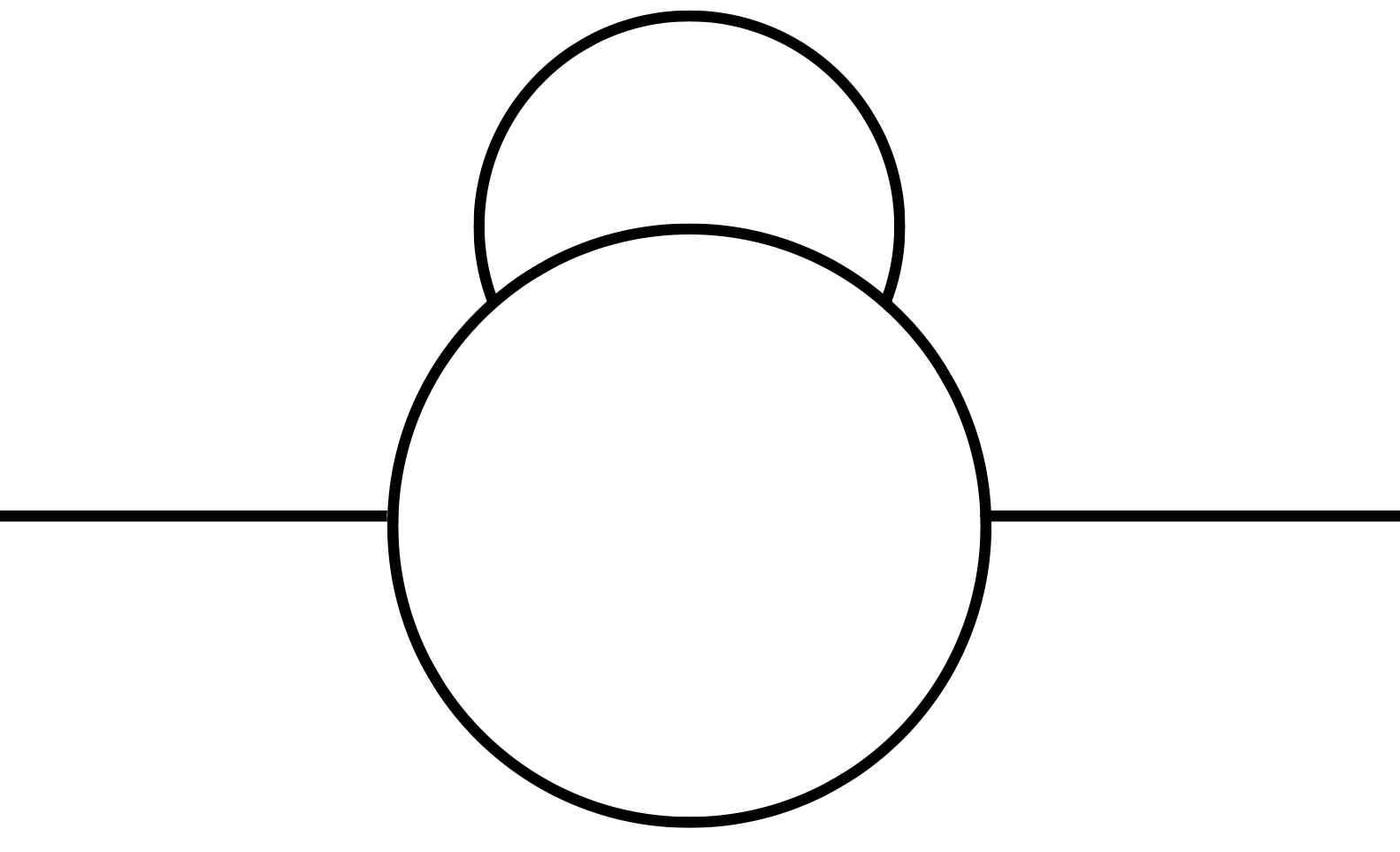}
\caption{Ring graphs $\log$-divergent}
\label{fig:2l2p_ring_graph}
\end{subfigure}
\hfill
\begin{subfigure}[t]{0.3\textwidth}
\centering
\includegraphics[width=\linewidth]{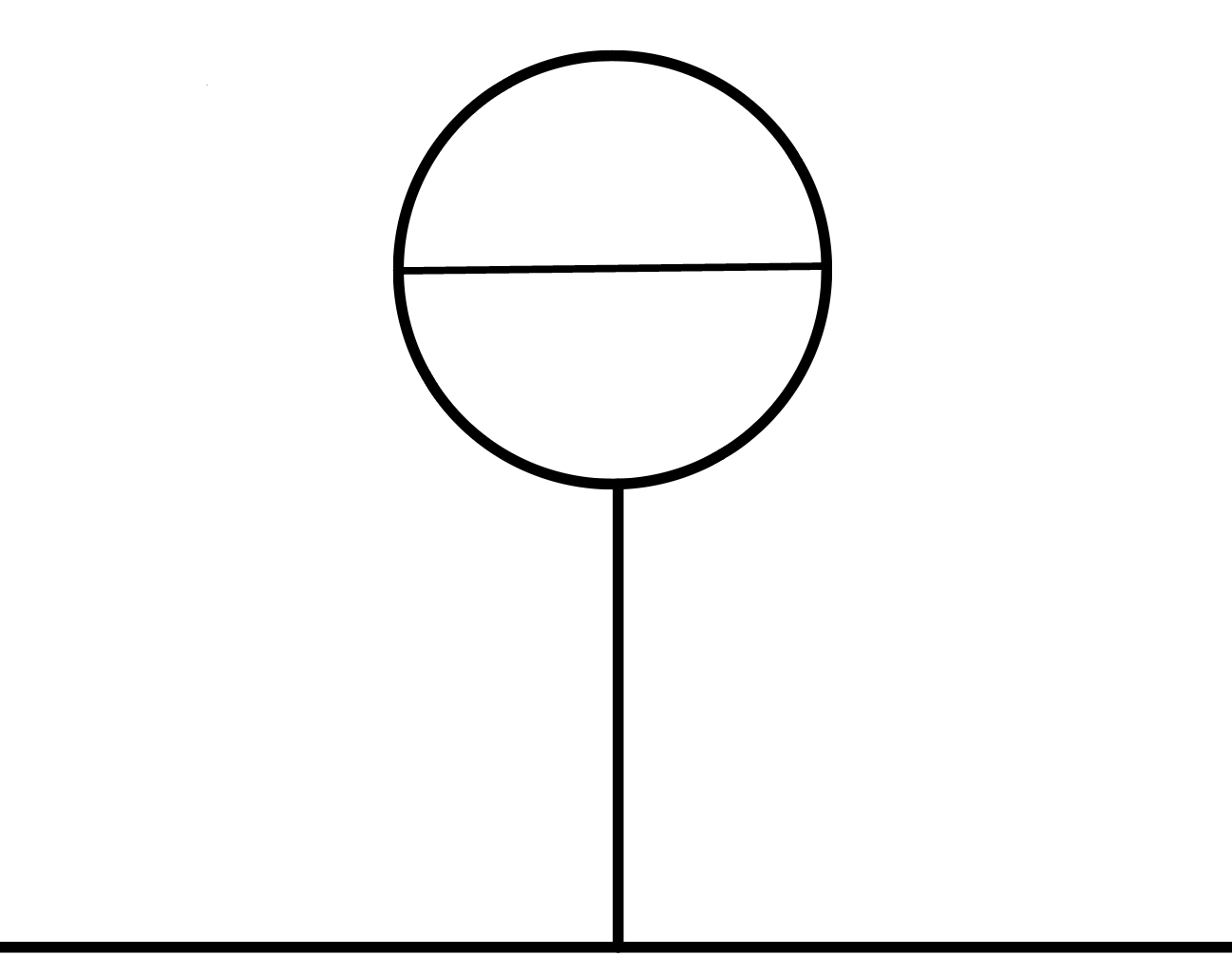}
\caption{Two-loop tadpole graph}
\label{fig:2l2p_tadpole_log_div_graph}
\end{subfigure}

\caption{All four Feynman diagrams.}
\label{fig:four_in_one}
\end{figure}

Thus, the classification of cones in global Schwinger space directly encodes 
the divergence structure of the corresponding Feynman diagrams. In particular, 
we find:  
\begin{itemize}
\item $\{\mathcal{C}_1,\mathcal{C}_2,\mathcal{C}_3,\mathcal{C}_4\}$ are dual to
$\log^2$–divergent graph. fig~\ref{fig:2l2p_log^2_graph}.
\item $\{\mathcal{C}_7,\mathcal{C}_9,\mathcal{C}_{13},\mathcal{C}_{14},\mathcal{C}_{17},\mathcal{C}_{19},\mathcal{C}_{22},\mathcal{C}_{24}\}$ 
are dual to $\log$ divergence graphs, fig~\ref{fig:2l2p_ring_graph}.
\item $\{\mathcal{C}_8,\mathcal{C}_{12},\mathcal{C}_{18},\mathcal{C}_{23}\}$ 
are dual to finite graphs,  fig~\ref{fig:2l2p_overlaping_graph}.
\item $\{\mathcal{C}_5,\mathcal{C}_6,\mathcal{C}_{10},\mathcal{C}_{11},\mathcal{C}_{15},\mathcal{C}_{16},\mathcal{C}_{20},\mathcal{C}_{21}\}$ 
are dual to graphs that include$\{ab\}$–tadpole curves,  fig~\ref{fig:2l2p_tadpole_log_div_graph}.
\end{itemize}

All the cones $\mathcal{C}_i$ are listed below and are grouped in classes having the same ${\cal U}$ polynomial.  For completeness, we also give the first  Symanzik polynomial 
corresponding to the graph ${\cal U}_{\textrm{FD}}(C_{i})$  which is dual to the cone.  These polynomials are written in terms of Schwinger parameters $a_{i}$ which are all positive.\footnote{Although we write ${\cal U}_{\textrm{FD}}(C_{i})$ in terms of $a_{1},\, \dots,\, a_{5}$, these variables simply label the edges of the corresponding graphs and hence ${\cal U}_{\textrm{FD}}(C_{i})$ and ${\cal U}_{\textrm{FD}}(C_{j})$ are independent polynomials $\forall\, i\, \neq\, j$.}\\

\underline{Cones over which ${\cal U} =\, t_{x}\, t_{y}$.}
\begin{enumerate}		
\item ${\cal C}_{1}$
\begin{align*}
&\text{g-vectors: } \{g_{1a}^{0}, g_{2a}^{0}, g_{2b}^{0}, g_{1b}^{1}, g_{12}^{a}\} \\
&\text{Inequalities: } 
\left( t_1 + t_w \geq 0 \right) \medcap 
\left( -t_1 - t_w - t_x \geq 0 \right) \medcap 
\left( -t_y - t_z \geq 0 \right) \medcap 
\left( t_z \geq 0 \right) \nonumber\\
&\hspace{23 mm}\medcap \left( t_w + t_x \geq 0 \right) \\
&\mathcal{U}_{\text{FD}}(C_{1})\, : a_1\,a_3 + a_2\,a_3 + a_1\,a_4 + a_2\,a_4
\end{align*}

\item  ${\cal C}_{2}$.
\begin{align*}
&\text{g-vectors: } \{g_{1a}^{0}, g_{1b}^{0}, g_{2b}^{0}, g_{2a}^{1}, g_{12}^{b}\} \\
&\text{Inequalities: } 
\left( t_w \geq 0 \right) \medcap 
\left( t_1 + t_z \geq 0 \right) \medcap 
\left( -t_1 - t_y - t_z \geq 0 \right) \medcap 
\left( -t_w - t_x \geq 0 \right) \nonumber\\
&\hspace{23 mm} \medcap 
\left( -t_z \geq 0 \right) \\
&\mathcal{U}_{\text{FD}}: a_1\,a_2 + a_1\,a_3 + a_2\,a_4 + a_3\,a_4
\end{align*}

\item$ {\cal C}_{3}$.
\begin{align*}
&\text{g-vectors: } \{g_{2a}^{0}, g_{1b}^{1}, g_{12}^{(b,1)}, g_{2b}^{1}, g_{1a}^{-1}\} \\
&\text{Inequalities: } 
\left( t_w \geq 0 \right) \medcap 
\left( t_1 + 2t_w + t_x + t_z \geq 0 \right) \medcap 
\left( t_w + t_x + t_y + t_z \geq 0 \right)  \nonumber\\
&\hspace{23 mm} \medcap 
\left( -t_1 - 2t_w - t_x - t_y - t_z \geq 0\right) \medcap 
\left( -t_w - t_x \geq 0 \right) \\
&\mathcal{U}: t_x\,t_y \\
&\mathcal{U}_{\text{FD}}: a_1\,a_2 + a_1\,a_4 + a_2\,a_5 + a_4\,a_5
\end{align*}

\item ${\cal C}_{4}$.
\begin{align*}
&\text{g-vectors: } \{g_{1b}^{0}, g_{1}^{1a}, g_{12}^{(a,1)}, g_{2a}^{1}, g_{2b}^{-1}\} \\
&\text{Inequalities: } 
\left( -t_y - t_z \geq 0 \right) \medcap 
\left( t_1 + t_w + t_y + 2t_z \geq 0 \right) \medcap 
\left( -t_w - t_z \geq 0 \right)\nonumber\\
&\hspace{23 mm} \medcap 
\left( -t_1 - t_w - t_x - t_y - 2t_z \geq 0 \right) \medcap 
\left( t_z \geq 0 \right) \\
&\mathcal{U}: t_x\,t_y \\
&\mathcal{U}_{\text{FD}}: a_1\,a_2 + a_1\,a_4 + a_2\,a_5 + a_4\,a_5
\end{align*}
\end{enumerate}	
\underline{Cones over which ${\cal U}\, =\, t_x\,t_y - t_z^2$.}	
\begin{enumerate}
\item ${\cal C}_{5}$.
\begin{align*}
&\text{g-vectors: } \{g_{1a}^{0}, g_{1b}^{0}, g_{1}^{1a}, g_{ab}, g_{11}^{ab}\} \\
&\text{Inequalities: } 
\left( t_w \geq 0 \right) \medcap 
\left( -t_y - t_z \geq 0 \right) \medcap 
\left( -t_w - t_x - t_z \geq 0 \right) \medcap 
\left( t_z \geq 0 \right)\nonumber\\
&\hspace{23 mm} \medcap 
\left( t_1 + t_w + t_x + t_y + 2t_z \geq 0 \right) \\
&\mathcal{U}_{\text{FD}}: a_1\,a_2 + a_2\,a_3 + a_1\,a_4 + a_2\,a_4 + a_3\,a_4
\end{align*}

\item ${\cal C}_{6}$.
\begin{align*}
&\text{g-vectors: } \{g_{2a}^{0}, g_{2b}^{0}, g_{2a}^{1}, g_{ab}, g_{22}^{ab}\} \\
&\text{Inequalities: } 
\left( t_w \geq 0 \right) \medcap 
\left( -t_y - t_z \geq 0 \right) \medcap 
\left( -t_w - t_x - t_z \geq 0 \right) \medcap 
\left( t_z \geq 0 \right) \nonumber\\
&\hspace{23 mm} \medcap 
\left( -t_1 - t_w \geq 0 \right) \\
&\mathcal{U}_{\text{FD}}: a_1\,a_2 + a_2\,a_3 + a_1\,a_4 + a_2\,a_4 + a_3\,a_4
\end{align*}

\item ${\cal C}_{7}$.
\begin{align*}
&\text{g-vectors: } \{g_{1a}^{0}, g_{2a}^{0}, g_{2b}^{0}, g_{2a}^{1}, g_{ab}\} \\
&\text{Inequalities: } 
\left( t_1 + t_w \geq 0 \right) \medcap 
\left( -t_1 \geq 0 \right) \medcap 
\left( -t_y - t_z \geq 0 \right) \medcap 
\left( -t_w - t_x - t_z \geq 0 \right) \nonumber\\
&\hspace{23 mm} \medcap 
\left( t_z \geq 0 \right) \\
&\mathcal{U}: t_x\,t_y - t_z^2 \\
&\mathcal{U}_{\text{FD}}: a_1\,a_3 + a_2\,a_3 + a_3\,a_4 + a_1\,a_5 + a_2\,a_5 + a_3\,a_5 + a_4\,a_5
\end{align*}

\item ${\cal C}_{8}$.
\begin{align*}
&\text{g-vectors: } \{g_{1a}^{0}, g_{1b}^{0}, g_{2b}^{0}, g_{2a}^{1}, g_{ab}\} \\
&\text{Inequalities: } 
\left( t_w \geq 0 \right) \medcap 
\left( t_1 \geq 0 \right) \medcap 
\left( -t_1 - t_y - t_z \geq 0 \right) \medcap 
\left( -t_w - t_x - t_z \geq 0 \right)\nonumber\\
&\hspace{23 mm} \medcap 
\left( t_z \geq 0 \right) \\
&\mathcal{U}_{\text{FD}}: a_1\,a_2 + a_1\,a_3 + a_2\,a_4 + a_3\,a_4 + a_1\,a_5 + a_2\,a_5 + a_3\,a_5 + a_4\,a_5
\end{align*}

\item ${\cal C}_{9}$.
\begin{align*}
&\text{g-vectors: } \{g_{1a}^{0}, g_{1b}^{0}, g_{1a}^{1}, g_{2a}^{1}, g_{ab}\} \\
&\text{Inequalities: } 
\left( t_w \geq 0 \right) \medcap 
\left( -t_y - t_z \geq 0 \right) \medcap 
\left( t_1 + t_y + t_z \geq 0 \right) \nonumber\\
&\hspace{23 mm} \medcap 
\left( -t_1 - t_w - t_x - t_y - 2t_z \geq 0 \right) \medcap 
\left( t_z \geq 0 \right) \\
&\mathcal{U}_{\text{FD}}: a_1\,a_2 + a_2\,a_3 + a_2\,a_4 + a_1\,a_5 + a_2\,a_5 + a_3\,a_5 + a_4\,a_5
\end{align*}
\end{enumerate}
\underline{Cones over which ${\cal U}\, =\, t_{x} t_{y} - ( t_{w} + t_{x} )^{2}$.}
\begin{enumerate}
\item ${\cal C}_{10}$.
\begin{align*}
&\text{g-vectors: } \{g_{1a}^{0}, g_{1b}^{0}, g_{1b}^{1}, g_{ab}, g_{11}^{ab}\} \\
&\text{Inequalities: } 
\left( t_w \geq 0 \right) \medcap 
\left( -t_y - t_z \geq 0 \right) \medcap 
\left( t_w + t_x + t_z \geq 0 \right) \medcap 
\left( -t_w - t_x \geq 0 \right)\nonumber\\
&\hspace{23 mm} \medcap 
\left( t_1 + t_y + t_z \geq 0 \right) \\
&\mathcal{U}_{\text{FD}}: a_1 a_2 + a_1 a_3 + a_1 a_4 + a_2 a_4 + a_3 a_4
\end{align*}

\item ${\cal C}_{11}$.
\begin{align*}
&\text{g-vectors: } \{g_{2a}^{0}, g_{2b}^{0}, g_{2b}^{1}, g_{ab}, g_{22}^{ab}\} \\
&\text{Inequalities: } 
\left( t_w \geq 0 \right) \medcap 
\left( -t_y - t_z \geq 0 \right) \medcap 
\left( t_w + t_x + t_z \geq 0 \right) \medcap 
\left( -t_w - t_x \geq 0 \right)\\
&\hspace{23mm} \medcap 
\left( -t_1 - 2t_w - t_x - t_z \geq 0 \right) \\
&\mathcal{U}_{\text{FD}}: a_1 a_2 + a_1 a_3 + a_1 a_4 + a_2 a_4 + a_3 a_4
\end{align*}

\item ${\cal C}_{12}$.
\begin{align*}
&\text{g-vectors: } \{g_{1a}^{0}, g_{2a}^{0}, g_{2b}^{0}, g_{1b}^{1}, g_{ab}\} \\
&\text{Inequalities: } 
\left( t_1 + t_w \geq 0 \right) \medcap 
\left( -t_1 \geq 0 \right) \medcap 
\left( -t_y - t_z \geq 0 \right) \medcap 
\left( t_w + t_x + t_z \geq 0 \right) \nonumber\\
&\hspace{23 mm} \medcap 
\left( -t_w - t_x \geq 0 \right) \\
&\mathcal{U}_{\text{FD}}: a_1 a_3 + a_2 a_3 + a_1 a_4 + a_2 a_4 + a_1 a_5 + a_2 a_5 + a_3 a_5 + a_4 a_5
\end{align*}

\item ${\cal C}_{13}$. 
\begin{align*}
&\text{g-vectors: } \{g_{1a}^{0}, g_{1b}^{0}, g_{2b}^{0}, g_{1b}^{1}, g_{ab}\} \\
&\text{Inequalities: } 
\left( t_w \geq 0 \right) \medcap 
\left( t_1 \geq 0 \right) \medcap 
\left( -t_1 - t_y - t_z \geq 0 \right) \medcap 
\left( t_w + t_x + t_z \geq 0 \right) \nonumber\\
&\hspace{23 mm} \medcap 
\left( -t_w - t_x \geq 0 \right) \\
&\mathcal{U}_{\text{FD}}: a_1 a_2 + a_1 a_3 + a_1 a_4 + a_1 a_5 + a_2 a_5 + a_3 a_5 + a_4 a_5
\end{align*}

\item ${\cal C}_{14}$. 
\begin{align*}
&\text{g-vectors: } \{g_{2a}^{0}, g_{2b}^{0}, g_{1b}^{1}, g_{2b}^{1}, g_{ab}\} \\
&\text{Inequalities: } 
\left( t_w \geq 0 \right) \medcap 
\left( -t_y - t_z \geq 0 \right) \medcap 
\left( t_1 + 2t_w + t_x + t_z \geq 0 \right) \medcap 
\left( -t_1 - t_w \geq 0 \right)\nonumber\\
&\hspace{23 mm} \medcap 
\left( -t_w - t_x \geq 0 \right) \\
&\mathcal{U}_{\text{FD}}: a_1 a_2 + a_1 a_3 + a_1 a_4 + a_1 a_5 + a_2 a_5 + a_3 a_5 + a_4 a_5
\end{align*}
\end{enumerate}

\underline{ Cones over which ${\cal U}\, =\, (t_w + t_x + t_z)(t_w + t_y + t_z) - (t_w + t_x + t_y + t_z)(2t_w + t_x + t_y + 2t_z)$.}
\begin{enumerate}
\item ${\cal C}_{15}$.
\begin{align*}
&\text{g-vectors: } \{g_{1a}^{0}, g_{1b}^{1}, g_{ab}, g_{11}^{ab}, g_{1a}^{-1}\} \\
&\text{Inequalities: } 
\left( t_w \geq 0 \right) \medcap 
\left( t_w + t_x + t_z \geq 0 \right) \medcap 
\left( -t_w - t_x - t_y - t_z \geq 0 \right) \medcap 
\left( t_1 \geq 0 \right)\nonumber\\
&\hspace{23 mm} \medcap 
\left( t_y + t_z \geq 0 \right) \\
&\mathcal{U}_{\text{FD}}: a_1 a_2 + a_1 a_3 + a_2 a_3 + a_2 a_5 + a_3 a_5
\end{align*}

\item ${\cal C}_{16}$.
\begin{align*}
&\text{g-vectors: } \{g_{2a}^{0}, g_{2b}^{1}, g_{ab}, g_{22}^{ab}, g_{2a}^{-1}\} \\
&\text{Inequalities: } 
\left( t_w \geq 0 \right) \medcap 
\left( t_w + t_x + t_z \geq 0 \right) \medcap 
\left( -t_w - t_x - t_y - t_z \geq 0 \right)\nonumber\\
&\hspace{23 mm} \medcap 
\left( -t_1 - 2t_w - t_x - t_y - 2t_z \geq 0 \right) \medcap 
\left( t_y + t_z \geq 0 \right) \\
&\mathcal{U}_{\text{FD}}: a_1 a_2 + a_1 a_3 + a_2 a_3 + a_2 a_5 + a_3 a_5
\end{align*}

\item ${\cal C}_{17}$. 
\begin{align*}
&\text{g-vectors: } \{g_{1a}^{0}, g_{2a}^{0}, g_{1b}^{1}, g_{ab}, g_{1a}^{-1}\} \\
&\text{Inequalities: } 
\left( t_1 + t_w \geq 0 \right) \medcap 
\left( -t_1 \geq 0 \right) \medcap 
\left( t_w + t_x + t_z \geq 0 \right) \medcap 
\left( -t_w - t_x - t_y - t_z \geq 0 \right) \nonumber\\
&\hspace{23 mm} \medcap 
\left( t_y + t_z \geq 0 \right) \\
&\mathcal{U}_{\text{FD}}: a_1 a_3 + a_2 a_3 + a_1 a_4 + a_2 a_4 + a_3 a_4 + a_3 a_5 + a_4 a_5
\end{align*}

\item ${\cal C}_{18}$.
\begin{align*}
&\text{g-vectors: } \{g_{2a}^{0}, g_{1b}^{1}, g_{2b}^{1}, g_{ab}, g_{1a}^{-1}\} \\
&\text{Inequalities: } 
\left( t_w \geq 0 \right) \medcap 
\left( t_1 + 2t_w + t_x + t_z \geq 0 \right) \medcap 
\left( -t_1 - t_w \geq 0 \right)\nonumber\\
&\hspace{23 mm} \medcap 
\left( -t_w - t_x - t_y - t_z \geq 0 \right) \medcap 
\left( t_y + t_z \geq 0 \right) \\
&\mathcal{U}_{\text{FD}}: a_1 a_2 + a_1 a_3 + a_1 a_4 + a_2 a_4 + a_3 a_4 + a_2 a_5 + a_3 a_5 + a_4 a_5
\end{align*}

\item ${\cal C}_{19}$.
\begin{align*}
&\text{g-vectors: } \{g_{2a}^{0}, g_{2b}^{1}, g_{ab}, g_{1a}^{-1}, g_{2a}^{-1}\} \\
&\text{Inequalities: } 
\left( t_w \geq 0 \right)\medcap 
\left( t_w + t_x + t_z \geq 0 \right) \medcap 
\left( -t_w - t_x - t_y - t_z \geq 0 \right)\nonumber\\
&\hspace{23 mm} \medcap 
\left( t_1 + 2t_w + t_x + t_y + 2t_z \geq 0 \right) \medcap 
\left( -t_1 - 2t_w - t_x - t_z \geq 0 \right) \\
&\mathcal{U}_{\text{FD}}: a_1 a_2 + a_1 a_3 + a_2 a_3 + a_2 a_4 + a_3 a_4 + a_2 a_5 + a_3 a_5
\end{align*}
\end{enumerate}

\underline{Cones over which ${\cal U}\, =\, t_x t_y - (t_w + t_z)^2$.}
\begin{enumerate}
\item ${\cal C}_{20}$.
\begin{align*}
&\text{g-vectors: } \{g_{1b}^{0}, g_{1a}^{1}, g_{ab}, g_{11}^{ab}, g_{1b}^{-1}\} \\
&\text{Inequalities: } 
\left( -t_y - t_z \geq 0 \right) \medcap 
\left( -t_w - t_x - t_z \geq 0 \right) \medcap 
\left( t_w + t_z \geq 0 \right) \nonumber\\
&\hspace{23 mm} \medcap 
\left( t_1 + 2t_w + t_x + t_y + 2t_z \geq 0 \right) \medcap 
\left( -t_w \geq 0 \right) \\
&\mathcal{U}_{\text{FD}}: a_1 a_2 + a_1 a_3 + a_2 a_3 + a_2 a_5 + a_3 a_5
\end{align*}

\item ${\cal C}_{21}$. 
\begin{align*}
&\text{g-vectors: } \{g_{2b}^{0}, g_{2a}^{1}, g_{ab}, g_{22}^{ab}, g_{2b}^{-1}\} \\
&\text{Inequalities: } 
\left( -t_y - t_z \geq 0 \right) \medcap 
\left( -t_w - t_x - t_z \geq 0 \right) \medcap 
\left( t_w + t_z \geq 0 \right) \medcap 
\left( -t_1 \geq 0 \right)\nonumber\\
&\hspace{23 mm} \medcap 
\left( -t_w \geq 0 \right) \\
&\mathcal{U}_{\text{FD}}: a_1 a_2 + a_1 a_3 + a_2 a_3 + a_2 a_5 + a_3 a_5
\end{align*}

\item ${\cal C}_{22}$.
\begin{align*}
&\text{g-vectors: } \{g_{1b}^{0}, g_{2b}^{0}, g_{2a}^{1}, g_{ab}, g_{2b}^{-1}\} \\
&\text{Inequalities: } 
\left( t_1 \geq 0 \right) \medcap 
\left( -t_1 - t_y - t_z \geq 0 \right) \medcap 
\left( -t_w - t_x - t_z \geq 0 \right) \medcap 
\left( t_w + t_z \geq 0 \right)\nonumber\\
&\hspace{23 mm} \medcap 
\left( -t_w \geq 0 \right) \\
&\mathcal{U}_{\text{FD}}: a_1 a_3 + a_2 a_3 + a_1 a_4 + a_2 a_4 + a_3 a_4 + a_3 a_5 + a_4 a_5
\end{align*}

\item ${\cal C}_{23}$.
\begin{align*}
&\text{g-vectors: } \{g_{1b}^{0}, g_{1a}^{1}, g_{2a}^{1}, g_{ab}, g_{2b}^{-1}\} \\
&\text{Inequalities: } 
\left( -t_y - t_z \geq 0 \right) \medcap  
\left( t_1 + t_y + t_z \geq 0 \right) \medcap 
\left( -t_1 - t_w - t_x - t_y - 2t_z \geq 0 \right) \nonumber\\
&\hspace{23 mm} \medcap 
\left( t_w + t_z \geq 0 \right) \medcap 
\left( -t_w \geq 0 \right) \\
&\mathcal{U}_{\text{FD}}: a_1 a_2 + a_1 a_3 + a_1 a_4 + a_2 a_4 + a_3 a_4 + a_2 a_5 + a_3 a_5 + a_4 a_5
\end{align*}

\item ${\cal C}_{24}$. 
\begin{align*}
&\text{g-vectors: } \{g_{1b}^{0}, g_{1a}^{1}, g_{ab}, g_{1b}^{-1}, g_{2b}^{-1}\} \\
&\text{Inequalities: } 
\left( -t_y - t_z \geq 0 \right) \medcap 
\left( -t_w - t_x - t_z \geq 0 \right) \medcap 
\left( t_w + t_z \geq 0 \right) \nonumber\\
&\hspace{23 mm} \medcap 
\left( t_1 + t_w + t_x + t_y + 2t_z \geq 0 \right) \medcap 
\left( -t_1 - 2t_w - t_x - t_y - 2t_z \geq 0 \right) \\
&\mathcal{U}_{\text{FD}}: a_1 a_2 + a_1 a_3 + a_2 a_3 + a_2 a_4 + a_3 a_4 + a_2 a_5 + a_3 a_5
\end{align*}
\end{enumerate}

%%%%%%%%%%%%%%%%%%%%%%%%%%%%%%%%
\subsection{ Derivation of ${\cal R}_{1}, \dots, {\cal R}_{5}$ in the case of $\Sigma_{2,2}$ }\label{sec:Decomposition_of_U}
Now that we have computed all the cones, we group together those cones that give rise to the same form of the first Symanzik polynomial $\mathcal{U}$, and take the union of their corresponding regions.  

As a first case, consider all cones in which $\mathcal{U}$ reduces to the form  
\begin{align}
\mathcal{U}\big|_{\mathcal{R}_1}=	t_x t_y. 
\end{align}
The region $\mathcal{R}_1$ is the union of four regions, corresponding to the four $\log^2$-divergent cones $\mathcal{C}_1-\mathcal{C}_4$:  
%\begin{align}
%\mathcal{R}_1 &= (t_y < 0) \, \medcap  \, (t_z < 0) \, \medcap \, (t_w > 0) \, \medcap \, (t_x < -t_w) 
%\, \medcap \, (-t_z < t_1 < -t_y - t_z)\nonumber \\
%&\, \medcup (t_y < 0) \, \medcap \, (0 < t_z < -t_y) \, \medcap \, (t_w < -t_z) \ \medcap \ (t_x < 0) \nonumber \\ 
%&\  \medcap \, (-t_w - t_y - 2t_z < t_1 < -t_w - t_x - t_y - 2t_z)
%\ \medcup (t_y < 0) \ \medcap \ (0 < t_z < -t_y)  \nonumber\\
%&\  \medcap \, 
%\bigl((t_w > 0) \ \medcap \ (-t_w < t_x < 0) \, \medcap \, (-t_w < t_1 < -t_w - t_x)\bigr) \nonumber\\
%&\ \medcup (t_y < 0) \, \medcap \, (t_z > -t_y)  
%\, \medcap \, (t_w > 0) \, \medcap \, 
%(-t_w - t_y - t_z < t_x < -t_w)\nonumber\\
%&\hspace*{1.6in}\medcap \, (-2t_w - t_x - t_z < t_1 < -2t_w - t_x - t_y - t_z).
%\end{align}
\begin{align}
	\mathcal{R}_1 &= \left\{(t_y < 0) \ \cap\ (t_z < 0) \ \cap\ (t_w > 0) \ \cap\ (t_x < -t_w)
	\ \cap\ (-t_z < t_1 < -t_y - t_z)\right\}\nonumber \\
	&\ \cup\left\{ (t_y < 0) \ \cap\ (0 < t_z < -t_y) \ \cap\ (t_w < -t_z) \ \cap\ (t_x < 0)\right.\nonumber\\
	&\left.\hspace*{1.6in}
	\ \cap\ (-t_w - t_y - 2t_z < t_1 < -t_w - t_x - t_y - 2t_z)\right\} \nonumber\\
	&\ \cup \left\{(t_y < 0) \ \cap\ (0 < t_z < -t_y) \ \cap\
	\bigl((t_w > 0) \ \cap\ (-t_w < t_x < 0) \ \cap\ (-t_w < t_1 < -t_w - t_x)\bigr)\right\} \nonumber\\
	&\ \cup \left\{(t_y < 0) \ \cap\ (t_z > -t_y) \ \cap\ (t_w > 0) \ \cap\
	(-t_w - t_y - t_z < t_x < -t_w)\right.\nonumber\\
	&\left.\hspace*{1.6in}\cap\ (-2t_w - t_x - t_z < t_1 < -2t_w - t_x - t_y - t_z)\right\}.
\end{align}

Next, consider the case where the first Symanzik polynomial takes the form  
\begin{align}
\mathcal{U}\big|_{\mathcal{R}_2}=	t_x t_y - t_z^2. 
\end{align} 
The region $\mathcal{R}_2$ consists of four $\log$-divergent cones $\mathcal{C}_5,\,\mathcal{C}_6,\,\mathcal{C}_7,\,\mathcal{C}_9$ and one convergent cone $\mathcal{C}_8$. The resulting region simplifies to  
\begin{align}
\mathcal{R}_2 
&= \Bigl\{\, (0 < t_{z} < - t_{y}) \;\medcap\; (t_{w} > 0) \;\medcap\; (t_{x} < - t_{w} - t_{z}) \,\Bigr\}.
\end{align}

Now consider the case where $\mathcal{U}$ takes the form  
\begin{align}
\mathcal{U}\big|_{\mathcal{R}_3}=- t_w^2 - 2 t_w t_x + t_x (t_y - t_x). 
\end{align}
The region $\mathcal{R}_3$ consists of four $\log$-divergent cones $\mathcal{C}_{10},\,\mathcal{C}_{11},\,\mathcal{C}_{13},\,\mathcal{C}_{14}$ and one convergent cone $\mathcal{C}_{12}$. It simplifies to  
\begin{align}
\mathcal{R}_3 
&= \Bigl\{\, (0 < t_{z} < -t_{y}) \;\medcap\; (-t_{w} - t_{z} < t_{x} < -t_{w}) \;\medcap\; (t_{w} > 0)\,\Bigr\}.
\end{align}

Next, consider the form  
\begin{align}
\mathcal{U}\big|_{\mathcal{R}_4}=(t_w + t_x + t_z)(t_w + t_y + t_z) 
- (t_w + t_x + t_y + t_z)(2 t_w + t_x + t_y + 2 t_z).
\end{align}
The region $\mathcal{R}_4$ consists of four $\log$-divergent cones $\mathcal{C}_{15},\,\mathcal{C}_{16},\,\mathcal{C}_{17},\,\mathcal{C}_{19}$ and one convergent cone $\mathcal{C}_{18}$. The simplified description is  
\begin{align}
\mathcal{R}_4 
&= \Bigl\{\, (t_z > 0) \;\medcap\; (t_w > 0) \;\medcap\; (-t_w - t_z < t_x < -t_w) \;\medcap\; (-t_z < t_y < -t_w - t_x - t_z) \,\Bigr\}.
\end{align}

Finally, consider the form  
\begin{align}
\mathcal{U}\big|_{\mathcal{R}_5}=- t_w^2 - 2 t_w t_z + t_x t_y - t_z^2.
\end{align}
The region $\mathcal{R}_5$ consists of four $\log$-divergent cones $\mathcal{C}_{20},\,\mathcal{C}_{21},\,\mathcal{C}_{22},\,\mathcal{C}_{24}$ and one convergent cone $\mathcal{C}_{23}$. The simplified form is  
\begin{align}
\mathcal{R}_5 
&= \Bigl\{\, (t_z > 0) \;\medcap\; (-t_z < t_w < 0) \;\medcap\; (t_x < -t_w - t_z) \;\medcap\; (t_y < -t_z) \,\Bigr\}.
\end{align}

In order to isolate the regions which lead to divergent contribution to the curve integrals, it will be convenient to use the following change of tropical parameters in each region. 
\begin{align}\label{r2tr4inc}
{\cal R}_{2}\,& =\, ( t_{x}, t_{y}, t_{w}, t_{z}, t_{1})\, \rightarrow\, ( s_{4,1}  =  \widetilde{t}_{w} + t_{z} < 0,\,  0 , t_{y} < 0,\, t_{w} > 0,\,  t_{y}  <  \widetilde{t}_{z} < 0 ,\,  t_{1}\, \in\, {\bf R}\, )\nonumber\\
{\cal R}_{3}\, &=\, (t_{x}, t_{y}, t_{w}, t_{z}, t_{1} )\, \rightarrow\, ( t_{y} < t_{x} < 0,\,  t_{y} < 0,\,  t_{x} < \widetilde{t}_{w} < 0,\,  (\, s_{4,1}\, >\, 0)\, \medcap\, ( \widetilde{t}_{w}\, <\, s_{4,1}\, <\, \widetilde{t}_{w} - t_{y} ),\, t_{1}\, \in\, {\bf R})\nonumber\\
&=\, (t_{x} < 0,\,  t_{y} < 0,\,  t_{x} < \widetilde{t}_{w} < 0,\, ( 0  <\, s_{4,1}\, <\, \widetilde{t}_{w} - t_{y} ),\, t_{1}\, \in\, {\bf R})\nonumber\\
{\cal R}_{4} &=  (\, t_{x}, t_{y}, t_{w}, t_{z}, t_{1}\, )\, \rightarrow\, (\, 0 <  s_{4,1} < - t_{y},  t_{y} <\, 0,\, 0 < t_{w} < - t_{x},\, \widetilde{t}_{z} > 0,\, t_{1} \in {\bf R}\, )\nonumber\\
{\cal R}_{5}\, &=\,  (\, t_{x}, t_{y}, t_{w}, t_{z}, t_{1}\, )\, \rightarrow\, (\, t_{x} < 0,\, t_{y} < 0,\, s_{4,1} - t_{y} < \widetilde{t}_{w} < t_{x},  t_{x} < s_{4,1} < 0,\,  t_{1} \in {\bf R})\nonumber\\
&=\, (\, t_{y} < t_{x} < 0,\, s_{4,1} - t_{y} < \widetilde{t}_{w} < t_{x},  t_{x} < s_{4,1} < 0,\,  t_{1} \in {\bf R})
\end{align}
where in the last equation we have used the fact that $t_{x} < - t_{z}  + \vert t_{w} \vert < 0$ and $t_{y} < - t_{z}$.
In eqn.(\ref{r2tr4inc}) the change of co-ordinates in each region is defined in the same order inside a set as the original co-ordinates. For e.g, in ${\cal R}_{2}$, 
\begin{align}
t_{x}, t_{y}, t_{w}, t_{z}\, \rightarrow\, s_{4,1}, t_{y}, t_{w}, \widetilde{t}_{z}
\end{align}
is the invertible co-ordinate transformation. The ordering of parameters on the right hand side informs us about the precise (linear) diffeomorphism that we use in that specific region.\\
We note that the limit $t_{x} \rightarrow\, 0, t_{w} + t_{x} \rightarrow\, 0$ is not a zero of ${\cal U}\vert_{{\cal R}_{4}}$. As a result the only zero of ${\cal U}\vert_{{\cal R}_{4}}$ is when 
$t_{y}\, \rightarrow\, 0$ as that leads to $s_{4,1}\, \rightarrow\, 0$ and in this limit ${\cal U}\vert_{{\cal R}_{4}}\, \sim O(t_{y})$. 

%%%%%%%%%%%%%%%%%%%%%%%%%%%%%%%%
\subsection{On overall normalization of the two-loop two-point amplitude}\label{ricones}
There are only \emph{six} Feynman diagrams for the two-loop two-point case, but  $24$ cones. This means that each Feynman diagram must correspond to \emph{four} cones, which is exactly what we observe in section~\ref{sec:Cones_for_2-point_2-loop}. For instance, in the case of the $\log^2$ Feynman diagram, the four cones are
\begin{align}
&\mathcal{C}^1=\{\, g^{0}_{1a},\; g^{0}_{2a},\; g^{0}_{2b},\; g^{1}_{1b},\; g^{a}_{12} \,\} \\[6pt]
&\mathcal{C}^2=\{\, g^{0}_{1a},\; g^{0}_{1b},\; g^{0}_{2b},\; g^{1}_{2a},\; g^{b}_{12} \,\} \\[6pt]
&\mathcal{C}^3=\{\,g^{1}_{2b},\; g^{1}_{1b},\; g^{-1}_{1a},\; g^{0}_{2a},\;  g^{(b,1)}_{12}\,\} \\[6pt]
&\mathcal{C}^4=\{\,g^{1}_{2a},\; g^{1}_{1a},\; g^{-1}_{2b},\;  g^{0}_{1b},\; g^{(a,1)}_{12} \,\}
\end{align}
The curve integral can be mapped to these cones through
\begin{align*}
\vec{X} = a_i\,\vec{g}_i, 
\qquad a_i > 0,
\end{align*}
where $\vec{X}$ is expressed in the basis 
$(e_{X},\, e_{W},\, e_{Z},\, e_{Y},\, e_{E_1})$, with 
$\{X, W, Z, Y, E_1\}$ denoting the Schwinger coordinates.  
The cones that corresponds to the $\log^2$ graphs are $\{\mathcal{C}_1,\mathcal{C}_2,\mathcal{C}_3,\mathcal{C}_4\}$ these cones reduce to the 
\begin{align}
	&I_{\mathcal{C}_1}=\int_{\mathbb{R}^{5}}\prod_{j=1}^5 da^{(1)}_{j}\;
	\frac{a^{(1)}_{123}}{a^{(1)}_{1234}}\,
	\frac{1}{\mathcal{U}^{(1)}}\,
	\exp\!\left[\,\frac{\mathcal{F}^{(1)}}{\mathcal{U}^{(1)}}-\mathcal{Z}^{(1)}\right]
	\\[1em]
	&I_{\mathcal{C}_2}=\int_{\mathbb{R}^{5}} \prod_{j=1}^5da^{(2)}_{j}\;
	\frac{a^{(2)}_{123}}{a^{(2)}_{1234}}\,
	\frac{1}{\mathcal{U}^{(2)}}\,
	\exp\!\left[\,\frac{\mathcal{F}^{(2)}}{\mathcal{U}^{(2)}}-\mathcal{Z}^{(2)}\right]
	\\[1em]
	&I_{\mathcal{C}_3}=\int_{\mathbb{R}^{5}}\prod_{j=1}^5 da^{(3)}_{j}\;
	\frac{a^{(3)}_{4}}{a^{(3)}_{1234}}\,
	\frac{1}{\mathcal{U}^{(3)}}\,
	\exp\!\left[\,\frac{\mathcal{F}^{(3)}}{\mathcal{U}^{(3)}}-\mathcal{Z}^{(3)}\right]
	\\[1em]
	&I_{\mathcal{C}_4}=\int_{\mathbb{R}^{5}}\prod_{j=1}^5 da^{(4)}_{j}\;
	\frac{a^{(4)}_{4}}{a^{(4)}_{1234}}\,
	\frac{1}{\mathcal{U}^{(4)}}\,
	\exp\!\left[\,\frac{\mathcal{F}^{(4)}}{\mathcal{U}^{(4)}}-\mathcal{Z}^{(4)}\right]
\end{align}

\noindent
where for each $i=1,2,3,4$,
\begin{align}
	\frac{\mathcal{F}^{(i)}}{\mathcal{U}^{(i)}} - \mathcal{Z}^{(i)} &=
	- p_{1}^{2}\,\frac{a^{(i)}_{1}a^{(i)}_{2}}{a^{(i)}_{12}}
	- p_{1}^{2}\,\frac{a^{(i)}_{3}a^{(i)}_{4}}{a^{(i)}_{34}}
	- m^{2}\,a^{(i)}_{1234}
	- (p_{1}^{2}+m^{2})\,a^{(i)}_{5}\, ,\\
	\mathcal{U}^{(i)}&=a^{(i)}_{12}a^{(i)}_{34}\, .
\end{align}
are Symanzik polynomials for each cone and $a_{i\dots j}$ is short hand for $a_i+\dots+a_j$. Now, if we sum over all four $\log^2$ cones, identifying $a^{(i)}_j=a_j \,\forall\, i$  , we obtain
\begin{align}
I_{\mathcal{C}_1}+I_{\mathcal{C}_2}+I_{\mathcal{C}_3}+I_{\mathcal{C}_4}=2\int_{\mathbb{R}^{5}}\prod_{j=1}^5 da_{j}\;
\frac{1}{\mathcal{U}}\,
\exp\!\left[\,\frac{\mathcal{F}}{\mathcal{U}}-\mathcal{Z}
\right]\, ,
\label{eq:2l2p_log^2_para_form}
\end{align}
where 
\begin{align}
	\frac{\mathcal{F}}{\mathcal{U}} - \mathcal{Z} &=
	- p_{1}^{2}\,\frac{a_{1}a_{2}}{a_{12}}
	- p_{1}^{2}\,\frac{a_{3}a_{4}}{a_{34}}
	- m^{2}\,a_{1234}
	- (p_{1}^{2}+m^{2})\,a_{5},\\
	\mathcal{U}&=a_{12} \, a_{34} \, . 
\end{align}
are the Symanzik variables for a two-loop $\log^2$ graph in $4$ dimensions. Note that there a factor of \emph{two} in equation~\ref{eq:2l2p_log^2_para_form}. Hence, the curve integral is in fact equal to \emph{twice} the amplitude, as explicitly verified for the $\log^2$ diagram.

%%%%%%%%%%%%%%%%%%%%%%%%%%%%
\bibliographystyle{bibstyle}
\bibliography{curve-int-renorm}

\end{document}